\documentclass[11pt]{article}
\usepackage[margin=.8in,left=.8in]{geometry}
\usepackage{amsmath}
\usepackage{amsfonts}
\usepackage{amssymb}
\usepackage{amsthm}
\usepackage{mathtools}
\usepackage{subcaption}
\usepackage{graphicx}
\usepackage{color}
\usepackage{xcolor}
\usepackage{xfrac}
\usepackage{stmaryrd}
\usepackage[utf8]{inputenc}
\usepackage{rotating}
\usepackage{hyperref}
\usepackage[all]{xy}

\usepackage{lmodern} 

\usepackage{tikz}
\usetikzlibrary{decorations.pathmorphing}
\usepackage{tikz-cd}
\usepackage{lipsum}
\usepackage{adjustbox}

\definecolor{plum}{rgb}{0.36078, 0.20784, 0.4}
\definecolor{chameleon}{rgb}{0.30588, 0.60392, 0.023529}
\definecolor{cornflower}{rgb}{0.12549, 0.29020, 0.52941}
\definecolor{scarlet}{rgb}{0.8, 0, 0}
\definecolor{brick}{rgb}{0.64314, 0, 0}
\definecolor{sunrise}{rgb}{0.80784, 0.36078, 0}
\definecolor{lightblue}{rgb}{0.15,0.35,0.75}
\definecolor{carolina}{RGB}{153, 186, 221}
\definecolor{darkblue}{rgb}{0.05,0.25,0.65}
\definecolor{greenii}{RGB}{20,140,10}

\usepackage{multirow}

\usepackage{stmaryrd}    

\usepackage{enumerate} 

\usepackage{helvet}   

\newcommand*{\Scale}[2][4]{\scalebox{#1}{$#2$}}%
\usepackage{mathptmx}
\usepackage{amsmath}
\usepackage{graphicx}

\usepackage{floatflt}  
\usepackage{array}     
\newcolumntype{L}[1]{>{\raggedright\let\newline\\\arraybackslash\hspace{0pt}}m{#1}}
\newcolumntype{C}[1]{>{\centering\let\newline\\\arraybackslash\hspace{0pt}}m{#1}}
\newcolumntype{R}[1]{>{\raggedleft\let\newline\\\arraybackslash\hspace{0pt}}m{#1}}

\usepackage{tabularx}   

\usepackage[new]{old-arrows}   

\newdir{> }{{}*!/10pt/@{>}}

\makeatletter
\newif\if@sup
\newtoks\@sups
\def\append@sup#1{\edef\act{\noexpand\@sups={\the\@sups #1}}\act}%
\def\reset@sup{\@supfalse\@sups={}}%
\def\mk@scripts#1#2{\if #2/ \if@sup ^{\the\@sups}\fi \else%
  \ifx #1_ \if@sup ^{\the\@sups}\reset@sup \fi {}_{#2}%
  \else \append@sup#2 \@suptrue \fi%
  \expandafter\mk@scripts\fi}
\def\tensor#1#2{\reset@sup#1\mk@scripts#2_/}
\def\multiscripts#1#2#3{\reset@sup{}\mk@scripts#1_/#2%
  \reset@sup\mk@scripts#3_/}
\makeatother

\makeatletter
\newbox\slashbox \setbox\slashbox=\hbox{$/$}
\def\itex@pslash#1{\setbox\@tempboxa=\hbox{$#1$}
  \@tempdima=0.5\wd\slashbox \advance\@tempdima 0.5\wd\@tempboxa
  \copy\slashbox \kern-\@tempdima \box\@tempboxa}
\def\slash{\protect\itex@pslash}
\makeatother

\def\clap#1{\hbox to 0pt{\hss#1\hss}}
\def\mathllap{\mathpalette\mathllapinternal}
\def\mathrlap{\mathpalette\mathrlapinternal}
\def\mathclap{\mathpalette\mathclapinternal}
\def\mathllapinternal#1#2{\llap{$\mathsurround=0pt#1{#2}$}}
\def\mathrlapinternal#1#2{\rlap{$\mathsurround=0pt#1{#2}$}}
\def\mathclapinternal#1#2{\clap{$\mathsurround=0pt#1{#2}$}}

\let\oldroot\root
\def\root#1#2{\oldroot #1 \of{#2}}
\renewcommand{\sqrt}[2][]{\oldroot #1 \of{#2}}

\DeclareSymbolFont{symbolsC}{U}{txsyc}{m}{n}
\SetSymbolFont{symbolsC}{bold}{U}{txsyc}{bx}{n}
\DeclareFontSubstitution{U}{txsyc}{m}{n}

\DeclareSymbolFont{stmry}{U}{stmry}{m}{n}
\SetSymbolFont{stmry}{bold}{U}{stmry}{b}{n}

\DeclareFontFamily{OMX}{MnSymbolE}{}
\DeclareSymbolFont{mnomx}{OMX}{MnSymbolE}{m}{n}
\SetSymbolFont{mnomx}{bold}{OMX}{MnSymbolE}{b}{n}
\DeclareFontShape{OMX}{MnSymbolE}{m}{n}{
    <-6>  MnSymbolE5
   <6-7>  MnSymbolE6
   <7-8>  MnSymbolE7
   <8-9>  MnSymbolE8
   <9-10> MnSymbolE9
  <10-12> MnSymbolE10
  <12->   MnSymbolE12}{}

\makeatletter
\def\Decl@Mn@Delim#1#2#3#4{%
  \if\relax\noexpand#1%
    \let#1\undefined
  \fi
  \DeclareMathDelimiter{#1}{#2}{#3}{#4}{#3}{#4}}
\def\Decl@Mn@Open#1#2#3{\Decl@Mn@Delim{#1}{\mathopen}{#2}{#3}}
\def\Decl@Mn@Close#1#2#3{\Decl@Mn@Delim{#1}{\mathclose}{#2}{#3}}
\Decl@Mn@Open{\llangle}{mnomx}{'164}
\Decl@Mn@Close{\rrangle}{mnomx}{'171}
\Decl@Mn@Open{\lmoustache}{mnomx}{'245}
\Decl@Mn@Close{\rmoustache}{mnomx}{'244}
\makeatother

\makeatletter
\DeclareRobustCommand\widecheck[1]{{\mathpalette\@widecheck{#1}}}
\def\@widecheck#1#2{%
    \setbox\z@\hbox{\m@th$#1#2$}%
    \setbox\tw@\hbox{\m@th$#1%
       \widehat{%
          \vrule\@width\z@\@height\ht\z@
          \vrule\@height\z@\@width\wd\z@}$}%
    \dp\tw@-\ht\z@
    \@tempdima\ht\z@ \advance\@tempdima2\ht\tw@ \divide\@tempdima\thr@@
    \setbox\tw@\hbox{%
       \raise\@tempdima\hbox{\scalebox{1}[-1]{\lower\@tempdima\box
\tw@}}}%
    {\ooalign{\box\tw@ \cr \box\z@}}}
\makeatother

\makeatletter
\def\udots{\mathinner{\mkern2mu\raise\p@\hbox{.}
\mkern2mu\raise4\p@\hbox{.}\mkern1mu
\raise7\p@\vbox{\kern7\p@\hbox{.}}\mkern1mu}}
\makeatother

\newcommand{\gt}{>}
\newcommand{\lt}{<}

\newcommand{\Z}{\ensuremath{\mathbb Z}}

\renewcommand{\(}{\begin{equation}}
\renewcommand{\)}{\end{equation}}
\newcommand{\bea}{\begin{eqnarray*}}
\newcommand{\eea}{\end{eqnarray*}}

\usepackage{cleveref}

\crefformat{section}{\S#2#1#3} 
\crefformat{subsection}{\S#2#1#3}
\crefformat{subsubsection}{\S#2#1#3}

\theoremstyle{italics}
\newtheorem{theorem}{Theorem}[section]

\newtheorem{prop}[theorem]{Proposition}
\newtheorem{cor}[theorem]{Corollary}
\newtheorem{folklore}[theorem]{Folklore}

\theoremstyle{definition}
\newtheorem{defn}[theorem]{Definition}

\newtheorem{example}[theorem]{Example}

\newtheorem{remark}[theorem]{Remark}
\newtheorem{note[theorem]}{Note}

\usepackage{amsfonts}

\begin{document}

\title{Equivariant Cohomotopy
 implies orientifold tadpole cancellation}

\author{ Hisham Sati, \; Urs Schreiber}

\maketitle

\begin{abstract}
There are fundamental open problems in the
precise global nature of RR-field tadpole cancellation
conditions in string theory.
Moreover, the non-perturbative lift as M5/MO5-anomaly cancellation in M-theory
had been based on indirect plausibility arguments,
lacking a microscopic underpinning in M-brane charge quantization.
We provide a framework for answering these
questions, crucial not only for mathematical consistency but also
for phenomenological accuracy of string theory,
by formulating the M-theory C-field on flat M-orientifolds in the
generalized cohomology theory called \emph{Equivariant Cohomotopy}.
This builds on our previous results for smooth but curved spacetimes, showing
in that setting that charge quantization in twisted Cohomotopy
rigorously implies a list of expected anomaly cancellation conditions.
Here we further expand this list by proving that
brane charge quantization in unstable equivariant Cohomotopy implies
the anomaly cancellation conditions for M-branes and D-branes on
flat orbi-orientifolds.
For this we
{\bf (a)} use an unstable refinement of
the equivariant Hopf-tom Dieck theorem to derive
local/twisted tadpole cancellation,
and in addition {\bf (b)} the lift to super-differential cohomology
to establish global/untwisted tadpole cancellation.
Throughout, we use
{\bf (c)} the unstable Pontrjagin-Thom theorem
to identify the brane/O-plane configurations
encoded in equivariant Cohomotopy
and {\bf (d)} the Boardman homomorphism
to equivariant K-theory to identify
Chan-Paton representations of D-brane charge.
We find  that unstable equivariant Cohomotopy,
but not its image in equivariant K-theory,
distinguishes D-brane charge
from the finite set of types of O-plane charges.
\end{abstract}

\tableofcontents

\newpage

\section{Introduction}
\label{Intro}

\vspace{-1mm}
Organizing and formalizing results in the string theory literature,
we start by noticing the following curious systematics, to be elaborated
upon throughout the paper.

\vspace{2mm}

\noindent {\bf Toroidal orientifolds with ADE-singularities -- A curious pattern.}
Consider type II superstring vacua compactified on
fluxless toroidal orientifolds
(\cite{Sagnotti88}\cite[p. 12]{DaiLeighPolchinski89}\cite{Mukhi97}\cite[3]{dBDHKMMS02};
see also \cite[5.3.4, 10.1.3]{IbanezUranga12}\cite[15.3]{BLT13})
with ADE-type singularities (\cite{AspinwallMorrison97}\cite{Intriligator97}; see \cite{ADE}),
hence on orbifold quotients $\mathbb{T}^{\mathbf{4}_{\mathbb{H}}} \!\sslash\! G$ (e.g. \cite[13]{Ratcliffe06})  of
4-tori by crystallographic point groups \eqref{CrystallographicGroups}.
These are finite subgroups
$G \subset \mathrm{SU}(2) \simeq \mathrm{Sp}(1)$
of the group unit quaternions acting by left multiplication on the
space $\mathbb{H} \simeq_{\mathbb{R}} \mathbb{R}^4$ of all quaternions \eqref{TheQuaternionicRepresentation}.

\medskip
The consistency condition on such compactifications
known as (Ramond-Ramond) \emph{RR-field tadpole anomaly cancellation}
(\cite[Sec. 3]{GimonPolchinski96}\cite[Sec. 9.3]{Witten12}; see \cite[4.4]{IbanezUranga12}\cite[9.4]{BLT13}),
essentially says that the joint D-brane and O-plane charge in such compact
orientifolds has to vanish, albeit with some subtle fine print.
Explicitly, we observe that
a case-by-case analysis of the string worldsheet
superconformal field theory shows (\hyperlink{Table1}{\it Table 1}) that,
for single wrapping number,
RR-field tadpole anomaly cancellation
is the following condition
on the $G$-representation of D-brane charge
and the $G$-set of O-planes:
\begin{enumerate}[{\bf (i)}]
  \vspace{-.2cm}
  \item
    {\bf Local/twisted tadpole cancellation:}
    D-brane charge is a combination of a regular representation $\mathbf{k}_{\mathrm{reg}}$
    and the trivial one  $\mathbf{1}_{\mathrm{triv}}$, with coefficients
     the number of integral and fractional branes, respectively:

         \vspace{-4mm}
              $$\mathbf{c}_{{}_{\mathrm{Dbra}}}
    = N_{\mathrm{brane} \atop \mathrm{integ}} \cdot \mathbf{k}_{\mathrm{reg}}
     + N_{\mathrm{brane} \atop \mathrm{frac}} \cdot \mathbf{1}_{\mathrm{triv}}\;.
     $$

      \vspace{-.4cm}
  \item {\bf Global/untwisted tadpole cancellation:}
    The dimension of D-brane charge is
    the cardinality of the $G$-set of O-planes:

    \vspace{-8mm}
    $$
     \mathrm{dim}(\mathbf{c}_{{}_{\mathrm{Dbra}}}) =
      \mathrm{card}( \mathbf{c}_{{}_{\mathrm{Opla}}}) \;.
      $$
\end{enumerate}

\vspace{-.1cm}

\noindent In particular, $\mathbf{c}_{{}_{\mathrm{Dbra}}}$ comes from,
and $\mathbf{c}_{{}_{\mathrm{Opla}}}$ gives rise to,
a permutation representation, in
the image of $\beta$ \eqref{BoardmanH} \cite{SS19b}:

\vspace{-7mm}
\begin{equation}
\hspace{-3mm}
  \xymatrix@R=-4pt@C=40pt{
    \mathbf{c}_{{}_{\mathrm{Dbra}}}
    \ar@{}[r]|{\in}
    &
    \mathrm{RO}(G)
    \ar@{}[r]|{\simeq}
    &
    \mathrm{KO}_G^0
    &
    \mathbb{S}_G^0
    \ar[l]_-{\beta}
    \ar@{}[r]|-{\simeq}
    &
    A(G)
    \ar@{<-}[r]
    &
    G \mathrm{Set}_{/\sim}
    \ar@{}[r]|-{\ni}
    &
    \mathbf{c}_{{}_{\mathrm{Opla}}}
    \\
    \mathclap{
    \mbox{
      \tiny
      \color{darkblue} \bf
      \begin{tabular}{c}
        D-brane charge
        \\
        on toroidal orientifold
      \end{tabular}
    }
    }
    &
    \mathclap{
    \mbox{
      \tiny
      \color{darkblue} \bf
      \begin{tabular}{c}
        representation
        \\
        ring
      \end{tabular}
    }
    }
    &
    \mathclap{
    \mbox{
      \tiny
      \color{darkblue} \bf
      \begin{tabular}{c}
        equivariant
        \\
        K-theory
      \end{tabular}
    }
    }
    &
    \mathclap{
    \mbox{
      \tiny
      \color{darkblue} \bf
      \begin{tabular}{c}
        equivariant
        \\
        stable Cohomotopy
      \end{tabular}
    }
    }
    &
    \mathclap{
    \mbox{
      \tiny
      \color{darkblue} \bf
      \begin{tabular}{c}
        Burnside
        \\
        ring
      \end{tabular}
    }
    }
    &
    \mathclap{
    \mbox{
      \tiny
      \color{darkblue} \bf
      \begin{tabular}{c}
        $G$-sets
        \\
        ($G$-permutations)
      \end{tabular}
    }
    }
    &
    \mathclap{
    \mbox{
      \tiny
      \color{darkblue} \bf
      \begin{tabular}{c}
        O-plane charge
        \\
        on toroidal orientifold
      \end{tabular}
    }
    }
  }
\end{equation}

\hypertarget{Table1}{}

\vspace{-1mm}
\hspace{-.7cm}
{\small
\begin{tabular}{|c|c|c|c|}
  \hline
  \begin{tabular}{c}
    \bf Single D-brane species
    \\
    \bf on toroidal orientifold
  \end{tabular}
  &
  \begin{tabular}{c}
    \bf Local/twisted
    \\
    \bf tadpole cancellation
    \\
    \bf condition
  \end{tabular}
  &
  \begin{tabular}{c}
    \bf Global/untwisted
    \\
    \bf tadpole cancellation
    \\
    \bf condition
  \end{tabular}
  &
  \bf Comments
  \\
  \hline
  \hline
  \begin{tabular}{l}
    Branes on
    \\
    $\mathbb{T}^{\mathbf{4}_{\mathbb{H}}} \!\sslash\! G^{\mathrm{ADE}}$
  \end{tabular}
  &
  \raisebox{17pt}{
  $
  \raisebox{-17pt}{
  $
    \begin{array}{rl}
      \mathbf{c}_{{}_{\mathrm{Dbra}}}
      =
      &
      \phantom{+} N_{\mathrm{brane} \atop \mathrm{int}}
        \cdot \mathbf{k}_{\mathrm{reg}}
      \\
      &
      + N_{\mathrm{brane} \atop \mathrm{frac}}
        \cdot \mathbf{1}_{\mathrm{triv}}
    \end{array}
  $
  }
  $
  }
  &
  $
  \begin{array}{l}
    \mathrm{dim}\big( \mathbf{c}_{\mathrm{tot}} \big)
    =
    \mathrm{card}\big( \mathbf{c}_{{}_{\mathrm{Opla}}} \big)
  \end{array}
  $
  &
  \begin{tabular}{l}
    The general pattern
    \\
    of the following
    \\
    case-by-case results
  \end{tabular}
  \\
  \hline
  \hline
  \begin{tabular}{l}
    D5/D9-branes
    \\
    on  $\mathbb{T}^{\mathbf{4}_{\mathbb{H}}} \!\sslash\! \mathbb{Z}_2$
  \end{tabular}
  &
  \begin{tabular}{l}
    $\mathbf{c}_{{}_{\mathrm{Dbra}}} = N \cdot \mathbf{2}_{\mathrm{reg}}$
    \\
    (\cite[(19)]{BST99})
  \end{tabular}
  &
  \begin{tabular}{l}
    $\mathbf{c}_{{}_{\mathrm{Dbra}}} = 16 \cdot \mathbf{2}_{\mathrm{reg}}$
    \\
    (\cite[(18)]{BST99})
  \end{tabular}
  &
  \multirow{2}{*}{
  \begin{tabular}{l}
    Following
    \\
    \cite{GimonPolchinski96}
    \\
    \cite{GC96}
  \end{tabular}
  }
  \\
  \cline{1-3}
  \begin{tabular}{l}
    D5/D9-branes
    \\
    on  $\mathbb{T}^{\mathbf{4}_{\mathbb{H}}} \!\sslash\! \mathbb{Z}_4$
  \end{tabular}
  &
  \begin{tabular}{l}
    $\mathbf{c}_{{}_{\mathrm{Dbra}}} = N \cdot \mathbf{4}_{\mathrm{reg}}$
    \\
    (\cite[(19)]{BST99})
  \end{tabular}
  &
  \begin{tabular}{l}
    $\mathbf{c}_{{}_{\mathrm{Dbra}}} = 8 \cdot \mathbf{4}_{\mathrm{reg}}$
    \\
    (\cite[(18)]{BST99})
  \end{tabular}
  &
  \\
  \hline
  \begin{tabular}{l}
    D4-branes
    \\
    on  $\mathbb{T}^{\mathbf{4}_{\mathbb{H}}} \!\sslash\! \mathbb{Z}_k$
  \end{tabular}
  &
  \begin{tabular}{l}
    $\mathbf{c}_{{}_{\mathrm{Dbra}}} = N \cdot \mathbf{k}_{\mathrm{reg}}$
    \\
    (\cite[4.2.1]{AFIRU00a})
  \end{tabular}
  &
  &
  \begin{tabular}{l}
    Re-derived via M5-branes
    \\
    below in \cref{M5MO5AnomalyCancellation}
  \end{tabular}
  \\
  \hline
  \begin{tabular}{l}
    D4-branes
    \\
    on  $\mathbb{T}^{\mathbf{4}_{\mathbb{H}}} \!\sslash\! \mathbb{Z}_3$
  \end{tabular}
  &
  \begin{tabular}{l}
    $\mathbf{c}_{{}_{\mathrm{Dbra}}} = N \cdot \mathbf{3}_{\mathrm{reg}}$
    \\
    \cite[(7.2)]{AFIRU00b}
  \end{tabular}
  &
  \begin{tabular}{l}
    $\mathbf{c}_{{}_{\mathrm{Dbra}}} = 4 \cdot \mathbf{3}_{\mathrm{reg}} + 4 \cdot \mathbf{1}_{\mathrm{triv}}$
    \\
    (\cite[(14)-(17)]{KataokaShimojo02},
    \\
  \end{tabular}
  &
  \begin{tabular}{l}
    The special case of $k = 3$
    \\
    (review in \cite[4]{Marchesano03})
  \end{tabular}
 \\
  \hline
  \begin{tabular}{l}
    D8-branes
    \\
    on  $\mathbb{T}^{{\mathbf{4}_{\mathbb{H}}}} \!\sslash\! \mathbb{Z}_3$
  \end{tabular}
  &
  \begin{tabular}{l}
    $\mathbf{c}_{{}_{\mathrm{Dbra}}} = N \cdot \mathbf{3}_{\mathrm{reg}}$
    \\
    \\
  \end{tabular}
  &
  \begin{tabular}{l}
    $\mathbf{c}_{{}_{\mathrm{Dbra}}} = 4 \cdot \mathbf{3}_{\mathrm{reg}} + 4 \cdot \mathbf{1}_{\mathrm{triv}}$
    \\
    (\cite[4]{Honecker01},
    \cite[(29)]{Honecker02})
  \end{tabular}
  &
  \begin{tabular}{l}
    Equivalent by T-duality
    \\
    to previous case
    \\
    (\cite[p.1 ]{Honecker01}, \cite[6]{Honecker02})
  \end{tabular}
  \\
  \hline
  \begin{tabular}{l}
    D3-branes
    \\
    on  $\mathbb{T}^{\mathbf{4}_{\mathbb{H  }}} \!\sslash\! \mathbb{Z}_k$
  \end{tabular}
  &
  \begin{tabular}{l}
    $\mathbf{c}_{{}_{\mathrm{Dbra}}} = N \cdot \mathbf{k}_{\mathrm{reg}} $
    \\
    (\cite[(25)]{FHKU01})
  \end{tabular}
  &
  &
  \\
  \hline
  \begin{tabular}{l}
    D7-branes
    \\
    on  $\mathbb{T}^{\mathbf{4}_{\mathbb{H  }}} \!\sslash\! \mathbb{Z}_k$
  \end{tabular}
  &
  \begin{tabular}{l}
    $\mathbf{c}_{{}_{\mathrm{Dbra}}} = N \cdot \mathbf{k}_{\mathrm{reg}} $
    \\
    (\cite[(5), (6)]{FHKU01})
  \end{tabular}
  &
  &
  \\
  \hline
  \begin{tabular}{l}
    D6-branes
    \\
    on  $\mathbb{T}^{6} \!\sslash\! \mathbb{Z}_4$
  \end{tabular}
  &
  &
  \begin{tabular}{l}
    $\mathbf{c}_{{}_{\mathrm{Dbra}}} = 8 \cdot \mathbf{k}_{\mathrm{reg}} $
    \\
    (\cite[(25)]{IKS99})
  \end{tabular}
  &
  \\
  \hline
\end{tabular}
}

\vspace{.15cm}

\noindent {\bf \footnotesize Table 1 -- Tadpole cancellation conditions
between D-branes and O-planes on toroidal ADE-orientifolds}
{\footnotesize as derived from case-by-case
analysis in
perturbative string theory.
The geometric content is shown
in \hyperlink{FigureA}{\it Figure A}.
The re-derivation from \hyperlink{HypothesisH}{\it Hypothesis H}
is in \cref{EquivariantCohomotopyChargeOfM5AtMO5}.}

\medskip
The D-brane species in \hyperlink{Table1}{\it Table 1}
with the most direct lift to M-theory are the D4-branes,
lifting to M5-branes under double dimensional reduction
(\cite[6]{APPS97a}\cite[6]{APPS97a}\cite{LPSS11});
see \hyperlink{Table7}{\it Table 7}.
With an actual formulation of M-theory lacking,
indirect plausibility arguments have been advanced
\cite{DasguptaMukhi95}\cite[3.3]{Witten95b}\cite[2.1]{Hori98}
that for M5-branes on M-theoretic orientifolds of the form
$\mathbb{T}^{ \mathbf{5}_{\mathrm{sgn}} }  \!\sslash\! \mathbb{Z}_2$,
anomaly cancellation implies \hyperlink{Table2}{\it Table 2}:

\vspace{-2mm}
{\small
\begin{center}
\hypertarget{Table2}{}
\begin{tabular}{|c|c|c|c|}
  \hline
  \begin{tabular}{c}
    \bf Single M-brane species
    \\
    \bf on toroidal orientifold
  \end{tabular}
  &
  \begin{tabular}{c}
    \bf Local/twisted
    \\
    \bf tadpole cancellation
    \\
    \bf condition
  \end{tabular}
  &
  \begin{tabular}{c}
    \bf Global/untwisted
    \\
    \bf tadpole cancellation
    \\
    \bf condition
  \end{tabular}
  &
  \bf Comments
  \\
  \hline
  \hline
  \multirow{2}{*}{
  \begin{tabular}{c}
    M5-branes
    \\
    on $\mathbb{T}^{\mathbf{5}_{\mathrm{sgn}}}
    \!\sslash\! \mathbb{Z}_2 $
  \end{tabular}
  }
  &
  \begin{tabular}{c}
    $
    \mathclap{\phantom{\vert^{\vert^{\vert}}}}
    \phantom{AAA}
    \mathbf{c}_{{}_{\mathrm{Mbra}}}
    = N \cdot \mathbf{2}_{\mathrm{reg}}
    \phantom{AAA}
    \mathclap{\phantom{\vert_{\vert_{\vert}}}}
    $
  \end{tabular}
  &
  \begin{tabular}{c}
    $
    \phantom{AAA}
    \mathbf{c}_{{}_{\mathrm{Mbra}}}
    = 16 \cdot \mathbf{2}_{\mathrm{reg}}
    \phantom{AAA}
    $
  \end{tabular}
  &
  \multirow{2}{*}{
    $\phantom{A}$
    \begin{tabular}{c}
      \phantom{a}
      \\
      plausibility arguments
    \end{tabular}
    $\phantom{A}$
  }
  \\
  &
  \multicolumn{2}{c|}{
    $\mathclap{\phantom{\vert^{\vert^{\vert}}}}$
    (\cite{DasguptaMukhi95} \cite[3.3]{Witten95b} \cite[2.1]{Hori98})
    $\mathclap{\phantom{\vert_{\vert_{\vert}}}}$
  }
  &
  \\
  \hline
\end{tabular}
\end{center}
}
\vspace{-.15cm}

\noindent {\bf \footnotesize Table 2 -- M5/MO5 anomaly cancellation in M-theory}
{\footnotesize according to Folklore \ref{AnomalyCancellationOnMTheoreticOrientifolds}.
While it has remained open in which cohomology theory the
charge $\mathbf{c}_{\mathrm{Mbra}}$ is quantized, the geometric
picture is again that illustrated in \hyperlink{FigureA}{\it Figure A}.}

\medskip

We highlight in \hyperlink{FigureA}{\it Figure A} the geometric interpretation of these
tadpole cancellation conditions from \hyperlink{Table1}{\it Table 1} and
\hyperlink{Table2}{\it Table 2}. The left side of  \hyperlink{FigureA}{\it Figure A}
shows a 2-dimensional slice through the toroidal orbifold
$\mathbb{T}^{\mathbf{4}_{\mathbb{H}}} \!\sslash\! \mathbb{Z}_4
= (\mathbb{R}^4 / \mathbb{Z}^4) \!\sslash\! \mathbb{Z}_4$
with transversal branes/O-plane charges appearing as points.
The O-plane charges (shown as open circles) are stuck one-to-one to the
fixed points of the point reflection subgroup $\mathbb{Z}_2 \hookrightarrow \mathbb{Z}_4$ (see also \hyperlink{TableRT}{\it Table RT})
and, in the example shown,
are permuted by the full orbifold group action of $\mathbb{Z}_4$ according to the permutation
representation $2 \cdot \mathbf{1}_{\mathrm{triv}} + 1 \cdot \mathbf{2}_{\mathrm{perm}}$.
The local/twisted tadpole cancellation condition says that the branes (shown as filled circles) appear
 in the vicinity of the O-planes with all their distinct mirror images under the full group action,
thus contributing Chan-Paton fields in the regular representation
$\mathbf{4}_{\mathrm{reg}}$. The global/untwisted tadpole cancellation
condition
says that the total charge of branes minus O-planes,
hence the net charge if all branes/O-planes could freely move
and pairwise annihilate, vanishes:

\begin{center}
\hypertarget{FigureA}{}
\begin{tikzpicture}[scale=0.8, decoration=snake]

  \begin{scope}[shift={(0,-.4)}]

  \node at (1.4,5.3)
    {$
      \overbrace{
        \phantom{------------------}
      }
    $};

  \node at (1.4,7)
    {
      \tiny
      \color{darkblue} \bf
      \begin{tabular}{c}
        local/twisted
        \\
        tadpole cancellation
      \end{tabular}
    };

  \node (EquivariantCocycle) at (1.4,5.3+.8)
    {\tiny $
      \mathllap{ \mathbf{c}_{\mathrm{tot}} = \; }
      4 \cdot
      \big(
      \,
      \overset{
        \mathbf{c}_{{}_{\mathrm{Dbra}}}
      }{
      \overbrace{
      1 \cdot \mathbf{4}_{\mathrm{reg}}
      }}
      -
      \overset{
        \beta\big( \mathbf{c}_{{}_{\mathrm{Opla}}}\big)
      }{
      \overbrace{
      (
        2 \cdot \mathbf{1}_{\mathrm{triv}}
        +
        1 \cdot \mathbf{2}_{\mathrm{perm}}
      )
      }
      }
      \;
      \big)
    $};

  \node at (1.4+8,7)
    {
      \tiny
      \color{darkblue} \bf
      \begin{tabular}{c}
        global/untwisted
        \\
        tadpole cancellation
      \end{tabular}
    };

  \node at (1.4+8,5.3)
    {$
      \overbrace{
        \phantom{------------------}
      }
    $};

  \node (PlainCocycle) at (1.4+8,5.3+.8)
    {\tiny
     \raisebox{-.6cm}{
     $
     \begin{aligned}
      &
      4 \cdot
      \big(
      \,
      1 \cdot 4
      -
      (
        2 \cdot 1 + 1 \cdot 2
      )
      \,
      \big)
      \\
      & = 0
      \end{aligned}
    $}};

  \draw[|->]
    (EquivariantCocycle)
    to node[above]{\tiny $\mathrm{dim}$ }
    (PlainCocycle);

  \end{scope}

  \begin{scope}
  \clip (-1.8,-1.8) rectangle (4.8,4.8);
  \draw[step=3, dotted] (-3,-3) grid (6,6);
  \draw[dashed] (-3,-3) circle (1);
  \draw[dashed] (0,-3) circle (1);
  \draw[dashed] (3,-3) circle (1);
  \draw[dashed] (6,-3) circle (1);
  \draw[dashed] (-3,0) circle (1);
  \draw[dashed] (0,0) circle (1);
  \draw[dashed] (3,0) circle (1);
  \draw[dashed] (-3,3) circle (1);
  \draw[dashed] (0,3) circle (1);
  \draw[dashed] (3,3) circle (1);
  \draw[dashed] (-3,6) circle (1);
  \draw[dashed] (0,6) circle (1);
  \draw[dashed] (3,6) circle (1);
  \draw[dashed] (6,6) circle (1);

  \draw[fill=white] (0,0) circle (.07);
  \draw[fill=white] (3,0) circle (.07);
  \draw[fill=white] (0,3) circle (.07);
  \draw[fill=white] (3,3) circle (.07);

  \draw[<->, dashed]
    (2.5,0)
    to[bend right=47]
    node
      {
        \colorbox{white}{
        \tiny
        \color{darkblue} \bf
        \begin{tabular}{c}
          orientifold
          \\
          action
        \end{tabular}
      }
      }
    (0,2.5);

  \draw (0,3)
    node[right]
      {
        \colorbox{white}{
        \hspace{-.3cm}
        \tiny
        \color{darkblue} \bf
        O-plane
        \hspace{-.3cm}
        }
      };
  \draw (3,0)
    node[right]
      {
        \colorbox{white}{
        \hspace{-.5cm}
        \tiny
        \color{darkblue} \bf
        \begin{tabular}{c}
          mirror
          \\
          O-plane
        \end{tabular}
        \hspace{-.3cm}
        }
      };

  \draw[fill=black] (17:.7) circle (.07);
  \draw[fill=black] (17+90:.7) circle (.07);
  \draw[fill=black] (17+180:.7) circle (.07);
  \draw[fill=black] (17+270:.7) circle (.07);

  \draw (17+90:.7)
    node[right]
      {
        \colorbox{white}{
        \hspace{-.3cm}
        \tiny
        \color{darkblue} \bf
        brane
        \hspace{-.3cm}
        }
      };
  \draw (17+180:.7)+(.58,.03)
    node[right, below]
      {
        {
        \hspace{-.3cm}
        \tiny
        \color{darkblue} \bf
        mirror branes
        \hspace{-.3cm}
        }
      };

  \end{scope}

  \begin{scope}[shift={(8,0)}]

  \clip (-1.8,-1.8) rectangle (4.8,4.8);
  \draw[step=3, dotted] (-3,-3) grid (6,6);
  \draw[dashed] (-3,-3) circle (1);
  \draw[dashed] (0,-3) circle (1);
  \draw[dashed] (3,-3) circle (1);
  \draw[dashed] (6,-3) circle (1);
  \draw[dashed] (-3,0) circle (1);
  \draw[dashed] (0,0) circle (1);
  \draw[dashed] (3,0) circle (1);
  \draw[dashed] (-3,3) circle (1);
  \draw[dashed] (0,3) circle (1);
  \draw[dashed] (3,3) circle (1);
  \draw[dashed] (-3,6) circle (1);
  \draw[dashed] (0,6) circle (1);
  \draw[dashed] (3,6) circle (1);
  \draw[dashed] (6,6) circle (1);

  \begin{scope}[shift={(1.1,1.1)}]
  \draw[fill=white] (0,0) circle (.07);
  \end{scope}

  \draw[->, decorate, lightgray] (0,0) to (.97,.97);

  \draw[fill=white] (3-1.1,0+1.1) circle (.07);
  \draw[->, decorate, lightgray] (3,0) to (3-.97,0+.97);

  \begin{scope}[shift={(1.1+.2,-1.1-.2)}]
  \draw[fill=white] (0,3) circle (.07);
  \end{scope}
  \draw[->, decorate, lightgray] (0,3) to (0+.97+.2,3-.97-.2);

  \begin{scope}[shift={(-1.2,-1.2)}]
  \draw[fill=white] (3,3) circle (.07);
  \end{scope}
  \draw[->, decorate, lightgray] (3,3) to (3-1.07,3-1.07);

  \draw[fill=black] (17:.7)+(1.36,1.36) circle (.07);
  \draw[fill=black] (17+90:.7)+(1.1,1.1) circle (.07);
  \draw[fill=black] (17+180:.7)+(1.5,1.5) circle (.07);
  \draw[fill=black] (17+270:.7)+(1.5,1.5) circle (.07);

  \draw[->, decorate, lightgray] (17:.7) to ++(1.23,1.23);
  \draw[->, decorate, lightgray] (17+90:.7) to ++(1.0,1.0);
  \draw[->, decorate, lightgray] (17+180:.7) to ++(1.37,1.37);
  \draw[->, decorate, lightgray] (17+270:.7) to ++(1.37,1.37);

  \end{scope}

  \begin{scope}[shift={(0,1.5)}]
  \draw (0,-3.5) node {\tiny $x_1 = 0$};
  \draw (3,-3.5) node {\tiny $x_1 = \tfrac{1}{2}$};

  \begin{scope}[shift={(8,0)}]
  \draw (0,-3.5) node {\tiny $x_1 = 0$};
  \draw (3,-3.5) node {\tiny $x_1 = \tfrac{1}{2}$};
  \end{scope}
  \end{scope}

  \draw (-3.1,0) node {\tiny $x_2 = 0$};
  \draw (-3.1,3) node {\tiny $x_2 = \tfrac{1}{2}$};

\end{tikzpicture}
\end{center}

\vspace{-.4cm}

\noindent {\bf \footnotesize Figure A -- Illustration of the geometric situation of tadpole cancellation on toroidal ADE-orientifolds} {\footnotesize according to \hyperlink{Table1}{\it Table 1},
shown for the case $G^{\mathrm{ADE}} = \mathbb{Z}_4$.
This is for single wrapping number of the branes
along any further compact dimensions; but the general statement
is just the tensor product of this situation with the cohomology
of these further compact spaces.}

\medskip
\medskip
\noindent In view of the evident pattern evidenced by
\hyperlink{Table1}{\it Table 1} and \hyperlink{Table2}{\it Table 2}, here we ask the following question:

\vspace{-1mm}
\begin{center}
\emph{
Is there a generalized cohomological brane charge quantization
which enforces tadpole anomaly cancellation?
}
\end{center}

\vspace{-1mm}
\noindent We show in this paper that (see \hyperlink{FigureU}{\it Figure U}), for fluxless toroidal ADE-orientifolds,
the answer to this question is \emph{unstable equivariant Cohomotopy} theory; see \eqref{EquivariantCohomotopySet} below. Before explaining this, we put the open problem in  perspective:

\vspace{5mm}
\noindent {\bf The open problem -- Systematic understanding of tadpole cancellation by charge quantization.} While the RR-field tadpole cancellation  conditions are thought to be crucial not just for mathematical consistency,
but also for phenomenological accuracy of string model building \cite[Sec. 4.4]{IbanezUranga12},
a real understanding of the conditions in full detail and generality
has remained an open problem; see
\cite[p. 2]{BDS05}\cite[4.6.1]{Moore14}\cite[p. 2]{HMSV19} for
critical discussion.
In particular, most of the existing literature on tadpole cancellation simply regards D-brane charge
as being in ordinary cohomology, while widely accepted arguments say that D-brane charge instead
   must be regarded in (a twisted differential enhancement of) K-theory;
in this context,  see \cite{SS19b} for review,
and see \cite{GS-AHSS}\cite{GS19A}\cite{GS19B} for
    detailed constructions and accounts of the twisted differential case.
  D-brane charge in K-cohomology may be understood as a generalized
  \emph{charge quantization} rule, in analogy  to how Dirac's classical argument for charge
  quantization \cite{Dirac31} (see \cite[16.4e]{Frankel97}) expresses the electromagnetic field as
  a cocycle in (the differential refinement of) ordinary cohomology; see \cite{Freed00}.
  Notice that cohomological charge quantization concerns the
  full non-perturbative structure of a physical theory, including
  its instanton/soliton charge content.

\medskip
 Accordingly, in \cite[5]{Uranga00} it was suggested that RR-tadpole cancellation
 must be a consistency condition expressed in K-theory.  Specifically, for orientifolds this could be
  Atiyah's \emph{R}eal  K-theory \cite{Atiyah66},  i.e., KR-theory restricting on O-planes to KO-theory,
  which has been argued to capture  D-brane charges on orientifolds in
  \cite[5]{Witten98c}\cite{Gukov99}\cite[\S 3]{BGS01};
  explicit constructions are given in \cite{DMDR1}\cite{DMDR2}\cite{HMSV15}\cite{HMSV19}\cite{GS19B}.
In more detail, D-brane charge on orbifolds  is traditionally expected
  \cite[5.1]{Witten98c}\cite[4.5.2]{dBDHKMMS02}\cite{GarciaCompean99}
  to be in equivariant K-theory  (see \cite{Greenlees05}).
Hence orientifolds are expected to have charge quantization in a combination
of these aspects in some Real equivariant K-theory
\cite{Moutuou11}\cite{Moutuou12}\cite{FreedMoore12}\cite{Gomi17}.

\medskip
However, before even formulating tadpole cancellation in Real equivariant K-theory,
the full formulation of O-plane charge has remained open:

\medskip
\noindent \textbf{\emph{Open issue 1: Single O-plane charge.}}
While O-plane charge is not supposed to vary over all
integers, perturbative string theory predicts it to vary in
the set $\{0, \pm 1\}$ (e.g. \cite[p. 2]{HIS00}),
illustrated in \hyperlink{FigureB}{\it Figure B}.

\begin{center}
\hypertarget{FigureB}{}
\begin{tikzpicture}[decoration=snake]

  \draw (-2.5,0) node {\tiny $x_2 = 0$};

  \begin{scope}[shift={(9.8,0)}]

  \draw[dashed] (0,0) circle (1);

  \draw[dotted] (0,2.5) to (0,-2.5);
  \draw[dotted] (1.9,0) to (-1.9,0);

  \draw (0,-2.7) node {\tiny $x_1 = 0$};

  \draw[<->, dashed] (120:2.3)
    to node[very near start]
    {
      \tiny
      \color{darkblue} \bf
      \begin{tabular}{c}
        orientifold
        \\
        action
      \end{tabular}
    }
    (120+180:2.3);

  \draw (0,-.3) node
    {
      \colorbox{white}{
        \tiny
        \color{darkblue} \bf
        $O^{{}^{0}}\!$-plane
      }
    };

  \end{scope}

  \begin{scope}[shift={(0,0)}]

  \draw[dashed] (0,0) circle (1);

  \draw[dotted] (0,2.5) to (0,-2.5);
  \draw[dotted] (1.9,0) to (-1.9,0);

  \draw (0,-2.7) node {\tiny $x_1 = 0$};

  \draw[<->, dashed] (120:2.3)
    to node[very near start]
    {
      \tiny
      \color{darkblue} \bf
      \begin{tabular}{c}
        orientifold
        \\
        action
      \end{tabular}
    }
    (120+180:2.3);

  \draw (0,-.3) node
    {
      \colorbox{white}{
        \tiny
        \color{darkblue} \bf
        $O^{{}^{-}}\!$-plane
      }
    };

    \draw[fill=white] (0,0) circle (.07);

  \end{scope}

  \begin{scope}[shift={(4.9,0)}]

  \draw[dashed] (0,0) circle (1);

  \draw[dotted] (0,2.5) to (0,-2.5);
  \draw[dotted] (1.9,0) to (-1.9,0);

  \draw (0,-2.7) node {\tiny $x_1 = 0$};

  \draw[<->, dashed] (120:2.3)
    to node[very near start]
    {
      \tiny
      \color{darkblue} \bf
      \begin{tabular}{c}
        orientifold
        \\
        action
      \end{tabular}
    }
    (120+180:2.3);

  \draw (0,-.3) node
    {
      \colorbox{white}{
        \tiny
        \color{darkblue} \bf
        $O^{{}^{+}}\!$-plane
      }
    };

    \draw[fill=black] (0,0) circle (.07);

  \begin{scope}[scale=.7, shift={(-1.9,-2.5)}]

  \draw[dashed] (0,0) circle (1);

  \draw[dotted] (0,1.2) to (0,-1.2);
  \draw[dotted] (1.2,0) to (-1.2,0);

  \draw[fill=white] (0,0) circle (.07);

  \draw[fill=black] (30:.85) circle (.07);
  \draw[fill=black] (30+180:.85) circle (.07);

  \draw[->, decorate, lightgray] (30:.8) to (30:.1);
  \draw[->, decorate, lightgray] (30+180:.8) to (30+180:.1);

  \draw (56:1.1)
    node
    {
      \begin{rotate}{56}
        \tiny
        \raisebox{-3pt}{
          $\simeq$
        }
      \end{rotate}
    };

  \end{scope}

  \end{scope}

\end{tikzpicture}
\end{center}

\vspace{-.4cm}

\noindent {\bf \footnotesize Figure B -- The charge carried by a single O-plane}
{\footnotesize takes values in the set $\{0, \pm 1\}$
(in units of corresponding integral D-brane charge),
visualized here following the geometric illustration of \hyperlink{FigureA}{\it Figure A}.
For O4-planes this situation lifts to MO5-planes in M-theory
\cite{Hori98}\cite{Gimon98} \cite[II.B]{AKY98}\cite[3.1]{HananyKol00}.
(The notation for $O^{{}^{0}}$ originates with \cite[p. 29]{Hori98}\cite[p. 4]{Gimon98}; see \hyperlink{FigureT}{\it Figure T} for more.)}

\medskip
But in plain KR-theory all O-planes are $\mathrm{O}^{{}^{-}}\!$-planes.
To capture at least the presence of $\mathrm{O}^{{}^{+}}\!$-planes
requires adding to KR-theory an extra sign choice \cite{DMDR1}.
In some cases this may be regarded as part of a twisting of KR-theory \cite{HMSV19}, but the situation remains inconclusive \cite[p. 2]{HMSV19}.
\footnote{Note that \cite[footnote 1]{HMSV19}
  claims a problem with the sign choice in \cite{DMDR1}, and hence also in \cite{Moutuou12}.
  These continuing issues with orbifold K-theory for D-brane charge
  may motivate but do not affect the discussion here,
  where instead we propose equivariant Cohomotopy theory
  for M-brane charge as an alternative. }

{\small
\begin{floatingtable}[r]
{
\hypertarget{Table3}{}
\begin{tabular}{|c|c|c|c|c|}
  \hline
  \begin{tabular}{c}
\bf    O-plane
    \\
  \bf  species
  \end{tabular}
  &
  \begin{tabular}{c}
  \bf  Charge
    \\
    $q_{{}_{\mathrm{O}p^-}}/q_{{}_{\mathrm{D}p}}$
  \end{tabular}
  &
  \begin{tabular}{c}
  \bf  Transverse
    \\
 \bf   orientifold
  \end{tabular}
  &
  \begin{tabular}{c}
 \bf   Number of
    \\
 \bf   singularities
  \end{tabular}
  \\
  \hline
  \hline
  $\mathrm{O}9^{-}$
  &
  -32
  &
  $\mathbb{T}^0 \!\sslash\! \mathbb{Z}_2$
  &
  1
  \\
  \hline
  $\mathrm{O}8^{-}$
  &
  -16
  &
  $\mathbb{T}^1 \!\sslash\! \mathbb{Z}_2$
  &
  2
  \\
  \hline
  $\mathrm{O}7^{-}$
  &
  -8
  &
  $\mathbb{T}^2 \!\sslash\! \mathbb{Z}_2$
  &
  4
  \\
  \hline
  $\mathrm{O}6^{-}$
  &
  -4
  &
  $\mathbb{T}^3 \!\sslash\! \mathbb{Z}_2$
  &
  8
  \\
  \hline
  $\mathrm{O}5^{-}$
  &
  -2
  &
  $\mathbb{T}^4 \!\sslash\! \mathbb{Z}_2$
  &
  16
  \\
  \hline
  $\mathrm{O}4^{-}$
  &
  -1
  &
  $\mathbb{T}^5 \!\sslash\! \mathbb{Z}_2$
  &
  32
  \\
  \hline
\end{tabular}
}
\\
{\bf \footnotesize   Table 3 -- Absolute O$p$-plane charge}
{\footnotesize
\cite[(5.52)]{IbanezUranga12}\cite[10.212]{BLT13}
$- 32$
is not implied by K-theory
\cite{BGS01},
but is implied by Cohomotopy.
}
\end{floatingtable}
}

\medskip
\noindent \textbf{\emph{Open issue 2: Total O-plane charge.}}
  As highlighted in \cite[p. 4, p. 25]{BGS01}, it remains open whether a putative
  formalization of tadpole cancellation  via Real K-theory reflects the \emph{absolute total}
  charge to be carried by O-planes. This is a glaring open problem, since the absolute total charge
  -32 of O$p$-planes in toroidal orientifolds (see \hyperlink{Table3}{\it Table 3})
  fixes the gauge algebra $\mathfrak{so}(32)$
  of type I string theory required for duality with heterotic string theory
  (see, e.g., \cite[p. 250]{BLT13} \cite{AntoniadisPartoucheTaylor97}) with Green-Schwarz anomaly
  cancellation. This core result of string theory, is the basis of the  ``first superstring revolution''
   \cite[p. 21]{Schwarz11}, and a successful formalization of tadpole cancellation ought to reproduce it.

\medskip
  A proposal for capturing absolute background charge of O-planes
  by equipping K-theory with a quadratic pairing
  has been briefly sketched in \cite{DFM11}, but the implications
  remain somewhat inconclusive \cite[p. 22]{Moore14}.
  We notice that the implications on M-brane charge quantization
  of analogous quadratic functions in M-theory \cite{HopkinsSinder05}
  are reproduced by charge quantization in twisted Cohomotopy theory \cite{FSS19b}.
  Here we further check this alternative proposal:
  That brane charge quantization
  is in \emph{Cohomotopy} cohomology theory,
  which lifts K-theory through the Boardman homomorphism;
  see \eqref{FromUnstableCohomotopyToEquivariantKTheory} below.

\medskip

\noindent {\bf The proposal -- Charge quantization on orientifolds in Equivariant Cohomotopy theory.}
When educated guesswork gets stuck,
it is desirable to identify principles from which to systematically \emph{derive} charge
quantization in M-theory, if possible, and seek the proper generalized cohomology theory to describe
the M-theory fields, as was advocated and initiated in \cite{Sa1}\cite{Sa2} \cite{Sa3}\cite{tcu}.
A first-principles analysis of super $p$-brane sigma-models in rational homotopy theory
shows \cite{S-top}\cite{FSS15}\cite{FSS16a}\cite{FSS16b}
\newline
 that rationalized M-brane charge is
quantized in rational \emph{Cohomotopy} cohomology theory; see \cite{FSS19a} for review.
This naturally suggests the following hypothesis about charge quantization in M-theory
\cite{S-top}\cite{FSS19b}\cite{FSS19c} \cite{SS19b}\cite{SS19c}
(for exposition see \cite{Schreiber20}):

\vspace{0mm}
\hypertarget{HypothesisH}{}

\begin{center}
\fbox{\noindent  {\bf Hypothesis H.} {\bf \it  The M-theory C-field is charge-quantized in Cohomotopy theory}.}
\end{center}

Applied to toroidal orbifolds, the relevant flavor of unstable
Cohomotopy theory is (see \hyperlink{Table4}{\it Table 4})
\emph{unstable equivariant Cohomotopy} (\cite[8.4]{tomDieck79}\cite{Cruickshank03}),
denoted $\pi^\bullet_G$ \eqref{EquivariantCohomotopySet}.
This is the cohomology theory whose degrees are labeled by
orthogonal linear $G$-representations, called the \emph{RO-degree} (see, e.g., \cite[3]{Blu17})

\vspace{-2mm}
\begin{equation}
  \label{RODegree}
  \xymatrix{
    {\phantom{V}}\mathllap{
      \mbox{
        \tiny
        \color{darkblue} \bf
        ``RO-degree''
      }
    }
    \ar@[white]@(ul,ur)^{
      \mathllap{
      \mbox{
        \tiny
        \color{darkblue} \bf
        orthogonal linear $G$-representation
      }
      }
    }
  }
  \xymatrix{V \ar@(ul,ur)^{G \subset \mathrm{O}(\mathrm{dim}(V))}  }
  \!\!\!\!\!\in \mathrm{RO}(G)
  \;\;\; \mbox{\tiny \color{darkblue} \bf representation ring}
\end{equation}
and whose value on
a topological $G$-space $X$
(representing a global $G$-quotient orbifold $X \!\sslash\! G$)
with specified point at infinity $\infty \in X$ -- see diagram \eqref{VanishingAtInfinity} --
is the
\emph{set} of $G$-homotopy classes \eqref{GHomotopy} of
pointed $G$-equivariant continuous functions \eqref{EquivariantFunction}
from $X$ to the
{$V$-representation sphere} $S^V$ \eqref{RepSpheres}
(see \cref{EquivariantCohomotopyAndTadpoleCancellation}
for details and illustration):
\begin{equation}
  \label{EquivariantCohomotopySet}
  \begin{array}{ccc}
  \pi^V_G
  \big(
    X
  \big)
  &
  \coloneqq
  &
  \left\{
    \raisebox{-6pt}{
    {\xymatrix{
      X
      \ar@(ul,ur)|-{\,G\,}
      \ar[rr]^-c
      &&
      S^V
      \ar@(ul,ur)|-{\,G\,}
    }
    }
    }
  \right\}_{\raisebox{2pt}{\tiny$\!\!\!\Big/\sim$}}
  \\
    \tiny
    \color{darkblue}
    \begin{tabular}{c}    \bf
      equivariant Cohomotopy set
      \\
   \bf   of the orbifold $X \!\sslash\! G$
      \\
  \bf    in RO-degree $V$
    \end{tabular}
    &&
    \tiny
    \color{darkblue}
    \begin{tabular}{c}
  \bf    set of $G$-homotopy classes
      \\
    \bf  of $G$-equivariant continuous functions
      \\
   \bf   from $X$ to $S^V$
    \end{tabular}
  \end{array}
\end{equation}
This
is the evident enhancement to unstable $G$-equivariant homotopy theory
(see \cite[1]{Blu17})
of unstable plain Cohomotopy theory $\pi^\bullet$
(\cite{Borsuk36}\cite{Spanier49}\cite{KMT12}\cite[3.1]{FSS19b}).

\vspace{-.3cm}
\begin{equation}
  \label{FromUnstableCohomotopyToEquivariantKTheory}
  \raisebox{0pt}{
  \xymatrix@R=-20pt@C=16pt{
    \mbox{
      \raisebox{-10pt}{\footnotesize
      \begin{minipage}[l]{7.2cm}
      Equivariant Cohomotopy
      is a non-abelian (i.e. ``unstable'') Cohomology theory
      \cite{SSS12}\cite{NSS12}
      that maps to equivariant K-theory
      via stabilization followed by the
      Boardman homomorphism,
      see \cref{StableEquivariantHopfDegree} and \cite{SS19b}.
      \end{minipage}
      }
    }
    &
    \pi^\bullet_G
    \ar[rr]^-{\Sigma^\infty}_-{
      \mbox{
        \tiny
        \color{darkblue} \bf
        \begin{tabular}{c}
          $\phantom{a}$
          \\
          stablilization
        \end{tabular}
      }
    }
    &&
    \mathbb{S}_G
    \ar[rr]^-{\beta}_-{
      \mbox{
        \tiny
        \color{darkblue} \bf
        \begin{tabular}{c}
          $\phantom{a}$
          \\
          Boardman
          \\
          homomorphism
        \end{tabular}
      }
    }
    &&
    \mathrm{KO}_G
    \\
    &
    \mbox{
      \tiny
      \color{darkblue} \bf
      \begin{tabular}{c}
        unstable
        \\
        equivariant
        \\
        Cohomotopy
      \end{tabular}
    }
    &&
    \mbox{
      \tiny
      \color{darkblue} \bf
      \begin{tabular}{c}
        stable
        \\
        equivariant
        \\
        Cohomotopy
      \end{tabular}
    }
    &&
    \mbox{
      \tiny
      \color{darkblue} \bf
      \begin{tabular}{c}
        \\
        equivariant
        \\
        K-theory
      \end{tabular}
    }
  }
  }
\end{equation}

\noindent {\bf The solution -- From Hypothesis H.}  In this article
we explain how lifting brane charge quantization to ADE-equivariant Cohomotopy,
regarded as the generalized Dirac charge
quantization of the M-theory C-field (e.g. \cite{Duff99B}) on toroidal M-orientifolds
(\cite{DasguptaMukhi95}\cite{Witten95b}\cite{Hori98}\cite{ADE}),
gives the local O-plane charges
in $\{0,\pm 1\}$ from \hyperlink{FigureB}{\it Figure B}
and enforces on D-brane charge in the underlying equivariant K-theory
\eqref{FromUnstableCohomotopyToEquivariantKTheory}
the RR-field tadpole cancellation constraints from
\hyperlink{Table1}{\it Table 1} via their M-theory lift from
\hyperlink{Table2}{\it Table 2}.


\medskip
\noindent {\bf Overall picture -- M-Theory and Cohomotopy.}
As we further explain in \cite{OrbifoldCohomology},
unstable equivariant Cohomotopy theory
is the incarnation on flat orbifolds of
\emph{unstable twisted Cohomotopy} cohomology theory,
which we showed in \cite{FSS19b}\cite{FSS19c}
implies a list of M-theory anomaly cancellation conditions
on non-singular (i.e., ``smooth'') but topologically non-trivial spacetimes;
see \hyperlink{Table4}{\it Table 4}:

\vspace{-1mm}
{\small
\hypertarget{Table4}{}
\begin{center}
\hspace{-4.5mm}
\begin{tabular}{cc}
\setlength\tabcolsep{.4em}
\begin{tabular}{|c||c|c|}
  \hline
  Spacetime & {\bf Flat} & {\bf Curved}
  \\
  \hline
  \hline
  {\bf Smooth} &
  \begin{tabular}{c}
    plain
    \\
    Cohomotopy
    \\
    (\cite{FSS15}\cite{BSS18})
  \end{tabular}
  &
  \begin{tabular}{c}
    twisted
    \\
    Cohomotopy
    \\
    (\cite{FSS19b}\cite{FSS19c})
  \end{tabular}
  \\
  \hline
  \begin{tabular}{c}
    \bf Orbi-
    \\
    \bf singular
  \end{tabular}
   &
   \begin{tabular}{c}
     equivariant
     \\
     Cohomotopy
     \\
     (\cite{ADE}\cite{SS19b} \cref{M5MO5AnomalyCancellation})
   \end{tabular}
   &
   \begin{tabular}{c}
     orbifold
     \\
     Cohomotopy
     \\
     (\cite{OrbifoldCohomology})
   \end{tabular}
  \\
  \hline
\end{tabular}
&
\begin{minipage}[l]{8cm}
  {\bf \footnotesize Table 4 -- M-theory anomaly cancellation by
    C-field
  charge quantization in Cohomotopy.} \hspace{-2mm}
  {\footnotesize On smooth but curved spacetimes, Cohomotopy theory
  is twisted via the J-homomorphism by the tangent bundle.
  On flat orbi-orientifolds the spacetime curvature is
  all concentrated in the $G$-singularities,
  around which  the tangent bundle becomes a $G$-representation
  and twisted Cohomotopy becomes equivariant Cohomotopy.
  In each case the respective charge quantization implies
  expected anomaly cancellation conditions. See also
  \hyperlink{Table8}{\it Table 8}}.
\end{minipage}
\end{tabular}
\end{center}
}
\vspace{1mm}
\noindent Each entry in \hyperlink{Table4}{\it Table 4}
supports \hyperlink{HypothesisH}{\it Hypothesis H}
in different corners of the expected phase space of M-theory.
This suggests that
\hyperlink{HypothesisH}{\it Hypothesis H} is a correct
assumption about the elusive mathematical foundation of M-theory.

\medskip
\hypertarget{RelevanceOfUnstable}{}
\noindent {\bf The necessity of unstable = non-abelian charge quantization for O-planes.}
We highlight that most authors
who discuss equivariant Cohomotopy consider \emph{stable} equivariant Cohomotopy
theory (e.g. \cite{Segal71}\cite{Carlsson84} \cite{Lueck}), represented by the
equivariant sphere spectrum $\mathbb{S}_G$ in equivariant
stable homotopy theory (\cite{LMS86}\cite[Appendix]{HHR16});
see \cref{StableEquivariantHopfDegree} below.
There are comparison homomorphisms \eqref{FromUnstableCohomotopyToEquivariantKTheory}
from equivariant unstable Cohomotopy
to stable Cohomotopy and further to K-theory
but each step forgets some information (has a non-trivial kernel)
and produces spurious information (has a non-trivial cokernel);
see \cite{SS19b}.
For the result presented here (just as for the previous discussion in
\cite{FSS19b}\cite{FSS19c}), it is crucial that we use the richer \emph{unstable}
version of the Cohomotopy theory, hence the \emph{non-abelian
Cohomology theory} \cite{SSS12}\cite{NSS12},
which is the one that follows from analysis of super $p$-brane cocycles
\cite{S-top}\cite{FSS19a}. We find that:
\begin{enumerate}[{\bf (a)}]
\vspace{-2mm}
\item the difference in
the behavior
between the O-plane charges and the D-brane charges
(in \hyperlink{Table1}{\it Table 1}, \hyperlink{TableMNTC}{\it Table 2}
and \hyperlink{FigureP}{Figure P}) and
\vspace{-2mm}
\item the unstable/non-abelian nature of O-plane charge itself
(\hyperlink{FigureOP}{\it Figure OP})
\end{enumerate}
\vspace{-2mm}
are reflected in the passage from
the unstable to the stable range in unstable  ADE-equivariant Cohomotopy, where
the O-plane charges are distinguished as being in the unstable range;
see \hyperlink{FigureC}{\it Figure C}:

\vspace{-9mm}

\begin{center}
\hypertarget{FigureC}{}
$$
  \xymatrix@C=20pt@R=1.2em{
    &
    \mbox{
      \tiny
      \color{darkblue} \bf
      \begin{tabular}{c}
        toroidal
        \\
        orientifold
      \end{tabular}
    }
    \ar@{}[rrrr]|-{
      \mbox{
        \tiny
        \color{darkblue} \bf
        \begin{tabular}{c}
          cocycle in
          \\
          equivariant Cohomotopy
        \end{tabular}
      }
    }
    &&&&
    \mbox{
      \tiny
      \color{darkblue} \bf
      \begin{tabular}{c}
        equivariant cohomotopy
        \\
        classifying space
        \\
        (representation sphere)
      \end{tabular}
    }
    \\
    &
    \mathbb{T}^{\mathbf{4}_{\mathbb{H}}}
    \ar@(ul,ur)|<<<<{ \mathbb{Z}_2 }
    \ar[rrrr]^{ c }
    &&&&
    S^{\mathbf{4}_{\mathbb{H}}}
    \ar@(ul,ur)|<<<<{\mathbb{Z}_2}
    \\
    \mbox{
      \tiny
      \color{darkblue} \bf
      \begin{tabular}{c}
        set of fixed points
        \\
        in orientifold
        \\
        (O-planes)
     \end{tabular}
    }
    &
    \big(
      \mathbb{T}^{\mathbf{4}_{\mathbb{H}}}
    \big)^{\mathbb{Z}_2}
    =
    \big\{
      0, \tfrac{1}{2}
    \big\}^4
    \ar[rrrr]^-{
      c^{(\mathbb{Z}_2)}
    }_{
      \mbox{
        \tiny
        \color{darkblue} \bf
        \begin{tabular}{c}
          O-plane charge
          \\
          (a subset)
        \end{tabular}
      }
    }
    \ar@{^{(}->}[dd]
    &&&&
    \underset{
      S^0
    }{
    \underbrace{
      \{0,\infty\}
    }
    }
    =
    \big( S^{\mathbf{4}_{\mathbb{H}}} \big)^{\mathbb{Z}_2}
    \ar@{^{(}->}[dd]
    &
    \mbox{
      \tiny
      \color{darkblue} \bf
      \begin{tabular}{c}
        set of fixed points
        \\
        in classifying space
        \\
        (the 0-sphere)
     \end{tabular}
    }
    \\
    \\
    \mbox{
      \tiny
      \color{darkblue} \bf
      \begin{tabular}{c}
        underlying
        \\
        plain 4-torus
      \end{tabular}
    }
    &
    \big(
      \mathbb{T}^{\mathbf{4}_{\mathbb{H}}}
    \big)^1
    =
    T^4
    \ar[rrrr]^-{
      (c)^1
    }_-{
      \mbox{
        \tiny
        \color{darkblue} \bf
        \begin{tabular}{c}
          net charge
          \\
          (an integer)
        \end{tabular}
      }
    }    &&&&
    S^4
    =
    \big( S^{\mathbf{4}_{\mathbb{H}}} \big)^{1}
    &
    \mbox{
      \tiny
      \color{darkblue} \bf
      \begin{tabular}{c}
        underlying
        \\
        plain 4-sphere
      \end{tabular}
    }
  }
$$
\end{center}

\vspace{-3mm}
\noindent {\bf \footnotesize Figure C -- A cocycle in unstable ADE-equivariant Cohomotopy on a toroidal orientifold}
{\footnotesize
according to \eqref{EquivariantCohomotopySet},
and its decomposition on fixed point strata
into Elmendorf stages; see \cite[1.3]{Blu17}\cite[3.1]{ADE}.}

\vspace{4mm}
{\bf Characterizing brane/O-plane charges -- Unstable (equivariant) differential topology.}
Since in \hyperlink{FigureC}{\it Figure C}
the fixed locus in the classifying space is just a 0-sphere,
and since the Hopf degree of maps $X^n \to S^n$ stabilizes only for $n \geq 1$ --
see diagram \eqref{HopfDegreesUnderSuspension} -- the fixed points in the
spacetime (= O-planes) carry ``unstable'' or ``non-linear'' charge:
not given by a group element, but by a subset, distinguishing
$O^{{}^{\pm}}\!$-planes from $O^{{}^{0}}\!$-planes as in
\hyperlink{FigureOP}{\it Figure OP}.
The further distinction between $O^{{}^{-}}\!$-planes and
$O^{{}^{+}}\!$-planes is implied by
normal framing that enters in the unstable Pontrjagin-Thom theorem
(discussed in \cref{NormalFramingAndBraneAntibraneAnnihilation}).
Moreover, the local/twisted tadpole cancellation condition
in the vicinity of O-planes is implied by the unstable
equivariant Hopf degree theorem (discussed in \cref{LocalTadpoleCancellation}).
Last but not least, it is the unstable Pontrjagin-Thom theorem,
discussed in \cref{Sec-Coh},
which identifies all these charges with \emph{sub}-manifolds,
hence with actual brane/O-plane worldvolumes
as shown in \hyperlink{FigureA}{\it Figure A},
(while the stable PT-theorem instead relates
stable Cohomotopy to manifolds equipped with any maps to spacetime).

\vspace{.4cm}

{\small
\hspace{-.8cm}
\setlength\tabcolsep{.25em}
\begin{tabular}{|c|c|c|c|}
  \hline
  \begin{tabular}{c}
    \bf
    Classical
    theorem
  \end{tabular}
  &
  Reference
  &
  \begin{tabular}{c}
    \bf Interpretation for
    brane charge quantization
    \\
    \bf in unstable Cohomotopy (\hyperlink{HypothesisH}{\it Hypothesis H})
  \end{tabular}
  &
  \begin{tabular}{c}
    Discussed in
  \end{tabular}
  \\
  \hline
  \hline
  \begin{tabular}{c}
    Unstable
    \\
    Pontrjagin-Thom theorem
  \end{tabular}
  &
  \begin{tabular}{c}
    \cite[IX (5.5)]{Kosinski93}
  \end{tabular}
  &
  \begin{tabular}{c}
    Cohomotopy charge is
    sourced by submanifolds
    \\
    hence by worldvolumes of branes and O-planes
  \end{tabular}
  &
  \cref{PTTheorem}
  \\
  \hline
  \begin{tabular}{c}
    Unstable
    \\
    Hopf degree theorem
  \end{tabular}
  &
  \begin{tabular}{c}
    \cite[IX (5.8)]{Kosinski93}
    \\
    \cite[7.5]{Kobin16}
  \end{tabular}
  &
  \begin{tabular}{c}
    Charge of flat transversal branes is integer
    \\
    while charge of flat transversal O-planes is in $\{0,1\}$
  \end{tabular}
  &
  \cref{NormalFramingAndBraneAntibraneAnnihilation}
  \\
  \hline
  \begin{tabular}{c}
    Unstable
    \\
    equivariant Hopf degree theorem
  \end{tabular}
  &
  \cite[8.4]{tomDieck79}
  &
  \begin{tabular}{c}
    Branes appear in regular reps around O-planes
    \\
    = local/twisted tadpole anomaly cancellation
  \end{tabular}
  &
  \cref{EquivariantCohomotopyAndTadpoleCancellation}
  \\
  \hline
\end{tabular}
}

\medskip
\noindent {\bf Organization of the paper.}
In \cref{Sec-Coh} we discuss how the classical unstable Pontrjagin-Thom isomorphism
says that plain Cohomotopy classifies charge carried by brane worldvolumes.
In \cref{EquivariantCohomotopyAndTadpoleCancellation} we introduce the enhancement
of this situation to equivariant Cohomotopy on toroidal orbifolds, where it encodes joint
D-brane and O-plane charge. We explain in \cref{EquivariantCohomotopyAndTadpoleCancellation}
that now the \emph{equivariant Hopf degree theorem} encodes the form of local/twisted
tadpole cancellation conditions, and explain in \cref{GlobalTadpoleCancellation} that
super-differential refinement at global Elmendorf stage encodes the form of global/untwisted
tadpole cancellation conditions as in \hyperlink{Table1}{\it Table 1} and
\hyperlink{Table2}{\it Table 2}. The Pontrjagin-Thom theorem now serves to map
these charges precisely to the geometric situations of the form shown in
\hyperlink{FigureA}{\it Figure A}.
Finally, in \cref{M5MO5AnomalyCancellation} we specify these
general considerations to the physics of M5-branes at MO5-planes
in toroidal ADE-orientifolds in M-theory,
with the C-field charge-quantized in equivariant Cohomotopy theory,
according to \hyperlink{HypothesisH}{\it Hypothesis H}.
To set the scene, we first recall in \cref{HeteroticMTheoryOnADEOrbifolds}
the situation of heterotic M-theory on ADE-orbifolds  and
highlight subtleties in the interpretation of MO5-planes.
With this in hand, we apply in \cref{EquivariantCohomotopyChargeOfM5AtMO5}
the general discussion of equivariant Cohomotopy
from \cref{EquivariantCohomotopyAndTadpoleCancellation}
to ADE-singularities intersecting MO9-planes in M-theory,
and find (Cor. \ref{EquivariantCohomotopyOfSemiComplementSpacetime}, Cor. \ref{GlobalM5MO5CancellationImplied}) that this correctly
encodes the expected anomaly cancellation of M5-branes at MO5-planes,
and this, upon double dimensional reduction (see \hyperlink{Table7}{\it Table 7} and \hyperlink{FigureU}{\it Figure U}), the RR-field tadpole anomaly cancellation for D-branes on ADE-orientifolds.

\section{Cohomotopy and brane charge }
\label{Sec-Coh}

\vspace{-1mm}
Before turning to equivariant/orbifold structure in \cref{EquivariantCohomotopyAndTadpoleCancellation},
we first discuss basics of plain unstable Cohomotopy on plain manifolds.
The key point is that the unstable \emph{Pontrjagin-Thom theorem},
reviewed in \cref{PTTheorem},
identifies cocycles in unstable Cohomotopy theory with
cobordism classes of submanifolds carrying
certain extra structure (normal framing).
These submanifolds are naturally identified with the
worldvolumes of branes that source the corresponding Cohomotopy charge,
and the normal structure they carry corresponds to the charge
carried by the branes, distinguishing
branes from anti-branes. In \cref{NormalFramingAndBraneAntibraneAnnihilation}
we highlight that coboundaries in unstable Cohomotopy
accordingly correspond to brane pair creation/annihilation processes.
This way the Pontrjagin-Thom theorem establishes Cohomotopy as a
natural home for brane charges, as proposed in \cite{S-top}.

\vspace{-1mm}
\subsection{Pontrjagin-Thom theorem and brane worldvolumes}
\label{PTTheorem}

\noindent {\bf Cohomotopy cohomology theory.}
The special case of unstable $G$-equivariant Cohomotopy \eqref{EquivariantCohomotopySet}
with $G = 1$ the trivial group
is unstable plain Cohomotopy theory (\cite{Borsuk36}\cite{Spanier49}\cite{KMT12}\cite[3.1]{FSS19b}), denoted
$\pi^\bullet \coloneqq \pi^\bullet_1$.
This is the unstable/non-abelian cohomology theory whose degrees are
natural numbers $n \in \mathbb{N}$ and which assigns to an
un-pointed topological space $X$ the
\emph{Cohomotopy set} of free
homotopy classes of continuous maps into the $n$-sphere:
\vspace{-2mm}
\begin{equation}
  \label{PlainCohomotopySet}
  \begin{array}{ccc}
  \pi^n
  (
    X
  )
  &
  \coloneqq
  &
  \big\{ \!\!\!\!\!
    \raisebox{+2pt}{
    {\xymatrix{
      X
      \ar[rr]^-c
      &&
      S^n
    }
    }
    }
  \!\!\!\!\! \big\}_{\raisebox{16pt}{$/\sim$}}
  \\[-12pt]
    \tiny
    \color{darkblue}
                  \begin{tabular}{c}
   \bf   Cohomotopy set
      \\
    \bf   of the space $X$
      \\
    \bf  in degree $n$
    \end{tabular}
    &&
    \tiny
    \color{darkblue}
    \begin{tabular}{c}
     \bf set of homotopy classes
      \\
      \bf of continuous functions
      \\
      \bf from $X$ to the $n$-sphere $S^n$
    \end{tabular}
      \end{array}
\end{equation}
The contravariant
assignment $X \mapsto \pi^n(X)$ is analogous to the assignment
$X \mapsto H^n(X, \mathbb{Z})$ of integral cohomology groups, or of the assignment
$X \mapsto K^n(X)$ of K-theory groups, and as such may be regarded
as a generalized but \emph{non-abelian} cohomology theory \cite{SSS12}\cite{NSS12}: For $n \geq 1$ we have
(as for any connected topological space)
a weak homotopy equivalence between the $n$-sphere and the
classifying space of its loop group,
$
  S^n \;\simeq_{\mathrm{whe}}\; B \big(\Omega S^{n}\big)
$,
which means that the Cohomotopy sets \eqref{PlainCohomotopySet}
\vspace{-2mm} 
$$
  \underset{
    \mathclap{
      \mbox{
        \tiny
        \color{darkblue}   \bf
        $n$-Cohomotopy set
      }
      \;
    }
  }{
    \pi^n(X)
  }
  \;\;\;\; \simeq\;\;\;\;
  \underset{
    \mathclap{
    \;
    \mbox{
      \tiny
      \color{darkblue} \bf
      \begin{tabular}{c}
        non-abelian cohomology set
        \\
        with coefficients in
        \\
        loop group of $n$-sphere
      \end{tabular}
    }
    }
  }{
    H^1(X, \Omega S^n)
  }
$$
are equivalently the non-abelian cohomology sets with coefficient in the
loop-group of the $n$-sphere, in direct generalization of
the familiar case of non-abelian cohomology
$H^1(X,G) \simeq G\mathrm{Bund}(X)_{/\sim}$ with coefficients
in a compact Lie group $G$.

\medskip
In this way we may think of \eqref{PlainCohomotopySet} as defining
a generalized cohomology theory, different from but akin to, say, K-theory,
and as such we call it \emph{Cohomotopy cohomology theory}, or
\emph{Cohomotopy theory} or just \emph{Cohomotopy}, for short.
The capitalization indicates that
this term is the proper name of a specific cohomology theory
(we might abbreviate further to \emph{$C$-theory} to bring out the analogy
with $K$-theory yet more) and \emph{not} on par with
\emph{homotopy theory},
which instead is the name of the general mathematical framework
within which we are speaking. In particular, \emph{Cohomotopy cohomology theory} is \emph{not} the dual concept of \emph{homotopy theory},
but is the dual concept of the unstable/non-abelian generalized
homology theory which assigns homotopy groups
$X \mapsto \pi_n(X)$ to pointed topological spaces $X$
(hence: \emph{Homotopy homology theory}, mostly familiar in its
stable form).

\medskip

\noindent {\bf Unstable Pontrjagin-Thom theorem.}
Thinking of $X$ here as spacetime,
we are interested in the case that $X = X^D$ admits the structure of closed smooth manifold
of dimension $D \in \mathbb{N}$. In this case, the unstable Pontrjagin-Thom
theorem \eqref{UnstablePTTheorem} identifies (see e.g. \cite[IX.5]{Kosinski93})
the degree-$n$ Cohomotopy set of $X^D$ \eqref{PlainCohomotopySet} with
the set of cobordism classes of normally framed codimension-$n$
closed submanifolds of $X^D$
(see e.g.\cite[IX.2]{Kosinski93}), hence of closed submanifolds $\Sigma^{d} \overset{i}{\hookrightarrow} X^D$
which are of dimension $d = D - n$ and  equipped with a choice of trivialization

\vspace{0mm}
\begin{equation}
  \label{TrivializationOfNormalVectorBundle}
  \xymatrix@R=-4pt@C=4em{
    N_i\Sigma
    \ar[rr]^-{ \mbox{\tiny \color{darkblue} \bf normal framing} }_-{\simeq}
    &&
    \Sigma \times \mathbb{R}^{n}
    \\
    \mathclap{
    \mbox{ \bf
      \tiny
      \color{darkblue}
      \begin{tabular}{c}
        normal bundle
        \\
        of codimension $n$ submanifold $\Sigma$
        \\
        inside ambient manifold $X$
      \end{tabular}
    }
    }
    &&
    \mbox{ \bf
      \tiny
      \color{darkblue}
      \begin{tabular}{c}
        trivial vector bundle
        \\
        of rank $n$
      \end{tabular}
    }
  }
\end{equation}
of their normal vector bundle:
\vspace{-3mm}
\begin{equation}
  \label{UnstablePTTheorem}
\hspace{-5mm}
  \xymatrix@R=-10pt@C=3.2em{
    \mbox{
      \raisebox{-10pt}{
      \begin{minipage}[l]{3cm}
       \footnotesize \bf Unstable \\ Pontrjagin-Thom \\ theorem
      \end{minipage}
      \hspace{-1.9cm}
    }
    }
    &
    \pi^n
    \big(
      X^D
    \big)
    \ar@<+6pt>[rrr]^-{
      \overset{
        \mbox{ \bf
          \tiny
          \color{darkblue}
          take pre-image at 0 of regular representative
        }
      }{
        \mathrm{fib}_0 \, \circ \,  \mathrm{reg}
      }
    }_-{
      \simeq
    }
    \ar@{<-}@<-6pt>[rrr]_-{
      \mbox{\tiny ``PT collapse''}
      \atop
      \mbox{\bf
        \tiny
        \color{darkblue}
        assign directed asymptotic distance
      }
    }
    &&&
    \left\{ \!\!\!\!\!\!\!\!
    \mbox{
      \raisebox{2pt}{\footnotesize
      \begin{tabular}{c}
        Closed submanifolds $\Sigma^d \overset{i}{\hookrightarrow} X^D$
        \\
        of dimension $d = D - n$
        \\
        and equipped with normal framing
      \end{tabular}
      }
    }
   \!\!\!\!\!\!\!\!\!  \right\}_{\raisebox{5pt}{\tiny $\!\!\!\Big/{\!\!\mathrm{cobordism}}$}}
    \\
    &
    \mbox{ \bf
      \tiny
      \color{darkblue}
      \begin{tabular}{c}
        Cohomotopy set in degree $n$
        \\
        of closed $D$-dim. manifold $X$
      \end{tabular}
    }
  }
\end{equation}
The construction which exhibits this bijection is
traditionally called the Pontrjagin-Thom \emph{collapse},
but a more suggestive description,
certainly for our application to
brane charges, is this:
\emph{
The Cohomotopy class corresponding
to a submanifold/brane is represented by the function which assigns
\emph{directed asymptotic distance} from the submanifold/brane,
as measured with respect to the given normal framing
\eqref{TrivializationOfNormalVectorBundle} upon identifying
the normal bundle with a tubular neighborhood and regarding
all points outside the tubular neighborhood as being at infinite
distance.} See \hyperlink{FigureD}{\it Figure D}:

\vspace{-4mm}
\begin{center}
{\hypertarget{FigureD}{}}
\begin{tikzpicture}[scale=0.69]
  \node (X) at (-4.5,6)  {\small $X$};
  \node (sphere) at (6,6)  {\small $S^n = (\mathbb{R}^n)^{\mathrm{cpt}}$};

  \draw[->] (X) to node[above] {\footnotesize $c$} (sphere);

  \node at (-4.5,5.4)  {\tiny \color{darkblue} \bf manifold};
  \node at (-4.5,4.4)  {$\overbrace{\phantom{--------------------------}}$};

  \node at (6,5.4)  {\tiny \color{darkblue} \bf \begin{tabular}{c} $n$-sphere \\ Cohomotopy coefficient \end{tabular}};
  \node at (6,4.4)  {$\overbrace{\phantom{--------------}}$};

  \node at (.25,5.7)  {\tiny \color{darkblue} \bf Cohomotopy cocycle};

  \begin{scope}[shift={(-6,-1.5)}]
    \clip (-2.9,-2.9) rectangle (5.9,5.9);
    \draw[step=3, dotted] (-3,-2) grid (6,6);
    \draw[very thick] (-4,1.3) .. controls (-1,-3.2) and (2.3,6.6)  .. (7,4.2);
    \begin{scope}[shift={(0,.9)}]
      \draw[dashed] (-4,1.3) .. controls (-1,-3.2) and (2.3,6.6)  .. (7,4.2);
    \end{scope}
    \begin{scope}[shift={(0,-.9)}]
      \draw[dashed] (-4,1.3) .. controls (-1,-3.2) and (2.3,6.6)  .. (7,4.2);
    \end{scope}
    \begin{scope}[shift={(0,.45)}]
      \draw[dashed, thick] (-4,1.3) .. controls (-1,-3.2) and (2.3,6.6)  .. (7,4.2);
    \end{scope}
    \begin{scope}[shift={(0,-.45)}]
      \draw[dashed, thick] (-4,1.3) .. controls (-1,-3.2) and (2.3,6.6)  .. (7,4.2);
    \end{scope}
  \end{scope}
  \begin{scope}[shift={(4,0)}]
    \draw (2,0) circle (2);
    \node at (+.6,0)
      {{\tiny $0$}
        \raisebox{.0cm}{
          $
          \mathrlap{
          \!\!\!\!\!\!\!
          \mbox{ \bf
          \tiny \color{darkblue}
            \begin{tabular}{c}
              regular
              \\
              value
            \end{tabular}
          }}
          $
        }
      };
    \node (zero) at (0,0) {$-$};
    \node (infinity) at (4,0) {\colorbox{white}{$\infty$}};

   \fill[black] (2,0) ++(40+180:2) node (minusepsilon)
     {\begin{turn}{-45} $)$  \end{turn}};
   \fill[black] (2,0) ++(180-40:2) node (epsilon)
     {\begin{turn}{45} $)$ \end{turn}};
   \fill[black] (2.3,0.25) ++(40+180:2) node
     { \tiny $-\epsilon$ };
   \fill[black] (2.3,-0.25) ++(-40-180:2) node
     { \tiny $+\epsilon$ };
  \end{scope}

  \draw[|->, thin, brown] (-5.1-.25,.95-.25)
    to[bend right=6.7]
    (epsilon);
  \draw[|->, thin, brown] (-5.1+.25,.05+.25)
    to[bend right=6.7]
    (minusepsilon);
  \draw[|->, thin, brown] (-5.1,.5)
    to[bend right=6.7]
    node
    {
      \tiny
      \color{darkblue}
      \colorbox{white}{\bf
        codimension $n$ submanifold
      }
      $\mathrlap{
        \;\;\;\;\;\;\;\;
        \raisebox{-47pt}{
        \begin{turn}{90}
          \colorbox{white}{
            \begin{tabular}{c}
              \tiny \bf tubular neighborhood
              \\
              \tiny \bf $\leftrightsquigarrow$ normal framing
            \end{tabular}
          }
        \end{turn}
        }
      }$
    }
    (zero);

  \draw[|->, thin, olive] (-5.1-.5,1.4-.45) to[bend left=26] (infinity);
  \draw[|->, thin, olive] (-4.9,3.2) to[bend left=26] (infinity);
  \draw[|->, thin, olive] (-5.1+.5,-.4+.45)
    to[bend right=33] node[below] {\colorbox{white}{\tiny \color{darkblue} \bf \begin{tabular}{c}
    constant on $\infty$ \\
    away from tubular neighborhood\\\end{tabular}}} (infinity);
  \draw[|->, thin, olive] (-4.7,-2.7) to[bend right=30] (infinity);

\end{tikzpicture}
\end{center}
\vspace{-9mm}
\noindent {\bf \footnotesize Figure D --  The Pontrjagin-Thom construction}
{\footnotesize which establishes the unstable Pontrjagin-Thom theorem \eqref{UnstablePTTheorem}. The cocycle $c$ in Cohomotopy \cref{PlainCohomotopySet} is the continuous function which sends
each point to its directed asymptotic distance from the given submanifold.}

\vspace{3mm}
\noindent {\bf One-point compactifications by adjoining the point at infinity.}
Here and in all of the following, we are making crucial use of the
fact that the $n$-sphere is the one-point compactification
$(-)^{\mathrm{cpt}}$ of the
Cartesian space $\mathbb{R}^n$,
\begin{equation}
  \label{SphereIsCompactificationOfCartesianSpace}
  S^n
    \;\simeq_{{}_{\mathrm{homeo}}}\;
  \big( \mathbb{R}^n\big)^{\mathrm{cpt}}
  \;\coloneqq\;
  \big(
    \{ x \in \mathbb{R}^n \;\mbox{or}\; x = \infty  \}
    ,
    \tau_{\mathrm{cpt}}
  \big)
  \phantom{AAAA}
  \mbox{for all $n \in \mathbb{N}$},
\end{equation}
as indicated on the right of \hyperlink{FigureD}{\it Figure D}.
Here the one-point compactification $X^{\mathrm{cpt}}$ of a topological space $X$
is defined (e.g. \cite[p. 150]{Kelly55}) by adjoining one point to the underlying set of
$X$ -- denoted ``$\infty$'' as it becomes literally the \emph{point at infinity} -- and by declaring
on the resulting set a topology $\tau_{{\mathrm{cpt}}}$ whose open subsets are those of
$X$, not containing $\infty$, and those containing $\infty$ but whose complement in
$X$ is compact. Notice that this construction also applies to topological spaces that already are
compact, in which case the point at infinity appears disconnected
\begin{equation}
  \label{BasepointFreelyAdjoined}
  X \;\mbox{already compact}
  \;\;\Rightarrow\;\;
  X^{\mathrm{cpt}} \;=\; X_+  \;\coloneqq\;
  X \sqcup \{\infty\}
  \,.
\end{equation}
This means that \eqref{SphereIsCompactificationOfCartesianSpace}
indeed holds also in the ``unstable range'' of $n = 0$:
\begin{equation}
  \label{OSphereAsCompactificationOfPoint}
  \big( \mathbb{R}^0 \big)^{\mathrm{cpt}}
  \;=\;
  \big( \{0\} \big)^{\mathrm{cpt}}
  \;=\;
  \{0\} \sqcup \{\infty\}
  \;=\;
  S^0
  \,.
\end{equation}

\noindent {\bf Cohomotopy charge vanishing at infinity.}
In view of the Pontrjagin-Thom theorem \eqref{UnstablePTTheorem},
it makes sense to say that a cocycle in Cohomotopy \emph{vanishes}
wherever it takes as value the point at infinity
$\infty \in \big( \mathbb{R}^n\big)^\mathrm{cpt} \simeq S^n$ in the
coefficient sphere, identified under \eqref{SphereIsCompactificationOfCartesianSpace}. This means
to regard the coefficient sphere as a pointed topological space,
with basepoint $\infty \in S^n$.
Given then a non-compact (spacetime) manifold $X$
(such as $X = \mathbb{R}^n$), a Cohomotopy cocycle
$X \longrightarrow S^n$
\emph{vanishes at infinity} if it extends to the one-point
compactification $X^{\mathrm{cpt}}$ \eqref{SphereIsCompactificationOfCartesianSpace}
such as to send the actual point at infinity
$\infty \in X^{\mathrm{cpt}}$ to the point at infinity in the coefficient
sphere.
\begin{equation}
  \label{VanishingAtInfinity}
  \hspace{-2cm}
  \mbox{
    \begin{minipage}[l]{9cm}
      \footnotesize
      A Cohomotopy cocycle on a non-compact space $X$
      which {\it vanishes at infinity} is
      a Cohomotopy cocycle on the one-point compactification
      $X^{\mathrm{cpt}}$ that sends the point
      at infinity in the domain to that in the coefficient
      $n$-sphere.
    \end{minipage}
  }
  \phantom{AAA}
  \raisebox{20pt}{
  \xymatrix@R=1.5em{
    X^{\mathrm{cpt}}
    \ar[rr]^-{c}
    &&
    \big( \mathbb{R}^n\big) \mathrlap{ \; \simeq S^n }
    \\
    \{\infty\}
    \ar@{^{(}->}[u]
    \ar[rr]_-{ c_{\vert_{ \{\infty\}}} }
    &&
    \{\infty\}
    \ar@{^{(}->}[u]
  }
  }
\end{equation}
\begin{example}[{\hyperlink{FigureE}{\it Figure E}}]
For $X = \mathbb{R}^n$, we have that Cohomotopy $n$-cocycles
on $X$ vanishing at infinity are equivalently maps from an $n$-sphere
to itself:
\end{example}

\vspace{-8mm}
\begin{center}
{\hypertarget{FigureE}{}}
\begin{tikzpicture}
  \begin{scope}[shift={(0,-1.3)}]
  \node (X) at (-4.5,6)  {\small $(\mathbb{R}^n)^{\mathrm{cpt}}$};
  \node (sphere) at (6,6)  {\small $S^n = (\mathbb{R}^n)^{\mathrm{cpt}}$};

  \draw[->] (X) to node[above] {\footnotesize $c  = 1 - 3 = -2$} (sphere);

  \node at (-4.5, 5.3)
    {\tiny \color{darkblue} \bf
      \begin{tabular}{c}
        Euclidean $n$-space
        \\
        compactified by
        \\
        a point at infinity
      \end{tabular}
    };
  \node at (-4.5,4.6)
    {$\overbrace{\phantom{----------------}}$};

  \node at (6,5.4)
    {\tiny \color{darkblue} \bf
      \begin{tabular}{c}
        $n$-sphere
        \\
        Cohomotopy coefficient
      \end{tabular}
    };
  \node at (6,4.6)  {$\overbrace{\phantom{--------------}}$};

  \node at (.55,5.4)
    {
      \tiny
      \color{darkblue} \bf
      \begin{tabular}{c}
        Cohomotopy cocycle
        \\
        counting net number
        \\
        of charged submanifolds
      \end{tabular}
    };
  \end{scope}

  \begin{scope}[shift={(-4.5,1.3)}]
    \draw (0,0) circle (2);
    \node (infinity1) at (2,0) {\colorbox{white}{$\infty$}};
    \node (submanifold1) at (180-20:2) {$\bullet$};
    \node (submanifold2) at (180+110:2) {};
    \draw[fill=white] (180+110:2) circle (.07);
    \node (submanifold3) at (180+130:2) {$\bullet$};
    \node (submanifold4) at (180+140:2) {$\bullet$};
  \end{scope}

  \draw[|->, thin, olive]
    (infinity1)
    to[bend right=40]
    node {\colorbox{white}{\tiny \color{darkblue} \bf cocycle vanishes at infinity}}
    (7.7,-.2);
  \node at (-.2,-1.5) { \colorbox{white}{$\phantom{{A A A}\atop {A A} }$} };
  \node at (5.1,-1.8) { \colorbox{white}{$\phantom{ A }$} };

  \begin{scope}[shift={(4,1.3)}]
    \draw (2,0) circle (2);
    \node at (+.5,0)
      {{\footnotesize $0$}
        \raisebox{.1cm}{
          $
          \mathrlap{
          \!\!\!\!\!\!\!
          \mbox{\bf
          \tiny \color{darkblue}
            \begin{tabular}{c}
              regular
              \\
              value
            \end{tabular}
          }}
          $
        }
      };
    \node (zero) at (0,0) {$-$};
    \node (infinity) at (4,0) {\colorbox{white}{$\infty$}};

   \fill[black] (2,0) ++(40+180:2) node (minusepsilon)
     {\begin{turn}{-45} $)$  \end{turn}};
   \fill[black] (2,0) ++(180-40:2) node (epsilon)
     {\begin{turn}{45} $)$ \end{turn}};
   \fill[black] (2.3,0.25) ++(40+180:2) node
     {\footnotesize $-\epsilon$ };
   \fill[black] (2.3,-0.25) ++(-40-180:2) node
     {\footnotesize $+\epsilon$ };
  \end{scope}

  \draw[|->, thin, olive]
    (submanifold2)
    to[bend right=16]
    node[very near start, below]
      {
        \tiny
        \color{darkblue} \bf
        \begin{tabular}{c}
          here with opposite
          \\
          normal framing
          \\
          (see \hyperlink{FigureF}{\it Figure F})
        \end{tabular}
      }
    (zero);

  \draw[|->, thin, olive]
    (submanifold3)
    to[bend right=16]
    (zero);

  \draw[|->, thin, olive]
    (submanifold4)
    to[bend right=16]
    node
      {
        \hspace{-4.4cm}
        \raisebox{-.9cm}{
          \colorbox{white}{
            \hspace{-.3cm}
            \tiny \color{darkblue} \bf more submanifolds
            \hspace{-.3cm}
          }
        }
      }
    (zero);

  \draw[white, line width=8pt] (2,.7) to (2.8,.7);

  \draw[|->, thin, brown]
    (submanifold1)+(.13,.3)
    to[bend left=11]
    (epsilon);
  \draw[|->, thin, brown]
    (submanifold1)+(-.1,-.4)
    to[bend left=11]
    (minusepsilon);
  \draw[|->, thin, brown]
    (submanifold1)
    to[bend left=11]
    node
    {
      \raisebox{-1cm}{
      \hspace{1.4cm}
      \tiny
      \color{darkblue}
      \colorbox{white}{ \bf
        codimension $n$ submanifold
      }
      $\mathrlap{
        \;\;\;\;\;\;\;\;\;
        \raisebox{-40pt}{
        \begin{turn}{90}
          \colorbox{white}{
            \hspace{-.5cm}
            \begin{tabular}{c}
          \bf    \tiny tubular neighborhood
              \\
          \bf    \tiny $\leftrightsquigarrow$ normal framing
            \end{tabular}
          }
        \end{turn}
        }
      }$
      }
    }
    (zero);
\end{tikzpicture}
\end{center}
\vspace{-1cm}
\noindent {\bf \footnotesize Figure E -- Cohomotopy in degree $n$
of Euclidean $n$-space vanishing at infinity }
{\footnotesize is given by Cohomotopy cocycles \eqref{PlainCohomotopySet} on the one-point compactification $(\mathbb{R}^n) \simeq S^n$ \eqref{SphereIsCompactificationOfCartesianSpace} that send $\infty$ to
$\infty$ \eqref{VanishingAtInfinity}. }

\medskip

Of course, this is just the cohomotopical version of \emph{instantons}
in ordinary gauge theory:

\medskip
\noindent {\bf Instantons and solitons.}
If $G$ is a compact Lie group with
classifying space $B G$ equipped with the canonical point
$\ast \simeq B \{e\} \longrightarrow B G$, then a
\emph{$G$-instanton sector} on Euclidean space $X = \mathbb{R}^n$
is the homotopy class of a continuous function from the
one-point compactification of $X$ to $B G$, which takes the base points
to each other \footnote{
An actual instanton in this instanton sector is a $G$-principal
connection on $X^{\mathrm{cpt}}$ whose underlying $G$-principal bundle has this
classifying map. Ultimately we are interested in such enhancement
to \emph{differential cohomology}, but this is beyond the scope of the
present article.}

\vspace{-2mm}
\begin{equation}
  \label{Instanton}
  \mbox{
    \begin{minipage}[l]{9cm}
      \footnotesize
      A $G$ \emph{instanton sector} is
      a cocycle in degree-1 $G$-cohomology
      which \emph{vanishes at infinity}
      in that it is a cocycle
      on the one-point compactification
      $X^{\mathrm{cpt}}$ \eqref{SphereIsCompactificationOfCartesianSpace} which sends the point
      at infinity in the domain to the base point in the
      classifying space $B G$.
    \end{minipage}
  }
  \phantom{AAA}
  \raisebox{20pt}{
  \xymatrix@R=1.5em{
    \big( \mathbb{R}^n\big)^{\mathrm{cpt}}
    \ar[rr]^-{c}
    &&
    B G
    \\
    \{\infty\}
    \ar@{^{(}->}[u]
    \ar[rr]_-{ c_{\vert_{\{\infty\}}} }
    &&
    B \{e\}
    \ar@{^{(}->}[u]
  }
  }
\end{equation}

\vspace{-3mm}
\noindent {\bf Cohomotopy and $\mathrm{SU}(N)$-instanton sectors.}
Specifically for $n = 4$ and $G = \mathrm{SU}(N)$ any
map $S^4 \overset{\epsilon}{\longrightarrow} B \mathrm{SU}(N)$
representing a generator $1 \in \mathbb{Z} \simeq \pi_4\big( B \mathrm{SU}(N) \big)$
of the 4th homotopy group of the classifying space exhibits a
bijection between the 4-Cohomotopy of $\mathbb{R}^4$ vanishing
at infinity \eqref{VanishingAtInfinity}, and the set of $\mathrm{SU}(N)$ instanton sectors
$$
  \pi^4\big( ( \mathbb{R}^n)^{\mathrm{cpt}}  \big)
  \;=\;
  \big\{
  \xymatrix{
    ( \mathbb{R}^4)^{\mathrm{cpt}}
    \ar[r]
    &
    S^4
  }
  \big\}_{\!\!\big/\sim}
 \;\;
  \overset{\epsilon_\ast}{\simeq}
 \;\;
  \big\{
  \xymatrix{
    ( \mathbb{R}^4)^{\mathrm{cpt}}
    \ar[r]
    &
    B \mathrm{SU}(N)
  }
  \big\}_{\!\!\big/\sim}
  \;\simeq\;
  \left\{ \!\!\!\!\!
  \mbox{ \footnotesize
    \begin{tabular}{c}
      $\mathrm{SU}(n)$-instanton sectors
      \\
      on $\mathbb{R}^4$
    \end{tabular}
  }
   \!\!\!\!\! \right\}.
$$

\vspace{-1mm}
\noindent Under this identification of $\mathrm{SU}(N)$-instanton sectors
with Cohomotopy vanishing at infinity, the Pontrjagin-Thom construction
\eqref{UnstablePTTheorem} produces precisely the distribution
of \emph{instanton center points}, again illustrated by the
left hand side in \hyperlink{FigureE}{\it Figure E}.
To see all this in more detail, we next turn to
further discussion of the charge structure encoded by Cohomotopy.

\subsection{Hopf degree theorem and brane-antibrane annihilation}
\label{NormalFramingAndBraneAntibraneAnnihilation}

{\bf The classical \emph{Hopf degree theorem}} describes
the $n$-Cohomotopy \eqref{PlainCohomotopySet} of orientable closed $D$-manifolds $X$
\eqref{UnstablePTTheorem} in the special case where $n = D$.
It says that, in the ``stable range'' $n \geq 1$, the Cohomotopy set is in bijection with the
set of integers, where the bijection is induced by sending the continuous
function representing a Cohomotopy coycle to its mapping degree (see, e.g., \cite[7.5]{Kobin16}):

\vspace{-3mm}
\begin{equation}
  \label{HopfDegreeTheorem}
  \hspace{-1cm}
  \mbox{\footnotesize
    \begin{tabular}{c}
      \bf
      Hopf degree
      \\
      \bf theorem
      \\
      in stable range $n \geq 1$
    \end{tabular}
  }
  \phantom{AA}
  \raisebox{30pt}{
  \xymatrix@R=-2pt@C=3em{
    \mbox{\bf
      \tiny
      \color{darkblue}
      \begin{tabular}{c}
        $n$-Cohomotopy
        \\
        of $n$-manifold
      \end{tabular}
    }
    \\
    \pi^{n}
    \big(
      X
    \big)
    \ar[rr]^
      {
        S^n \overset{\epsilon_\ast}{\longrightarrow} K(\mathbb{Z}, n )
      }_-{\simeq}
    &&
    H^{n}
    \big(
      X, \mathbb{Z}
    \big)
    \ar@{}[r]|-{\simeq}
    &
    \mathbb{Z}
    \\
    \big[X^n \overset{c}{\longrightarrow}S^n\big]
    \ar@{|->}[rr]
    &&
    \big[X^n \overset{c}{\longrightarrow}S^n \overset{\epsilon}{\longrightarrow}  K(\mathbb{Z},n) \big]
    \ar@{}[r]|-{ \eqqcolon }
    &
    \mathrm{deg}(c)
  }
  }
\end{equation}
Under the Pontrjagin-Thom theorem \eqref{UnstablePTTheorem} the Hopf degree
theorem \eqref{HopfDegreeTheorem} translates into the following geometric situation
for signed (charged) points in $X^n$ (see \cite[IX.4]{Kosinski93}):
A codimension-$n$ submanifold  in an $n$-manifold $X^n$
is a set of points in $X^n$, and a choice of normal framing
\eqref{TrivializationOfNormalVectorBundle}
is, up to normally framed cobordism,
the same as choice of sign (charge) in $\{\pm 1\}$ for each point,
as shown in \hyperlink{FigureF}{\it Figure F}:

\vspace{-4mm}
\begin{center}
\hypertarget{FigureF}{}
\begin{tikzpicture}[scale=0.75]
  \draw (1.5,6.7) node {$\overbrace{\phantom{------------------------}}$};
  \draw (11,6.7) node {$\overbrace{\phantom{---------------}}$};
  \draw (1.5,7.3) node {\tiny \color{darkblue} \bf manifold};
  \draw (11,7.4)
    node
    {
      \tiny
      \color{darkblue} \bf
      \begin{tabular}{c}
        sphere
        \\
        Cohomotopy coefficient
      \end{tabular}
    };
  \draw (5.5,7.6) node { \tiny \color{darkblue} \bf Cohomotopy cocycle  };
 \begin{scope}[shift={(0, 0.6)}]
 \begin{scope}
  \clip (-2.9,-2.9) rectangle (5.9,5.9);
  \draw[step=3, dotted] (-3,-3) grid (6,6);

  \draw[dashed] (0+2.3+.3,4.9) circle (1.8);
  \draw (0+2.3+.2-.8,4.9-1.2)
    node
    {
      {
      \tiny \color{darkblue} \bf
      \begin{tabular}{c}
        tubular
        \\
        neighborhood
      \end{tabular}
      }
    };

  \draw[dashed] (0+2.3+.3,-2.4) circle (1.8);
  \draw (0+2.3+.2-.6,-2.4+1.2)
    node
    {
      {
      \tiny \color{darkblue} \bf
      \begin{tabular}{c}
        tubular
        \\
        neighborhood
      \end{tabular}
      }
    };

  \end{scope}

  \node at (11,2) {\colorbox{white}{$\phantom{a}$}};

  \draw[dashed] (11-.1,2) circle (2);
  \node (zero) at (11,2) {\tiny $0$};
  \node (infinity) at (11-.1,2+2.1) {\tiny $\infty$};
  \node (leftinfinity) at (11-.1-2.2,2) {\tiny $\infty$};
  \node (bottominfinity) at (11-.1,2-2.1) {\tiny $\infty$};
  \node (rightinfinity) at (11-.1+2.2,2) {\tiny $\infty$};
  \node (torus) at (1.5,7.35)
    {\raisebox{0pt}{\small $
      X^n
    $}};
  \node (sphere) at (11,7.35)
    {\raisebox{0pt}{\small $
      S^{n}
        =
      D(
        \mathbb{R}^{n}
      )/S(\mathbb{R}^{n})
    $}};
    \draw[->, thin] (torus) to node[above]{\footnotesize $c$} (sphere);

   \begin{scope}[shift={(.13,0)}]

   \draw[fill=black] (0+2.3,4.9) circle (.07);

   \draw[|->, olive] (0+2.3+.2+.55,4.9-.05) to (11-.2+.55,2+.05);
   \draw[|->, olive] (0+2.3+.2+1.1,4.9-.05) to (11-.2+1.1,2+.05);
   \draw[|->, olive] (0+2.3+.2+1.65,4.9-.05) to (11-.2+1.65,2+.05);

   \draw[|->, olive] (0+2.3+.2-.55,4.9-.05) to (11-.2-.55,2+.05);
   \draw[|->, olive] (0+2.3+.2-1.1,4.9-.05) to (11-.2-1.1,2+.05);
   \draw[|->, olive] (0+2.3+.2-1.65,4.9-.05) to (11-.2-1.65,2+.05);

   \draw[|->, olive]
     (0+2.3+.2,4.9-.05)
     to
     node {\colorbox{white}{\tiny \color{darkblue} \bf positively charged submanifold} }
     (11-.2,2+.05);

   \draw[fill=white] (0+2.3,-2.4) circle (.07);

   \draw[|->, olive] (0+2.3+.2+.55,-2.4+.05) to (11-.2-.55,2-.05);
   \draw[|->, olive] (0+2.3+.2+1.1,-2.4+.05) to (11-.2-1.1,2-.05);
   \draw[|->, olive] (0+2.3+.2+1.65,-2.4+.05) to (11-.2-1.65,2-.05);

   \draw[|->, olive] (0+2.3+.2-.55,-2.4+.05) to (11-.2+.55,2-.05);
   \draw[|->, olive] (0+2.3+.2-1.1,-2.4+.05) to (11-.2+1.1,2-.05);
   \draw[|->, olive] (0+2.3+.2-1.65,-2.4+.05) to (11-.2+1.65,2-.05);

   \draw[|->, olive]
     (0+2.3+.2,-2.4+.05)
     to
     node {\colorbox{white}{\tiny \color{darkblue} \bf negatively charged submanifold} }
     (11-.2,2-.05);

   \end{scope}

   \begin{scope}[shift={(-.04,.2)}]
   \draw[fill=black] (0.4,1.4) circle (.07);
   \draw[fill=black] (0.6,1.05) circle (.07);

   \draw[fill=white] (0.4,1.2) circle (.07);
   \draw[fill=white] (0.7,1.3) circle (.07);

   \draw[dashed] (0.57,1.2) circle (.6);

   \draw (1.55,1.2) node {$\simeq$};
  \draw (1.55,1.2-.9)
    node
    {
      \tiny
      \color{darkblue} \bf
      \begin{tabular}{c}
        opposite charges
        \\
        cancel each other
      \end{tabular}
    };

   \draw[dashed] (2.57,1.2) circle (.6);

   \end{scope}

   \draw[|->, olive]
     (4,1.4+.3)
     to[bend right=6]
     (leftinfinity);
   \draw[|->, olive]
     (4,1.4-.3)
     to[bend right=6]
     (leftinfinity);
   \draw[|->, olive]
     (4,1.4)
     to[bend right=6]
     node
       {
         \colorbox{white}{
         \hspace{-.3cm}
         \tiny
         \color{darkblue} \bf
         \begin{tabular}{c}
           no charge here
         \end{tabular}
         \hspace{-.3cm}
         }
       }
     (leftinfinity);

\end{scope}

\end{tikzpicture}
\end{center}
\vspace{-.4cm}
\noindent {\bf \footnotesize Figure F -- Charge in Cohomotopy
carried by submanifolds, under the PT-isomorphism \eqref{UnstablePTTheorem}}
{\footnotesize is encoded in their normal framing \eqref{TrivializationOfNormalVectorBundle}. In full codimension
the normal framing is a normal orientation and hence a choice in $\{\pm 1\}$, which we indicate graphically by
$
\renewcommand{\arraystretch}{.4}
\begin{array}{ccc}
  \bullet  &\leftrightarrow& -1
  \\
  \circ &\leftrightarrow& +1
\end{array}
$ }

\vspace{2mm}
\noindent Under this geometric translation, we have the correspondence
\vspace{-2mm}
$$
  \xymatrix{
    \mbox{\footnotesize
      \begin{tabular}{c}
        Hopf degree
        \\
        of Cohomotopy cocycle on $X$
      \end{tabular}
    }
    \ar@{<->}[rr]^-{\mbox{\tiny PT}}
    &&
    \mbox{\footnotesize
      \begin{tabular}{c}
      Net number of $\pm$-charges
      \\
      carried by points in $X$
      \end{tabular}
    }
  }
$$

\vspace{-2mm}
\noindent The mechanism which implements this on the geometric right hand side is that
points of opposite sign/normal framing are cobordant to the empty
collection of points, hence mutually annihilate each other
via coboundaries in Cohomotopy,
as shown in \hyperlink{FigureG}{\it Figure G}:

\vspace{-3mm}
\begin{center}
\hypertarget{FigureG}{}
\begin{tikzpicture}[scale=.8]

  \begin{scope}[shift={(0,-.7)}]

  \node (X) at (-4.5,6)  {\small $[0,1] \times X $};
  \node (sphere) at (6,6)  {\small $S^n = (\mathbb{R}^n)^{\mathrm{cpt}}$};

  \draw[->] (X)
    to
    node[above]
    {
      \tiny
      $0 \simeq (-1) + (+1)$
    }
    (sphere);

  \node at (-4.5,5.4)
    {
      \tiny
      \color{darkblue} \bf
      \begin{tabular}{c}
        product space
        \\
        of interval
                with manifold
      \end{tabular}
    };
  \node at (-4.5, 4.7)  {$\overbrace{\phantom{----------------------}}$};

  \node at (6,5.4)  {\tiny \color{darkblue} \bf \begin{tabular}{c} $n$-sphere \\ Cohomotopy coefficient \end{tabular}};
  \node at (6,4.7)  {$\overbrace{\phantom{--------------}}$};

  \node at (.5,5.7)  {\tiny \color{darkblue} \bf Cohomotopy coboundary};

  \end{scope}

  \begin{scope}[shift={(-1.5,0.0)}]

  \begin{scope}[shift={(-6,-1.5)}]
    \clip (0,-2.9) rectangle (6,5.9);
    \draw[step=3, dotted] (-3,-2.5) grid (6,5.5);
  \end{scope}

  \begin{scope}[shift={(-6,-1.5)}]
    \draw (6,-2.8) node {\tiny $\{1\} \times X $  };
    \draw (0,-2.8) node {\tiny $\{0\} \times X $  };
  \end{scope}

 \begin{scope}[rotate=-90]

 \draw[very thick] (-3,0) .. controls (-3,-2.5-.85) and (3,-2.5-.85) .. (+3,0);

 \draw[dashed, thick] (-3+.4,0) .. controls (-3+.4,-2.5+.5-.85) and (3-.4,-2.5+.5-.85) .. (+3-.4,0);
 \draw[dashed, thick] (-3-.4,0) .. controls (-3-.4,-2.5-.5-.85) and (3+.4,-2.5-.5-.85) .. (+3+.4,0);

 \draw[dashed] (-3+.4+.4,0) .. controls (-3+.4+.4,-2.5+.5+.5-.85) and (3-.4-.4,-2.5+.5+.5-.85) .. (+3-.4-.4,0);
 \draw[dashed] (-3-.4-.4,0) .. controls (-3-.4-.4,-2.5-.5-.5-.85) and (3+.4+.4,-2.5-.5-.5-.85) .. (+3+.4+.4,0);

 \end{scope}

  \end{scope}

  \begin{scope}[shift={(4,0)}]
    \draw (2,0) circle (2);
    \node at (+.6,0)
      {{\tiny $0$}
        \raisebox{.0cm}{
          $
          \mathrlap{
          \!\!\!\!\!\!\!
          \mbox{
          \tiny \color{darkblue} \bf
            \begin{tabular}{c}
              regular
              \\
              value
            \end{tabular}
          }}
          $
        }
      };
    \node (zero) at (0,0) {$-$};
    \node (infinity) at (4,0) {\colorbox{white}{$\infty$}};

   \fill[black] (2,0) ++(40+180:2) node (minusepsilon)
     {\begin{turn}{-45} $)$  \end{turn}};
   \fill[black] (2,0) ++(180-40:2) node (epsilon)
     {\begin{turn}{45} $)$ \end{turn}};
   \fill[black] (2.3,0.25) ++(40+180:2) node
     { \tiny $-\epsilon$ };
   \fill[black] (2.3,-0.25) ++(-40-180:2) node
     { \tiny $+\epsilon$ };
  \end{scope}

  \draw[|->, olive] (-1.35,2.9) to (3.8,0+0.05);
  \draw[|->, olive] (-1.35,2.9+.4) to (3.9,0+.4+0.05);
  \draw[|->, olive] (-1.35,2.9-.4) to (3.9,0-.4+0.05);

  \draw[|->, olive] (-1.35,3+.8) to[bend right=-17] (7.7,0.15);

  \draw[|->, olive] (-1.35,-2.9) to (3.8-.0,0-0.05);
  \draw[|->, olive] (-1.35,-2.9-.4) to (3.9-.0,0+.4-0.05);
  \draw[|->, olive] (-1.35,-2.9+.4) to (3.9-.0,0-.4-0.05);

  \draw[|->, olive] (-1.35,-3-.8) to[bend left=-17] (7.7,-0.15);

  \draw[fill=black] (-1.5,3) circle (.07);

  \draw[fill=white] (-1.5,-3) circle (.07);

  \draw (-1.5,3)+(1.4,-.2)
    node
    {
      \colorbox{white}{
      \hspace{-.4cm}
      \tiny
      \color{darkblue} \bf
      \begin{tabular}{c}
        positively charged
        \\
        submanifold
      \end{tabular}
      \hspace{-.4cm}
      }
    };

  \draw (-1.5,-3)+(1.4,+.2)
    node
    {
      \colorbox{white}{
      \hspace{-.4cm}
      \tiny
      \color{darkblue} \bf
      \begin{tabular}{c}
        negatively charged
        \\
        submanifold
      \end{tabular}
      \hspace{-.4cm}
      }
    };

  \draw (-8,0)
    node
    {
      \colorbox{white}{
      \hspace{-.4cm}
      \tiny
      \color{darkblue} \bf
      \begin{tabular}{c}
        no
        \\
        submanifold
      \end{tabular}
      \hspace{-.4cm}
      }
    };

  \draw (-3.8,0) node
    {
      \colorbox{white}
      {
        \hspace{-.3cm}
        \tiny
        \color{darkblue} \bf
        cobordism
        \hspace{-.3cm}
      }
    };

\end{tikzpicture}
\end{center}
\vspace{-5mm}
\noindent  {\footnotesize \bf Figure G -- Cobordisms between submanifolds of opposite normal framing}
{\footnotesize as in \hyperlink{FigureF}{\it Figure F} exhibit their pair creation/annihilation.
This is the geometric mechanism which underlies the
Hopf degree theorem \eqref{HopfDegreeTheorem} when translating
via the Pontrjagin-Thom theorem \eqref{UnstablePTTheorem}
between Cohomotopy charge and the submanifolds sourcing it,
as in \hyperlink{FigureD}{\it Figure D}.
}

\medskip

\noindent {\bf Hopf degree in unstable range.}
The classical Hopf degree theorem \eqref{HopfDegreeTheorem}
is stated only in the stable range $n \geq 1$, but it is immediate
to extend it to the unstable range.
While this is a simple statement in itself,
it is necessary to conceptually complete the discussion of the
equivariant Hopf degree theorem in \cref{LocalTadpoleCancellation} below,
where the ordinary Hopf degree appears jointly in stable and unstable range,
with the distinction being responsible for the difference in
nature between O-plane charge (unstable range) and D-brane charge
(stable range):
For $X = X^0$ a compact 0-manifold, hence a finite set,
and $X^{\mathrm{cpt}} =  X_+ = X \sqcup \{\infty\}$
the same set with a ``point at infinity'' adjoined
\eqref{BasepointFreelyAdjoined}, its
unstable Cohomotopy classes \eqref{PlainCohomotopySet} in degree 0,
being functions to the 0-sphere
hence to the 2-element set $S^0 = \{0,\infty\}$
that take $\infty \mapsto \infty$
$$
  \pi^0\big( X^{\mathrm{cpt}} \big)
  \;=\;
  \big\{
    X
      \xrightarrow{ \;\;\; c \;\;\; }
    S^0
  \big\},
$$
are in bijection
to the subsets $S \subset X$ of $X$, by the assignment that
sends $c$ to the pre-image
$c^{-1}\big( \{0\}\big)$
of $0 \in S^0$ under $c$.
We may think of these subsets as elements of
the power set $\{0,1\}^X$ and as such call them the
sets $\mathrm{deg}(c)$ of Hopf degrees in $\{0,1\}$ for $n = 0$:
\vspace{-3mm}
\begin{equation}
  \label{UnstableRangeHopfDegreeTheorem}
  \hspace{-1cm}
  \mbox{\footnotesize
    \begin{tabular}{c}
      \bf
      Hopf degree
      \\
      \bf theorem
      \\
      in unstable range $n = 0 $
    \end{tabular}
  }
  \phantom{AA}
  \raisebox{30pt}{
  \xymatrix@R=-2pt@C=3em{
    \mbox{
      \tiny
      \color{darkblue} \bf
      \begin{tabular}{c}
        $0$-Cohomotopy
        \\
        of $0$-manifold
      \end{tabular}
    }
    && &
      \mathclap{
      \mbox{
        \tiny
        \color{darkblue} \bf
        \begin{tabular}{c}
         sets of
         \\
         unstable Hopf degrees
        \end{tabular}
      }
      }
    \\
    \pi^{0}
    \big(
      X^{\mathrm{cpt}}
    \big)
    \ar[rr]^
      {
        S^0 = \{0,\infty\}
      }_-{\simeq}
    &&
    \mathrm{Subsets}(X)
    \ar@{}[r]|-{\simeq}
    &
    \{0,1\}^{X}
    \\
    \big[X^0 \overset{c}{\longrightarrow}S^0\big]
    \ar@{|->}[rr]
    &&
    \big[ c^{-1}\big( \{0\}\big) \subset X  \big]
    \ar@{}[r]|-{ \eqqcolon }
    &
    \mathrm{deg}(c)
  }
  }
\end{equation}

\begin{example}
For $X = \{0\}$ the single point
so that, with \eqref{OSphereAsCompactificationOfPoint},
$X^{\mathrm{cpt}}$ is the 0-sphere, we have
$
  \pi^0\big( \{0\}^{\mathrm{cpt}} \big)
  \simeq
  \{0,1\}
  $,
as illustrated in the following figure:
\end{example}

\vspace{-.5cm}

\begin{center}
\hyperlink{FigureH}{}
\begin{tikzpicture}[scale=1.2]

  \begin{scope}[shift={(6,0)}]

  \begin{scope}

  \draw (0,2.3)
    node
    {$
      \big( \mathbb{R}^0\big)^{\mathrm{cpt}}
    $};

  \draw (0,1.7) node {$\overbrace{\phantom{--}}$};

  \draw
    (-.01,1.3) to (.01,1.3);
  \draw (-.2,1.3) node {\tiny $\infty$};

  \draw (0,.5) circle (.07);
  \draw (-.2,.5) node {\tiny $0$};

  \end{scope}

  \draw[|->, olive]
   (0.1,1.3)
   to
   node
   {\colorbox{white}{\tiny \color{darkblue} \bf vanishing at infinity}}
   (3-.3,1.3);

  \draw[|->, olive]
    (0.1,.5)
    to
    node
    {\colorbox{white}{ \tiny \color{darkblue} \bf charge }}
    (3-.3,.5);

  \draw[->] (.7,2.25) to node[above]{\small $c = 1$}  (3-.3,2.25);

  \begin{scope}[shift={(3,0)}]

  \draw (0,2.3)
    node
    {$
      S^0
    $};

  \draw (0,1.7) node {$\overbrace{\phantom{--}}$};

  \draw
    (-.01,1.3) to (.01,1.3);
  \draw (-.14,1.3) node {\tiny $\infty$};

  \draw
    (-.01,.5) to (.01,.5);
  \draw (-.14,.5) node {\tiny $0$};

  \end{scope}

  \end{scope}

  \begin{scope}

  \begin{scope}

  \draw (0,2.3)
    node
    {$
      \big( \mathbb{R}^0\big)^{\mathrm{cpt}}
    $};

  \draw (0,1.7) node {$\overbrace{\phantom{--}}$};

  \draw
    (-.01,1.3) to (.01,1.3);
  \draw (-.2,1.3) node {\tiny $\infty$};

  \draw
    (-.01,.5) to (.01,.5);
  \draw (-.2,.5) node {\tiny $0$};

  \end{scope}

  \draw[|->, olive]
   (0.1,1.3)
   to
   node
   {\colorbox{white}{\tiny \color{darkblue} \bf vanishing at infinity}}
   (3-.3,1.3);

  \draw[|->, olive]
    (0.1,.5)
    to
    node
      {\colorbox{white}{\tiny \color{darkblue} \bf no charge}}
    (3-.3,1.3);

  \draw[->] (.7,2.25) to node[above]{\small $c = 0$}  (3-.3,2.25);

  \begin{scope}[shift={(3,0)}]

  \draw (0,2.3)
    node
    {$
      S^0
    $};

  \draw (0,1.7) node {$\overbrace{\phantom{--}}$};

  \draw
    (-.01,1.3) to (.01,1.3);
  \draw (-.14,1.3) node {\tiny $\infty$};

  \draw
    (-.01,.5) to (.01,.5);
  \draw (-.14,.5) node {\tiny $0$};

  \end{scope}

  \end{scope}

\end{tikzpicture}
\end{center}

\vspace{-.5cm}
\noindent {\footnotesize \bf Figure H -- Hopf degree in the unstable range}
{\footnotesize takes values in the set $\{0,1\}$ \eqref{UnstableRangeHopfDegreeTheorem}, corresponding to the
binary choice of there being or not being a unit charge at the single
point.}

\medskip
The point of unstable Hopf degree in $\{0,1\}$ is that it exhibits
{\it homogeneous behavior under suspension} $\Sigma^1$ \eqref{EquSuspension}
across the unstable and stable range of Hopf degrees,
with the unstable Hopf degrees in $\{0,1\}$ injecting into
the full set of integers in the stable range:
\begin{equation}
  \label{HopfDegreesUnderSuspension}
  \raisebox{20pt}{
  \xymatrix@R=13pt{
    \overset{
      \mathclap{
      \mbox{
        \tiny
        \color{darkblue} \bf
        \begin{tabular}{c}
          unstable
          \\
          Hopf degrees
        \end{tabular}
      }
      }
    }
    {
      \{0,1\}
    }
    \ar@{}[d]|-{
      \mathllap{
      \mbox{ \tiny \color{darkblue} \bf  \eqref{UnstableRangeHopfDegreeTheorem} }
      \;
      }
      \begin{rotate}{270}
        $\!\!\!\!\!\simeq$
      \end{rotate}
    }
    \ar@{^{(}->}[rr]^-{ \mbox{ \tiny \color{darkblue} \bf injection } }
    &&
    \overset{
      \mathclap{
      \mbox{
        \tiny
        \color{darkblue} \bf
        \begin{tabular}{c}
          stable
          \\
          Hopf degrees
        \end{tabular}
      }
      }
    }{
      \mathbb{Z}
    }
    \ar@{}[d]|-{
      \mathllap{
      \mbox{ \tiny \color{darkblue} \bf \eqref{HopfDegreeTheorem} }
      \;
      }
      \begin{rotate}{270}
        $\!\!\!\!\!\simeq$
      \end{rotate}
    }
    \ar[rr]^-{=}
    &&
    \overset{
      \mathclap{
      \mbox{ \bf
        \tiny
        \color{darkblue}
        \begin{tabular}{c}
          stable
          \\
          Hopf degrees
        \end{tabular}
      }
      }
    }{
      \mathbb{Z}
    }
    \ar@{}[d]|-{
      \mathllap{
      \mbox{ \tiny \color{darkblue} \bf \eqref{HopfDegreeTheorem} }
      \;
      }
      \begin{rotate}{270}
        $\!\!\!\!\!\simeq$
      \end{rotate}
    }
    \ar[rr]^-{=}
    &&
    \overset{
      \mathclap{
      \mbox{
        \tiny
        \color{darkblue} \bf
        \begin{tabular}{c}
          stable
          \\
          Hopf degrees
        \end{tabular}
      }
      }
    }{
      \mathbb{Z}
    }
    \ar@{}[d]|-{
      \mathllap{
      \mbox{ \tiny \color{darkblue} \bf \eqref{HopfDegreeTheorem} }
      \;
      }
      \begin{rotate}{270}
        $\!\!\!\!\!\simeq$
      \end{rotate}
    }
    \ar[r]
    &
    \cdots
    \ar@{}[d]|-{\vdots}
    \\
    \pi^0\big( S^0 \big)
    \ar[rr]_-{ \Sigma^1 }^-{ \mbox{\tiny \color{darkblue} \bf suspension } }
    &&
    \pi^{1}\big( S^1 \big)
    \ar[rr]_-{ \Sigma^1 }
    &&
    \pi^{2}\big( S^2 \big)
    \ar[rr]_-{ \Sigma^1 }
    &&
    \pi^{3}\big( S^3 \big)
    \ar[r]
    &
    \cdots
  }
  }
\end{equation}

As we next turn from plain to equivariant Cohomotopy in \cref{EquivariantCohomotopyAndTadpoleCancellation},
we find that unstable and stable Hopf degrees unify in the
equivariant Hopf degree theorems, and that
the {\it D-brane charge is what appears in the stable range},
while the {\it O-plane charge is what appears in the unstable range}
(in particular, via the proof of Theorem \ref{CharacterizationOfStabilizationOfUnstableCohomotopy} below).

\section{Equivariant Cohomotopy and tadpole cancellation}
\label{EquivariantCohomotopyAndTadpoleCancellation}

We now turn to the equivariant enhancement \eqref{EquivariantCohomotopySet} of Cohomotopy theory.
We discuss in \cref{LocalTadpoleCancellation}
and in \cref{GlobalTadpoleCancellation}, respectively,
how this captures the form of the
local/twisted (see Diagram \eqref{KernelOfTheGlobalElmendorfStageProjection} in \cref{GlobalTadpoleCancellation})
and of the global/untwisted
tadpole cancellation conditions (see \cref{HeteroticMTheoryOnADEOrbifolds})
according to \hyperlink{Table1}{\it Table 1}
and \hyperlink{Table2}{\it Table 2},
by appeal to the equivariant enhancement
of the Hopf degree theorem applied to representation spheres,
which we state as
Theorem \ref{UnstableEquivariantHopfDegreeTheorem} and
Theorem \ref{CharacterizationOfStabilizationOfUnstableCohomotopy}.

\medskip

\noindent {\bf Basic concepts of unstable equivariant homotopy theory.}
To set up notation,
we start with reviewing a minimum of underlying concepts from
unstable equivariant homotopy theory (see \cite[1]{Blu17}\cite[3.1]{ADE} for more).

\noindent {\it Topological $G$-spaces.}
For $G$ a finite group, a \emph{topological $G$-space}
$\xymatrix{ X
\ar@(ul,ur)|-{\,G\,}}
$
(or just \emph{$G$-space}, for short) is
a topological space $X$ equipped with a continuous $G$-action,
hence with a continuous function
${G \times X \xrightarrow{\cdot}   X}$
such that for all $g_i \in G$ and $x \in X$
we have $g_1 \cdot (g_2 \cdot x) = (g_1 g_2) \cdot x$ and
$e \cdot x = x$ (where $e \in G$ is the neutral element).

\medskip

Here we are concerned with the
{\bf classes of examples of $G$-spaces}
shown in \hyperlink{Table5}{\it Table 5}:\footnote{
  For our purposes here, the covering $G$-space $X$
  is all we need to speak about the corresponding
  orbifold $X \!\sslash\! G$. For a dedicated
  discussion of geometric orbifolds we refer to
  \cite[13]{Ratcliffe06}\cite{OrbifoldCohomology}.
  Note that \cite[13]{Ratcliffe06} says ``Euclidean orbifold''
  for any flat orbifold.
}
{\small
\begin{center}
\hypertarget{Table5}{}
\begin{tabular}{|c||c|c||c|c|}
  \hline
  {\bf $G$-representation}
  &
  {\bf $G$-space}
  &
  {\bf $G$-orbifold}
  &
  \multicolumn{2}{c|}{
    {\bf Terminology}
  }
  \\
  \hline
  \hline
  \multirow{2}{*}{
    \begin{tabular}{c}
    $
      \underset
      {
        \mbox{
          \tiny
          \color{darkblue} \bf
          \begin{tabular}{c}
            finite group
          \end{tabular}
        }
      }
      {
        G
      }
    $
    \\
    $\phantom{-}$
    \\
    $
      \underset
      {
        \mbox{
          \tiny
          \color{darkblue} \bf
          \begin{tabular}{c}
            orthogonal linear
            \\
            $G$-representation
          \end{tabular}
        }
      }
      {
        V \;\in\; \mathrm{RO}(G)
      }
    $
    \end{tabular}
  }
  &
      $
    \underset{
      \mbox{ \bf
        \tiny
        \color{darkblue}
        \begin{tabular}{c}
          Euclidean
          \\
          $G$-space \eqref{EuclideanGSpace}
        \end{tabular}
      }
    }{
    \xymatrix{
      \mathbb{R}^V
      \ar@(ul,ur)|{\,G\,}
    }
    }
  $
  &
  $
    \underset{
      \mbox{\bf
        \tiny
        \color{darkblue}
        \begin{tabular}{c}
          Euclidean orbifold
        \end{tabular}
      }
    }{
      \mathbb{R}^V \!\sslash\! G
    }
  $
  &
  \begin{tabular}{l}
    {\bf singularity}
  \end{tabular}
  &
  \multirow{2}{*}{
    \begin{tabular}{c}
      {\bf ADE-singularities}
      \\
      \\
      $
      \underset{
        \mbox{\bf
          \tiny
          \color{darkblue}
          \begin{tabular}{c}
            finite subgroup of $\mathrm{SU}(2)$
            \\
            \eqref{ADESubgroups}
          \end{tabular}
        }
      }{
        G \subset \mathrm{SU}(2)
      }
      $
      \\
      $\phantom{-}$
      \\
      $
      \underset{
        \mbox{\bf
          \tiny
          \color{darkblue}
          \begin{tabular}{c}
            quaternionic
            representation
            \\
            \eqref{TheQuaternionicRepresentation}
          \end{tabular}
        }
      }{
        V = \mathbf{4}_{\mathbb{H}}
      }
      $
    \end{tabular}
  }
  \\
  \cline{2-4}
  &
  $
    \underset{
      \mbox{\bf
        \tiny
        \color{darkblue}
        \begin{tabular}{c}
          $G$-representation
          \\
          sphere \eqref{RepSpheres}
        \end{tabular}
      }
    }{
    \xymatrix{
      S^V
      \ar@(ul,ur)|{\,G\,}
    }
    }
  $
  &
  $
    \underset{
      \!\!\!\!\!\!\!\!\!\!
      \mbox{\bf
        \tiny
        \color{darkblue}
        \begin{tabular}{c}
          Euclidean orbifold
          \\
          including point at infinity \eqref{VanishingAtInfinity}
        \end{tabular}
      }
    }{
      S^V \!\sslash\! G
      =
      \big( \mathbb{R}^V \!\sslash\! G \big)^{\mathrm{cpt}}
    }
  $
  &
  \begin{tabular}{l}
    {\bf vicinity of}
    \\
    {\bf singularity}
  \end{tabular}
  &
  \\
  \cline{2-4}
    $
      \underset
      {
        \mbox{\bf
          \tiny
          \color{darkblue}
          \begin{tabular}{c}
            crystallographic group \eqref{CrystallographicGroups})
          \end{tabular}
        }
      }
      {
        G \rtimes \mathbb{Z}^{\mathrm{dim}(V)}
        \subset
        \mathrm{Iso}\big( \mathbb{R}^V \big)
      }
    $
  &
  $
    \underset{
      \mbox{\bf
        \tiny
        \color{darkblue}
        \begin{tabular}{c}
          $G$-representation
          \\
          torus \eqref{RepresentationTorus}
        \end{tabular}
      }
    }{
    \xymatrix{
      \mathbb{T}^V
      \ar@(ul,ur)|{\,G\,}
    }
    }
  $
  &
  $
    \underset{
      \!\!\!\!\!\!\!\!\!\!\!\!\!\!\!\!\!\!\!\!\!\!\!\!\!\!\!\!\!\!
      \mbox{\bf
        \tiny
        \color{darkblue}
        \begin{tabular}{c}
          toroidal orbifold
        \end{tabular}
      }
    }{
      \mathbb{T}^V \!\sslash\! G
      =
      \big( \mathbb{R}^V \!\sslash\! G \big)/\mathbb{Z}^{\mathrm{dim}(V)}
    }
  $
  &
  \begin{tabular}{l}
    {\bf flat, compact}
    \\
    {\bf singular space}
  \end{tabular}
  &
  \\
  \hline
\end{tabular}
\end{center}

\vspace{-.2cm}
\noindent
{\bf \footnotesize Table 5 -- Flat $G$-orbifolds and the $G$-spaces covering them.} \
{\footnotesize Examples arising in application to
M-theory are discussed in \cref{HeteroticMTheoryOnADEOrbifolds}.}
}

\medskip

\noindent {\bf Orbifold terminology.} As common in string theory, we will
be thinking of $G$-spaces $X$ as stand-ins for their homotopy quotients
$X \!\sslash\! G$, which are the actual orbifolds. This is mathematically
fully justified by the fact that
the proper notion of generalized cohomology of such global quotient orbifolds
$X \!\sslash\! G$ is equivalently the $G$-equivariant
generalized cohomology of the space $X$.
We relegate a comprehensive discussion
of this technical point to \cite{OrbifoldCohomology},
but this is mathematical folklore:
see \cite[\S 1]{PronkScull10}\cite[p. 1]{Schwede17}\cite[p. ix-x]{Schwede18}.
Moreover, in the specific application to M-theory,
below in \cref{M5MO5AnomalyCancellation}, the relevant orbifolds
are always part of \emph{orbi-orientifolds}, in that a subgroup $\mathbb{Z}_2^{\mathrm{refl}}$
of the orbifold quotient group $G = G^{\mathrm{ADE}}$ combines
with a Ho{\v r}ava-Witten-involution $\mathbb{Z}_2^{\mathrm{HW}}$
to an orientation-changing involution $\mathbb{Z}_2^{\mathrm{HW} + \mathrm{refl}}$ which fixes an ``$\mathrm{MO5}$-plane''.
This is made precise in \cref{HeteroticMTheoryOnADEOrbifolds}
below; see \eqref{OrbiOrientifoldGroupSequence} there.
Since, with passage to the $\mathbb{Z}_2^{\mathrm{HW}}$-fixed
locus (the ``$\mathrm{MO9}$-pane'') understood \eqref{SemiComplement},
the further localization to the $\mathrm{MO5}$-plane coincides with the
orbifold singularity, we will often refer here
to orbifold fixed points as orientifold fixed points,
wherever this serves the preparation of the application in \cref{M5MO5AnomalyCancellation}.
Accordingly, the orbifold singularities in the applications below
in \cref{M5MO5AnomalyCancellation} are always inside an O-plane,
so that the relevant flavor of equivariant K-theory 
considered below in Prop. \ref{TheoremLocalTadpoleCancellation} and 
in \hyperlink{FigureP}{\it Figure P}, \hyperlink{FigureM}{\it Figure M}
is $\mathrm{KO}$.

\medskip

\noindent {\bf Linear $G$-representations.}
The $G$-spaces of interest for the discussion of toroidal orbifolds all come from
{\it orthogonal linear $G$-representations} $V$: finite-dimensional Euclidean vector spaces
equipped with a linear action by $G$ factoring through the canonical action of the orthogonal
group. We will denote concrete examples of such $V$ of dimension $n \in \mathbb{N}$
and characterized by some label ``$\mathrm{l}$'' in the form
$V = \mathbf{n}_{\mathrm{l}}$, and also refer to them as an
\emph{RO-degree} \eqref{RODegree}.

\medskip
 The {\it key class of examples} of interest here are finite subgroups
(see, e.g., \cite[A.1]{SS19b})
\begin{equation}
  \label{ADESubgroups}
  G^{\mathrm{ADE}}
  \;\subset\;
  \mathrm{SU}(2)
  \;\simeq\;
  \mathrm{Sp}(1)
  \;\simeq\;
  U(1,\mathbb{H})
  \;\simeq\;
  S(\mathbb{H})
\end{equation}
of the multiplicative group of unit norm elements
$q \in S(\mathbb{H})$
in the vector space $\mathbb{H} \simeq_{{}_{\mathbb{R}}} \mathbb{R}^4$ of quaternions,
and their defining 4-dimensional linear representation on this
space (by left quaternion multiplication), which we denote by
\begin{equation}
  \label{TheQuaternionicRepresentation}
  \mathbf{4}_{\mathbb{H}}
  \;\in\;
  \mathrm{RO}\big( G^{\mathrm{ADE}} \big)
  \,.
\end{equation}
All of these, except the cyclic groups of odd order, contain the
subgroup
\begin{equation}
  \label{PointReflectionSubgroup}
  \mathbb{Z}_2^{\mathrlap{\mathrm{refl}}}
  \;\;\;\coloneqq\;
  \big\langle -1 \in S(\mathbb{H})\big\rangle
  \;\subset\;
  G^{\mathrm{A}_{\mathrm{ev}}\mathrm{DE}}
\end{equation}
generated by the quaternion $-1 \in \mathbb{H}$. This
acts on the 4-dimensional quaternionic representation
\eqref{TheQuaternionicRepresentation} by point reflection at the
origin, hence as the 4-dimensional sign representation
$$
  \xymatrix{
    \mathbb{R}^{\mathbf{4}_{\mathbb{H}}}
    \ar@(ul,ur)|-{\; \mathbb{Z}_2^{\mathrm{refl}}\!\!\!\! }
  }
  \;\simeq\;
  \xymatrix{
    \mathbb{R}^{\mathbf{4}_{\mathrm{sgn}}}
    \ar@(ul,ur)|-{ \; \mathbb{Z}_2 }
  }
  \,,
$$
as illustrated for 2 of 4 dimensions in \hyperlink{FigureI}{\it Figure I}.

\medskip

\noindent {\bf Euclidean $G$-Spaces.}
The underlying Euclidean space of a linear $G$-representation $V$ is of course a $G$-space, hence
a {\it Euclidean $G$-space}, which we
suggestively denote by $\mathbb{R}^V$:

\vspace{-.5cm}

\begin{equation}
  \label{EuclideanGSpace}
  \mbox{\tiny \color{darkblue} \bf linear $G$-representation }
  \;\;\;\;
    V \in \mathrm{RO}(G)
  \;\;\;\;\;\; \Rightarrow  \;\;\;\;\;\;
  \xymatrix{
    \mathbb{R}^V\ar@(ul,ur)|{\,G\,}
  }
  \;\;\;\;
  \mbox{ \tiny \color{darkblue} \bf Euclidean $G$-space }
\end{equation}

\begin{example}[{\hyperlink{FigureI}{\it Figure I}}]
With $G = \mathbb{Z}_2$ and $V = \mathbf{2}_{\mathrm{sgn}}$
its 2-dimensional sign representation, the Euclidean $G$-spaces
$\mathbb{R}^{\mathbf{2}_{\mathrm{sgn}}}$ is the Cartesian plane
equipped with the action of point reflection at the origin:
\end{example}

\begin{center}
\hypertarget{FigureI}{}
\begin{tikzpicture}[scale=0.5]

  \draw
    (-16,0)
    node
    {\footnotesize
      \begin{minipage}[l]{6.4cm}
        {  \bf Figure I -- The Euclidean $\mathbb{Z}_2$-space}
        \eqref{EuclideanGSpace}
        of the 2-dimensional sign representation
        $\mathbf{2}_{\mathrm{sgn}}$. The underlying
        topological space is the Euclidean plane $\mathbb{R}^2$,
        with group action by point reflection at the origin.
      \end{minipage}
    };

  \begin{scope}
  \clip (-1.8-1.7,-3.4) rectangle (4.8-1.5,3.4);
  \draw[step=3, dotted] (-6,-6) grid (6,6);

  \draw[<->, dashed, color=darkblue]
    (-2,2)
    to node[near start]
    {
      \colorbox{white}{ \bf
        \tiny
        \color{darkblue}
        \begin{tabular}{c}
          $\mathbb{Z}_2$
          \\
          action
        \end{tabular}
      }
    }
    (2,-2);
  \draw[<->, dashed, color=darkblue] (2,2) to (-2,-2);

  \end{scope}

  \draw (-6,0) node {$\mathbb{R}^{\mathbf{2}_{\mathrm{sgn}}} = $};

  \begin{scope}[shift={(0,-1.4)}]

  \draw (-3,-2.6) node {\tiny $x_1 = -\tfrac{1}{2}$};
  \draw (0,-2.6) node {\tiny $x_1 = 0$};
  \draw (3,-2.6) node {\tiny $x_1 = \tfrac{1}{2}$};

  \end{scope}

  \draw (-4.1,0) node {\tiny $x_2 = 0$};
  \draw (-4.1,3) node {\tiny $x_2 = \tfrac{1}{2}$};
  \draw (-4.1,-3) node {\tiny $x_2 = -\tfrac{1}{2}$};

\end{tikzpicture}
\end{center}

\vspace{-2mm}
Notice that for $V, W \in \mathrm{RO}(G)$ two orthogonal linear
$G$-representations, with $V \oplus W \in \mathrm{RO}(G)$
their direct sum representation, the Cartesian product of
their Euclidean $G$-spaces \eqref{EuclideanGSpace} is the Euclidean $G$-space of their direct sum:
\begin{equation}
  \label{CartesianProductOfEuclideanGSpaces}
  \mathbb{R}^V \times \mathbb{R}^W
  \;\simeq\;
  \mathbb{R}^{V \oplus W}
  \,.
\end{equation}

\noindent {\bf $G$-Representation spheres.}
The one-point compactification \eqref{SphereIsCompactificationOfCartesianSpace}
of a Euclidean space $\mathbb{R}^V$ \eqref{EuclideanGSpace}
becomes itself a $G$-space,
with the point at infinity
declared to be fixed by all group elements;
this is called the {\it representation sphere} of $V$
(see, e.g., \cite[1.1.5]{Blu17}):
\vspace{-3mm}
\begin{equation}
  \label{RepSpheres}
  \raisebox{10pt}{
  \xymatrix@R=-4pt{
    &
    \mathclap{
    \mbox{\bf
      \tiny
      \color{darkblue}
      \begin{tabular}{c}
        one-point compactification
        \\
        of Euclidean space $V$
      \end{tabular}
    }
    }
    &&
    \mathclap{
    \mbox{\bf
      \tiny
      \color{darkblue}
      \begin{tabular}{c}
        unit sphere
        \\
        in product of $V$
        \\
        with the 1d trivial representation
      \end{tabular}
    }
    }
    \\
    S^V
    \ar@{}[r]|-{ \coloneqq }
    &
    \big( \mathbb{R}^V \big)^{\mathrm{cpt}}
    \ar@{}[r]|-{ \simeq }
    &
    D\big(\mathbb{R}^V\big)/S\big( \mathbb{R}^V\big)
    \ar@{}[r]|-{ \simeq }
    &\
    S\big( \mathbb{R}^{\mathbf{1}_{\mathrm{triv}} \oplus V} \big)\;.
    \\
    \mathclap{
    \mbox{\bf
      \tiny
      \color{darkblue}
      \begin{tabular}{c}
        representation
        \\
        sphere
      \end{tabular}
    }
    }
    &&
    \mathclap{
    \mbox{\bf
      \tiny
      \color{darkblue}
      \begin{tabular}{c}
        unit disk in $V$
        \\
        with boundary collapsed
        \\
        to the point at infinity
      \end{tabular}
    }
    }
    }
  }
\end{equation}

\vspace{-2mm}

\begin{example}[{\hyperlink{FigureJ}{\it Figure J}}]
With $G \coloneqq \mathbb{Z}_2$ the group of order 2
and $\mathbf{1}_{\mathrm{sgn}}$ its 1-dimensional sign-representation,
the corresponding representation sphere \eqref{RepSpheres}
is the circle equipped with the $\mathbb{Z}_2$-action that reflects
across an equator:
\end{example}

\begin{center}
\hypertarget{FigureJ}{}
\begin{tikzpicture}[scale=0.6]

  \draw (-13,0)
   node
   {\begin{minipage}[l]{6.7cm} \footnotesize
     {\bf Figure J -- The $\mathbb{Z}_2$-representation sphere}
     {of the 1-dimensional sign representation
     $\mathbf{1}_{\mathrm{sgn}}$
     is the $\mathbb{Z}_2$-space whose underlying topological
     space is the circle, and equipped with the $\mathbb{Z}_2$-action
     that reflects points across the equator through $0$ and the
     point at infinity.}
     \end{minipage}
   };

    \draw (-4,.5)
      node
      {$
        \xymatrix{
          S^{\mathbf{1}_{\mathrm{sgn}}}
                    \ar@(ul,ur)|{\,\mathbb{Z}_2\,}
        }
        =
      $};

    \draw (0,0) circle (2);
    \node (infinity1) at (2,0) {\colorbox{white}{$\infty$}};
    \node
      (submanifold1)
      at (-2-.3,0)
      {\footnotesize $0$};
    \draw  (180-0:2) node {$-$};
    \node (submanifold2) at (180+50:2) {};
    \node (mirrorsubmanifold2) at (180-50:2) {};
    \node (submanifold3) at (180+35:2) {};
    \node (mirrorsubmanifold3) at (180-35:2) {};

    \node (submanifold4) at (90:2) {};
    \node (mirrorsubmanifold4) at (-90:2) {};
    \node (submanifold5) at (45:2) {};
    \node (mirrorsubmanifold5) at (-45:2) {};

    \draw[<->, dashed, darkblue]
      (submanifold3)
      to
      (mirrorsubmanifold3);
    \draw[<->, dashed, darkblue]
      (submanifold2)
      to
      node[near start]
        {
          \raisebox{1.2cm}{
          \tiny \bf
          \color{darkblue}
          \hspace{.2cm}
          \begin{tabular}{c}
            $\mathbb{Z}_2$
            \\
            action
          \end{tabular}
          }
        }
      (mirrorsubmanifold2);

    \draw[<->, dashed, darkblue]
      (submanifold4)
      to
      (mirrorsubmanifold4);
    \draw[<->, dashed, darkblue]
      (submanifold5)
      to
      (mirrorsubmanifold5);

\end{tikzpicture}
\end{center}

\vspace{-3mm}
\noindent {\bf $G$-Representation tori.}
Similarly, consider the linear $G$-representation $V$ such that
$G \subset \mathrm{Iso}\big( \mathbb{R}^{\mathrm{dim}(V)}\big)$
is the point group of a crystallographic group $C$
(see, e.g., \cite{Farkas81}) of the underlying Euclidean space
$\mathbb{R}^{\mathrm{dim}(V)}$ with corresponding translational sub-lattice $\mathbb{Z}^n \subset \mathrm{Iso}(n)$ inside the Euclidean group
in $n = \mathrm{dim}(V)$ dimensions. This means we have an
exact sequence of this form:
\vspace{-2mm}
\begin{equation}
  \label{CrystallographicGroups}
  \raisebox{25pt}{
  \xymatrix@R=-2pt{
    &
    \mathclap{
    \mbox{\bf
      \tiny
      \color{darkblue}
      \begin{tabular}{c}
        lattice of translations
        \\
        normal subgroup
      \end{tabular}
    }}
    &&
    \mathclap{
    \mbox{\bf
      \tiny
      \color{darkblue}
      \begin{tabular}{c}
        crystallographic
        \\
        group
      \end{tabular}
    }
    }
    &&
    \mathclap{
    \mbox{\bf
      \tiny
      \color{darkblue}
      \begin{tabular}{c}
        point group
        \\
        $\simeq C/\mathbb{Z}^n$
      \end{tabular}
    }
    }
    \\
    1
    \ar[r]
    &
    \mathbb{Z}^n
    \ar@{^{(}->}[dddddd]
    \ar@{^{(}->}[rr]
    &&
    C
    \ar@{^{(}->}[dddddd]
    \ar@{->>}[rr]
    &&
    {\color{darkblue} G }
    \ar@{^{(}->}[dddddd]
    \ar[r]
    &
    1
    \\
    \\
    \\
    \\
    \\
    \\
       1
    \ar[r]
    &
    \mathbb{R}^n
    \ar@{^{(}->}[rr]
    &&
    \mathrm{Iso}(n)
    \ar@{->>}[rr]
    &&
    \mathrm{O}(n)
    \ar[r]
    &
    1
    \\
    &
    \mathclap{
    \mbox{\bf
      \tiny
      \color{darkblue}
      \begin{tabular}{c}
        translation group
      \end{tabular}
    }
    }
    &&
    \mathclap{
    \mbox{\bf
      \tiny
      \color{darkblue}
      \begin{tabular}{c}
        Euclidean group
        \\
        (isometries of $\mathbb{R}^n$)
      \end{tabular}
    }
    }
    &&
    \mathclap{
    \mbox{\bf
      \tiny
      \color{darkblue}
      \begin{tabular}{c}
        orthogonal group
      \end{tabular}
    }
    }
  }
  }
\end{equation}
Then the corresponding torus $\mathbb{T}^n \coloneqq \mathbb{R}^n/\mathbb{Z}^n$
inherits a $G$-action from $\mathbb{R}^V$. We may call the resulting
$G$-space the {\it representation torus} of $V$. This is
the type of $G$-space whose global quotients are {\it toroidal orbifolds}:
\begin{equation}
  \label{RepresentationTorus}
  \begin{array}{ccccccc}
    V \in \mathrm{RO}(G)
    &\Rightarrow&
    \xymatrix{
      \mathbb{R}^V
      \ar@(ul,ur)|{\,G\,}
    }
    &\Rightarrow&
    \xymatrix@R=-6pt{
      \mathbb{T}^V
      \ar@(ul,ur)|-{\,G\,}
      \ar@{}[r]|-{\coloneqq}
      &
      \mathbb{R}^V
      \ar@(ul,ur)|-{\,G\,}
      &
      \!\!\!\!\!\!\!\!\!\!\!\!\!\!\!\!\!\!\!\!\!\!\!\!\!\!\!\!\!\!\!/\mathbb{Z}^n
    }
    &\Rightarrow&
    \mathbb{T}^V \!\sslash\!G\;.
    \\
    \mbox{\bf
      \tiny
      \color{darkblue}
      \begin{tabular}{c}
        orthogonal linear
        \\
        $G$-representation
      \end{tabular}
    }
    &&
    \mbox{\bf
      \tiny
      \color{darkblue}
      \begin{tabular}{c}
        Euclidean
        \\
        $G$-space
      \end{tabular}
    }
    &
    \mbox{\bf
      \tiny
      \color{darkblue}
      \begin{tabular}{c}
        if $G$ is point group
        \\
        of crystallographic group
      \end{tabular}
    }
    &
    \mbox{\bf
      \tiny
      \color{darkblue}
      \begin{tabular}{c}
        representation torus
      \end{tabular}
    }
    &&
    \mbox{\bf
      \tiny
      \color{darkblue}
      \begin{tabular}{c}
        toroidal
        \\
        orbifold
      \end{tabular}
    }
  \end{array}
\end{equation}

\vspace{-2mm}
\begin{example}[{\hyperlink{FigureK}{\it Figure K}}]
\label{FixedPointsInSignRepTorus}
For $G = \mathbb{Z}_4$ the cyclic group of
order 4 and $\mathbf{2}_{\mathrm{rot}}$ its 2-dimensional
linear representation given by rotations around the origin
by integer multiples of $\pi/2$,
this action descends to the 2-torus quotient
to give the representation torus $\mathbb{T}^{\mathbf{2}_{\mathrm{rot}}}$:
\end{example}

\vspace{-8mm}
\begin{center}
\hypertarget{FigureK}{}
\begin{tikzpicture}[scale=0.85]

  \draw
    (-12,.5)
    node
    {\footnotesize
      \begin{minipage}[l]{9cm}
        {  \bf Figure K -- The $\mathbb{Z}_4$-representation torus}
        \eqref{RepresentationTorus}
        of the 2-dimensional rotational representation
        $\mathbf{2}_{\mathrm{rot}}$. The underlying
        topological space is the 2-torus $T^2 = \mathbb{R}^2/\mathbb{Z}^2$,
        of which we show the canonical covering $\mathbb{R}^2$-coordinate chart.
        Due to the coordinate identifications
        $$
          \big([x_1], [x_2]\big) = \big([x_1 + n], [x_2 + m]\big)
          \;\in\; \mathbb{T}^2 = \mathbb{R}^2 / \mathbb{Z}^2
        $$
        the fixed point set \eqref{FixedLoci}
        of the $\mathbb{Z}_2$-subgroup has four points
        is
        $$
          \big(
            \mathbb{T}^{\mathbf{2}_{\mathrm{rot}}}
          \big)^{\mathbb{Z}_2}
          \;=\;
          \Big\{
            \big([0],[0]\big),
            \big([\tfrac{1}{2}], [\tfrac{1}{2}]\big),
            \big([0],[\tfrac{1}{2}]\big),
            \big([\tfrac{1}{2}], [0]\big)
          \Big\}
          \subset \mathbb{T}^2\,.
        $$
        while that of the full group has two points
        $$
          \big(
            \mathbb{T}^{\mathbf{2}_{\mathrm{rot}}}
          \big)^{\mathbb{Z}_4}
          \;=\;
          \Big\{
            \big([0],[0]\big),
            \big([\tfrac{1}{2}], [\tfrac{1}{2}]\big)
          \Big\}
          \subset \mathbb{T}^2\,.
        $$
      \end{minipage}
    };

  \draw (-5.5,.6)
    node
    {
      $
      \xymatrix{
        \mathbb{T}^{\mathbf{2}_{\mathrm{rot}}}
        \ar@(ul,ur)|{\,\mathbb{Z}_4\,}
             }
      \;=\;
      $
    };

  \begin{scope}
  \clip (-1.8-1.7,-1.4) rectangle (4.8-1.5,3.5);
  \draw[step=3, dotted] (-3.4,-3) grid (6,6);

  \draw[<->, dashed, darkblue] (-.8,.8) to (.8,-.8);

  \draw[<->, dashed, greenii]
    (0-45+3:1.3)
    arc
    (0-45+3:90-45-3:1.3);

  \draw[<->, dashed, greenii]
    (0+45+3:1.3)
    arc
    (0+45+3:90+45-3:1.3);

  \draw[<->, dashed, greenii]
    (0-30+3:1.6)
    arc
    (0-30+3:90-30-3:1.6);

  \draw[<->, dashed, greenii]
    (0-30+3+90:1.6)
    arc
    (0-30+3+90:90-30-3+90:1.6);

  \draw[<->, dashed, darkblue] (0-30:1.5) to (0-30+180:1.5);

  \draw[<->, dashed, greenii]
    (0-20+3:1.9)
    arc
    (0-20+3:90-20-3:1.9);

  \draw[<->, dashed, greenii]
    (0-20+3+90:1.9)
    arc
    (0-20+3+90:90-20-3+90:1.9);

  \draw[<->, dashed, darkblue] (0-20:1.8) to (0-20+180:1.8);

  \draw[<->, dashed, greenii]
    (0-10+3:2.2)
    arc
    (0-10+3:90-10-3:2.2);

  \draw[<->, dashed, greenii]
    (0-10+3+90:2.2)
    arc
    (0-10+3+90:90-10-3+90:2.2);

  \draw[<->, dashed, darkblue] (0-10:2.1) to (0-10+180:2.2);

  \draw[<->, dashed, greenii]
    (0+3:2.5)
    arc
    (0+3:90-3:2.5);

  \draw[<->, dashed, greenii]
    (0+3+90:2.5)
    arc
    (0+3+90:90-3+90:2.5);

  \draw[<->, dashed, darkblue] (0:2.4) to (0+180:2.4);

  \draw[<->, dashed, greenii]
    (0+10+3:2.8)
    arc
    (0+10+3:90+10-3:2.8);

  \draw[<->, dashed, greenii]
    (0+10+3+90:2.8)
    arc
    (0+10+3+90:90+10-3+90:2.8);

  \draw[<->, dashed, darkblue] (0+10:2.7) to (0+10+180:2.7);

  \draw[<->, dashed, greenii]
    (0+20+3:3.1)
    arc
    (0+20+3:90+20-3:3.1);

  \draw[<->, dashed, greenii]
    (0+20+3+90:3.1)
    arc
    (0+20+3+90:90+20-3+90:3.1);

  \draw[<->, dashed, darkblue] (0+20:3) to (0+20+180:3);

  \draw
    (2,2)
    node
      {
        \colorbox{white}{
        \hspace{-.4cm}
        \tiny
        \color{greenii}
        \begin{tabular}{c}
          $\mathbb{Z}_4$
          \\
          action
        \end{tabular}
        \hspace{-.4cm}
      }
      }
    (0,2.5);

  \draw
    (-1.3,0)
    node
      {
        \colorbox{white}{
        \hspace{-.4cm}
        \tiny
        \color{darkblue}
        \begin{tabular}{c}
          $\mathbb{Z}_2 \subset \mathbb{Z}_4$
          \\
          action
        \end{tabular}
        \hspace{-.4cm}
      }
      }
    (0,2.5);

  \end{scope}

  \begin{scope}[shift={(0,1)}]
  \draw (-3,-2.6) node {\tiny $x_1 = -\tfrac{1}{2}$};
  \draw (0,-2.6) node {\tiny $x_1 = 0$};
  \draw (3,-2.6) node {\tiny $x_1 = \tfrac{1}{2}$};

  \draw (-3,-2.8) .. controls (-3,-2.8-.8) and (3,-2.8-.8) .. node[below] {\tiny $\sim$} (3,-2.8);

  \end{scope}

  \draw (-3.9,0) node {\tiny $x_2 = 0$};
  \draw (-3.9,3) node {\tiny $x_2 = \tfrac{1}{2}$};

\end{tikzpicture}
\end{center}

\vspace{-1mm}
\noindent {\bf $H$-Fixed subspaces and isotropy groups.}
For $\xymatrix{X \ar@(ul,ur)|{\, G\,}}$ a $G$-space and $H \subset G$ any
subgroup, the {\it $H$-fixed subspace}
\begin{equation}
  \label{FixedLoci}
  X^H
    \;\coloneqq\;
  \big\{
    x \in X \big\vert h\cdot x = x
    \;
    \mbox{for all}\;
    h \in H
  \big\}
  \;\subset\; X
\end{equation}
is the topological subspace of $X$ on those points which are fixed
by the action of $H$. In particular, for $1 \subset G$ the trivial
group we have $X^1 = X$. We also write
\begin{equation}
  \label{IsotropySubgroups}
  \mathrm{Isotr}_X(G)
  \;\coloneqq\;
  \big\{
    \mathrm{Stab}_G(x) \subset G
    \,\big\vert\,
    x \in X
  \big\}
\end{equation}
for the set of {\it isotropy subgroups} of $G$, hence those that
appear as stabilizer groups of some point,
namely as maximal subgroups fixing a point:
$
  \mathrm{Stab}_G(x)
  \;\coloneqq\;
  \big\{
    g \in G
    \vert
    g \cdot x = x
  \big\}
  \;\subset\;
  G
  \,.
$
It is the isotropy subgroups \eqref{IsotropySubgroups},
but not necessarily the generic
subgroups, which serve to filter a $G$-space in a non-degenerate
way,
since if one isotropy subgroup is strictly larger than
another, then its fixed subspace \eqref{FixedLoci} is strictly smaller
$$
  H_1 \subsetneq H_2
  \;\in\; \mathrm{Isotr}_X(G)
  \phantom{AAA}
  \Rightarrow
  \phantom{AAA}
  X^{H_2} \subsetneq X^{H_1}.
$$
\begin{example}[fixed subspaces of ADE-singularities]
\label{FixedSubspacesOfADESingularities}
The non-trivial fixed subspaces of the
Euclidean $G$-space \eqref{EuclideanGSpace}
of the quaternionic representation $\mathbf{4}_{\mathbb{H}}$
\eqref{TheQuaternionicRepresentation}
are all the singleton sets consisting of the origin:
\begin{equation}
  \label{FixedSubspacesOfQuaternionRepresentation}
  \big(
    \mathbb{R}^{\mathbf{4}_{\mathbb{H}}}
  \big)^H
  \;=\;
  \left\{
  \begin{array}{cc}
    \mathbb{R}^4 & \mbox{if}\;H = 1
    \\
    \{0\} & \mbox{otherwise}.
  \end{array}
  \right.
  \end{equation}
\end{example}

\begin{example}[{\hyperlink{FigureK}{\it Figure K}}]
\label{ExampleRT}
For $G = \mathbb{Z}_2$ and $\mathbf{n}_{\mathrm{sgn}}$
the $n$-dimensional sign representation, the corresponding
representation torus \eqref{RepresentationTorus}
has as $\mathbb{Z}_2$-fixed space \eqref{FixedLoci} the 0-dimensional space
which is the set of points whose canonical coordinates are
all either 0 mod $\mathbb{Z}$ or $\tfrac{1}{2}$ mod $\mathbb{Z}$:
\vspace{-2mm}
\begin{equation}
  \label{RepresentationTorusOfSignRep}
  \mathbb{T}^{\mathbf{n}_{\mathrm{sgn}}}
  \;\coloneqq\;
  \xymatrix{
    (
    \mathbb{R}^n
    \ar@(ul,ur)^{\footnotesize [x] \mapsto [-x]  }
    &
    \!\!\!\!\!\!\!\!\!\!\!\!\!\!\!\!\!\!\!\!\!\!\!\!\!\!\!\!\!\!
    /\mathbb{Z}^n)
  }
  \phantom{AAA}
  \Longrightarrow
  \phantom{AAA}
  \big(
    \mathbb{T}^{\mathbf{n}_{\mathrm{sgn}}}
  \big)^{\mathbb{Z}_2}
  \;=\;
  \big\{
    [0], [\tfrac{1}{2}]
  \big\}^n
  \;\subset\;
    \mathbb{T}^n = \mathbb{R}^n/\mathbb{Z}^n.
\end{equation}
\end{example}

\begin{example}[\bf Kummer surface]
\label{KummerSurface}
The reflection ADE-action \eqref{PointReflectionSubgroup}
$
  \xymatrix{
    \mathbb{R}^{\mathbf{4}_{\mathbb{H}}}
    \ar@(ul,ur)|-{\,
      \mathbb{Z}^{\mathrm{refl}}_2
    \!\!\!}
  }
$
is clearly crystallographic \eqref{CrystallographicGroups}.
The orbifold
$
  \mathbb{T}^{\mathbf{4}_{\mathbb{H}}}
    \!\sslash\!
  \mathbb{Z}^{\mathrm{refl}}_2
  \simeq
  \mathbb{T}^{\mathbf{4}_{\mathrm{sgn}}}
    \!\sslash\!
  \mathbb{Z}_2
$
presented by the corresponding representation torus \eqref{RepresentationTorus}
(when equivalently thought of as an orbifold of the complex 2-dimensional
torus) known as the {\it Kummer surface} (e.g. \cite[5.5]{BDP17}).
The cardinality of its fixed point set \eqref{FixedLoci}
is (by Example \ref{ExampleRT})
$$
  \left\vert
    \big(
      \mathbb{T}^{\mathbf{4}_{\mathbb{Z}}}
    \big)^{\mathbb{Z}_2^{\mathrm{refl}}}
  \right\vert
  \;=\;
  \left\vert
    \{[0],[\tfrac{1}{2}]\}^4
  \right\vert
  \;=\;
  16
  .
$$
\end{example}

\medskip

\noindent {\bf Residual action on fixed spaces.}
There is a residual group action on any $H$-fixed subspace
$X^H$ \eqref{FixedLoci}
inherited from the $G$-action on all of $X$,
with the residual group  being the
``Weyl group'' \cite[p. 13]{May96}
\begin{equation}
  \label{WeylGroup}
  W_G(H)
  \;\coloneqq\;
  N_G(H) / H
\end{equation}{
which is the quotient group
of the maximal subgroup $N_G(H) \subset G$ for which $H$ is a normal
subgroup (the normalizer of $H$ in $G$) by $H$ itself.
Thereby any $H$-fixed subspace becomes itself a $W_G(H)$-space:
\begin{equation}
  \label{ResidualActionOnFixedSubspaces}
  \begin{array}{ccccc}
  \xymatrix{
    X
    \ar@(ul,ur)|{\, G\,}
  }
  &~~~\colon~~~&
  (H \subset G)
  &~~ \longmapsto ~~&
  \xymatrix{
    X^H
    \ar@(ul,ur)^{W_G(H)}
  }\!\!.
  \\
  \mathclap{
  \mbox{\bf
    \tiny
    \color{darkblue}
    A $G$-space induces
  }}
  &&
  \mathclap{
  \mbox{\bf
    \tiny
    \color{darkblue}
    for each subgroup $H$
  }}
  &&
  \mathclap{
  \mbox{\bf
    \raisebox{-2pt}{
    \tiny
    \color{darkblue}
    \begin{tabular}{l}
      the $H$-fixed space with
      \\
      residual $W_G(H)$-action
    \end{tabular}
    }
  }}
  \end{array}
\end{equation}
Notice the two extreme cases of the Weyl group \eqref{WeylGroup}:
\begin{equation}
  \label{ExtremeCasesOfWeylGroups}
  W_G(1) = G
  \phantom{AAA}
  \mbox{and}
  \phantom{AAA}
  W_G(G) = 1
  \,.
\end{equation}

\medskip

\noindent {\bf Maps between $G$-spaces and their Elmendorf stages.}
The relevant {\it morphisms between $G$-spaces}
are
continuous functions between the underlying spaces
that are $G$-equivariant:
\begin{equation}
  \label{EquivariantFunction}
  \xymatrix{
   X \ar@(ul,ur)|{\, G\,}
   \ar[rr]^-{f}
   &&
   Y \ar@(ul,ur)|{\, G\,}
  }
  \phantom{AAA}
  \Leftrightarrow
  \phantom{AAA}
  \xymatrix{
   X
   \ar[rr]^-{f}
   &&
   Y
  }
  \;\;
  \mbox{\footnotesize
   \begin{tabular}{l}
     such that $f(g\cdot x) = g \cdot f(x)$
     \\
     for all $g \in G$  and all $x \in X$.
   \end{tabular}
  }
\end{equation}
This $G$-equivariance implies that $H$-fixed points are sent to $H$-fixed points,
for every subgroup $H \subset G$, hence that every
$G$-equivariant continuous function \eqref{EquivariantFunction}
induces a system of plain continous functions
$f^H := f_{\vert X^H}$ between
$H$-fixed point spaces \eqref{FixedLoci},
which are each equivariant with respect to the residual
$W_G(H)$-action \eqref{WeylGroup}
and compatible with each other with respect to
inclusions $H_i \subset H_j$ of subgroups:
\begin{equation}
  \label{SystemOfMapsOnHFixedSubspaces}
  \begin{array}{ccc}
  {\xymatrix{
    X
    \ar@(ul,ur)|{\,G\,}
    \ar[rr]^{f}
    &&
    Y\ar@(ul,ur)|{\,G\,}
  }}
  &
  \phantom{AAA}
    \Rightarrow
  \phantom{AA}
  &
  \raisebox{60pt}{
  \xymatrix{
    X
    \ar@(ul,dl)_G
    \ar[rr]^-f
    &&
    Y
    \ar@(ur,dr)^G
    &{\phantom{AAAAA}}&
    1
    \ar@{^{(}->}[d]
    \\
    X^{H_i}
    \ar@(ul,dl)_{W_G(H_i)}
    \ar@{^{(}->}[u]
    \ar[rr]^-{f^{H_i}}
    &&
    Y^{H_i}
    \ar@(ur,dr)^{W_G(H_i)}
    \ar@{^{(}->}[u]
    &
       \ar@{}[d]|{ \mbox{\hspace{1.3cm} for all} }
    &
    H_i
    \ar@{^{(}->}[d]
    \\
    X^{H_j}
    \ar@(ul,dl)_{W_G(H_j)}
    \ar@{^{(}->}[u]
    \ar[rr]^-{f^{H_j}}
    &&
    Y^{H_j}
    \ar@(ur,dr)^{W_G(H_j)}
    \ar@{^{(}->}[u]
    &&
    H_j
    \ar@{^{(}->}[d]
    \\
    X^{G}
    \ar@{^{(}->}[u]
    \ar[rr]^-{f^{G}}
    &&
    Y^{G}
    \ar@{^{(}->}[u]
    &&
    G
  }
  }
  \\
  \mathclap{
  \mbox{\bf
    \tiny
    \color{darkblue}
    \begin{tabular}{c}
      $G$-equivariant function
      between $G$-spaces
    \end{tabular}
  }
  }
  &
  \mathclap{
  \mbox{\bf
    \tiny
    \color{darkblue}
    \begin{tabular}{c}
      induces
    \end{tabular}
  }
  }
  &
  \mathclap{
  \mbox{\bf
    \tiny
    \color{darkblue}
    \begin{tabular}{c}
      system of
      $W_G(H)$-equivariant functions
      between $H$-fixed subspaces
    \end{tabular}
  }
  }
  \end{array}
\end{equation}
We will refer to the component $f^H$ here as
the {\it Elmendorf stage} labeled by $H$ \cite[1.3]{Blu17}\cite[3.1]{ADE}.

\medskip

Finally, a {\it $G$-homotopy} between two $G$-equivariant functions
$f_1, f_2$ \eqref{EquivariantFunction}
\begin{equation}
  \label{GHomotopy}
  \xymatrix{
    X
    \ar@(ul,dl)_G
    \ar@/^1pc/[rr]^-{f_0}_-{\ }="s"
    \ar@/_1pc/[rr]_-{f_1}^-{\ }="t"
    &&
    Y
    \ar@(ur,dr)^G
    \ar@{=>}^\eta "s"; "t"
  }
\end{equation}

\vspace{-.5cm}
\noindent is a homotopy
$
  [0,1] \times X
  \overset{ \eta }{\longrightarrow}
  X
$
between the underlying continuous functions, hence such that
$f_i = \eta(i,-)$,
which is equivariant as a function on the product $G$-space
$X \times [0,1]$, where the $G$-action on the interval $[0,1]$
is taken to be trivial.

\medskip

\subsection{Equivariant Hopf degree on spheres and Local tadpole cancellation}
\label{LocalTadpoleCancellation}

We discuss the unstable
(Theorem \ref{UnstableEquivariantHopfDegreeTheorem})
and the stabilized
(Theorem \ref{CharacterizationOfStabilizationOfUnstableCohomotopy})
equivariant Hopf degree theorem for representation spheres,
which characterizes equivariant Cohomotopy
in compatible RO-degree (Def. \ref{CompatibleRODegree} below),
on Euclidean $G$-spaces and vanishing at infinity, hence
of the vicinity of $G$-singularities inside flat Euclidean space
(Def. \ref{CohomotopyOfVicinityOfSingularity} below).
Using this we show (Prop. \ref{TheoremLocalTadpoleCancellation})
that equivariant Cohomotopy implies the
form of the local/twisted tadpole cancellation conditions from
\hyperlink{Table1}{\it Table 1}, \hyperlink{Table2}{\it Table 2}.

\medskip

\subsubsection{Unstable equivariant Hopf degree}
\label{UnstableEquivariantHopfDegree}

For stating the equivariant Hopf degree theorem, we need the following concept of
\emph{compatible RO-degree} for equivariant Cohomotopy. This condition is really a
reflection of the structure of \emph{J-twisted} Cohomotopy (as in \cite{FSS19b}\cite{FSS19c})
in its version on flat orbifolds, and as such is further developed in \cite{OrbifoldCohomology}.

\vspace{-.3cm}

\begin{defn}[\bf Compatible RO-degree]
\label{CompatibleRODegree}
Given a $G$-space $\xymatrix{ X \ar@(ul,ur)|{\,G \,}}$
such that each $H$-fixed subspace $X^H$ \eqref{FixedLoci}
for isotropy groups $H \in \mathrm{Isotr}_X(G)$ \eqref{IsotropySubgroups}
admits the structure of an orientable manifold, we say that an orthogonal linear $G$-representation
$V$ is a \emph{compatible RO-degree for equivariant Cohomotopy of $X$}
if for each isotropy subgroup $H \in \mathrm{Isotr}_X(G)$ \eqref{IsotropySubgroups}
the following two conditions hold:
\footnote{
These conditions are a specializations of the conditions stated in
\cite[p. 212-213]{tomDieck79}, streamlined here for our purpose.}
\begin{enumerate}[{\bf (i)}]
\vspace{-2mm}
\item
 {\bf Compatible fixed space dimensions:}
 the dimension of the $H$-fixed subspace of $V$ equals that of
 the $H$-fixed subspace of $X$:
 \begin{equation}
   \label{FixedSpacesOfCompatibleDimension}
   \mathrm{dim}\big(  X^H \big)
   \;=\;
   \mathrm{dim}\big( V^H \big).
 \end{equation}

 \vspace{-3mm}
\item
  {\bf Compatible orientation behavior:}
  the action \eqref{ResidualActionOnFixedSubspaces}
  of an element $[g] \in W_G(H)$ \eqref{WeylGroup}
  on $V^H$ is orientation preserving or reversing,
  respectively, precisely if it is so on $X^H$
  \begin{equation}
    \label{OrientationBehaviousCompatible}
    \mathrm{orient}
    \left(
    \raisebox{-10pt}{
    \xymatrix{
      X^H
      \ar@(ul,ur)^{ [g] \in W_H(H) }
    }}
    \right)
    \;=\;
    \mathrm{orient}
    \left(
    \raisebox{-10pt}{
    \xymatrix{
      (S^V)^H
      \ar@(ul,ur)^{ [g] \in W_H(H) }
    }
    }
    \right).
  \end{equation}
\end{enumerate}
\end{defn}

\begin{example}[\bf Compatible RO-degree for representation-spheres and -tori]
\label{ExamplesOfCompatibleRODegree}
We observe that every real linear $G$-representation
$V$ is a compatible RO-degree (Def. \ref{CompatibleRODegree})
\begin{enumerate}[{\bf (i)}]
\vspace{-2mm}
\item for the corresponding representation
sphere $S^V$ \eqref{RepSpheres};
\vspace{-2mm}
\item and for the corresponding representation torus $\mathbb{T}^{V}$
\eqref{RepresentationTorus}
\end{enumerate}
\vspace{-2mm}
If the latter exists, hence if $G$ is the point group of a crystallographic group on $\mathbb{R}^V$ \eqref{CrystallographicGroups}.
\end{example}

For brevity, we introduce the following
terminology, following \hyperlink{Table5}{\it Table 5},
for the situation in which we will now
consider equivariant Cohomotopy in compatible RO-degree:
\begin{defn}[Cohomotopy of vicinity of the singularity]
  \label{CohomotopyOfVicinityOfSingularity}
  Given a finite group $G$ and an orthogonal linear
  $G$-representation $V \in \mathrm{RO}(G)$,
  we say that the {\it Cohomotopy of the vicinity of the singularity}
  is
  the unstable $G$-equivariant Cohomotopy \eqref{EquivariantCohomotopySet}
  $$
    \pi^V_G
    \big(
      (\mathbb{R}^V)^{\mathrm{cpt}}
    \big)
    \;=\;
    \pi^V_G
    \big(
      S^V
    \big)
  $$
  in compatible RO-degree $V$ (Def. \ref{CompatibleRODegree}, Example
\ref{ExamplesOfCompatibleRODegree})
of the Euclidean $G$-space $\mathbb{R}^V$ \eqref{EuclideanGSpace}
and vanishing at infinity \eqref{VanishingAtInfinity},
hence of the representation sphere $S^V$ \eqref{RepSpheres}
and preserving the point at infinity.
\end{defn}

The key implication of the first clause \eqref{FixedSpacesOfCompatibleDimension} on
compatible RO-degrees is that
each Elmendorf stage $c^H$
\eqref{SystemOfMapsOnHFixedSubspaces}
of a $G$-equivariant Cohomotopy cocycle $c$
is a cocycle in ordinary Cohomotopy
\eqref{PlainCohomotopySet}
to which
the ordinary Hopf degree theorem applies,
either in its stable range \eqref{HopfDegreeTheorem}
or in the unstable range \eqref{UnstableRangeHopfDegreeTheorem}:

\vspace{-2mm}
\begin{equation}
  \label{ElmedorfStageWiseHopfDegrees}
  \hspace{-2mm}
  \!\!\!\!\!\!\!\!\!
  \begin{array}{ccc}
  \underset{
    \mbox{\bf
      \tiny
      \color{darkblue}
      \begin{tabular}{c}
        equivariant Cohomotopy cocycle \eqref{EquivariantCohomotopySet}
        \\
        in compatible RO-degree $V$ \eqref{FixedSpacesOfCompatibleDimension}
      \end{tabular}
    }
  }{
  {\xymatrix{
    X
    \ar@(ul,ur)|{\, G\,}
    \ar[rr]^-{c}
    &&
    S^V\ar@(ul,ur)|{\,G\,}
  }}
  }
  &
  \phantom{}
    \Rightarrow
  \phantom{A}
  &
  \raisebox{83pt}{
  \xymatrix{
    X
    \ar[rr]^-c
    &&
    S^{\,\mathrm{dim}(X) \gt 0}
    &
    &
    \mathrm{deg}(f) \in \mathbb{Z}
    \\
    \\
    X^{H}
    \ar@{^{(}->}[uu]|-{ \raisebox{2pt}{$\vdots$} }
    \ar[rr]^-{c^{H}}
    &&
    S^{\,\mathrm{dim}(X^{H}) \gt 0}
    \ar@{^{(}->}[uu]|-{ \raisebox{2pt}{$\vdots$} }
    &
    &
    \underset{
      \tiny
      \color{darkblue}  \bf
      \begin{tabular}{c}
     \bf    ordinary stable Hopf degree \eqref{HopfDegreeTheorem}
      \end{tabular}
    }{
      \mathrm{deg}(c^{H}) \in \mathbb{Z}
    }
    \\
    &
    \ddots
    \ar@{^{(}->}[ul]|-{\raisebox{5pt}{$\ddots$}}
    \ar[rr]
    &&
    \ddots
    \ar@{^{(}->}[ul]|-{\raisebox{5pt}{$\ddots$}}
    \\
    \mathllap{
      \mbox{\bf
        \tiny
        \color{darkblue}
        Elmendorf stages \eqref{SystemOfMapsOnHFixedSubspaces}
      }
      \;\;\;
    }
    X^{K}
    \ar@{^{(}->}[uu]|-{ \raisebox{2pt}{$\vdots$} }
    \ar[rr]^>>>>>>>{c^{K}}
    &&
    S^{\,\mathrm{dim}(V^{K})\gt 0}
    \ar@{^{(}->}[uu]|-{ \raisebox{2pt}{$\vdots$} }
    &
    &
    \mathrm{deg}(c^{K}) \in \mathbb{Z}
    \\
    &
    \ddots
    \ar@{^{(}->}[ul]|-{\raisebox{5pt}{$\ddots$}}
    \ar[rr]
    \ar@{^{(}->}[uu]|-{ \raisebox{2pt}{$\vdots$} }
    &&
    \ddots
    \ar@{^{(}->}[ul]|-{\raisebox{5pt}{$\ddots$}}
    \ar@{^{(}->}[uu]|-{ \raisebox{2pt}{$\vdots$} }
    \\
    X^{J}
    \ar@{^{(}->}[uu]|-{ \raisebox{2pt}{$\vdots$} }
    \ar[rr]^>>>>>>>{c^{J}}
    &
    \ar@{..>}[u]
    &
    S^{\,\mathrm{dim}(V^J)= 0}
    \ar@{^{(}->}[uu]|-{ \raisebox{2pt}{$\vdots$} }
    &
    \ar@{..>}[u]
    &
    \underset{
      \mbox{\bf
        \tiny
        \color{darkblue}
        \phantom{AA} ordinary unstable Hopf degree \eqref{UnstableRangeHopfDegreeTheorem}
      }
    }{
      \mathrm{deg}(c^J)
      \in
      \mathrlap{ \{0,1\}^{(X^J)} }
      \phantom{\mathbb{Z}}
    }
    \\
    {\phantom{ {A \atop A} \atop {A \atop A} }}
    \ar@{..>}[u]
    &
    &
    {\phantom{ {A \atop A} \atop {A \atop A} }}
    \ar@{..>}[u]
    &
  }
  }
  \end{array}
\end{equation}

\vspace{-1.2cm}

\begin{theorem}[\bf Unstable equivariant Hopf degree theorem for
representation spheres]
\label{UnstableEquivariantHopfDegreeTheorem}
The unstable Cohomopotopy of the vicinity of a $G$-singularity $\mathbb{R}^V$
(Def. \ref{CohomotopyOfVicinityOfSingularity}) is in bijection to the product set of
 one copy of the integers for each isotropy group \eqref{IsotropySubgroups} with positive-dimensional
 fixed subspace $\mathrm{Isotr}^{d_{\mathrm{fix}} \gt 0}_X(G)$ \eqref{FixedLoci}, and
 one copy of $\{0,1\}$ if there is an isotropy group with 0-dimensional fixed subspace
 $\mathrm{Isotr}^{d_{\mathrm{fix}} = 0}_X(G)$
(which is then necessarily unique and, in fact, the group $G$ itself):
\vspace{-3mm}
\begin{equation}
  \label{UnstableEquivariantCohomotopyOfRepresentationSphereInCompatibleDegree}
  \xymatrix{
    \pi^V_G\big( \big(\mathbb{R}^V\big)^{\mathrm{cpt}} \big)
    \ar[rrr]^-{
      c
      \;\mapsto\;
      ( H \mapsto {\color{darkblue} \bf N_H}(c) )
    }_-{\simeq}
    &&&
    \mathbb{Z}^{{}^{ \mathrm{Isotr}^{d_{\mathrm{fix}} \gt 0}_X(G) } }
    \times
    \{0,1\}^{{}^{ \mathrm{Isotr}^{d_{\mathrm{fix}} = 0}_X(G) } }
  },
\end{equation}
where, for $H \in \mathrm{Isotr}^{d_{\mathrm{fix}} \gt 0  }_X(G)$,
the ordinary Hopf degree at Elmendorf stage $H$  \eqref{ElmedorfStageWiseHopfDegrees} is of the form
\begin{equation}
  \label{TheWeylGroupMultiples}
  \xymatrix@R=-2pt{
    \mathrm{deg}\big( c^{H} \big)
    &
    \!\!\!
    \!\!\!
    \!\!\!
    \!\!\!
    =
    \!\!\!
    \!\!\!
    \!\!\!
    \!\!\!
    &
    \phi_H\big(
       \{ \mathrm{deg}\big( c^K \big) \big\vert K \supsetneq H \in \mathrm{Isotr}_X(G)  \}
    \big)
    &
    \!\!\!
    \!\!\!
    \!\!\!
    \!\!\!
    -
    \!\!\!
    \!\!\!
    \!\!\!
    \!\!\!
    &
    {\color{darkblue} \bf N_H}(c) \cdot \big| \big( W_G(H)\big) \big|
    &
    \!\!\!\!\!\!\!\!\!
    \in
    \mathbb{Z}.
    \\
    \mathclap{
    \mbox{\bf
      \tiny
      \color{darkblue}
      \begin{tabular}{c}
        The ordinary Hopf degree \eqref{HopfDegreeTheorem}
        \\
        at Elmendorf stage $K$ \eqref{SystemOfMapsOnHFixedSubspaces}
      \end{tabular}
    }
    }
    &
    \!\!\!\!\!\!
    \!\!\!\!\!\!
    \!\!\!\!\!\!
    \!\!\!\!\!\!
    &
    \mathclap{
    \mbox{\bf
      \tiny
      \color{darkblue}
      \begin{tabular}{c}
        offset, being a function $\phi_H$ of
        \\
        the Hopf degrees
        at all lower stages.
      \end{tabular}
    }
    }
    &
    \!\!\!\!\!\!
    \!\!\!\!\!\!
    \!\!\!\!\!\!
    \!\!\!\!\!\!
    &
    \mathclap{
    \mbox{\bf
      \tiny
      \color{darkblue}
      \begin{tabular}{c}
        an integer multiple of
        \\
        the order of the Weyl group \eqref{WeylGroup}
      \end{tabular}
    }
    }
  }
\end{equation}
The isomorphism
\eqref{UnstableEquivariantCohomotopyOfRepresentationSphereInCompatibleDegree}
is exhibited by sending an equivariant Cohomotopy cocycle $c$ to the sequence of the
integers ${\color{darkblue} \bf N_H}(c)$ from \eqref{TheWeylGroupMultiples}
in positive fixed subspace dimensions,
together with  possibly the choice of an element of $\{0,1\}$,
which is the unstable Hopf degree in dimension 0 \eqref{UnstableRangeHopfDegreeTheorem},
at Elmendorf stage $G$ (if $\mathrm{dim}(V^G) = 0$).
\end{theorem}
\begin{proof}
  In the special case that no subgroup $H \subset G$
  has a fixed subspace of vanishing dimension,
  this is \cite[Theorem 8.4.1]{tomDieck79}
  (the assumption of positive dimension is made
  ``for simplicity'' in \cite[middle of p. 212]{tomDieck79}).
  Hence we just need to convince ourselves that the proof
  given there generalizes:
in the present case of representation spheres,
  the only possible 0-dimensional fixed subspace is
  the 0-sphere. Hence we need to consider the case that
  $( S^V)^G = S^0$.

  To generalize the inductive argument in \cite[p. 214]{tomDieck79} to this case,
  we just need to see that every function $( S^V)^G \to ( S^V)^G$
  extends to a $W_G(H)$-equivariant function $( S^V)^H \to ( S^V)^H$ on
  a next higher Elmendorf stage $H$. But this holds in the present case:  every function
  from $S^0 = \{0,\infty\}$ to itself
  (as in \hyperlink{FigureH}{\it Figure H}) readily extends even to a $G$-equivariant function
  $S^V \to S^V$, and by assumption of vanishing at infinity \eqref{VanishingAtInfinity}
  one of exactly two extensions will work, namely either the identity function
  or the function constant on $\infty \in S^V$:
  \begin{equation}
    \label{InductionStartForRepSpheres}
    \xymatrix@R=-4pt{
      &
      \{0,1\}
      \ar@{<-}[rr]^-{ \mathrm{deg}\left( (-)^G \right) }
      &&
      \pi^V\big( S^V \big)
      \\
      \mbox{\bf
        \tiny
        \color{darkblue}
        \begin{tabular}{c}
          configuration of
          \\
          a single point in $S^0$
          \\
          sitting at $0 \in S^0$
        \end{tabular}
      }
      &
      \big[
        S^0 \xrightarrow{\mathrm{id}_{S^0}} S^0
      \big]
      \ar@{}[rr]|-{\longmapsfrom}
      &&
      \big[
        S^V \xrightarrow{c = \mathrm{id}_{S^V}} S^V
      \big]
      &
      \mbox{\bf
        \tiny
        \color{darkblue}
        \begin{tabular}{c}
          configuration of
          \\
          a single charged point in $S^V$
          \\
          which is sitting at $0 \in S^V$
        \end{tabular}
      }
      \\
      \mbox{\bf
        \tiny
        \color{darkblue}
        \begin{tabular}{c}
          configuration of
          \\
          no point in $S^0$
        \end{tabular}
      }
      &
      \big[
        S^0 \xrightarrow{\mathrm{const}_\infty} S^0
      \big]
      \ar@{}[rr]|-{\longmapsfrom}
      &&
      \big[
        S^V \xrightarrow{c = \mathrm{const}_{\infty}} S^V
      \big]
      &
      \mbox{\bf
        \tiny
        \color{darkblue}
        \begin{tabular}{c}
          configuration of
          \\
          no point in $S^V$
        \end{tabular}
      }
    }
  \end{equation}
    From this induction forward, the  proof of \cite[8.4.1]{tomDieck79} applies verbatim
  and shows that on top of  this initial Hopf degree number of -1 (a charge at $0 \in S^0$) or
  $0$ (no charge at $0 \in S^0$) there may now be further
  $N_H \cdot \vert W_G(H)\vert$-worth
  of Hopf degree at the next higher Elmendorf stage $H$, and so on.
\end{proof}

\begin{example}[$\Z_2$-equivariant Cohomotopy]
Consider
$$
  c
  \;\in\;
  \pi^{\mathbf{n}_{\mathrm{sgn}}}_{\mathbb{Z}_2}
    \big(
      (\mathbb{R}^{\mathbf{n}_{\mathrm{sgn}}})^{\mathrm{cpt}}
    \big)
$$
(i.e., a cocycle in $\mathbb{Z}_2$-equivariant
Cohomotopy vanishing at infinity \eqref{VanishingAtInfinity}
of the $n$-dimensional Euclidean orientifold
$\mathbb{R}^{\mathbf{n}_{\mathrm{sgn}}}$ \eqref{EuclideanGSpace}
underlying the $n$-dimensional sign representation $\mathbf{n}_{\mathrm{sgn}}$,
as in \hyperlink{FigureI}{\it Figure I},
hence the equivariant Cohomotopy of the representation sphere
$S^{\mathbf{n}_{\mathrm{sgn}}}$
\eqref{RepSpheres}, as in \hyperlink{FigureJ}{\it Figure J},
in compatible RO-degree $\mathbf{n}_{\mathrm{sgn}}$, by Example \ref{ExamplesOfCompatibleRODegree}).
Then the unstable equivariant Hopf degree theorem \ref{UnstableEquivariantHopfDegreeTheorem}
says, when translated to a geometric situation via
the unstable Pontrjagin-Thom theorem \eqref{UnstablePTTheorem}, that:

\begin{enumerate}[{\bf (i)}]
\vspace{-2mm}
\item there either is, or is not, a single charge sitting at the finite fixed point $0 \in S^{\mathbf{n}_{\mathrm{sgn}}}$,
    corresponding, with \eqref{InductionStartForRepSpheres},
    to an offset of $- 1$ or $0$, respectively,
    in \eqref{TheWeylGroupMultiples};
    \vspace{-2mm}
\item
  in addition, there is any integer number
  (the $N_{1} \in \mathbb{N}$ in \eqref{TheWeylGroupMultiples})
  of orientifold mirror pairs
  (since $\vert W_{\mathbb{Z}_2}(1)\vert = \vert \mathbb{Z}_2\vert = 2$,
   by \eqref{ExtremeCasesOfWeylGroups}) of charges floating in the vicinity.
\end{enumerate}

\vspace{-4mm}
\begin{center}
{\hypertarget{FigureL}{}}
\begin{tikzpicture}[scale=0.75]
  \begin{scope}[shift={(0,-1.3)}]
  \node
    (X)
    at (-4.5,6)
    {
    \raisebox{44pt}{
    $
      (
        \xymatrix{
          \mathbb{R}^{\mathbf{n}_{\mathrm{sgn}}}
          \ar@(ul,ur)^{
            \overset{
              \mathclap{
              \mbox{\bf
                \tiny
                \color{darkblue}
                \begin{tabular}{c}
                  sign
                  \\
                  representation
                \end{tabular}
              }
              }
            }
            {
              \mathbb{Z}_2
            }
          }
        }
      )^{\mathrm{cpt}}
    $}};
  \node
    (sphere)
    at (6,6)
    {
    \raisebox{44pt}{
    $
      S^{\mathbf{n}_{\mathrm{sgn}}}
        =
      (
        \xymatrix{
          \mathbb{R}^{\mathbf{n}_{\mathrm{sgn}}}
          \ar@(ul,ur)^{
            \overset{
              \mathclap{
              \mbox{\bf
                \tiny
                \color{darkblue}
                \begin{tabular}{c}
                  sign
                  \\
                  representation
                \end{tabular}
              }
              }
            }{
              \mathbb{Z}_2
            }
          }
        }
      )^{\mathrm{cpt}}
    $}};

  \draw[->]
    (X)
    to node[above]
    {$c$}
    (sphere);

  \node at (-4.5,4.9)
    {\tiny \color{darkblue} \bf
      \begin{tabular}{c}
        Euclidean $n$-space
        \\
        around orientifold singularity
        \\
        compactified by
        a point at infinity
      \end{tabular}
    };
  \node at (-4.5,3.9)
    {$\overbrace{\phantom{----------------}}$};

  \node at (6,4.9)
    {\tiny \color{darkblue} \bf
      \begin{tabular}{c}
        representation sphere
        \\
        equivariant Cohomotopy coefficient
      \end{tabular}
    };
  \node at (6,3.9)  {$\overbrace{\phantom{--------------}}$};

  \node
    at (.25,5.7)
    {
      \tiny
      \color{darkblue} \bf
      equivariant Cohomotopy cocycle
    };

  \end{scope}

  \begin{scope}[shift={(-4.5,0)}]
    \draw (0,0) circle (2);
    \node (infinity1) at (2,0) {\colorbox{white}{$\infty$}};
    \node
      (submanifold1)
      at (180-0:2)
      {$
        \mathllap{
        \mbox{ \bf
          \tiny
          \color{darkblue}
          \begin{tabular}{c}
            orientifold
            \\
            singularity
          \end{tabular}
        }
        \;
        }
      $};
    \draw[fill=white] (180-0:2) circle (.07);
    \node (submanifold2) at (180+50:2) {$\bullet$};
    \node (mirrorsubmanifold2) at (180-50:2) {$\bullet$};
    \node (submanifold3) at (180+35:2) {$\bullet$};
    \node (mirrorsubmanifold3) at (180-35:2) {$\bullet$};

    \draw[<->, dashed, darkblue]
      (submanifold3)
      to
      (mirrorsubmanifold3);
    \draw[<->, dashed, darkblue]
      (submanifold2)
      to
      node[near start]
        {
          \raisebox{.6cm}{ \bf
          \tiny
          \color{darkblue}
          \hspace{.2cm}
          \begin{tabular}{c}
            orientifold
            \\
            action
          \end{tabular}
          }
        }
      (mirrorsubmanifold2);
  \end{scope}

  \draw[|->, thin, olive]
    (infinity1)
    to[bend right=40]
    node {
      \colorbox{white}{\bf
        \tiny
        \color{darkblue}
        \hspace{-.5cm}
        \begin{tabular}{c}
          cocycle vanishes
          \\
          at infinity
          \\
          (far away from the singularity)
        \end{tabular}
        \hspace{-.5cm}
      }
    }
    (7.7,-.2);
  \node at (-.2,-1.5) { \colorbox{white}{$\phantom{{A A A}\atop {A A} }$} };
  \node at (5.1,-1.8) { \colorbox{white}{$\phantom{ A }$} };

  \begin{scope}[shift={(4,0)}]
    \draw (2,0) circle (2);
    \node at (+.5,0)
      {$\!\!\!\!\!\!0$};
    \node (zero) at (0,0) {$-$};
    \node (infinity) at (4,0) {\colorbox{white}{$\infty$}};

   \fill[black] (2,0) ++(40+180:2) node (minusepsilon)
     {\begin{turn}{-45} $)$  \end{turn}};
   \fill[black] (2,0) ++(180-40:2) node (epsilon)
     {\begin{turn}{45} $)$ \end{turn}};
   \fill[black] (2.3,0.25) ++(40+180:2)
     node (label+epsilon)
     { \tiny $-\epsilon$ };
   \fill[black] (2.3,-0.25) ++(-40-180:2)
     node (label-epsilon)
     { \tiny $+\epsilon$ };

   \draw[<->, dashed, darkblue]
     (label+epsilon)
     to
     node {\tiny $\mathbb{Z}_2$}
     (label-epsilon);
  \end{scope}

  \draw[|->, olive]
    (mirrorsubmanifold2)
    to[bend left=18]
    (zero);
  \draw[|->, olive]
    (mirrorsubmanifold3)
    to[bend left=18]
    node {
      \colorbox{white}{\bf
        \tiny
        \color{darkblue}
        \hspace{-.5cm}
        \begin{tabular}{c}
          submanifolds away from
          \\
          fixed point/singularity
        \end{tabular}
        \hspace{-.5cm}
      }
    }
    (zero);

  \draw[|->, thin, olive]
    (submanifold2)
    to[bend right=18]
    (zero);
  \draw[|->, thin, olive]
    (submanifold3)
    to[bend right=18]
    node[near end]
      {
        \colorbox{white}{\bf
        \tiny
        \color{darkblue}
        \begin{tabular}{c}
          mirror submanifolds
        \end{tabular}
        }
      }
    (zero);

  \draw[|->, thin, brown]
    (submanifold1)
    to[bend left=11]
    node[near end]
    {
      \hspace{-.6cm}
      \colorbox{white}{\bf
        \tiny
        \color{darkblue}
        \hspace{-.6cm}
        \begin{tabular}{c}
          submanifold inside
          \\
          fixed point/singularity
        \end{tabular}
        \hspace{-.6cm}
      }
    }
    (zero);
\end{tikzpicture}
\end{center}
\vspace{-.6cm}
\noindent {\bf \footnotesize Figure L  --
The $\mathbb{Z}_2$-Equivariant Cohomotopy of
Euclidean $n$-orientifolds vanishing at infinity}
{\footnotesize according to the unstable equivariant
Hopf degree theorem \ref{UnstableEquivariantHopfDegreeTheorem}
applied to sign-representation spheres (\hyperlink{FigureJ}{\it Figure J})
and visuallized by the corresponding configurations of
charged points via the unstable Pontrjagin-Thom construction
\eqref{UnstablePTTheorem}, in equivariant enhancement of the
situation show in \hyperlink{FigureE}{\it Figure E}.
The same situation, just crossed with an interval, appears in
the application to M5/MO5 charge in
\hyperlink{FigureV}{\it Figure V}.
}

\end{example}

\medskip

It is possible and instructive to make this fully explicit in the simple special case
of the 1-dimensional sign representation, where the statement of the equivariant Hopf
degree theorem \ref{UnstableEquivariantHopfDegreeTheorem} may be found in
elementary terms: It is readily checked that all the continuous functions $c^1 : S^1 \to S^1$
which take $0$ to either of $0, \infty \in S^1$ and wind around at constant parameter speed
are $\mathbb{Z}_2$-equivariant, hence are Elmendorf stages \eqref{SystemOfMapsOnHFixedSubspaces}
of $\mathbb{Z}_2$-equivariant cocycles $c$:
\begin{equation}
  \label{EquivariantMapS1sgnToItself}
  \xymatrix@R=2pt{
    \mathllap{
      \big( \mathbb{R}^{\mathbf{1}_{\mathrm{sgn}}}\big)^{\mathrm{cpt}}
      \simeq
      \;
    }
    S^{\mathbf{1}_{\mathrm{sgn}}}
    \ar[rr]^-c
    &&
    S^{\mathbf{1}_{\mathrm{sgn}}}
    \\
    S^1
    \ar[rr]^-{c^1}
    &&
    S^1
    \\
    \\
    S^0
    \ar@{^{(}->}[uu]
    \ar[rr]^-{c^{\mathbb{Z}_2}}
    &&
    S^0
    \ar@{^{(}->}[uu]
  }
  \,.
\end{equation}
If such a function vanishes at infinity \eqref{VanishingAtInfinity},
in that it takes $\infty \mapsto \infty$ as shown in
\hyperlink{FigureL}{\it Figure L}, then
we have one of two cases:
\begin{enumerate}[{\bf (i)}]
\vspace{-2mm}
  \item
    either $c^1$ {\it  winds an odd number} of times,
    so that \eqref{TheWeylGroupMultiples} reads:
    \vspace{-2mm}
    $$
      \;
      \mathrm{deg}(c^1)
      =
      \;
      \overset{
        \mathclap{
        \mbox{\bf
          \tiny
          \color{darkblue}
          offset
        }
        }
      }
      {
        \overbrace{1}
      }
      \;-\;
      N_{1} \cdot 2
      \,,
    $$
        in which case it satisfies $c^1(0) = 0$,
    so that under the PT-theorem \eqref{UnstablePTTheorem}
    there is precisely one
    charge at the singular fixed point, together with
    the even integer number
    $2 \cdot N_1 \in \mathbb{Z}$ of net charges in its ``vicinity''
    (namely: away from infinity)
    which are arranged in $\mathbb{Z}_2$-mirror pairs,
    due to the $\mathbb{Z}_2$-equivariance of $c$;
    this is what is shown on the
    left of  \hyperlink{FigureL}{\it Figure L};
    \vspace{-2mm}
  \item
    or $c^1$ {\it winds an even number} of times
    so that \eqref{TheWeylGroupMultiples} reads:
    \vspace{-2mm}
    $$
      \;
      \mathrm{deg}(c^1)
      =
      \overset{
        \mathclap{
        \mbox{\bf
          \tiny
          \color{darkblue}
          offset
        }
        }
      }{
        \overbrace{0}
      }
      \;-\;
      N_{1} \cdot 2
      \,,
    $$
        in which case it satisfies $c^1(0) = \infty$,
    so that under the PT-theorem \eqref{UnstablePTTheorem}
    there is no
    charge at the singular fixed point,
    but a net even integer number
    $2 \cdot N_1 \in \mathbb{Z}$ of charges in its vicinity,
    as before.
\end{enumerate}
\begin{remark} [Number of branes and offset]
Notice that:
\begin{enumerate}[{\bf (i)}]
\vspace{-3mm}
\item For $N_1 = 0$ (no branes)
this is the situation of \eqref{InductionStartForRepSpheres}:
either there is a non-vanishing charge associated with the
singular fixed point (O-plane charge), or not.
\vspace{-3mm}
\item
  Furthermore, if there is, it is either +1 or -1,
  so that in general
  the charge associated with the singular fixed point is
  in $\{0, \pm 1\}$, as befits O-plane charge according to
  \hyperlink{FigureOP}{\it Figure OP}.
\vspace{-3mm}
\item The offset is relevant only modulo 2, so that we could have
chosen an offset of $+1$ instead of as $-1$ in the first case.
This choice just fixes the sign convention for D-brane/O-plane
charge.
\end{enumerate}

\end{remark}

\noindent {\bf Characterizing the brane content around a singularity.}
In the above example
in RO-degree $\mathbf{1}_{\mathrm{sgn}}$ \eqref{EquivariantMapS1sgnToItself},
it is clear that
the configurations of branes
implied by the unstable equvariant Hopf degree theorem
(Theorem \ref{UnstableEquivariantHopfDegreeTheorem})
appear in multiples of the
regular $G$-set around a fixed O-plane charge
stuck in the singularity,
as illustrated in \hyperlink{FigureL}{\it Figure L} and
as demanded by the local/twisted tadpole cancellation conditions
according to \hyperlink{Table1}{\it Table 1}.
In order to prove that this is the case generally, we
now turn to the stabilized equivariant Hopf degree theorem
(Theorem \ref{CharacterizationOfStabilizationOfUnstableCohomotopy} below),
which concretely characterizes the (virtual) $G$-sets of branes
that may appear classified by equivariant Cohomotopy.

\medskip

\subsubsection{Stable equivariant Hopf degree}
\label{StableEquivariantHopfDegree}
In a homotopy-theoretic incarnation of perturbation theory,
we may approximate unstable equivariant Cohomotopy (Theorem \ref{CharacterizationOfStabilizationOfUnstableCohomotopy})
by its homotopically linearized, namely stabilized (see \cite{BSS18}) version.
We briefly recall the basics of  stable equivariant Cohomotopy in RO-degree 0
(\cite{Segal71}, see \cite[7.6 \& 8.5]{tomDieck79}\cite{Lueck}) before applying this in
Theorem \ref{CharacterizationOfStabilizationOfUnstableCohomotopy} and Prop. \ref{TheoremLocalTadpoleCancellation} below.

\medskip

\noindent {\bf Equivariant suspension.}
For $V,W \in \mathrm{RO}(G)$
two orthogonal linear $G$-representations, and for
$$
  \big[
    S^V
    \simeq
    \big( \mathbb{R}^V\big)^{\mathrm{cpt}}
    \overset{c}{\longrightarrow}
    \big( \mathbb{R}^V\big)^{\mathrm{cpt}}
    \simeq
    S^V
  \big]
  \;\in\;
  \pi^V_G
  \big(
    \big( \mathbb{R}^V \big)^{\mathrm{cpt}}
  \big)
$$
the class of a cocycle in the equivariant Cohomotopy \eqref{EquivariantCohomotopySet}
of the Euclidean $G$-space $\mathbb{R}^V$ \eqref{EuclideanGSpace}
in compatible $\mathrm{RO}$-degree $V$ (Example \ref{ExamplesOfCompatibleRODegree})
and vanishing at infinity \eqref{VanishingAtInfinity},
we obtain the class of a cocycle vanishing at infinity
on the product $G$-space $\mathbb{R}^{V \oplus W}$
\eqref{CartesianProductOfEuclideanGSpaces}
in compatible degree $V \oplus W$,
simply by forming the
Cartesian product of $c$ with the identity on $\mathbb{R}^W$.
This is the {\it equivariant suspension} of $c$ by RO-degree $W$:
\begin{equation}
  \label{EquSuspension}
  \underset{
    \mathclap{
    \mbox{\bf
     \tiny
     \color{darkblue}
     \begin{tabular}{c}
       equivariant suspension
       \\
       by RO-degree $W$
       \\
       of equivariant Cohomotopy cocycle
     \end{tabular}
    }
    }
  }{
    \Sigma^W c
  }
  \hspace{1cm} \coloneqq \hspace{.2cm}
    \Big[
    \big(\mathbb{R}^V \times \mathbb{R}^W\big)^{\mathrm{cpt}}
    \underset{
      \mathclap{
      \mbox{\bf
        \tiny
        \color{darkblue}
        \begin{tabular}{c}
          \phantom{a}
          \\
          compactified Cartesian product
          \\
          of cocycle with identity on $\mathbb{R}^W$
        \end{tabular}
      }
      }
    }{
      \xrightarrow{\;\;
        c \,\times\, \mathrm{id}_{\mathbb{R}^W}
      }
    }
    \big(\mathbb{R}^V \times \mathbb{R}^W\big)^{\mathrm{cpt}}
  \Big]
  \;\in\;
  \pi^{V \oplus W}_G
  \big(
    \big( \mathbb{R}^{V \oplus W} \big)^{\mathrm{cpt}}
  \big)
  \,.
\end{equation}
Note that this reduces to the ordinary suspension operation \cref{HopfDegreesUnderSuspension}
for $G = 1$ the trivial group, hence for RO-degrees $\mathbf{n}_{\mathbf{triv}} = n$.
These equivariant suspension operations
form a directed system on the collection
of equivariant Cohomotopy sets \eqref{EquivariantCohomotopySet},
indexed by inclusions of orthogonal linear representations:
\begin{equation}
  \label{DirectedSystemOfEquivariantSuspensionMaps}
  \big(
    V \hookrightarrow V \oplus W
  \big)
  \;\;\longmapsto\;\;
  \Big(
  \pi^V_G\big( \big(\mathbb{R}^V\big)^{\mathrm{cpt}} \big)
  \xrightarrow{ \;\Sigma^W}
  \pi^{V \oplus W}_G\big( \big(\mathbb{R}^{V \oplus W}\big)^{\mathrm{cpt}} \big)
  \Big)
  \,.
\end{equation}

\medskip

\noindent {\bf Stable equivariant Cohomotopy.}
As a consequence of the above, one may consider the union
of all unstable equivariant Cohomotopy sets
of representation spheres in all compatible degrees,
with respect to the identifications along
the equivariant suspension maps \eqref{DirectedSystemOfEquivariantSuspensionMaps}
(the colimit of this system):
\begin{equation}
  \label{StableEquivariantCohomotopyOfThePoint}
  \hspace{-1.8cm}
  \underset{
    \mbox{\bf
      \tiny
      \color{darkblue}
      \begin{tabular}{c}
        unstable equivariant Cohomotopy set
        \\
        in compatible RO-degree $V$
      \end{tabular}
    }
  }{
    \pi^V_G
    \big(
      \big(
        \mathbb{R}^V
      \big)^{\mathrm{cpt}}
    \big)
  }
  \;\;
  \underset{
    \mbox{\bf
      \tiny
      \color{darkblue}
      \begin{tabular}{c}
        stabilization map
        \\
        (coprojection into colimit)
      \end{tabular}
    }
  }{
    \xymatrix{ \;\;\;\; \ar[rr]^{\Sigma^\infty} &&   }
  }
  \;\;\;\;\;\;\;\;\;\;
  \underset{
    \mathclap{
    \mbox{\bf
      \tiny
      \color{darkblue}
      \begin{tabular}{c}
        union (colimit) of
        \\
        all unstable Cohomotopy sets
        \\
        in compatible RO-degrees
        \\
        identified along equivariant suspensions
      \end{tabular}
    }
    }
  }{
    \underset{
      \underset{W}{\longrightarrow}
    }{\mathrm{lim}}
    \;
    \pi^W_G
    \big(
      \big(
        \mathbb{R}^W
      \big)^{\mathrm{cpt}}
    \big)
  }
  \hspace{.7cm}= \hspace{.8cm}
  \underset{
    \mathclap{
    \mbox{\bf
      \tiny
      \color{darkblue}
      \begin{tabular}{c}
        stable
        \\
        equivariant Cohomotopy ring
        \\
        in RO-degree 0
      \end{tabular}
    }
    }
  }{
    \mathbb{S}^0_G
  }.
\end{equation}
Since the resulting union/colimit is, by construction, stable under taking further such suspensions,
this is called the {\it stable equivariant Cohomotopy in degree 0}
(\cite[p. 1]{Segal71}, see \cite[p. 9-10]{Lueck})
also called the \emph{0th stable $G$-equivariant homotopy group of spheres}
or the \emph{$G$-equivariant stable 0-stem} or similar
(see \cite[IX.2]{May96}\cite[3]{Schwede}).
Notice that here the stable RO-degree is
the formal difference of the unstable RO-degree by
the RO-degree of the singularity, so that vanishing stable RO-degree
is another expression of compatibility of unstable degree, in the sense of
Example \ref{ExamplesOfCompatibleRODegree}:

\vspace{-5mm}
$$
  -
  \underset{
    \mbox{\bf
      \tiny
      \color{darkblue}
      \begin{tabular}{c}
        RO-degree of
        \\
        singularity
      \end{tabular}
    }
  }{
    \underbrace{
      V
    }
  }
  +
  \underset{
    \mbox{\bf
      \tiny
      \color{darkblue}
      \begin{tabular}{c}
        compatible RO-degree of
        \\
        unstable Cohomotopy
      \end{tabular}
    }
  }{
    \underbrace{
      V
    }
  }
  =
  \underset{
    \mbox{\bf
      \tiny
      \color{darkblue}
      \begin{tabular}{c}
        RO-degree of
        \\
        stable Cohomotopy
      \end{tabular}
    }
  }{
    \underbrace{
      0
    }
  }
$$
via
$$
\pi_V(S^W)=[S^V, S^W]=\pi^W(S^V) \xrightarrow{\Sigma^\infty}
S^{W- V}.
$$
Extensive computation of stable $\mathbb{Z}_2$-equivariant Cohomotopy
of representation spheres in non-vanishing RO-degrees,
i.e., computation of
the abelian groups
$\mathbb{S}^{\mathbf{n}_{\mathrm{sgn}} + \mathbf{m}_{\mathrm{triv}}}_{\mathbb{Z}_2}$,
is due to \cite{ArakiIriye82}\cite{Iriye82};
see also \cite[5]{DuggerIsaksen16}\cite[p. 10-15]{Dugger08}.
Under \hyperlink{HypothesisH}{\it Hypothesis H}, these
groups are relevant for tadpole cancellation with branes
wrapping orientifold singularities non-transversally.
This is of interest to us but goes beyond the scope of this article.

\medskip

\noindent {\bf Equivalence to the Burnside ring.}
Due to the stabilization, the stable equivariant Cohomotopy
set \eqref{StableEquivariantCohomotopyOfThePoint}
has the structure of an abelian group,
in fact
the structure of a ring. As such, it is isomorphic to the
{\it Burnside ring} $A(G)$ of virtual $G$-sets
(\cite{Burnside01}\cite{Solomon67}\cite[1]{tomDieck79},
for exposition in our context see \cite[2]{SS19b}):
\vspace{-2mm}
\begin{equation}
  \label{IsoToAG}
\hspace{1cm}
  \overset{
    \mathclap{
    \mbox{\bf
      \tiny
      \color{darkblue}
      \begin{tabular}{c}
        stable
        \\
        equivariant Cohomotopy
      \end{tabular}
    }
    }
  }{
    \mathbb{S}_G^0
  }
  \hspace{8mm} \simeq \hspace{3mm}
  \overset{
    \mathclap{
    \mbox{\bf
      \tiny
      \color{darkblue}
      \begin{tabular}{c}
        Burnside ring
      \end{tabular}
    }
    }
  }{
    A(G)
   }
  \;\;=\;\;
  \big\{ \!\!
    \mbox{
        Virtual $G$-sets
    }
 \!\! \big\}.
\end{equation}
This result is due to \cite[p. 2]{Segal71};
see \cite[7.6.7 \& 8.5.1]{tomDieck79}\cite[1.13]{Lueck}, we highlight its geometric
meaning below; see \hyperlink{FigureM}{\it Figure M}.
This is a non-linear analog
(more precisely, the analog over the absolute base ``field'' $\mathbb{F}_1$
\cite[p. 3]{Cohn04}\cite[2.5.6]{Durov07}) of the fact that the
equivariant K-theory in degree 0 is the representation ring
of virtual linear $G$-representations
over the field of real numbers (see, e.g., \cite[3]{Greenlees05}):
\begin{equation}
  \label{ROG}
  \overset{
    \mathclap{
    \mbox{\bf
      \tiny
      \color{darkblue}
      \begin{tabular}{c}
        equivariant K-theory
      \end{tabular}
    }
    }
  }{
    \mathrm{KO}_G^0
  }
  \hspace{4mm} \simeq \hspace{4mm}
  \overset{
    \mathclap{
    \mbox{\bf
      \tiny
      \color{darkblue}
      \begin{tabular}{c}
        representation ring
      \end{tabular}
    }
    }
  }{
    \mathrm{RO}(G)
   }
 \; \;=\;
  \big\{ \!\!
    \mbox{
        Virtual $G$-representations
    }
 \!\! \big\}.
\end{equation}
In fact, the operation $S \mapsto \mathbb{R}[S]$ that sends a (virtual) $G$-set $S \in \mathrm{A}(G)$
to its linearization, hence to its linear span $\mathbb{R}[S]$,
hence to the (virtual) permutation representation that it induces
(see \cite[4]{tomDieck79}\cite[2]{SS19b}),  is
a ring homomorphism from the Burnside ring to the representation
ring.  Furthermore, it exhibits the value on the point of
unique multiplicative morphism
from equivariant stable Cohomotopy theory to equvariant K-theory,
called the \emph{Boardman homomorphism} \cite[II.6]{Adams74},
which is the Hurewicz homomorphism generalized from
ordinary cohomology to generalized cohomology theories:

\vspace{-2mm}
\begin{equation}
  \label{BoardmanH}
  \xymatrix@C=1.6em@R=2pt{
    \mbox{\bf
      \tiny
      \color{darkblue}
      \begin{tabular}{c}
        equivariant stable Cohomotopy
      \end{tabular}
    }
    &
    \mathbb{S}_G^0
    \ar@{}[dd]|-{
      \begin{rotate}{270}
        $\!\!\simeq$
      \end{rotate}
    }
    \ar[rrrr]^{
      \overset{
        \mbox{\bf
          \tiny
          \color{darkblue}
          \begin{tabular}{c}
            Boardman homomorphism
          \end{tabular}
        }
      }{
        \beta
      }
    }
    &&&&
    \mathrm{KO}_G^0
    \ar@{}[dd]|-{
      \begin{rotate}{270}
        $\!\!\simeq$
      \end{rotate}
    }
    &
    \mbox{\bf
      \tiny
      \color{darkblue}
      \begin{tabular}{c}
        equivariant K-theory
      \end{tabular}
    }
    \\
    \\
    \mbox{\bf
      \tiny
      \color{darkblue}
      \begin{tabular}{c}
        Burnside ring
      \end{tabular}
    }
    &
    \mathrm{A}(G)
    \ar[rrrr]_{
      \underset{
        \mathclap{
        \mbox{\bf
          \tiny
          \color{darkblue}
          \begin{tabular}{c}
            linearization
            \\
            sending $G$-sets $S$ to
            \\
            linear $G$-representations $\mathbb{R}[S]$
          \end{tabular}
        }
        }
      }{
        S \, \mapsto \, \mathbb{R}[S]
      }
    }
    &&&&
    \mathrm{RO}(G)
    &
    \mbox{\bf
      \tiny
      \color{darkblue}
      \begin{tabular}{c}
        representation ring
      \end{tabular}
    }
  }
\end{equation}

In summary, the composite of the stabilization morphism
\eqref{StableEquivariantCohomotopyOfThePoint}
with the isomorphism \eqref{IsoToAG} to the Burnside ring explicitly
extracts from any cocycle $c$ in unstable equivariant Cohomotopy
a virtual $G$-set $\{\mathrm{branes}\}$, hence a
virtual $G$-permutation representation
$\mathbb{R}[ \{\mathrm{branes}\}]$.
The following theorem explicitly identifies this $G$-set
$\{branes\}$ in terms of the Elmendorf stage-wise Hopf degrees
of the cocycle $c$; see \hyperlink{FigureM}{\it Figure M} below for
illustration.

\begin{theorem}[\bf Stabilized equivariant Hopf degree theorem for representation spheres]
 \label{CharacterizationOfStabilizationOfUnstableCohomotopy}
Consider a cocycle $c$ in
unstable Cohomotopy of the vicinity of a $G$-singularity $\mathbb{R}^V$
(Def. \ref{CohomotopyOfVicinityOfSingularity}).
Its image under stabilization in equivariant stable Cohomotopy
\eqref{StableEquivariantCohomotopyOfThePoint} is,
under the identification \eqref{IsoToAG} with the Burnside ring,
precisely that virtual $G$-set $\{\mathrm{branes}\} \in A(G)$
whose net number of $H$-fixed points (``Burnside marks'', see \cite[2]{SS19b}),
equals the Hopf degree of $c$
at any Elmendorf stage $H \in \mathrm{Isotr}_X(G)$ \eqref{ElmedorfStageWiseHopfDegrees}.
Hence if $H = \langle g\rangle$ is a cyclic group generated
by an element $g \in G$, this number also equals the character value at $g$
(i.e., the trace of the linear action of $g$)
on the linear representation $\mathbb{R}\big[\{\mathrm{branes}\}\big]$:

\vspace{-8mm}
\begin{equation}
  \label{MorphismFromUnstableEquivariantCohomotopyToRepresentationRing}
  \xymatrix@R=1pt{
    &
    \overset{
      \mathclap{
      \mbox{\bf
        \tiny
        \color{darkblue}
        \begin{tabular}{c}
          cocycle in
          \\
          unstable equivariant Cohomotopy
        \end{tabular}
      }
      }
    }{
      c
    }
        \ar@{}[dd]|-{
      \begin{rotate}{270}
        $\!\!\in$
      \end{rotate}
    }
  \ar@{|->}[rr]
    &&
    \overset{
      \mathclap{
      \mbox{\bf
        \tiny
        \color{darkblue}
        \begin{tabular}{c}
          cocycle in
          stable equivariant Cohomotopy
          \\
          $\simeq$ virtual $G$-set of $\mathrm{branes}$
        \end{tabular}
      }
      }
    }{
      \{\mathrm{branes}\}
    }
    \ar@{}[dd]|-{
      \begin{rotate}{270}
        $\!\!\in$
      \end{rotate}
    }
    \ar@{|->}[rr]
    &&
    \overset{
      \mathclap{
      \mbox{\bf
        \tiny
        \color{darkblue}
        \begin{tabular}{c}
          virtual linear $G$-representation
          \\
          spanned by virtual $G$-set of branes
        \end{tabular}
      }
      }
    }{
      \mathbb{R}[ \{\mathrm{branes}\}]
    }
    \ar@{}[dd]|-{
      \begin{rotate}{270}
        $\!\!\in$
      \end{rotate}
    }
    \\
    \\
    &
    \pi^V_G
    \big(
      \big(
        \mathbb{R}^V
      \big)^{\mathrm{cpt}}
    \big)
    \ar@{}[dd]|-{
      \mathllap{
        \mbox{\bf
          \tiny
          \color{darkblue}
          \eqref{UnstableEquivariantCohomotopyOfRepresentationSphereInCompatibleDegree}
        }
        \;
      }
      \begin{rotate}{270}
        $\!\!\!\!\!\simeq$
      \end{rotate}
    }
    \ar[rr]^-{
      \overset{
        \mathclap{
        \mbox{\bf
          \tiny
          \color{darkblue}
          stabilization
        }
        }
      }{
        \Sigma^\infty
      }
    }
    &&
    \mathbb{S}_G^0
    \ar@{}[dd]|-{
      \mathllap{
        \mbox{\bf
          \tiny
          \color{darkblue}
          \eqref{IsoToAG}
        }
        \;
      }
      \begin{rotate}{270}
        $\!\!\!\!\!\simeq$
      \end{rotate}
    }
    \ar[rr]^-{\beta}
    &&
    \mathrm{KO}_G^0
    \ar@{}[dd]|-{
      \mathllap{
        \mbox{\bf
          \tiny
          \color{darkblue}
          \eqref{ROG}
        }
        \;
      }
      \begin{rotate}{270}
        $\!\!\!\!\!\simeq$
      \end{rotate}
    }
    \\
    \\
    &
    \mathbb{Z}^{{}^{\mathrm{Isotr}^{d_{\mathrm{fix}}\gt 0 }_X(G)}}
    \times
    \{0,1\}^{{}^{\mathrm{Isotr}^{d_{\mathrm{fix}}= 0 }_X(G)}}
    \ar[rr]
    &&
    \mathrm{A}(G)
    \ar[rr]_-{
      \underset{
        \mbox{\bf
          \tiny
          \color{darkblue}
          linearization
        }
      }
      {
        \mathbb{R}[-]
      }
    }
    &&
    \mathrm{RO}(G)
    \\
    \underset
    {
      \mathclap{
      \mbox{\bf
        \tiny
        \color{darkblue}
        \begin{tabular}{c}
          Elmendorf
          \\
          stage \eqref{SystemOfMapsOnHFixedSubspaces}
        \end{tabular}
      }
      }
    }
    {
      G \supset H
      \mathrlap{\; \colon}
    }
    &
    \underset{
      \mathclap{
      \mbox{\bf
        \tiny
        \color{darkblue}
        \begin{tabular}{c}
          Hopf degree at stage $H$
          \eqref{ElmedorfStageWiseHopfDegrees}
          \\
          of Cohomotopy cocycle
        \end{tabular}
      }
      }
    }{
      \mathrm{deg}\big(
        c^H
      \big)
    }
    \ar@{=}[rr]
    &&
    \underset{
      \mathclap{
      \mbox{\bf
        \tiny
        \color{darkblue}
        \begin{tabular}{c}
          net number of $H$-fixed elements
          \\
          = Burnside marks at stage $H$
          \\
          in virtual set of $\mathrm{branes}$
        \end{tabular}
      }
      }
    }{
      \Big\vert \{\mathrm{branes}\}^H \Big\vert
    }
    \\
    \underset
    {
      \mathclap{
      \mbox{\bf
        \tiny
        \color{darkblue}
        \begin{tabular}{c}
          Elmendorf stage
          \\
          at cyclic subgroup
          \\
          generated by $g \in G$
        \end{tabular}
      }
      }
    }
    {
      H = \langle g\rangle
      \mathrlap{\; \colon}
    }
    &
    \underset{
      \mathclap{
      \mbox{\bf
        \tiny
        \color{darkblue}
        \begin{tabular}{c}
          Hopf degree at stage $H = \langle g\rangle$
          \eqref{ElmedorfStageWiseHopfDegrees}
          \\
          of Cohomotopy cocycle
        \end{tabular}
      }
      }
    }{
      \mathrm{deg}\big(
        c^{\langle g \rangle}
      \big)
    }
    \ar@{=}[rr]
    &&
    \underset{
      \mathclap{
      \mbox{\bf
        \tiny
        \color{darkblue}
        \begin{tabular}{c}
          virtual number of fixed points
          \\
          under action of $g \in G$
          \\
          on virtual set of branes
        \end{tabular}
      }
      }
    }{
      \Big\vert \{\mathrm{branes}\}^{\langle g \rangle} \Big\vert
    }
    \ar@{=}[rr]
    &&
    \underset{
      \mathclap{
      \mbox{\bf
        \tiny
        \color{darkblue}
        \begin{tabular}{c}
          character value at $g$
          \\
          of virtual permutation representation
          \\
          spanned by $\mathrm{branes}$
        \end{tabular}
      }
      }
    }{
      \mathrm{Tr}_{g}\big( \mathbb{R}[ \{\mathrm{branes} \} ] \big)
    }
  }
\end{equation}
\end{theorem}
\begin{proof}
  For the case that all fixed subspace dimensions
  are positive, this is essentially the statement
  of \cite[8.5.1]{tomDieck79}, after unwinding the definitions there
  (see \cite[p. 190]{tomDieck79}).
  We just need to see that the statement generalizes
  as claimed to the case where
  the full fixed subspace $\big( S^V\big)^G = S^0$ is the 0-sphere.
  But, under stabilization map $\Sigma^\infty$
  \eqref{StableEquivariantCohomotopyOfThePoint}, the image
  of a Cohomotopy cocycle $S^V \xrightarrow{\;\;c\;\;} S^V$
  and its equivariant suspension \eqref{EquSuspension}
  $
    S^{\mathbf{1}_{\mathrm{triv}}\oplus V }
    \xrightarrow{
      \Sigma^{\mathbf{1}_{\mathrm{triv}}}c
    }
    S^{\mathbf{1}_{\mathrm{triv}} \oplus V}
  $
  by, in particular, the trivial 1-dimensional representation,
  have the same image
  $
    \Sigma^\infty ( c )
    \simeq
    \Sigma^\infty \big( \Sigma^{\mathbf{1}_{\mathrm{triv}}} c \big)
      $.
  Now to the suspended cocycle $\Sigma^{\mathbf{1}_{\mathrm{triv}}} c$
  the theorem \cite[8.5.1]{tomDieck79} applies, and hence the claim follows from
  the fact \eqref{HopfDegreesUnderSuspension} that the
  unstable Hopf degree in $\{0,1\}$ injects
  under suspension into the stable Hopf degrees:
  $$
    \big[ S^0 = \big( S^V\big)^G \xrightarrow{\;c^G\;} \big( S^V\big)^G = S^0 \big]
    \;\in\;
    \{0,1\}
    \hookrightarrow \mathbb{Z}
    \,.
  $$

  \vspace{-7mm}
\end{proof}

\medskip

\noindent {\it For ADE-singularities}, this implies the following
(see \hyperlink{FigureM}{\it Figure M}):

\begin{prop}[\bf Classification of Cohomotopy charge in the vicinity of ADE-singularities]
\label{TheoremLocalTadpoleCancellation}
Consider $G = G^{\mathrm{ADE}} \subset \mathrm{SU}(2)$ a finite ADE-group
\eqref{ADESubgroups} and $\mathbf{4}_{\mathbb{H}}$ its canonical quaternionic
representation \eqref{TheQuaternionicRepresentation}. Then  the homomorphism \eqref{MorphismFromUnstableEquivariantCohomotopyToRepresentationRing}
from Theorem \ref{CharacterizationOfStabilizationOfUnstableCohomotopy}
identifies the unstable Cohomotopy  of the vicinity of the $G^{\mathrm{ADE}}$-singularity
$\mathbb{R}^{\mathbf{4}_{\mathbb{H}}}$ (Def. \ref{CohomotopyOfVicinityOfSingularity})
with its image in the representation ring
\vspace{-3mm}
$$
  \xymatrix{
  \pi^{\mathbf{4}_{\mathbb{H}}}_{G^{\mathrm{ADE}}}
  \big(
    \big(
      \mathbb{R}^{\mathbf{4}_{\mathbb{H}}}
    \big)^{\mathrm{cpt}}
  \big)
  \; \ar@{^{(}->}[rr]^{ \beta \circ \Sigma^\infty }
  &&
  \mathrm{KO}_G^0
  \simeq\mathrm{RO}(G)
  }
$$
which consists of all the virtual representations of the form
\begin{equation}
  \mathbb{R}[\{\mathrm{branes}\}]
  \;=\;
  N_{\color{darkblue} \bf \mathrm{Opla}} \cdot \mathbf{1}_{\mathrm{triv}}
  -
  N_{\color{darkblue} \bf \mathrm{brane} \atop \mathrm{int}} \cdot \mathbf{k}_{\mathrm{reg}}
  \phantom{AAAA}
  \mbox{for}
  \phantom{AA}
  \begin{array}{rcl}
    N_{\color{darkblue} \bf \mathrm{brane} \atop \mathrm{int}}
    &\in& \mathbb{N}
    \,,
    \\
    N_{\color{darkblue} \bf \mathrm{Opla}} &\in& \{0, 1\}
  \end{array}
\end{equation}
hence of the form of the local/twisted tadpole cancellation conditions in
\hyperlink{Table1}{\it Table 1} and
\hyperlink{Table2}{\it Table 2}.
\end{prop}
\begin{proof}
By \eqref{FixedSubspacesOfQuaternionRepresentation},  the representation
$\mathbf{4}_{\mathbb{H}}$  is such that \emph{every} non-trivial subgroup
$1 \neq H \subset G$ has a 0-dimensional fixed space:
$$
  \mathrm{dim}
  \left(
  \big(
    \mathbb{R}^{\mathbf{4}_{\mathbb{H}}}
  \big)^H
  \right)
  \;=\;
  \left\{
  \begin{array}{cc}
    4 & \mbox{if}\;H = 1
    \\
    0 & \mbox{otherwise}.
  \end{array}
  \right.
  $$
This means that for
$
  c \in \pi^{\mathbf{4}_{\mathbb{H}}}
  \big(
    \big( \mathbb{R}^{\mathbf{4}_{\mathbb{H}}}\big)^{\mathrm{cpt}}
  \big)
$
an equivariant Cohomotopy cocycle in the vicinity of an ADE-singularity,
its only Elmendorf stage-wise Hopf degree \eqref{ElmedorfStageWiseHopfDegrees}
in positive dimension is, by equation \eqref{TheWeylGroupMultiples} in Theorem \ref{UnstableEquivariantHopfDegreeTheorem}, of the form
$$
  \mathrm{deg}\big( c^1 \big)
  \;=\;
   \overset{
     \mbox{\tiny $\in \{0, 1\}$}
   }{
     \overbrace{ Q_{\mathrm{Opla}} }
   }
   -
   N_1 \cdot \vert G \vert
  \,,
$$
where we used the fact that $W_G(1) = G$ \eqref{ExtremeCasesOfWeylGroups}.
But, by Theorem \ref{CharacterizationOfStabilizationOfUnstableCohomotopy},
this implies that the virtual $G$-set
$\{\mathrm{branes}\}$ of branes corresponding
to $c$ has the following Burnside marks
$$
  \{\mathrm{branes}\}^H
  \;=\;
  \left\{
  \begin{array}{cc}
    Q_{\mathrm{Opla}}
    -
    N_1 \cdot \left\vert G \right\vert
    & \mbox{if}\; H = 1
    \\
    Q_{\mathrm{Opla}} & \mbox{otherwise}
    \,,
  \end{array}
  \right.
$$
hence that the corresponding permutation representation of branes
has the following characters:
$$
  \mathrm{Tr}_g\big( \mathbb{R}\{\mathrm{branes}\}\big)
  \;=\;
  \left\{
  \begin{array}{cc}
    Q_{\mathrm{Opla}}
    -
    N_1 \cdot \left\vert G \right\vert
    & \mbox{if}\; g = e
    \\
    Q_{\mathrm{Opla}}
    &
    \mbox{otherwise}\,.
    \end{array}
  \right.
 $$
The unique $G$-set/$G$-representation with these
Burnside marks/characters is the sum of
the $N_1$-fold multiple of the regular $G$-set/$G$-representation
and the $Q_{\mathrm{Opla}}$-fold multiple of the trivial representation
(see \hyperlink{FigureM}{\it Figure M}):
$$
  \mathbb{R}[\{\mathrm{branes}\}]
  \;=\;
  Q_{\mathrm{Opla}} \cdot \mathbf{1}_{\mathrm{triv}}
  \;-\;
  N_1 \cdot \mathbf{k}_{\mathrm{reg}}
  \,,
$$

\vspace{-7mm}
\end{proof}

\medskip
\noindent The situation is illustrated by \hyperlink{FigureM}{\it Figure M}:
\begin{center}
\hypertarget{FigureM}{}
\begin{tikzpicture}[scale=.8]

  \begin{scope}[shift={(0,1)}]

  \draw node at (0,7.2)
    {
      \tiny
      \color{darkblue} \bf
      \begin{tabular}{c}
        equivariant Cohomotopy
        \\
        vanishing at infinity
        \\
        of Euclidean $G$-space
        \\
        in compatible RO-degree $V$
      \end{tabular}
    };

  \draw node at (0,6)
    {
      $
      \pi^{V}_{{}_{G}}
      \big(
        \big(
          \mathbb{R}^V
        \big)^{\mathrm{cpt}}
      \big)
    $
    };

  \draw[->] (0+2,6)
    to
      node {\colorbox{white}{\small $\Sigma^\infty$}}
      node[above]
        {
          \raisebox{.3cm}{
          \tiny
          \color{darkblue} \bf
          stabilization
          }
        }
    (6-.8,6);

  \draw node at (6,6.9)
    {
      \tiny
      \color{darkblue} \bf
      \begin{tabular}{c}
        stable
        \\
        equivariant
        \\
        Cohomotopy
      \end{tabular}
    };

  \draw node at (6,6)
    {
      $
        \mathbb{S}_G^0
      $
    };

  \draw[->] (6+.7,6)
    to
      node{\colorbox{white}{\footnotesize $\beta$}}
      node[above]
        {
          \raisebox{.3cm}{
          \tiny
          \color{darkblue} \bf
          \begin{tabular}{c}
            Boardman
            \\
            homomorphism
          \end{tabular}
          }
        }
    (11-.9,6);

  \draw node at (6,5.5)
    {
      \begin{rotate}{270}
        $\!\!\simeq$
      \end{rotate}
    };

  \draw node at (6,5)
    {
      $
        A_G
      $
    };

  \draw node at (6,4.3)
    {
      \tiny
      \color{darkblue} \bf
      \begin{tabular}{c}
        Burnside
        \\
        ring
      \end{tabular}
    };

  \draw[->] (6+.7,5) to
    node
    {
      \colorbox{white}
      {
        \small
        $
        \underset
        {
          \mbox{
            \tiny
            \color{darkblue} \bf
            linearization
          }
        }
        {\footnotesize
          \mathbb{R}[-]
        }
        $
      }
    }
    (11-.7,5);

  \draw node at (11,6.8)
    {
      \tiny
      \color{darkblue} \bf
      \begin{tabular}{c}
        equivariant
        \\
        K-theory
      \end{tabular}
    };

  \draw node at (11,6)
    {
      $
        \mathrm{KO}_G^0
      $
    };

  \draw node at (11,5.5)
    {
      \begin{rotate}{270}
        $\!\!\simeq$
      \end{rotate}
    };

  \draw node at (11.2,5)
    {
      $
        \mathrm{RO}(G)
      $
    };

  \draw node at (11,4.3)
    {
      \tiny
      \color{darkblue} \bf
      \begin{tabular}{c}
        representation
        \\
        ring
      \end{tabular}
    };

  \end{scope}

\begin{scope}[shift={(0,.4)}]

  \draw node at (0,4)
    {$
      \overset{
      }{
        \mbox{
          \tiny
          e.g. one $O^{{}^{-}}\!\!$-plane
          and two branes
        }
      }
    $};

  \draw node at (0,3)
    {
      $\overbrace{\phantom{AAAAAAAAAAAAAAAAAAA}}$
    };

  \draw node at (6,4)
    {$
      \overset{
      }{
        \mbox{
          \tiny
          \begin{tabular}{c}
            minus the trivial $G$-set
            \\
            with two regular $G$-sets
          \end{tabular}
        }
      }
    $};

  \draw node at (11,3)
    {
      $\overbrace{\phantom{AAAAAA}}$
    };

  \draw node at (11,4)
    {$
      \overset{
      }{
        \mbox{
          \tiny
          \begin{tabular}{c}
            minus the trivial $G$-representation
            \\
            plus two times the regular $G$-representation
          \end{tabular}
        }
      }
    $};

  \draw node at (6,3)
    {
      $\overbrace{\phantom{AAAAAA}}$
    };

  \draw[dashed]
    (0,0) circle (2);

  \draw (0,2+.2) node {\footnotesize $\infty$};
  \draw (0,-2-.2) node {\footnotesize $\infty$};
  \draw (2+.3,0) node {\footnotesize $\infty$};
  \draw (-2-.3,0) node {\footnotesize $\infty$};

  \draw[fill=white]
    (0,0) circle (.07);

  \draw[fill=black]
    (18:.6) circle (.07);
  \draw[fill=black]
    (18+90:.6) circle (.07);
  \draw[fill=black]
    (18+180:.6) circle (.07);
  \draw[fill=black]
    (18+270:.6) circle (.07);

  \draw[fill=black]
    (58:1.2) circle (.07);
  \draw[fill=black]
    (58+90:1.2) circle (.07);
  \draw[fill=black]
    (58+180:1.2) circle (.07);
  \draw[fill=black]
    (58+270:1.2) circle (.07);

  \draw[|->] (3.6,0) to ++(.5,0);
  \draw[|->] (8.4,0) to ++(.5,0);

  \begin{scope}[shift={(6,.2)}]

    \draw[fill=white] (0,2) circle (0.07);

    \draw (5,2) node {\small $+\mathbf{1}_{{}_{\mathrm{triv}}}$ };

    \draw[|->, olive] (-.05,2.12) arc (30:325:.2);

    \begin{scope}[shift={(0,.3)}]
    \draw[fill=black] (0,1) circle (0.07);
    \draw[fill=black] (0,.5) circle (0.07);
    \draw[fill=black] (0,0) circle (0.07);
    \draw[fill=black] (0,-.5) circle (0.07);

    \draw[|->, olive] (0-.1,1-.05) arc (90+6:270-16:.2);
    \draw[|->, olive] (0-.1,.5-.05) arc (90+6:270-16:.2);
    \draw[|->, olive] (0-.1,0-.05) arc (90+6:270-16:.2);

    \draw[|->, olive] (0+.1,-.5+.1) to[bend right=60] (0+.1,1-.03);

    \draw (5,0.25) node {\small $-\mathbf{4}_{{}_{\mathrm{reg}}}$ };

    \end{scope}

    \begin{scope}[shift={(0,-1.7)}]
    \draw[fill=black] (0,1) circle (0.07);
    \draw[fill=black] (0,.5) circle (0.07);
    \draw[fill=black] (0,0) circle (0.07);
    \draw[fill=black] (0,-.5) circle (0.07);

    \draw[|->, olive] (0-.1,1-.05) arc (90+6:270-16:.2);
    \draw[|->, olive] (0-.1,.5-.05) arc (90+6:270-16:.2);
    \draw[|->, olive] (0-.1,0-.05) arc (90+6:270-16:.2);

    \draw[|->, olive] (0+.1,-.5+.1) to[bend right=60] (0+.1,1-.03);

    \draw (5,0.25) node {\small $-\mathbf{4}_{{}_{\mathrm{reg}}}$ };

    \end{scope}

\end{scope}

  \end{scope}
\end{tikzpicture}
\end{center}
\vspace{-4mm}
\noindent {\bf \footnotesize Figure M -- Virtual $G$-representations of brane configurations classified by
equivariant Cohomotopy} {\footnotesize
in the vicinity of ADE-singularities
(Def. \ref{CohomotopyOfVicinityOfSingularity}), according to
Prop. \ref{TheoremLocalTadpoleCancellation},
following
Theorem \ref{UnstableEquivariantHopfDegreeTheorem} and
Theorem \ref{CharacterizationOfStabilizationOfUnstableCohomotopy}.
The results reproduces the form of the local/twisted tadpole cancellation
conditions in \hyperlink{Table1}{\it Table 1}, \hyperlink{Table2}{\it Table 2}.
Shown is a situation for $G = \mathbb{Z}_4$ and
$V = \mathbf{2}_{\mathrm{rot}}$ as
in \hyperlink{FigureK}{\it Figure K}.}

\subsection{Equivariant Hopf degree on tori and Global tadpole cancellation}
\label{GlobalTadpoleCancellation}

We now globalize the characterization of equivariant Cohomotopy
from the vicinity of singular fixed points to compact toroidal
orbifolds, in Theorem \ref{UnstableEquivariantHopfDegreeTheoremForTori} below. Prop. \ref{PushforwardOfVicinityOfSingularityToRepresentationTorus}
below shows that the two are closely related, implying that the
local/twisted tadpole cancellation carries over to toroidal orbifolds.
Then we informally discuss the enhancement of unstable equivariant
Cohomotopy to a super-differential cohomology theory \eqref{DifferentialEquivariantCohomotopyPullback}
and show that its implications \eqref{KernelOfTheGlobalElmendorfStageProjection}
on the underlying equivariant Cohomotopy
enforce the form of the global/untwisted tadpole cancellation conditions.

\medskip

\noindent {\bf Globalizing from Euclidean orientifolds to toroidal orientifolds.}
In \cref{LocalTadpoleCancellation} we discussed the
characterization of  equivariant Cohomotopy in the vicinity
of singularities (according to \hyperlink{Table5}{\it Table 5}).
We may globalize this to compact toroidal orientifolds by
applying this local construction in the vicinity
of each singularity, using that the condition of
``vanishing at infinity'' \eqref{VanishingAtInfinity}
with respect to any one singularity
means that the local constructions may be glued together.
This \emph{local-to-global} construction is
indicated in \hyperlink{FigureN}{\it Figure N}:

\vspace{-3mm}
\begin{center}
{\hypertarget{FigureN}{}}
\begin{tikzpicture}[scale=0.75]
  \draw (1.5,6.4) node {$\overbrace{\phantom{------------------------}}$};
  \draw (11,6.4) node {$\overbrace{\phantom{---------------}}$};
  \draw (1.5,7.3) node {\tiny \color{darkblue} \bf toroidal orientifold};
  \draw (11,7.3)
    node
    {
      \tiny
      \color{darkblue} \bf
      \begin{tabular}{c}
        representation sphere
        \\
        equivariant Cohomotopy coefficient
      \end{tabular}
    };
  \draw (5.8,7.6) node { \tiny \color{darkblue} \bf  $\mathbb{Z}_2$-equivariant Cohomotopy cocycle  };
  \begin{scope}
  \clip (-2.9,-2.9) rectangle (5.9,5.9);
  \draw[step=3, dotted] (-3,-3) grid (6,6);
  \draw[dashed] (-3,-3) circle (1);
  \draw[dashed] (0,-3) circle (1);
  \draw[dashed] (3,-3) circle (1);
  \draw[dashed] (6,-3) circle (1);
  \draw[dashed] (-3,0) circle (1);
  \draw[dashed] (0,0) circle (1);
  \draw[dashed] (3,0) circle (1);
  \draw[dashed] (-3,3) circle (1);
  \draw[dashed] (0,3) circle (1);
  \draw[dashed] (3,3) circle (1);
  \draw[dashed] (-3,6) circle (1);
  \draw[dashed] (0,6) circle (1);
  \draw[dashed] (3,6) circle (1);
  \draw[dashed] (6,6) circle (1);

  \draw (0,.2) node {\colorbox{white}{\tiny $(0,0)$}};
  \draw[fill=white] (0,0) circle (.07);

  \draw[<->, dashed, darkblue]
    (0-1.9,3+1.6)
    to
    node[near start, above]
      {
        \tiny
        \color{darkblue} \bf
        \begin{tabular}{c}
          orientifold
          \\
          action
        \end{tabular}
      }
    (0+1.5,3-1.5);

  \draw (3,0) node {\colorbox{white}{\tiny $(\tfrac{1}{2},0)$}};
  \draw (0,3) node {\colorbox{white}{\tiny $(0,\tfrac{1}{2})$}};
  \draw (0,3-.25) node {{\tiny \color{darkblue} \bf a fixed point}};
  \draw (0,3-.65) node {\colorbox{white}{\tiny \color{darkblue} \bf \begin{tabular}{c}disk around fixed point\end{tabular}}};

  \draw
    (3,3)
    node
    {
      \colorbox{white}{
        \tiny
        $(\tfrac{1}{2},\tfrac{1}{2})$
      }
    };

  \draw
    (3+.25,3+.51)
    node {$\bullet$};

  \draw
    (3-.25,3-.51)
    node {$\bullet$};

  \end{scope}

  \draw (0,-3.4) node {\tiny $x_1 = 0$};
  \draw (3,-3.4) node {\tiny $x_1 = \tfrac{1}{2}$};
  \draw (-3.7,0) node {\tiny $x_2 = 0$};
  \draw (-3.7,3) node {\tiny $x_2 = \tfrac{1}{2}$};
  \draw[<->, dashed, darkblue]
    (11-1,2+1)
    to
    node[below, near end] {\tiny $\mathbb{Z}_2$}
    (11+1,2-1);
  \node at (11,2) {\colorbox{white}{$\phantom{a}$}};

  \draw[dashed] (11,2) circle (2);
  \node (zero) at (11,2) {\tiny $0$};
  \node (infinity) at (11,2+2.1) {\tiny $\infty$};
  \node (bottominfinity) at (11,2-2.1) {\tiny $\infty$};
  \node (rightinfinity) at (11+2.2,2) {\tiny $\infty$};
  \node (torus) at (1.5,8)
    {\raisebox{42pt}{$
      \mathbb{T}^{\mathbf{n}_{\mathrm{sgn}}}
        =
      \xymatrix{
        \mathbb{R}^{\mathbf{n}_{\mathrm{sgn}}}
        \ar@(ul,ur)^{
          \overset{
            \mathclap{
            \mbox{
              \tiny
              \color{darkblue} \bf
              \begin{tabular}{c}
                sign
                \\
                representation
              \end{tabular}
            }
            }
          }
          {
            \mathbb{Z}_2
          }
        }
      }
      /\mathbb{Z}^n
    $}};
  \node (sphere) at (11,8)
    {\raisebox{42pt}{$
      S^{\mathbf{n}_{\mathrm{sgn}}}
        =
      D(
        \xymatrix{
          \mathbb{R}^{\mathbf{n}_{\mathrm{sgn}}}
          \ar@(ul,ur)^{
            \overset{
              \mathclap{
              \mbox{
                \tiny
                \color{darkblue} \bf
                \begin{tabular}{c}
                  sign
                  \\
                  representation
                \end{tabular}
              }
              }
            }
            {
              \mathbb{Z}_2
            }
          }
        }
      )/S(\mathbb{R}^{\mathbf{n}_{\mathrm{sgn}}})
    $}};
    \draw[->, thin] (torus) to node[above]{c} (sphere);

    \draw[|->, thin, olive] (3+.4,3+.5)
      to[bend left=5]
      node
        {
          \hspace{-1.3cm}
          \colorbox{white}{
            \tiny
            \color{darkblue} \bf
            some charge at this singularity
          }
        }
      (zero);

    \draw[|->, thin, olive] (3-.1,3-.5)
      to[bend left=5]
      node
        {
          \hspace{-1.3cm}
          \colorbox{white}{
            \tiny
            \color{darkblue} \bf
            mirror charge
          }
        }
      (zero);

    \draw[|->, thin, olive] (0+.2,0+0) to[bend left=6.7] (zero);
    \draw[|->, thin, olive] (0-0,0-.5) to[bend left=6] (11-.1,2-.3);
    \draw[|->, thin, olive] (0+1,0+0)
      to[bend left=6.7]
      node
        {
          \hspace{-.8cm}
          \colorbox{white}{
            \tiny
            \color{darkblue} \bf
            O-plane unit charge at this singularity
          }
        }
      (11+1.9,2+0);
    \draw[|->, thin, olive] (1.5+0,4+0)
      to[bend left=13]
      node[below] {\hspace{.4cm}\colorbox{white}{\tiny \color{darkblue}
      \bf vanishing charge far away from all singularities}}
      (infinity);
    %
    \draw[|->, thin, olive] (3+.2,0-.5) to[bend left=7] (bottominfinity);
    \draw[|->, thin, olive] (3+.4,0+0)
      to[bend left=7]
      node {\colorbox{white}{\tiny \color{darkblue} \bf no charge at this singularity}}
      (bottominfinity);
    \draw[|->, thin, olive] (3+.6,0-.8) to[bend left=7] (bottominfinity);
\end{tikzpicture}
\end{center}
\vspace{-5mm}
\noindent {\bf \footnotesize Figure N --
Equivariant Cohomotopy cocycle on toroidal orbifolds glued
from local cocycles in the vicinity of singularities,}
{\footnotesize as formalized in the proof of Theorem \ref{UnstableEquivariantHopfDegreeTheoremForTori}.
Shown is a situation for $G = \mathbb{Z}_2$, as in
\hyperlink{FigureI}{\it Figure I} and \hyperlink{FigureJ}{\it Figure J}.}

\medskip

\noindent {\bf Well-isolated singularities.}
In order to formalize this local-to-global construction
conveniently, we make the following
sufficient assumption on the $G$-spaces to which we apply it:
\begin{defn}
We say that a $G$-space $\xymatrix{X \ar@(ul,ur)|{\, G\,}}$
has {\it well-isolated singularities} if all the
\emph{minimal} subgroups
with 0-dimensional fixed subspaces \eqref{FixedLoci}
are in the center of $G$, i.e. if the following condition holds:
\begin{equation}
  \label{WellIsolatedFixedPoints}
  H \subset G
  \;\mbox{minimal such that}\;
  \mathrm{dim}\big( X^G\big) = 0
  \phantom{AA}
  \Rightarrow
  \phantom{AA}
  H \subset \mathrm{Center}(G)\;.
\end{equation}
\end{defn}
\begin{example} The ADE-singularites (\hyperlink{Table5}{\it Table 5})
with well-isolated fixed points in the sense of \eqref{WellIsolatedFixedPoints} are all those in the $\mathbb{A}$-series,
as well as the generalized quaternionic ones in the $\mathbb{D}$-series --
see \hyperlink{Table6}{\it Table 6}.
 This is because, for ADE-singularities, all non-trivial subgroups
have 0-dimensional fixed space \eqref{FixedSubspacesOfQuaternionRepresentation}, so that here
the condition of well-isolated singularities \eqref{WellIsolatedFixedPoints}
requires that all non-trivial minimal elements in the subgroup lattice be in the center.
This is trivially true for the cyclic groups in the $\mathbb{A}$-series, since they are abelian.
For the generalized quaternionic groups in the $\mathbb{D}$-series there
is in fact a unique minimal non-trivial subgroup, and it in
fact it is always the orientifold action $H_{\mathrm{min}}  = \mathbb{Z}_2$
which coincides with the center, as shown for the first few cases
in \hyperlink{Table6}{\it Table 6}.
\end{example}

\medskip

The point of the notion of well-isolated fixed points
\eqref{WellIsolatedFixedPoints} is that it is sufficient
to guarantee that the action of the full group restricts to the
union of the 0-dimensional fixed subspaces, since then
\begin{equation}
  \label{SetOfIsolatedFixedPointsIsIndeedFixed}
  H \cdot x_{\mathrm{fixed}}
  \;=\;
  x_{\mathrm{fixed}}
  \phantom{AA}
  \Rightarrow
  \phantom{AA}
  H \cdot (g \cdot x_{\mathrm{fixed}})
  \;=\;
  (H \cdot g) \cdot x_{\mathrm{fixed}}
  \;=\;
  (g \cdot H) \cdot x_{\mathrm{fixed}}
  \;=\;
  g \cdot (H \cdot x_{\mathrm{fixed}})
  \;=\;
  g \cdot x_{\mathrm{fixed}},
\end{equation}
for all $g \in G$.
Hence, with \eqref{WellIsolatedFixedPoints},
the quotient set
\begin{equation}
  \label{SetOfWellIsolatedFixedPoints}
  \mathrm{IsolSingPts}_G(X)
  \;\coloneqq\;
  \Bigg(
  \underset
  { \Scale[0.6]
    {
      H \subset G
 ,
      \mathrm{dim}( X^H ) = 0
    }
  }
  {
    \bigcup
  }
  X^H
  \Bigg)
  \big/ G
\end{equation}
exists and
is the set of isolated singular
points in the orbifold $X \!\sslash\! G$.

{\footnotesize
\begin{center}
\hypertarget{Table6}{}
\hspace{-.2cm}
\small\addtolength{\tabcolsep}{-5pt}
\begin{tabular}{|c||c|c|c|c|c|}
  \hline
  \begin{tabular}{c}
    {\bf Dynkin}
    \\
 {\bf  label}
  \end{tabular}
  &
  $\mathbb{A}3 = \mathbb{D}3$
  &
  $\mathbb{D}4$
  &
  $\mathbb{D}6$
  &
  $\mathbb{D}10$
  &
  $\mathbb{D}18$
  \\
  \hline
  \raisebox{-20pt}{
  \begin{tabular}{c}
    $G \subset \mathrm{Sp}(1)$
    \\
    \phantom{A}
    \\
    \phantom{A}
    \\
    \begin{rotate}{+90}
      $
      \mathllap{
      \mbox{\bf
      subgroup lattice
      }
      }
      $
    \end{rotate}
  \end{tabular}
  }
  &
  $
    \xymatrix@C=5pt{
      &
      {\phantom{Q_8}}
      \\
      &
      \mathbb{Z}_4
      \mathrlap{ \, \mbox{\tiny = center} }
      \\
      &
      {\color{darkblue} \bf \mathbb{Z}^{\mathrlap{\mathrm{refl}}}_2 }
      \ar@{^{(}->}[u]
      &
      {\phantom{\mathbb{Z_2}}}
      \\
      & 1
      \ar@{^{(}->}[u]
    }
  $
  &
  $
    \xymatrix@C=12pt{
      & Q_8
      \ar@{=}[d]
      \\
      &
      2 D_4
      \\
      \mathbb{Z}_4
      \ar@{^{(}->}[ur]
      &
      \mathbb{Z}_4
      \ar@{^{(}->}[u]
      &
      \mathbb{Z}_4
      \ar@{^{(}->}[ul]
      \\
      &
      {\color{darkblue} \bf \mathbb{Z}^{\mathrlap{\mathrm{refl}}}_2 }
      \mathrlap{ \;\;\; \mbox{\tiny = center} }
      \ar@{^{(}->}[ul]
      \ar@{^{(}->}[u]
      \ar@{^{(}->}[ur]
      \\
      & 1
      \ar@{^{(}->}[u]
    }
  $
  &
  $
    \xymatrix@C=12pt{
      &
      Q_{16}
      \ar@{=}[d]
      \\
      &
      2 D_8
      \\
      2 D_4
      \ar@{^{(}->}[ur]
      &
      \mathbb{Z}_8
      \ar@{^{(}->}[u]
      &
      2 D_4
      \ar@{^{(}->}[ul]
      \\
      \mathbb{Z}_4
      \ar@{^{(}->}[u]
      &
      \mathbb{Z}_4
      \ar@{^{(}->}[ul]
      \ar@{^{(}->}[u]
      \ar@{^{(}->}[ur]
      &
      \mathbb{Z}_4
      \ar@{^{(}->}[u]
      \\
      &
      {\color{darkblue} \bf \mathbb{Z}^{\mathrlap{\mathrm{refl}}}_2 }
      \mathrlap{ \;\;\; \mbox{\tiny = center} }
      \ar@{^{(}->}[ul]
      \ar@{^{(}->}[u]
      \ar@{^{(}->}[ur]
      \\
      &
      1
      \ar@{^{(}->}[u]
    }
  $
  &
  $
  \xymatrix@C=12pt{
    &
    Q_{32}
    \ar@{=}[d]
    \\
    &
    2 D_{16}
    \\
    2 D_8
    \ar@{^{(}->}[ur]
    &
    \mathbb{Z}_{16}
    \ar@{^{(}->}[u]
    &
    2 D_8
    \ar@{^{(}->}[ul]
    \\
    2 D_4
    \ar@{^{(}->}[u]
    &
    \mathbb{Z}_8
    \ar@{^{(}->}[ul]
    \ar@{^{(}->}[u]
    \ar@{^{(}->}[ur]
    &
    2 D_4
    \ar@{^{(}->}[u]
    \\
    \mathbb{Z}_4
    \ar@{^{(}->}[u]
    &
    \mathbb{Z}_4
    \ar@{^{(}->}[ul]
    \ar@{^{(}->}[u]
    \ar@{^{(}->}[ur]
    &
    \mathbb{Z}_4
    \ar@{^{(}->}[u]
    \\
    &
      {\color{darkblue} \mathbf \mathbb{Z}^{\mathrlap{\mathrm{refl}}}_2 }
      \mathrlap{ \;\;\; \mbox{\tiny = center} }
    \ar@{^{(}->}[ul]
    \ar@{^{(}->}[u]
    \ar@{^{(}->}[ur]
    \\
    &
    1
    \ar@{^{(}->}[u]
  }
  $
  &
  \xymatrix@C=12pt{
    &
    Q_{64}
    \ar@{=}[d]
    \\
    &
    2 D_{32}
    \\
    2 D_{16}
    \ar@{^{(}->}[ur]
    &
    \mathbb{Z}_{32}
    \ar@{^{(}->}[u]
    &
    2 D_{16}
    \ar@{^{(}->}[ul]
    \\
    2 D_{8}
    \ar@{^{(}->}[u]
    &
    \mathbb{Z}_{16}
    \ar@{^{(}->}[ul]
    \ar@{^{(}->}[u]
    \ar@{^{(}->}[ur]
    &
    2 D_{8}
    \ar@{^{(}->}[u]
    \\
    2 D_{4}
    \ar@{^{(}->}[u]
    &
    \mathbb{Z}_{8}
    \ar@{^{(}->}[ul]
    \ar@{^{(}->}[u]
    \ar@{^{(}->}[ur]
    &
    2 D_{4}
    \ar@{^{(}->}[u]
    \\
    \mathbb{Z}_4
    \ar@{^{(}->}[u]
    &
    \mathbb{Z}_4
    \ar@{^{(}->}[ul]
    \ar@{^{(}->}[u]
    \ar@{^{(}->}[ur]
    &
    \mathbb{Z}_4
    \ar@{^{(}->}[u]
    \\
    &
      {\color{darkblue} \bf \mathbb{Z}^{\mathrlap{\mathrm{refl}}}_2 }
      \mathrlap{ \;\;\; \mbox{\tiny = center} }
    \ar@{^{(}->}[ul]
    \ar@{^{(}->}[u]
    \ar@{^{(}->}[ur]
    \\
    &
    1
    \ar@{^{(}->}[u]
  }
  \\
  \hline
\end{tabular}
\end{center}
}
\vspace{-2mm}
\noindent {\footnotesize \bf Table 6 -- The ADE-Singularities
$\xymatrix{\mathbb{R}^{\mathbf{4}_{\mathbb{H}}}\ar@(ul,ur)|{\;\;\;\; G^{\mathrm{ADE}}\, }}$ with well-isolated fixed points}
{\footnotesize according to
\eqref{WellIsolatedFixedPoints} are those in the $\mathbb{A}$-series $G^{\mathrm{ADE}} = \mathbb{Z}_n$ and the quaternionic groups $Q_{2^{n+2}} = 2 D_{2^{n} + 2}$
in the $\mathbb{D}$-series. For the latter and for the even-order
cyclic groups, the minimal non-trivial central subgroup is
unique and given by the point reflection group
$\mathbb{Z}_2^{\mathrm{refl}}$ \eqref{PointReflectionSubgroup}.
}


\vspace{4mm}
\noindent {\bf Unstable equivariant Hopf degree of representation tori.}
With these preliminaries in hand, we may now state and prove the unstable
equivariant Hopf degree theorem for representation tori with well-isolated singularities,
Theorem \ref{UnstableEquivariantHopfDegreeTheoremForTori} below.
Its statement and proof are directly analogous to the case for
representation spheres in Theorem \ref{UnstableEquivariantHopfDegreeTheorem}.
The difference here, besides the passage from
spheres to tori, is the extra assumption on well-isolated singularities
and the fact that the proof here invokes the construction of the
previous proof around each one of the well-isolated singularities.

\begin{theorem}[\bf Unstable equivariant Hopf degree theorem for
representation tori]
\label{UnstableEquivariantHopfDegreeTheoremForTori}
The unstable equivariant Cohomotopy \eqref{EquivariantCohomotopySet}
of a $G$-representation torus $\mathbb{T}^V$ \eqref{RepresentationTorus}
with well-isolated singularities \eqref{WellIsolatedFixedPoints}
and with a point at infinity adjoined \eqref{BasepointFreelyAdjoined}
in compatible RO-degree $V$ (Example \ref{ExamplesOfCompatibleRODegree})
is in bijection to the product set of
one copy of the
integers for each isotropy group \eqref{IsotropySubgroups}
with positive dimensional fixed subspace \eqref{FixedLoci},
and one copy of $\{0,1\}$ for each well-isolated fixed point
\eqref{SetOfWellIsolatedFixedPoints}
\begin{equation}
  \label{UnstableEquivariantCohomotopyOfRepresentationTorusInCompatibleDegree}
  \xymatrix{
    \pi^V_G
    \big(
      \big( \mathbb{T}^V \big)_+
    \big)
    \ar[rrr]^-{ 
      c 
      \;\mapsto\;
      ( H \mapsto {\color{darkblue} \bf N_H}(c) )   
    }_-{\simeq}
    &&&
    \mathbb{Z}^{{}^{ \mathrm{Isotr}^{d_{\mathrm{fix}} \gt 0}_X(G) } }
    \times
    \{0,1\}^{{}^{
      \mathrm{IsolSingPts}_G(X)
    }}
  },
\end{equation}
where, for $H \in \mathrm{Isotr}^{d_{\mathrm{fix}} \gt 0  }_X(G)$,
the ordinary Hopf degree at Elmendorf stage $H$  \eqref{ElmedorfStageWiseHopfDegrees}
is of the form
\begin{equation}
  \label{TheWeylGroupMultiplesForRepresentationTori}
  \xymatrix@R=-2pt{
    \mathrm{deg}\big( c^{H} \big)
    &=&
    \phi_H\big(
       \{ \mathrm{deg}\big( c^K \big) \big\vert K \supsetneq H \in \mathrm{Isotr}_X(G)  \}
    \big)
    &-&
    {\color{darkblue} \bf N_H}(c) \cdot \big| \big( W_G(H)\big) \big|
    &
    \!\!\!\!\!\!\!\!\!\!\!\!
    \in
    \mathbb{Z}.
    \\
    \mathclap{
    \mbox{\bf
      \tiny
      \color{darkblue}
      \begin{tabular}{c}
        The ordinary Hopf degree \eqref{HopfDegreeTheorem}
        \\
        at Elmendorf stage $K$ \eqref{SystemOfMapsOnHFixedSubspaces}
      \end{tabular}
    }
    }
    &
    &
    \mathclap{
    \mbox{\bf
      \tiny
      \color{darkblue}
      \begin{tabular}{c}
        a fixed offset, being a function of
        \\
        the Hopf degrees
        at all lower stages.
      \end{tabular}
    }
    }
    &
    &
    \mathclap{
    \mbox{\bf
      \tiny
      \color{darkblue}
      \begin{tabular}{c}
        an integer multiple of
        \\
        the order of the Weyl group \eqref{WeylGroup}
      \end{tabular}
    }
    }
  }
\end{equation}
The isomorphism \eqref{UnstableEquivariantCohomotopyOfRepresentationTorusInCompatibleDegree}
is exhibited by sending an equivariant Cohomotopy cocycle to the sequence of the
integers ${\color{darkblue} \bf N_H}$ from \eqref{TheWeylGroupMultiplesForRepresentationTori}
in positive fixed subspace dimensions, together with the collection of elements in $\{0,1\}$,
which are the unstable Hopf degrees in dimension 0 \eqref{UnstableRangeHopfDegreeTheorem},
at Elmendorf stage $G$ at each one of the well-isolated singularities.
\end{theorem}
\begin{proof}
  In the special case when no subgroup $H \subset G$
  has a fixed subspace of vanishing dimension,
  this is \cite[Theorem 8.4.1]{tomDieck79}, where
  the assumption of positive dimension is made
  ``for simplicity'' in \cite[middle of p. 212]{tomDieck79}.
  Hence we just need to convince ourselves that the proof
  given there generalizes.

  To that end, assume that $\mathrm{dim}\big( V^G \big) = 0$.
  To generalize the inductive argument in \cite[p. 214]{tomDieck79} to this case,
  we just need to see that every
  $G$-invariant function on the isolated fixed
  \eqref{SetOfWellIsolatedFixedPoints}
  \begin{equation}
    \label{CohomotopyCocycleOnWellIsolatedFixedPoints}
    \xymatrix@R=1.2em{
      \mathrm{IsolSingPts}_G(X)
      \ar[rr]
      &&
      S^0
      \\
       \underset
       {
         \mathclap{
         \Scale[0.6]
         {
           H \subset G
           ,
           \mathrm{dim}( X^H ) = 0
         }
         }
       }
       {
         \bigcup
       }
       \;\;\;\;\;
       X^H
       \ar[rr]^{ (c^H) }
      \ar[u]^q
      &&
      S^0
      \ar@{=}[u]
    }
  \end{equation}
  extends to a $W_G(K)$-equivariant function $( S^V)^K \to ( S^V)^K$ on
  the next higher Elmendorf stage $K \in \mathrm{Isotr}^{d_{\mathrm{fix}}\gt 0}_X(G)$.
  For this, consider a $G$-equivariant tubular neighborhood around the well-isolated
  fixed points.  This is guaranteed to exist on general grounds by the equivariant tubular
  neighborhood theorem, since, by assumption \eqref{WellIsolatedFixedPoints},
  the set of points (in the bottom left of \eqref{CohomotopyCocycleOnWellIsolatedFixedPoints})
  is an equivariant (and of course closed) subspace, by \eqref{SetOfIsolatedFixedPointsIsIndeedFixed}.
  In fact, in the present specific situation of global
  \emph{orthogonal} linear actions on a Euclidean space we
  obtain a concrete such equivariant tubular neighborhood
  by forming the union of Euclidean open balls of radius $\epsilon$ around each
  of the points, for any small enough positive real number $\epsilon$.
  This kind of tubular neighborhood is indicated by the collection of dashed circles
  in \hyperlink{FigureA}{\it Figure A} and  \hyperlink{FigureN}{\it Figure N}.
  Given this or any choice of equivariant tubular neighborhood,
  the extensions \eqref{InductionStartForRepSpheres} in the proof of
  Theorem \ref{UnstableEquivariantHopfDegreeTheorem} apply to the
  vicinity of any one fixed points. This is a choice in $\{0,1\}$ for each element in
  $\mathrm{IsolSingPts}_G(X)$ \eqref{SetOfWellIsolatedFixedPoints},
  hence in total is the choice of an element in
  $\{0,1\}^{{}^{ \mathrm{IsolSingPts}_G(X) }}$, as it appears in
  \eqref{UnstableEquivariantCohomotopyOfRepresentationSphereInCompatibleDegree}.
  Since all these local extensions to the vicinity of any of the singularities ``vanish at infinity'' \eqref{VanishingAtInfinity}, i.e., at some distance $\gt \epsilon$ from any and all of the
  well-isolated fixed points, they may jointly be further extended to a
  global cocycle $\mathbb{T}^V \overset{c}{\longrightarrow} S^V$ by declaring that $c$
  sends every other point in $\mathbb{T}^V$ outside the given tubular neighborhood to $\infty \in S^V$
  (shown in \hyperlink{FigureN}{\it Figure N}).
   From this induction onwards, the  proof of \cite[8.4.1]{tomDieck79} applies verbatim
  and shows that on top of  this initial Hopf degree
  choice in $\{0,1\}^{{}^{\mathrm{IsolSingPts}_G(X)}}$
  there may now be further
  $N_H \cdot \vert W_G(H)\vert$-worth
  of Hopf degree at the next higher Elmendorf stage $H$, and so on.
\end{proof}

\vspace{1mm}
\noindent {\bf Stable equivariant Hopf degree of representation tori.}
Note that the unstable equivariant Hopf degrees of
representation spheres (Theorem \ref{UnstableEquivariantHopfDegreeTheorem}) and of representation tori (Theorem \ref{UnstableEquivariantHopfDegreeTheoremForTori})
have the same form,
\eqref{UnstableEquivariantCohomotopyOfRepresentationSphereInCompatibleDegree}
and \eqref{UnstableEquivariantCohomotopyOfRepresentationTorusInCompatibleDegree},
respectively,  away from the unstable Hopf degrees in vanishing fixed space dimensions.
It follows immediately that, up to equivariant homotopy,
all brane charge may be thought of as concentrated in the
vicinity of the ``central'' singularity (see \hyperlink{FigureO}{\it Figure O}):

\vspace{-2mm}
\begin{prop}[\bf Pushforward in unstable equivariant Cohomotopy]
\label{PushforwardOfVicinityOfSingularityToRepresentationTorus}
Let$\xymatrix{\mathbb{T}^V\ar@(ul,ur)|-{G}}$ be a $G$-representation torus
\eqref{RepresentationTorus}
with well-isolated singularities \eqref{WellIsolatedFixedPoints}, and write
$
      D_\epsilon ( \mathbb{R}^V)
    \overset{i}{\hookrightarrow}
       \mathbb{T}^V
  $,
  $
    0 \hookrightarrow x_0
  $,
for the inclusion of the $G$-equivariant tubular neighborhood around the fixed
point $x_0 \in \mathbb{T}^V$ covered by $0 \in \mathbb{R}^V$ that is given
by the open $\epsilon$-ball around the point, for any small enough positive radius
$\epsilon$. Then pushforward along $i$ from the unstable equivariant Cohomotopy
of the vicinity of this fixed point (as in Theorem \ref{UnstableEquivariantHopfDegreeTheorem})
to that of the full representation torus (as in Theorem \ref{UnstableEquivariantHopfDegreeTheoremForTori})
$$
  \xymatrix@R=-2pt{
  \mathclap{
  \mbox{\bf
    \tiny
    \color{darkblue}
    \begin{tabular}{c}
      unstable equivariant Cohomotopy
      \\
      of vicinity of $G$-singularity
    \end{tabular}
  }
  }
  \ar@{}[rrr]|-{
    \mbox{\bf
      \tiny
      \color{darkblue}
      \begin{tabular}{c}
        identify with
        vicinity of $x_0$
        \\
        ${\phantom{\vert}}$
      \end{tabular}
    }
  }
  &&&
  \mathclap{
  \mbox{\bf
    \tiny
    \color{darkblue}
    \begin{tabular}{c}
      unstable equivariant Cohomotopy
      \\
      of $G$-representation torus
    \end{tabular}
  }
  }
  \\
  \pi^V_G
  \big(
    \big(
      \mathbb{R}^{V}
    \big)^{\mathrm{cpt}}
  \big)
  \ar@{^{(}->}[rrr]^-{
    i_\ast
  }_-{ \simeq_{{}_{ (d_{\mathrm{fix}} \gt 0 } ) } }
  &&&
  \pi^V_G
  \big(
    \big(
      \mathbb{T}^{V}
    \big)_+
  \big)
  \\
  \left[
  \!\!\!\!\!\!
  \scalebox{.95}{
  $
  {\begin{array}{ccc}
    \mathbb{R}^{V}
    &\mathrlap{\xrightarrow{\phantom{---}c\phantom{---}}}&
    S^V
    \\
    x
    &\!\!\!\longmapsto\!\!\!&
    \left\{
    \!\!\!
    \scalebox{.9}{
    $
    \begin{array}{cl}
      c(x) & \mbox{if $d(x,0) \lt \epsilon$}
      \\
      \infty & \mbox{otherwise}
    \end{array}
    $
    }
    \right.
  \end{array}}
  $
  }
  \!\!\!\!\!\!\!\!\!\!\!\!
  \right]
  \ar@{}[rrr]|-{ \longmapsto }
  &&&
  \left[
  \!\!\!
  \scalebox{.95}{
  $
  {\begin{array}{ccc}
    \mathbb{T}^{V}
    &\mathrlap{\xrightarrow{\phantom{---}i_\ast(c)\phantom{---}}}&
   \;\;\;\;\;\;\;\;\;\; S^V
    \\
    x
    &\!\!\!\longmapsto\!\!\!&
    \left\{
    \!\!\!\!\!\!
    \scalebox{.9}{
    $
    \begin{array}{cl}
      c(x) & \mbox{if $d(x,x_0) \lt \epsilon$}
      \\
      \infty & \mbox{otherwise}
    \end{array}
    $
    }
    \right.
  \end{array}}
  $
  }
  \!\!\!\!\!\!\!\!\!\!\!\!
  \right]
  }
$$
is an isomorphism on Hopf degrees at Elmendorf stages $H_{\gt 0}$
of non-vanishing fixed space dimension and an injection
on the unstable Hopf degree set at Elemendorf stages 
$H_{= 0}$ with vanishing fixed subspace dimension:
$$
  i_\ast
  \;\colon\;
  \left\{
    \begin{array}{cl}
      N_{H_{= 0}}(c) \! & \hookrightarrow \; N_{H_{= 0}}(i_\ast(c))
      \\
      N_{H_{\gt 0}}(c) & \mapsto \; N_{H_{\gt 0}}(i_\ast(c))\,.
    \end{array}
  \right.
$$
\end{prop}
\noindent This is illustrated by \hyperlink{FigureO}{\it Figure O}:

\begin{center}
\hypertarget{FigureO}{}
\begin{tikzpicture}[scale=0.8, decoration=snake]

  \begin{scope}[shift={(0,0)}]

  \begin{scope}[shift={(0,-.6)}] 

  \node at (1.4,8)
    {
      \tiny
      \color{darkblue} \bf
      \begin{tabular}{c}
        unstable equivariant Cohomotopy
        \\
        of vicinity of singularity
      \end{tabular}
    };

  \node (EquivariantCohomotopy) at (1.4,5.3+1.6)
    {$
      \pi^{ \mathbf{4}_{\mathbb{H}} }_{\mathbb{Z}_4}
      \big(
        \big(
          \mathbb{R}^{ \mathbf{4}_{\mathbb{H}} }
        \big)^{\mathrm{cpt}}
      \big)
    $};

  \end{scope}

  \node at (1.4,5.3)
    {$
      \overbrace{
        \phantom{------------------}
      }
    $};

  \begin{scope}[shift={(0,.8)}]

  \draw[<->, dashed, darkblue]
    (2.5,0)
    to[bend right=47]
    node
      {
        \colorbox{white}{\bf
        \tiny
        \color{darkblue}
        \begin{tabular}{c}
          orientifold
          \\
          action
        \end{tabular}
      }
      }
    (0,2.5);

  \end{scope}

  \begin{scope}[shift={(0, .8)}]
  \begin{scope}
  \clip (-1.8,-1.5) rectangle (1.5,1.5);
  \draw[step=3, dotted] (-3,-2) grid (6,6);
  \draw[dashed] (-3,-3) circle (1);
  \draw[dashed] (0,-3) circle (1);
  \draw[dashed] (3,-3) circle (1);
  \draw[dashed] (6,-3) circle (1);
  \draw[dashed] (-3,0) circle (1);
  \draw[dashed] (0,0) circle (1);
  \draw[dashed] (3,0) circle (1);
  \draw[dashed] (-3,3) circle (1);
  \draw[dashed] (0,3) circle (1);
  \draw[dashed] (3,3) circle (1);
  \draw[dashed] (-3,6) circle (1);
  \draw[dashed] (0,6) circle (1);
  \draw[dashed] (3,6) circle (1);
  \draw[dashed] (6,6) circle (1);

  \draw[fill=white] (0,0) circle (.07);
  \draw[fill=white] (3,0) circle (.07);
  \draw[fill=white] (0,3) circle (.07);
  \draw[fill=white] (3,3) circle (.07);

  \draw (0,3)
    node[right]
      {
        \colorbox{white}{
        \hspace{-.3cm}
        \tiny
        \color{darkblue} \bf
        O-plane
        \hspace{-.3cm}
        }
      };
  \draw (3,0)
    node[right]
      {
        \colorbox{white}{
        \hspace{-.5cm}
        \tiny
        \color{darkblue} \bf
        \begin{tabular}{c}
          mirror
          \\
          O-plane
        \end{tabular}
        \hspace{-.3cm}
        }
      };

  \end{scope}

  \begin{scope}[shift={(0,0)}]

  \draw[fill=black] (17:.8) circle (.07);
  \draw[fill=black] (17+90:.8) circle (.07);
  \draw[fill=black] (17+180:.8) circle (.07);
  \draw[fill=black] (17+270:.8) circle (.07);

  \begin{scope}[rotate=7]

  \draw[fill=black] (17:.3) circle (.07);
  \draw[fill=black] (17+90:.3) circle (.07);
  \draw[fill=black] (17+180:.3) circle (.07);
  \draw[fill=black] (17+270:.3) circle (.07);

  \end{scope}

  \draw (17+90+16:.62)
    node[right]
      {
        {
        \hspace{-.3cm}
        \tiny
        \color{darkblue} \bf
        branes
        \hspace{-.3cm}
        }
      };
  \draw (17+180:.7)+(.58,.03)
    node[right, below]
      {
        {
        \hspace{-.3cm}
        \tiny
        \color{darkblue} \bf
        mirror branes
        \hspace{-.3cm}
        }
      };

  \end{scope}

  \begin{scope}[shift={(0,1.6)}]
  \draw (0,-3.4) node {\tiny $x_1 = 0$};

  \end{scope}

  \draw (-2.5,0) node {\tiny $x_2 = 0$};

\end{scope}

\end{scope}

  \begin{scope}[shift={(10,0)}]

  \begin{scope}[shift={(0,-.6)}] 

  \node at (1.4,8)
    {
      \tiny
      \color{darkblue} \bf
      \begin{tabular}{c}
        unstable equivariant Cohomotopy
        \\
        of representation torus
      \end{tabular}
    };

  \node (EquivariantCohomotopy) at (1.4,5.3+1.6)
    {$
      \pi^{ \mathbf{4}_{\mathbb{H}} }_{\mathbb{Z}_4}
      \big(
        \big(
          \mathbb{T}^{ \mathbf{4}_{\mathbb{H}} }
        \big)_+
      \big)
    $};

   \end{scope}

  \node at (1.4,5.3)
    {$
      \overbrace{
        \phantom{------------------}
      }
    $};

  \begin{scope}[shift={(0,.8)}]

  \draw[<->, dashed, darkblue]
    (2.5,0)
    to[bend right=47]
    node
      {
        \colorbox{white}{
        \tiny
        \color{darkblue} \bf
        \begin{tabular}{c}
          orientifold
          \\
          action
        \end{tabular}
      }
      }
    (0,2.5);

  \end{scope}

  \begin{scope}[shift={(0, .8)}]
  \begin{scope}
  \clip (-1.8,-1.5) rectangle (4.8,4.4);
  \draw[step=3, dotted] (-3,-2) grid (6,6);
  \draw[dashed] (-3,-3) circle (1);
  \draw[dashed] (0,-3) circle (1);
  \draw[dashed] (3,-3) circle (1);
  \draw[dashed] (6,-3) circle (1);
  \draw[dashed] (-3,0) circle (1);
  \draw[dashed] (0,0) circle (1);
  \draw[dashed] (3,0) circle (1);
  \draw[dashed] (-3,3) circle (1);
  \draw[dashed] (0,3) circle (1);
  \draw[dashed] (3,3) circle (1);
  \draw[dashed] (-3,6) circle (1);
  \draw[dashed] (0,6) circle (1);
  \draw[dashed] (3,6) circle (1);
  \draw[dashed] (6,6) circle (1);

  \draw[fill=white] (0,0) circle (.07);

  \end{scope}

  \begin{scope}[shift={(0,0)}]

  \draw[fill=black] (17:.8) circle (.07);
  \draw[fill=black] (17+90:.8) circle (.07);
  \draw[fill=black] (17+180:.8) circle (.07);
  \draw[fill=black] (17+270:.8) circle (.07);

  \begin{scope}[rotate=7]

  \draw[fill=black] (17:.3) circle (.07);
  \draw[fill=black] (17+90:.3) circle (.07);
  \draw[fill=black] (17+180:.3) circle (.07);
  \draw[fill=black] (17+270:.3) circle (.07);

  \end{scope}

  \draw (17+90+16:.62)
    node[right]
      {
        {
        \hspace{-.3cm}
        \tiny
        \color{darkblue} \bf
        branes
        \hspace{-.3cm}
        }
      };
  \draw (17+180:.7)+(.58,.03)
    node[right, below]
      {
        {
        \hspace{-.3cm}
        \tiny
        \color{darkblue} \bf
        mirror branes
        \hspace{-.3cm}
        }
      };

  \end{scope}

  \begin{scope}[shift={(0,1.6)}]
  \draw (0,-3.4) node {\tiny $x_1 = 0$};
  \draw (3,-3.4) node {\tiny $x_1 = \tfrac{1}{2}$};

  \end{scope}

  \draw (-2.5,0) node {\tiny $x_2 = 0$};
  \draw (-2.5,3) node {\tiny $x_2 = \tfrac{1}{2}$};

\end{scope}

\end{scope}

 \begin{scope}[shift={(0,-.6)}] 

 \draw[->]
   (3.6,7)
   to
   node[above] {\tiny $i_\ast$}
   node[below] {\tiny $\simeq_{{}_{d_{\mathrm{fix}} \gt 0}}$ }
   (9.3,7);

  \end{scope}

\end{tikzpicture}
\end{center}

\vspace{-.3cm}

\noindent {\footnotesize \bf Figure O -- Pushforward in equivariant Cohomotopy from the vicinity of a singularity to the full toroidal orientifold}
{\footnotesize is an isomorphism on brane charges and an injection on
O-plane charges, by Prop. \ref{PushforwardOfVicinityOfSingularityToRepresentationTorus}.
Shown is a case with $G = \mathbb{Z}_{4}$, as in \hyperlink{FigureM}{\it Figure M}. All integer number of branes (black dots) are in the image
of the map, but only the O-plane at $(x_1, x_2) = (0,0)$ is in the image. }

\medskip

\noindent {\bf Local tadpole cancellation in toroidal ADE-orientifolds.}
Under the identification from Prop. \ref{PushforwardOfVicinityOfSingularityToRepresentationTorus},
the stabilized equivariant Hopf degree theorem
for representation spheres (Theorem \ref{CharacterizationOfStabilizationOfUnstableCohomotopy})
applies also to representation tori,
and hence so does Prop. \ref{TheoremLocalTadpoleCancellation},
showing now for the case of toroidal orbifolds
with ADE-singularities that the brane charges
classified by equivariant Cohomotopy are necessarily
multiples of the regular representation.
This result is visualized in \hyperlink{FigureP}{\it Figure P}:

\vspace{-.4cm}
\begin{center}
\hypertarget{FigureP}{}
\begin{tikzpicture}[scale=0.8, decoration=snake]

  \node at (1.4,5.3)
    {$
      \overbrace{
        \phantom{------------------}
      }
    $};

  \node at (1.4,8)
    {
      \tiny
      \color{darkblue} \bf
      \begin{tabular}{c}
        equivariant Cohomotopy
        \\
        of representation torus
        \\
        (orientifold Cohomotopy)
      \end{tabular}
    };

  \node (EquivariantCohomotopy) at (1.4,5.3+1.6)
    {$
      \pi^{ \mathbf{4}_{\mathbb{H}} }_{\mathbb{Z}_4}
      \big(
        \big(
          \mathbb{T}^{ \mathbf{4}_{\mathbb{H}} }
        \big)_+
      \big)
    $};

  \node (EquivariantCocycle) at (1.4,5.3+.8)
    {\tiny $
      4 \cdot [\mathbb{Z}_4/\mathbb{Z}_4]
      - 3 \cdot [\mathbb{Z}_4/1]
    $};

  \node at (1.4+8,7.8)
    {
      \tiny
      \color{darkblue} \bf
      \begin{tabular}{c}
        equivariant K-theory
        \\
        of representation torus
        \\
        = representation ring
      \end{tabular}
    };

  \node at (1.4+8,5.3)
    {$
      \overbrace{
        \phantom{--------}
      }
    $};

  \node (PlainCohomotopy) at (1.4+8,5.3+1.6)
    {$
      \mathrm{KO}_{\mathbb{Z}_4}^0
      \;\simeq\;
      \mathrm{RO}(\mathbb{Z}_4)
    $};

  \node (PlainCocycle) at (1.4+8,5.3+.8)
    {\tiny
     \raisebox{-.0cm}{
     $
     \begin{aligned}
      &
      4 \cdot \mathbf{1}
      - 3 \cdot \mathbf{4}_{\mathrm{reg}}
      \end{aligned}
    $}};

  \draw[->]
    (EquivariantCohomotopy)
    to
    node[above]
    {
      \tiny
      \color{darkblue} \bf
      stabilize and linearize
    }
    (PlainCohomotopy);
  \draw[|->] (EquivariantCocycle) to (PlainCocycle);

  \draw[<->, dashed, darkblue]
    (2.5,0)
    to[bend right=47]
    node
      {
        \colorbox{white}{
        \tiny
        \color{darkblue} \bf
        \begin{tabular}{c}
          orientifold
          \\
          action
        \end{tabular}
      }
      }
    (0,2.5);

  \begin{scope}[shift={(0, .8)}]
  \begin{scope}
  \clip (-1.8,-1.5) rectangle (4.8,4.4);
  \draw[step=3, dotted] (-3,-2) grid (6,6);
  \draw[dashed] (-3,-3) circle (1);
  \draw[dashed] (0,-3) circle (1);
  \draw[dashed] (3,-3) circle (1);
  \draw[dashed] (6,-3) circle (1);
  \draw[dashed] (-3,0) circle (1);
  \draw[dashed] (0,0) circle (1);
  \draw[dashed] (3,0) circle (1);
  \draw[dashed] (-3,3) circle (1);
  \draw[dashed] (0,3) circle (1);
  \draw[dashed] (3,3) circle (1);
  \draw[dashed] (-3,6) circle (1);
  \draw[dashed] (0,6) circle (1);
  \draw[dashed] (3,6) circle (1);
  \draw[dashed] (6,6) circle (1);

  \draw[fill=white] (0,0) circle (.07);
  \draw[fill=white] (3,0) circle (.07);
  \draw[fill=white] (0,3) circle (.07);
  \draw[fill=white] (3,3) circle (.07);

  \draw (0,3)
    node[right]
      {
        \colorbox{white}{
        \hspace{-.3cm}
        \tiny
        \color{darkblue} \bf
        O-plane
        \hspace{-.3cm}
        }
      };
  \draw (3,0)
    node[right]
      {
        \colorbox{white}{
        \hspace{-.5cm}
        \tiny
        \color{darkblue} \bf
        \begin{tabular}{c}
          mirror
          \\
          O-plane
        \end{tabular}
        \hspace{-.3cm}
        }
      };

  \draw[fill=black] (38:.8) circle (.07);
  \draw[fill=black] (38+90:.8) circle (.07);
  \draw[fill=black] (38+180:.8) circle (.07);
  \draw[fill=black] (38+270:.8) circle (.07);

  \draw[fill=black] (70:.4) circle (.07);
  \draw[fill=black] (70+90:.4) circle (.07);
  \draw[fill=black] (70+180:.4) circle (.07);
  \draw[fill=black] (70+270:.4) circle (.07);

  \end{scope}

  \begin{scope}[shift={(3,3)}]

    \draw[fill=black] (17:.7) circle (.07);
  \draw[fill=black] (17+90:.7) circle (.07);
  \draw[fill=black] (17+180:.7) circle (.07);
  \draw[fill=black] (17+270:.7) circle (.07);

  \draw (17+90:.7)
    node[right]
      {
        \colorbox{white}{
        \hspace{-.3cm}
        \tiny
        \color{darkblue} \bf
        brane
        \hspace{-.3cm}
        }
      };
  \draw (17+180:.7)+(.58,.03)
    node[right, below]
      {
        {
        \hspace{-.3cm}
        \tiny
        \color{darkblue} \bf
        mirror branes
        \hspace{-.3cm}
        }
      };

  \end{scope}

  \begin{scope}[shift={(0,1.6)}]
  \draw (0,-3.4) node {\tiny $x_1 = 0$};
  \draw (3,-3.4) node {\tiny $x_1 = \tfrac{1}{2}$};

  \end{scope}

  \draw (-2.5,0) node {\tiny $x_2 = 0$};
  \draw (-2.5,3) node {\tiny $x_2 = \tfrac{1}{2}$};

  \draw[|->] (6.8,1.5) to ++(.6,0);

  \draw (1.4+8,1.5) node
   {\footnotesize $
     \begin{array}{c}
      4 \cdot \mathbf{1}_{{}_{\mathrm{triv}}}
      \\
      - 3 \cdot \mathbf{4}_{{}_{\mathrm{reg}}}
     \end{array}
   $};

\end{scope}
\end{tikzpicture}
\end{center}

\vspace{-.6cm}

\noindent {\bf \footnotesize Figure P -- Local/twisted tadpole cancellation in a toroidal ADE-orientifold is enforced by equivariant Cohomotopy}
{\footnotesize according to Prop. \ref{PushforwardOfVicinityOfSingularityToRepresentationTorus},
which reduces to the situation in the vicinity of a single singularity,
as in \cref{LocalTadpoleCancellation}. Shown is a case with $G = \mathbb{Z}_4$ as in \hyperlink{FigureO}{\it Figure O}.}

\noindent This is the local/twisted tadpole cancellation
in toroidal ADE-orentifolds according to
\hyperlink{Table1}{\it Table 1} and \hyperlink{Table2}{\it Table 2}.

\medskip

\noindent {\bf Global/untwisted tadpole cancellation from super-differential Cohomotopy.}
This concludes our discussion of local tadpole cancellation
in global (i.e. toroidal) ADE-orientifolds implied by
C-field charge quantization in equivariant Cohomotopy.
Finally, we turn to discuss how the global/untwisted tadpole cancellation
condition on toroidal orbifolds follows from charge quantization
in super-differential equivariant Cohomotopy.
We state the concrete condition below in \eqref{KernelOfTheGlobalElmendorfStageProjection},
but first we explain how this condition arises from super-differential
refinement:

\medskip

\noindent {\bf Super-differential enhancement of unstable equivariant Cohomotopy theory.}
Given any generalized cohomology theory for charge quantization,
it is its corresponding enhancement to a \emph{differential cohomology theory} which classifies
not just the topological soliton/instanton sectors, but the actual
geometric higher gauge field content, hence including the flux densities.
For stable/abelian cohomology theories this is discussed
for instance in \cite{Freed00}\cite{Bunke12}, while in the broader
generality of unstable/non-abelian cohomology theories
this is discussed in \cite{FSS10}\cite{SSS12}\cite{FSS12}\cite{FSS15}.
For example, ordinary degree-2 integral cohomology theory 
$B U(1) \simeq B^2 \mathbb{Z}$
classifies magnetic charge sectors,
but it is its differential cohomology enhancement $\mathbf{B}U(1)_{\mathrm{conn}}$
(Deligne cohomology) which is the universal moduli for actual electromagnetic field configurations.
Similarly, plain (twisted) K-theory $K U$ and $K O$ classifies
topological RR-charge sectors, but it is differential K-theory
which classifies the actual RR-fields; see \cite{GS-AHSS}\cite{GS19A}\cite{GS19B}.

\medskip
Hence with \hyperlink{HypothesisH}{\it Hypothesis H}
we are ultimately to consider the refinement of
ADE-equivariant Cohomotopy theory $\pi^{\mathbf{4}_{\mathbb{H}}}_G$,
discussed so far,
to some \emph{differential} equivariant Cohomotopy theory,
denoted
$\big( \pi^{\mathbf{4}_{\mathbb{H}}}_G \big)_{\mathrm{conn}}$
and characterized as completing a homotopy pullback diagram
of geometric unstable cohomology theories of the
following form:
\begin{equation}
  \label{DifferentialEquivariantCohomotopyPullback}
  \hspace{-10mm}
  \raisebox{43pt}{
  \xymatrix@C=6em@R=18pt{
    \mathclap{
    \mbox{
      \tiny
           \begin{tabular}{c}
        super-differential
        \\
        unstable equivariant Cohomotopy
        \\
        \cite{OrbifoldCohomology}
        = \cite{FSS15} $\wedge$ \cite{ADE}
      \end{tabular}
    }
    }
    &
    \big(\pi^\bullet_G\big)_{\mathrm{conn}}
    \ar@{}[ddrr]|-{
      \mbox{\bf
        \tiny
        \color{darkblue}
        \begin{tabular}{c}
          {\color{black} \large
          \begin{rotate}{-140}
            $\!\!\Rightarrow$
          \end{rotate}
          }
          \\
          \\
          universal homotopy
        \end{tabular}
      }
    }
    \ar[rr]^{
      \mbox{\bf
        \tiny
        \color{darkblue}
        \begin{tabular}{c}
          forget topological data,
          \\
          retain only flux super-forms
          \\
        \end{tabular}
      }
    }
    \ar[dd]|-{
      \mbox{\bf
        \tiny
        \color{darkblue}
        \begin{tabular}{c}
          forget flux forms,
          \\
          retain underlying cocycle in
          \\
          plain equivariant Cohomotopy
        \end{tabular}
      }
      \;\;\;\;\;\;\;
      \;\;\;\;\;\;\;
      \;\;\;\;\;\;\;
    }
    &&
    \Big\{ \big( \mu_{{}_{\rm M2/M5}}\big)_G  \Big\}
    \ar[dd]|-{
      \;\;\;\;\;\;\;
      \;\;\;\;\;\;\;
      \;\;\;\;\;\;\;
      \mbox{\bf
        \tiny
        \color{darkblue}
        \begin{tabular}{c}
          inject this cocycle, thereby
          \\
          enforce 11d SuGra torsion constraint
        \end{tabular}
      }
    }
    &
    \mathclap{
    \mbox{
      \tiny
           \begin{tabular}{c}
        $G$-equivariant enhancement \cite[5]{ADE}
        \\
        of M2/M5-brane super WZW-terms
        \\
        jointly regarded as a cocycle in
        \\
        super-rational 4-Cohomotopy
        \\
        \cite[3]{FSS15}\cite[2.3]{FSS16a}\cite[(57)]{FSS19a}
      \end{tabular}
    }
    }
    \\
    \\
    \mathclap{
    \mbox{
      \tiny
       \begin{tabular}{c}
      unstable equivariant Cohomotopy
      \\
      cohomology theory
      \\
      \eqref{EquivariantCohomotopySet}
      \end{tabular}
    }
    }
    &
    \pi^\bullet_G
    \ar[rr]_-{
      \mbox{\bf
        \tiny
        \color{darkblue}
        \begin{tabular}{c}
          rationalize, i.e.:
          \\
          forget all torsion subgroups
          \\
          in homotopy/cohomology groups
        \end{tabular}
      }
    }
    &&
    \mbox{\bf $\Omega$}_G(-, \mathfrak{l}S^4_G)
    &
    \mathclap{
    \mbox{
      \tiny
           \begin{tabular}{c}
        super-rational
        \\
        unstable equivariant Cohomotopy theory
        \\
        \cite[3.2]{ADE}
      \end{tabular}
    }
    }
    }
  }
\end{equation}

\noindent
Discussing this construction $\big( \pi^{\mathbf{4}_{\mathbb{H}}}_G \big)_{\mathrm{conn}}$ in detail requires invoking concepts from
$\infty$-stacks and $L_\infty$-algebroids \cite{FSS10}\cite{SSS12},
as well as their application to super-geometric orbifolds \cite{OrbifoldCohomology},
which is beyond the scope of this article.
However, for the present purpose of seeing the global tadpole
cancellation condition arise,
all that matters are the following implications of
super-differential refinement, which we make explicit
by themselves:

\medskip

\noindent {\bf Rational flux constraints from equivariant enhancement of M2/M5-cocycle.}
The homotopy pullback construction \eqref{DifferentialEquivariantCohomotopyPullback}
amounts to equipping the rationalization of cocycles in
plain unstable equivariant Cohomotopy \eqref{EquivariantCohomotopySet}
with equivalences (connection data) to prescribed
flux super-forms in super-rational equivariant Cohomotopy theory
\cite[3.2]{ADE}. The flux super-forms relevant
for charge-quantization of the M-theory C-field
according to \hyperlink{HypothesisH}{\it Hypothesis H}
are $G$-equivariant enhancements of the joint M2/M5-brane cocycle \cite[3]{FSS15}\cite[2.3]{FSS16a}\cite[3.42]{ADE}\cite[(57)]{FSS19a}
with coefficients in the rationalized 4-sphere $\mathfrak{l}S^4$:

\vspace{-3mm}
\begin{equation}
  \label{TheM2M5Cocycle}
  \xymatrix{
    \underset{
      \mathclap{
      \mbox{\bf
        \tiny
        \color{darkblue}
        \begin{tabular}{c}
          $D = 11$, $\mathcal{N} = 1$
          \\
          super-Minkowski spacetime
        \end{tabular}
      }
      }
    }{
      \mathbb{R}^{10,1\vert \mathbf{32}}
    }\;\;\;\;
    \ar[rrrrr]^-{
      \mu_{{}_{\rm M2/M5}}
      \coloneqq
      \left(
        {
          {
            \frac{i}{2}
            \overline{\psi}\Gamma_{a_1 a_2} \psi
            \wedge
            e^{a_1} \wedge e^{a_2}
            \,,
          }
        \atop
          {
            \frac{1}{5}
            \overline{\psi}\Gamma_{a_1 \cdots a_5} \psi
            \wedge
            e^{a_1} \wedge \cdots \wedge e^{a_2}
          }
        }
      \right)
    }_-{
      \;\;\;\;\;
      \mbox{\bf
        \tiny
        \color{darkblue}
        \begin{tabular}{c}
          M2/M5-brane super-cocycle
          \\
          (joint M2/M5 WZW-term curvatures)
        \end{tabular}
      }
    }
    &&&&&
    \underset{
      \mathclap{
      \mbox{\bf
        \tiny
        \color{darkblue}
        \begin{tabular}{c}
          rationalized
          \\
          4-sphere
        \end{tabular}
      }
      }
    }{\;\;\;
      \mathfrak{l} S^4
    }
  }\;.
\end{equation}
Specifically, for
$G^{\mathrm{ADE}}$-equivariance \eqref{ADESubgroups} at ADE-singularities
$\mathbf{4}_{\mathbb{H}}$ \eqref{TheQuaternionicRepresentation},
a choice of equivariant extension of
this cocycle is a choice of extension
to an Elmendorf-stage diagram
as in \eqref{ElmedorfStageWiseHopfDegrees} --
see \cite[5]{ADE}:\footnote{
  For more general actions this involves extension to a functor on the
  \emph{orbit category}; see \cite[Lemma 5.4]{ADE}.
}
\begin{equation}
  \label{ElmendorfStagesOfEquivariantM2M5Cocycle}
  \xymatrix@R=6pt@C=5em{
    &
  (  \mathbb{R}^{10,1\vert \mathbf{32}}
    \ar@(ul,ur)|-{\;\;G^{\mathrm{ADE}}\!\!\!\!})
    \ar[rr]^{ (\mu_{{}_{\rm M2/M5}})_{G_{\mathrm{ADE}}} }
    &&
   ( \mathfrak{l}
    S^{\mathbf{4}_{\mathbb{H}}}
    \ar@(ul,ur)|-{\;\;G^{\mathrm{ADE}}\!\!\!\!})
    &
    \mbox{\bf
     \tiny
     \color{darkblue}
     \begin{tabular}{c}
       $G^{\mathrm{ADE}}$-equivariant enhancement
       \\
       of M2/M5-brane super-cocycle
     \end{tabular}
    }
    \\
    (-)^{H = 1}
    \ar@{}[dd]|-{
      \mathclap{
      \mbox{\bf
       \tiny
       \color{darkblue}
       \begin{tabular}{c}
         Elmendorf stages
         \\
         \eqref{SystemOfMapsOnHFixedSubspaces}
       \end{tabular}
      }
      }
    }
    &
    \mathbb{R}^{10,1\vert \mathbf{32}}
    \ar[rr]^{ \mu_{{}_{\rm M2/M5}} }
    &&
    \mathfrak{l}
    S^{\mathbf{4}_{\mathbb{H}}}
    &
    \mbox{\bf
      \tiny
      \color{darkblue}
      \begin{tabular}{c}
        M2/M5-brane super-cocycle
        \\
        \eqref{TheM2M5Cocycle}
      \end{tabular}
    }
    \\
    \\
    {(-)}^{H = G^{\mathrm{ADE}}}
    &
    \underset{
      \mathclap{
      \mbox{\bf
      \tiny
      \color{darkblue}
      \begin{tabular}{c}
        MK6 super-embedding
        \\
        (see \cite[Thm. 4.3]{ADE} and
        Rem. \ref{TheRoleOfMK6EndingOnM5})
      \end{tabular}
      }
      }
    }{
      \mathbb{R}^{6,1\vert \mathbf{16}}
    }
    \ar@{^{(}->}[uu]
    \ar[rr]^{ \in \{0,1\} }
    &&
    \mathfrak{l}S^0
    \ar@{^{(}->}[uu]
    &
    \mbox{\bf
      \tiny
      \color{darkblue}
      \begin{tabular}{c}
        charge at fixed planes
      \end{tabular}
    }
  }
\end{equation}
This involves a binary choice
at lowest (and hence any other, by Example \ref{FixedSubspacesOfADESingularities})
Elmendorf stage. The homotopy in the diagram \eqref{DifferentialEquivariantCohomotopyPullback}
enforces this local choice of rationalized flux
globally onto the rationalized
fluxes of the equivariant Cohomotopy cocycles. This
has two effects:

\medskip
\noindent {\bf 1. Super-differential enhancement at global Elmendorf stage implies vanishing total flux.}
Note the M2/M5-brane super-cocycle $\mu_{{}_{\rm M2/M5}}$ \eqref{TheM2M5Cocycle}
appearing at global Elmendorf stage in \eqref{ElmendorfStagesOfEquivariantM2M5Cocycle}
has vanishing bosonic flux ( $\mu_{{}_{\rm M2/M5}}\vert_{\psi = 0} = 0$ by \eqref{TheM2M5Cocycle}).
Also, the infinitesimal fermionic component $\psi$ does
not contribute to the topology seen by plain equivariant Cohomotopy
(see \cite{OrbifoldCohomology} for details). Hence
the homotopy in \eqref{DifferentialEquivariantCohomotopyPullback}
forces the underlying classes in plain equivariant Cohomotopy to
be \emph{pure torsion} at global Elmendorf stage.
But, since in compatible RO-degree
(as in Example \ref{ExamplesOfCompatibleRODegree})
the Hopf degree theorem \eqref{HopfDegreeTheorem} implies non-torsion Cohomotopy groups
at all positive Elmendorf stages \eqref{ElmedorfStageWiseHopfDegrees},
this means
that super-differential refinement \eqref{DifferentialEquivariantCohomotopyPullback}
of equivariant Cohomotopy in compatible RO-degree enforces
\emph{vanishing} Hopf degrees at global Elmendorf stage $H = 1$
\eqref{ElmedorfStageWiseHopfDegrees}.

\medskip
Explicitly, this means that the super-differential enhancement \eqref{DifferentialEquivariantCohomotopyPullback}
forces the underlying plain equivariant Cohomotopy cocycles
of ADE-orientifolds in compatible RO-degree to be in the kernel of the forgetful map $(-)^1$ \eqref{ElmedorfStageWiseHopfDegrees}
from equivariant to ordinary Cohomotopy, which projects out
the global Elmendorf stage at $H = 1$:
\begin{equation}
  \label{KernelOfTheGlobalElmendorfStageProjection}
  \xymatrix@R=1pt@C=3em{
    \overset{
      {
      \mathclap{
      \mbox{\bf
        \tiny
        \color{darkblue}
        \begin{tabular}{c}
          unstable equivariant Cohomotopy
          \\
          admitting super-differential refinement
        \end{tabular}
      }
      }
      }
    }{
      \pi^{\mathbf{4}_{\mathbb{H}}}_{G^{\mathrm{ADE}}}
      \big(
        \big(
          \mathbb{T}^{\mathbf{4}_{\mathbb{H}}}
        \big)_+
      \big)_{
        {}_{
          \underset{
            \mathclap{
            \mbox{\bf
              \tiny
              \color{darkblue}
              \begin{tabular}{c}
                ``super-differentiable''
                \\
                \eqref{DifferentialEquivariantCohomotopyPullback}
              \end{tabular}
            }
            }
          }{\mathrm{Sdiffble}}
        }}
    }
    \ar@{^{(}->}[dd]_-{ \mathrm{kernel} }
    \ar[rrrr]
    \ar@{}[ddrrrr]|<<<<<<<<<<<<<<<<<<<{ \mbox{\tiny (pb)} }
    &&&&
    \{0\}
    \ar@{^{(}->}[dd]
    \\
    \\
    \;\;\;
    \underset{
      \mathclap{
      \mbox{\bf
        \tiny
        \color{darkblue}
        \begin{tabular}{ccc}
          $\phantom{A}$ &&
          \\
       &&   equivariant Cohomotopy \eqref{EquivariantCohomotopySet}
          \\
        &&  of toroidal orbifold \eqref{RepresentationTorus}
          \\
        &&  with ADE-singularities \eqref{TheQuaternionicRepresentation}
        \end{tabular}
      }
      }
    }{
      \pi^{\mathbf{4}_{\mathbb{H}}}_{G^{\mathrm{ADE}}}
    }
    \big(
      \big(
        \mathbb{T}^{\mathbf{4}_{\mathbb{H}}}
      \big)_+
    \big)
    \;\;\; \ar[rrr]^-{
        \left\vert Q_{\mathrm{tot}}\right\vert \coloneqq (-)^1
    }_-{
        \mbox{\bf
          \tiny
          \color{darkblue}
          \begin{tabular}{c}
            project out
            \\
            global charge =
            \\
            Hopf degree at global
            Elmendorf stage
            \\
            \eqref{ElmedorfStageWiseHopfDegrees}
          \end{tabular}
        }
    }
    &&& \;\;\;\;
    \underset{
      \mathclap{
      \mbox{\bf
        \tiny
        \color{darkblue}
        \begin{tabular}{c}
          plain Cohomotopy
          \\
          of plain 4-torus
          \\
          \eqref{PlainCohomotopySet}
        \end{tabular}
      }
      }
    }{
      \pi^4
      \big(
        \big(
          \mathbb{T}^4
        \big)_+
      \big)
    }
    \ar@{}[r]|-{\simeq}
    &
    \underset{
      \mathclap{
      \mbox{\bf
        \tiny
        \color{darkblue}
        \begin{tabular}{c}
        \\
          global Hopf degree
          \\
          = net brane/O-plane charge
          \\
          \eqref{HopfDegreeTheorem}
        \end{tabular}
      }
      }
    }{
      \mathbb{Z}
    }
  }
\end{equation}

\medskip
\noindent It is now immediate, from
Theorem \ref{UnstableEquivariantHopfDegreeTheorem} and
Theorem \ref{UnstableEquivariantHopfDegreeTheoremForTori},
that this enforces the condition of vanishing net brane/O-plane charge,
precisely in the form of the global/untwisted tadpole cancellation
condition from \hyperlink{Table1}{\it Table 1} and
\hyperlink{Table2}{\it Table 2}
in the way illustrated in \hyperlink{FigureA}{\it Figure A}.

\vspace{4mm}
\noindent {\bf 2. Super-differential enhancement at lower Elmendorf stage implies choice of O-plane charge.}
The globalization via \eqref{KernelOfTheGlobalElmendorfStageProjection} of the lower $S^0$-valued Elmendorf stage
in the equivariantized M2/M5-brane cocycle \eqref{ElmendorfStagesOfEquivariantM2M5Cocycle} means to
impose the chosen charge $\in \{0,1\}$ to all O-planes,
via Prop. \ref{TheoremLocalTadpoleCancellation}
as illustrated in \hyperlink{FigureH}{\it Figure H}.
We will denote the ADE-equivariant Cohomotopy sets
which admit super-differential refinement with the choice
$-1 \in \{0,1\}$ in \eqref{ElmendorfStagesOfEquivariantM2M5Cocycle} by a
subscript $(-)_-$:
\begin{example}[\bf Super-differentiable equivariant Cohomotopy of ADE-orbifolds]
\label{SuperDifferentiableEquivariantCohomotopyOfADEOrbifolds}
Locally, the super-differentiable equivariant Cohomotopy
of the vicinity of an ADE-singularity (\hyperlink{Table5}{\it Table 5})
with respect to the choice $-1 \in \{-0,-1\}$
in the equivariant enhancement \eqref{ElmendorfStagesOfEquivariantM2M5Cocycle}
of the super-flux form \eqref{TheM2M5Cocycle} is
\begin{equation}
  \label{SuperDifferentiableLocalCohomotopyCharge}
  \pi^{\mathbf{4}_{\mathbb{H}}}_{G^{\mathrm{ADE}}}
  \big(
    \big(
    \mathbb{R}^{\mathbf{4}_{\mathbb{H}}}
    \big)^{\mathrm{cpt}}
  \big)_-
  \;=\;
  \Big\{
    \underset{
      \mbox{
        \tiny
        \color{darkblue} \bf
        \begin{tabular}{c}
          local charge structure
          \\
          (Prop. \eqref{TheoremLocalTadpoleCancellation})
        \end{tabular}
      }
    }{
      1 \cdot \mathbf{1}_{\mathrm{triv}}
      -
      N_{\mathrm{brane}} \cdot \mathbf{k}_{\mathrm{reg}}
     }
     \;\Big\vert\;
     N_{\mathrm{brane}} \in \mathbb{Z}
  \Big\}
  \,.
\end{equation}

Globally,
the super-differentiable equivariant
Cohomotopy specifically of the Kummer surface
ADE-orbifold
$\mathbb{T}^{\mathbf{4}_{\mathbb{H}}}\sslash \mathbb{Z}^{\mathrm{refl}}_2$
(Example \ref{KummerSurface}) is

\begin{equation}
  \label{SuperDifferentiableEquivariantCohomotopyOfKummerSurface}
  \underset{
    \mbox{
      \tiny
      \color{darkblue} \bf
      \begin{tabular}{c}
        ADE-equivariant Cohomotopy
        \\
        admitting super-differential lift
        \\
        \eqref{DifferentialEquivariantCohomotopyPullback}
      \end{tabular}
    }
  }{
  \pi^{\mathbf{4}_{\mathbb{H}}}_{\mathbb{Z}_2^{\mathrm{refl}}}
  \big(
    \big(
    \mathbb{T}^{\mathbf{4}_{\mathbb{H}}}
    \big)_+
  \big)_{\mathrm{Sdiffble}_-}
  }
  \;=\;
  \Big\{
    \underset{
      \mbox{
        \tiny
        \color{darkblue} \bf
        \begin{tabular}{c}
        \\
          super-differentiability
          \\
          at low Elmendorf stage
          \\
          \eqref{ElmendorfStagesOfEquivariantM2M5Cocycle}
        \end{tabular}
      }
    }{
      16 \cdot \mathbf{1}_{\mathrm{triv}}
    }
    -
    \underset{
      {
      \mbox{
        \tiny
        \color{darkblue} \bf
        \begin{tabular}{c}
                \\
          local charge structure
          \\
          (Prop. \ref{TheoremLocalTadpoleCancellation},
          Prop. \ref{PushforwardOfVicinityOfSingularityToRepresentationTorus})
        \end{tabular}
      }
      }
    }{
      N_{\mathrm{brane} } \cdot \mathbf{2}_{\mathrm{reg}}
    }
    \;\big\vert\;
    \underset{
      \mbox{
        \tiny
        \color{darkblue} \bf
        \begin{tabular}{c}
        \\
          super-differentiablity
          \\
          at global Elmendorf stage
          \\
          \eqref{KernelOfTheGlobalElmendorfStageProjection}
        \end{tabular}
      }
    }{
      2 N_{\mathrm{brane} } - 16 = 0
    }
  \Big\}.
\end{equation}
\end{example}

\section{M5/MO5 anomaly cancellation}
\label{M5MO5AnomalyCancellation}

We now apply the general discussion of equivariant Cohomotopy in
\cref{EquivariantCohomotopyAndTadpoleCancellation} to cohomotopical charge quantization
of the M-theory C-field, according to \hyperlink{HypothesisH}{\it Hypothesis H},
for compactifications of heterotic M-theory on toroidal orbifolds with ADE-singularities. The
resulting M5/MO5-anomaly cancellation is discussed in \cref{EquivariantCohomotopyChargeOfM5AtMO5}
below. In order to set the scene and to sort out some fine print,
we first discuss in \cref{HeteroticMTheoryOnADEOrbifolds}
relevant folklore regarding heterotic M-theory on ADE-orbifolds.

\subsection{Heterotic M-theory on ADE-orbifolds}
\label{HeteroticMTheoryOnADEOrbifolds}

We now explain how the singularity structure (as in \hyperlink{Table5}{\it Table 5}), which
must really be meant when speaking of MO5-planes \eqref{TheMO5} coinciding with black
M5-branes \eqref{M5Singularity}, is that  of ``$\tfrac{1}{2}\mathrm{M5}$-branes''
\eqref{TheHalfM5} \cite[2.2.7]{ADE}\cite[4]{FSS19d}; see
\hyperlink{FigureS}{\it Figure S} below.
This singularity structure  goes back to \cite[3]{Sen97} with further discussion and development in
\cite{FLO99}\cite{KSTY99}\cite{FLO00a}\cite{FLO00b}\cite{FLO00c}\cite{CabreraHananySperling19};
the type IIA perspective is considered in \cite{GKST01} and also briefly in \cite[p. 4]{KataokaShimojo02}.
We highlight the systematic picture behind the resulting {\it heterotic M-theory on ADE-orbifolds}
and its string theory duals, further below in
\hyperlink{Table7}{\it Table 7}.

\medskip

\noindent
{\bf Critique of pure $\mathrm{MO5}$-planes.} We highlight the following:
\begin{enumerate}[{\bf (i)}]
\vspace{-2mm}
\item Seminal literature on M-theoretic orientifolds speaks of
M5-branes parallel and/or coincident to {\it MO5} {\it singularities}
\cite{DasguptaMukhi95}\cite[3.3]{Witten95b}\cite[2.1]{Hori98},
namely to Euclidean $\mathbb{Z}_2$-orientifolds \eqref{EuclideanGSpace}
of the form (see \cite[2.2.2]{ADE}):

\vspace{-.6cm}

\begin{equation}
  \label{TheMO5}
  {\color{darkblue}\tiny \bf \mathrm{MO5}}
  \phantom{AAAAAA}
  \mathbb{R}^{5,1}
 \; \xymatrix{\ar@{^{(}->}[r]&}
  \mathbb{R}^{5,1}
  \times
  \xymatrix{
    \mathbb{R}^{\mathbf{5}_{\mathrm{sgn}}}
    \ar@(ul,ur)|-{\, \mathbb{Z}_2}
  }
  \,,
\end{equation}
where $\mathbb{R}^{\mathbf{5}_{\mathrm{sgn}}}$
is the Euclidean singularity \eqref{EuclideanGSpace} of the 5-dimensional sign representation of the group $\mathbb{Z}_2$.

\vspace{-2mm}
\item But $\sfrac{1}{2}$BPS M5-brane solutions of $D=11$ supergravity themselves have been classified
\cite[8.3]{MF10} and found to be given, in their singular far horizon limit \cite[3]{AFCS99},
by singularities for finite subgroups $G^{\mathrm{ADE}} \subset \mathrm{SU}(2) \simeq \mathrm{Sp}(1)$ \eqref{ADESubgroups} of the type

\vspace{-.6cm}

\begin{equation}
  \label{M5Singularity}
  {\color{darkblue}\tiny \mathrm{M5}}
  \phantom{AAAAAA}
  \mathbb{R}^{5,1}
 \; \xymatrix{\ar@{^{(}->}[r]&}
  \mathbb{R}^{5,1}
  \times
  \mathbb{R}^1
  \times
  \xymatrix{
    \mathbb{R}^{\mathbf{4}_{\mathbb{H}}}
    \ar@(ul,ur)|{\;\;\; G^{\mathrm{ADE}} }
  }
  \,,
\end{equation}
where the
last factor is an ADE-singularity \eqref{TheQuaternionicRepresentation}.

\vspace{-3mm}
\item As orbifold singularities, this coincides with the far horizon geometry
of coincident KK-monopole solutions to 11d supergravity
(e.g. \cite[(47)]{IMSY98}\cite[(18)]{Asano00};  see \cite[2.2.5]{ADE})

\vspace{-.3cm}

\begin{equation}
  \label{MK6Singularity}
  {\color{darkblue}\tiny \mathrm{MK6}}
  \phantom{AAAAAA}
  \mathbb{R}^{6,1}
   \; \xymatrix{\ar@{^{(}->}[r]&}
  \mathbb{R}^{6,1}
  \times
  \xymatrix{
    \mathbb{R}^{\mathbf{4}_{\mathbb{H}}}
    \ar@(ul,ur)|{\;\;\; G^{\mathrm{ADE}} }
  }
  \,,
\end{equation}
which, from the perspective of type IIA theory, reflects the fact
that NS5-branes are domain walls inside D6-branes
(e.g. \cite[p. 5]{EGKRS00}, see \cite[3.3.1. 3.3.2]{Fazzi17}).
This is illustrated by the central dot on the vertical axis in \hyperlink{FigureS}{\it Figure S}.  Hence for the special case that
$G^{\mathrm{ADE}} = \mathbb{Z}^{\mathrm{refl}}_2$ \eqref{PointReflectionSubgroup},
this yields the product $\mathbb{R}^1 \times \mathbb{R}^{\mathbf{4}_{\mathrm{sgn}}}$ of
the 4-dimensional sign representation with the trivial 1-dimensional representation, instead of the
5-dimensional sign representation in \eqref{TheMO5}.

\vspace{-3mm}
\item
In order to allow M5-singularities \eqref{M5Singularity}
to coincide with MO5-singularities \eqref{TheMO5}
we have to consider intersecting a
$\sfrac{1}{2}$BPS 5-brane solution with an $\mathrm{MO9}$ locus
fixed by a Ho{\v r}ava-Witten involution $\mathbb{Z}_2^{\mathrm{HW}}$
(\cite{HoravaWitten96a}, see \cite[2.2.1]{ADE}):

\vspace{-4mm}
\begin{equation}
  \label{TheMO9}
  {\color{darkblue}\tiny \mathrm{MO9}}
  \phantom{AAAAAA}
  \mathbb{R}^{9,1}
   \; \xymatrix{\ar@{^{(}->}[r]&}
  \mathbb{R}^{9,1}
  \times
  \xymatrix{
    \mathbb{R}^{\mathbf{1}_{\mathrm{sgn}}}
    \ar@(ul,ur)|-{\, \mathbb{Z}^{\mathrm{HW}}_2\!\!\!\!\!}
  }
  \,.
\end{equation}

\vspace{-3mm}
\item This intersection is called the $\tfrac{1}{2}\mathrm{M5}$ in \cite[2.2.7]{ADE}\cite[4]{FSS19d}

\vspace{-.5cm}

\begin{equation}
  \label{TheHalfM5}
  {\color{darkblue}
    \tfrac{1}{2}\mathrm{M5}
  }
  = \mathrm{MK6} \cap \mathrm{MO9}
  \phantom{AAAAAA}
  \mathbb{R}^{5,1}
  \; \xymatrix{\ar@{^{(}->}[r]&}
  \mathbb{R}^{5,1}
  \times
  \xymatrix{
    \mathbb{T}^{\mathbf{1}_{\mathrm{sgn}}}
    \ar@(ul,ur)^{\mathbb{Z}_2^{\mathrm{HW}}}
  }
  \xymatrix{
    \times
    \ar@[white]@(ul,ur)^{\times}
  }
  \xymatrix{
    \mathbb{T}^{\mathbf{4}_{\mathbb{H}}}
    \ar@(ul,ur)^{ G^{\mathrm{ADE}} }
  }
\end{equation}
since its type IIA incarnation
is known as the $\tfrac{1}{2}\mathrm{NS5}$
\cite[6]{GKST01}\cite[p. 18]{ApruzziFazzi17}.
This is the brane configuration thought to geometrically engineer
$D=6$, $\mathcal{N} = (1,0)$ field theories \cite{HananyZaffaroni97}\cite{HKLY15}\cite[6]{DHTV14}.
\end{enumerate}
\vspace{.0cm}

\noindent \hspace{-5pt}
$
  \mbox{
    \hyperlink{FigureR}{}
    \begin{minipage}[l]{9.3cm}
Since the fixed point set of the
toroidal orbifolds \eqref{RepresentationTorus}
for both the $\tfrac{1}{2}\mathrm{M5}$ \eqref{TheHalfM5}
and the $\mathrm{MO5}$ \eqref{TheMO5}
is the same set \eqref{RepresentationTorusOfSignRep}
of 32 points, all arguments about
$\mathrm{MO5}$ \eqref{TheMO5} which depend only on the
set of isolated orientifold fixed points, such as in \cite{DasguptaMukhi95}\cite[3.3]{Witten95b}\cite[2.1]{Hori98},
apply to $\tfrac{1}{2}\mathrm{M5}$ \eqref{TheHalfM5} as well.
But the $\tfrac{1}{2}\mathrm{M5}$ orientifold has in addition
fixed lines, namely the $\mathrm{MK6}$ loci, and fixed 4-planes, namely
the $\mathrm{MO9}$, as shown on the right of
\hyperlink{FigureS}{Figure S}.
This reflects the fact that, by the classification of \cite[8.3]{MF10},
 the black $\mathrm{M5}$ not only may, but must appear as a domain wall inside an $\mathrm{MK6}$ singular locus.

We {\bf conclude} from this that:
\emph{The $\tfrac{1}{2}\mathrm{M5}$
\eqref{TheHalfM5} orientifold is the
correct model of orientifolded M5/MO5 geometry, while
the pure $\mathrm{MO5}$ \eqref{TheMO5} is
just its restriction along the diagonal subgroup inclusion
\eqref{Z2ReflHW}, as shown in \hyperlink{FigureR}{\it Figure R}}
\end{minipage}
  }
  \phantom{AA}
  \raisebox{70pt}{\small
  \xymatrix@C=-5pt@R=-1pt{
    &
    \mathclap{
    \overset{
      \mbox{\bf
        \tiny
        \color{darkblue}
        \begin{tabular}{c}
          orientifold
          \\
          subgroup
          \\
          ${\phantom{A}}$
        \end{tabular}
      }
    }{
      H \subset G
    }
    }
    &
    &
    &
    &
    \mathclap{
    \overset{
      \mbox{\bf
        \tiny
        \color{darkblue}
        \begin{tabular}{c}
          fixed/singular
          subspace \eqref{FixedLoci}
          \\
          ${\phantom{A}}$
        \end{tabular}
      }
    }{
      \mathbb{R}^{5,1} \times
      \big(
        \mathbb{R}^{
          \mathbf{1}^{\mathrm{HW}}_{\mathrm{sgn}}
          +
          \mathbf{4}^{\mathrm{ADE}}_{\mathbb{H}}
        }
      \big)^H
    }
    }
    \\
    &
      \mathbb{Z}_2^{\mathrm{HW}}
      \times
      G^{\mathrm{ADE}}
    &&&&
    \overset{
      \mbox{\bf
        \tiny
        \color{darkblue}
        $\tfrac{1}{2}\mathrm{M5}$
      }
    }{
      \mathrlap{\phantom{\vert \atop \vert}}
      \mathbb{R}^{5,1}
    }
    \\
    \mathbb{Z}_2^{\mathrm{HW}}
    \ar@{^{(}->}[ur]
    &
    &
    \mathbb{Z}_2^{\mathrm{refl}}
    \ar@{^{(}->}[ul]
    &
    {\phantom{AA}}
    &
    \overset{
      \mbox{\bf
        \tiny
        \color{darkblue}
        \begin{tabular}{c}
          MO9
        \end{tabular}
      }
    }{
      \mathrlap{\phantom{\vert \atop \vert}}
      \mathbb{R}^{9,1}
    }
    &&
    \overset{
      \mbox{\bf
        \tiny
        \color{darkblue}
        \begin{tabular}{c}
          MK6
        \end{tabular}
      }
    }{
      \mathrlap{\phantom{\vert \atop \vert}}
      \mathbb{R}^{6,1}
    }
    \\
    &
    \mathclap{\phantom{\vert \atop \vert}}
    \mathbb{Z}_2^{ \mathrm{refl}+\mathrm{HW} }
    \ar@{^{(}->}[uu]|-{\mathrm{diag}}
    &&&&
    \underset{
      \mbox{\bf
        \tiny
        \color{darkblue}
        $\mathrm{MO}5$
      }
    }{
      \mathrlap{\phantom{\vert \atop \vert}}
      \mathbb{R}^{5,1}
    }
    \ar@{^{(}->}[ul]
    \ar@{_{(}->}[ur]
    \\
    &
    &
    &&&
    \\
    &
    1
    \ar@/^1pc/@{^{(}->}[uuul]
    \ar@{^{(}->}[uu]
    \ar@/_1pc/@{_{(}->}[uuur]
    &&&&
    \\
    \mathrlap{
    \!\!\!\!\!\!\!\!\!\!
    \mbox{
      \begin{minipage}[l]{7cm}
      {\footnotesize \bf
      Figure R --
      Fixed subspaces in the $\tfrac{1}{2}\mathrm{M5}$-singularity
      \eqref{TheMO5}  }
      {\footnotesize with MO5 \eqref{TheMO5} in the intersection
       of MK6 \eqref{MK6Singularity} with MO9 \eqref{TheMO9},
       illustrated in \hyperlink{FigureS}{\it Figure S}. }
      \end{minipage}
    }
    }
  }
  }
$

\vspace{2mm}
\begin{equation}
  \label{Z2ReflHW}
\xymatrix{
    \underset{
      \mbox{\bf
        \tiny
        \color{darkblue}
        \eqref{TheMO5}
      }
    }{
      \mathbb{Z}_2^{\mathrm{refl}+\mathrm{HW}}
    }
    \; \ar@{^{(}->}[rr]^-{\small  \mathrm{diag} }
    &&
    \underset{
      \mbox{\bf
        \tiny
        \color{darkblue}
        \eqref{TheMO9}
      }
    }{
      \mathbb{Z}^{\mathrm{HW}}_2
    }
    \times
    \underset{
      \mbox{\bf
        \tiny
        \color{darkblue}
        \eqref{PointReflectionSubgroup}
      }
    }{
      \mathbb{Z}^{\mathrm{refl}}_2
    }
  \;  \ar@{^{(}->}[rr]
    &&
    \underset{
      \mbox{\bf
        \tiny
        \color{darkblue}
        \eqref{TheMO9}
      }
    }{
      \mathbb{Z}_2^{\mathrm{HW}}
    }
      \times
    \underset{
      \mbox{\bf
        \tiny
        \color{darkblue}
        \eqref{ADESubgroups}
      }
    }{
      G^{\mathrm{ADE}}
    }
}.
\end{equation}

\noindent In summary, this data arranges into a short exact sequence
of orbi-/orienti-fold group actions (as in \cite[p. 4]{DFM11})
\begin{equation}
 \label{OrbiOrientifoldGroupSequence}
\hspace{-4mm}
  \xymatrix@C=2.5em@R=0pt{
    1
    \ar[r]
    &
    \underset{
      \mbox{\bf
        \tiny
        \color{darkblue}
        \begin{tabular}{c}
          \phantom{a}
          \\
          orbifold
        \end{tabular}
      }
    }{
      G^{\mathrm{ADE}}
    }
    \ar@{^{(}->}[rr]^-{
      \mbox{
        \tiny
        \begin{tabular}{c}
          index-2 subgroup
        \end{tabular}
      }
    }
    &&
    \overset{
      \{\mathrm{e}, \sigma\}
    }{
      \overbrace{
        \mathbb{Z}_2^{\mathrm{HW}}
      }
    }
    \times
    G^{\mathrm{ADE}}
    \ar[rr]^-{
      \scalebox{.7}{
      $
      \begin{aligned}
        (\mathrm{e}, q) & \mapsto (\mathrm{e}, + q)
        \\
        (\sigma, q) & \mapsto ( R , - q )
      \end{aligned}
      $
      }
    }_-{ \simeq }
    &
    \underset{
      \mathclap{
      \mbox{\bf
        \tiny
        \color{darkblue}
        \begin{tabular}{c}
          \phantom{a}
          \\
          ---\!---\!---\!---\!---\!---\!---\!---\!---\!---\!---
          orbi-orientifold
          ---\!---\!---\!---\!---\!---\!---\!---\!---\!---\!---
        \end{tabular}
      }
      }
    }
    {\phantom{\mathclap{A}}}
    &
    \overset{
      \{\mathrm{e}, R\}
    }{
      \overbrace{
        \mathbb{Z}_2^{\mathrm{HW} + \mathrm{refl}}
      }
    }
    \times
    G^{\mathrm{ADE}}
    \ar@{->>}[r]
    &
    \underset{
      \mbox{\bf
        \tiny
        \color{darkblue}
        \begin{tabular}{c}
          \phantom{A}
          \\
          orientifold
        \end{tabular}
      }
    }{
      \mathbb{Z}_2^{\mathrm{HW} + \mathrm{refl}}
    }
    \ar[r]
    &
    1
    \\
    &
    \mathbb{R}^{\mathbf{1}_{\mathrm{triv}} + \mathbf{4}_{\mathbb{H}}}
    &&
    \mathbb{R}^{ \mathbf{1}_{\mathrm{sgn}} + \mathbf{4}_{\mathbb{H}} }
    &&
    \mathbb{R}^{ \mathbf{5} }
    &&
    \mathbb{R}^{ \mathbf{5}_{\mathrm{sgn}} }
  }
\end{equation}

\noindent This situation is illustrated by the following figure:

\vspace{3mm}
{\small
\hypertarget{FigureS}{}
\begin{tabular}{cc}
\begin{tabular}{|cc||c|c|}
  \hline
  \multicolumn{2}{|c||}{
    $\mbox{\bf Orientifold}$
  }
  &
  $\mathclap{\phantom{\vert \atop \vert}}
  \mathrm{\bf MO5}$ & $\tfrac{1}{2}\mathrm{\bf M5}$
  \\
  \hline
  \hline
  \begin{tabular}{c}
    Global quotient
    \\
    group
  \end{tabular}
    &
  $G=$
    &
  $\mathbb{Z}_2$
    &
  $\mathclap{\phantom{A \atop A}}$
  $\mathbb{Z}_2^{\mathrm{HW}} \times G^{\mathrm{ADE}}$
  \\
  \hline
  \begin{tabular}{c}
    Global quotient
    \\
    group action
  \end{tabular}
  &
  $\xymatrix{ \mathbb{T}^V \ar@(ul,ur)^G } = $
  &
  $
  \xymatrix{
    \mathbb{T}^{\mathbf{5}_{\mathrm{sgn}}}
    \ar@(ul,ur)^{\mathbb{Z}_2}
  }
  $
  &
  $
  \xymatrix{
    \mathbb{T}^{\mathbf{1}_{\mathrm{sgn}}}
    \ar@(ul,ur)^{\mathbb{Z}_2^{\mathrm{HW}}}
  }
  \xymatrix{
    \times
    \ar@[white]@(ul,ur)^{\times}
  }
  \xymatrix{
    \mathbb{T}^{\mathbf{4}_{\mathrm{sgn}}}
    \ar@(ul,ur)^{ G^{\mathrm{ADE}} }
  }
  $
  \\
  \hline
  \begin{tabular}{c}
    Fixed/singular
    \\
    points
  \end{tabular}
  &
  $\left( T^V\right)^G = $
  &
  \multicolumn{2}{c|}{ $\{0,\tfrac{1}{2}\}^5 = \overline{32}$ }
  \\
  \hline
  \multicolumn{2}{|c||}{
    \begin{tabular}{c}
      Far horizon-limit
      \\
      of M5 SuGra solution?
    \end{tabular}
  }
  &
  no
  & yes
  \\
  \hline
\end{tabular}
   \hspace{-.1cm}
   \scalebox{.76}{
   \raisebox{-96pt}{
   \includegraphics[width=.5\textwidth]{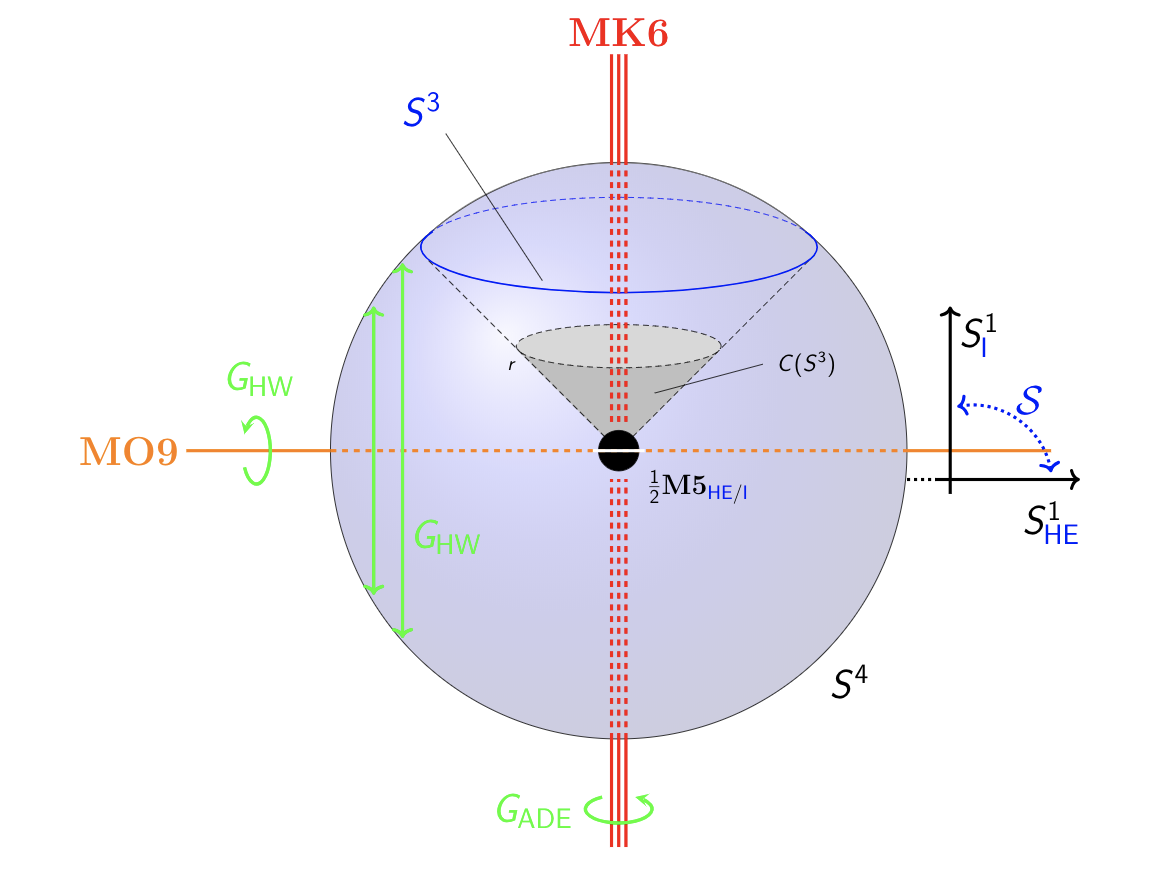}
   }}
\end{tabular}
}

\noindent {\footnotesize \bf Figure S
-- Singularity structure of heterotic M-theory on ADE-singularities},
{\footnotesize as in \hyperlink{FigureR}{Figure R}, \cite[2.2.2, 2.2.7]{ADE}.
The corresponding toroidal orbifolds
(as per \hyperlink{Table5}{\it Table 5})
are illustrated in \hyperlink{FigureV}{\it Figure V} and
\hyperlink{Table8}{\it Table 8}.}

\vspace{5mm}
\noindent {\bf $\mathrm{O}^0$-planes and M2-brane CS level.}
There is one more ingredient to the $G$-space structure
of heterotic M-theory on ADE-orbifolds (see \hyperlink{Table7}{\it Table 7}
below for the full picture): While the MO5-planes \eqref{TheMO5}
are supposed to be the M-theory lifts of the charged $\mathrm{O4}^{\pm}$-planes
\cite[3]{Hori98}\cite[III.A]{Gimon98}\cite[3.1.1]{HananyKol00},
the M-theory lift of the un-charged $\mathrm{O4}^0$-planes
(see \hyperlink{FigureOP}{\it Figure OP})
involves one more group action on spacetime \cite[III.B]{Gimon98},
being rotation of the circle fiber in M/IIA-duality,
which we hence indicate as follows:

\vspace{-.5cm}

\begin{equation}
  \label{TheIIAZero}
  {\color{darkblue}\tiny  \bf \mathrm{IIA}^0}
  \phantom{AAAAAA}
  \mathbb{R}^{9,1}
  \times \varnothing
   \; \xymatrix{\ar@{^{(}->}[r]&}
  \mathbb{R}^{9,1}
  \times
  \xymatrix{
    S^1
    \ar@(ul,ur)|-{\, \mathbb{Z}^{\mathrm{rot}}_k \!\!\!}
  }
  .
\end{equation}
Here on the right we have the circle regarded as a
$\mathbb{Z}_k$-space (\cref{EquivariantCohomotopyAndTadpoleCancellation})
via rigid rotation by multiples of $2 \pi/k$, for any
$k \in \mathbb{Z} \setminus \{0\}$.
This is of course a free action
(in particular, not a representation sphere \eqref{RepSpheres})
hence with empty fixed subspace \eqref{FixedLoci},
whence the superscript $(-)^0$ and the 
empty set $\varnothing$ of fixed points in \eqref{TheIIAZero}.
But passing along the unique $\mathbb{Z}_k^{\mathrm{rot}}$-equivariant function
\eqref{EquivariantFunction}
\begin{equation}
  \label{KKReductionOnShiftCircle}
  \xymatrix{
    S^1
    \ar@(ul,ur)|-{\, \mathbb{Z}^{\mathrm{rot}}_k\!\!\!}
    \ar[rrr]^-{ \mathrm{KK}_{S^1_{\mathrm{rot}}} }_-{
      \mbox{
        \tiny
        \color{darkblue} \bf
        \begin{tabular}{c}
          KK-reduction on
          $S^1_{\mathrm{rot}}$
        \end{tabular}
      }
    }
    &&&
    \ast
    \ar@(ul,ur)|-{\, \mathbb{Z}^{\mathrm{rot}}_k \!}
  }
\end{equation}
from the circle to the point $\ast$ with its
necessarily trivial $\mathbb{Z}^{\mathrm{rot}}_k$-action,
as befits KK-reduction from M-theory to type IIA string theory
(see \cite{BSS18} for discussion in the context of \hyperlink{HypothesisH}{\it Hypothesis H}), we
obtain a non-empty fixed subspace:
\vspace{-2mm}
\begin{equation}
  \label{TheIIA}
  {\color{darkblue}\tiny  \bf \mathrm{IIA}}
  \phantom{AAAAAA}
  \mathbb{R}^{9,1}
  \; \xymatrix{\ar@{^{(}->}[r]&}
  \mathbb{R}^{9,1}
  \times
  \xymatrix{
    \ast
    \ar@(ul,ur)|-{\, \mathbb{Z}^{\mathrm{rot}}_k\!}
  }
  .
\end{equation}
In these terms, we may phrase the core of
M/IIA duality as saying that
\begin{center}
\emph{The lift of $\mathrm{IIA}$ \eqref{TheIIA}
through $\mathrm{KK}_{S^1_{\mathrm{rot}}}$ \eqref{KKReductionOnShiftCircle}
is $\mathrm{IIA}^0$ \eqref{TheIIAZero}}.
\end{center}
Notice in the case that the global 11d-spacetime is $\mathrm{AdS}_3$
times $S^7$ regarded as an $S^1_{\mathrm{rot}}$-fibration
$$
  \xymatrix@R=10pt{
    S^1\ar@(ul,ur)|-{\, \mathbb{Z}^{\mathrm{rot}}_k \!\!\!} \ar[r]
    &
    S^7
    \ar@(ul,ur)|-{\, \mathbb{Z}_k^{\mathrm{rot}}  \!\!\!}
    \ar[r]
        & \mathbb{C}P^3
  }
$$
the order $k$ of $\mathbb{Z}^{\mathrm{rot}}_k$ in \eqref{TheIIAZero}
is the level of the dual 3d Chern-Simons-matter theory \cite{ABJM08}.

\medskip
The argument in \cite[III.B]{Gimon98}, together with our discussion above,
suggests that the analogous statement for $\mathrm{O4}^0$-planes is this:

\vspace{-2mm}
\begin{center}
\emph{The lift of $\mathrm{O4}^0$
through $\mathrm{KK}_{S^1_{\mathrm{rot}}}$  \eqref{KKReductionOnShiftCircle}
is $\mathrm{MO5}^0$ \eqref{TheMO5Zero}
}.
\end{center}

\vspace{-2mm}
\noindent Hence we take $\mathrm{MO5}^0$ to be the following $G$-space/orbifold, combining $\mathrm{MO5}$ \eqref{TheMO5} with $\mathrm{IIA}^0$ \eqref{TheIIAZero}:

\vspace{-.2cm}

\begin{equation}
  \label{TheMO5Zero}
  {\color{darkblue}\tiny  \bf \mathrm{MO5}^0}
  \phantom{AAAAAA}
  \mathbb{R}^{4,1} \times \varnothing
 \; \xymatrix{\ar@{^{(}->}[r]&}
  \mathbb{R}^{4,1}
  \times
  \xymatrix{
    S^1
    \ar@(ul,ur)|-{\, \mathbb{Z}^{\mathrm{rot}}_k\!\!}
  }
  \times
  \xymatrix{
    \mathbb{T}^{\mathbf{5}_{\mathrm{sgn}}}
    \ar@(ul,ur)^-{\, \mathbb{Z}_2^{\mathrm{refl}+\mathrm{HW}}}
  }
  \,.
\end{equation}
As before in \eqref{TheIIAZero}, the fixed subspace of the diagonal group action (now for $k =2$, as in \cite[(3.2)]{Gimon98})
$$
  \mathbb{Z}_2^{\mathrm{refl}+\mathrm{rot}+\mathrm{HW}}
  \; \xymatrix{\ar@{^{(}->}[r]^{\mathrm{diag}}&}
  \mathbb{Z}_2^{\mathrm{refl}+\mathrm{HW}}
  \times
  \mathbb{Z}_2^{\mathrm{rot}}
$$
in \eqref{TheMO5Zero} is actually empty,
since the action of $\mathbb{Z}_2^{\mathrm{rot}}$
and hence that of $\mathbb{Z}_2^{\mathrm{refl}+\mathrm{HW}+\mathrm{rot}}$
is free, whence the superscript $(-)^0$.
But, as before in \eqref{TheIIA}, under M/IIA KK-reduction
\eqref{KKReductionOnShiftCircle} we have an equivariant projection
map to the orbifold
\begin{equation}
  \label{TheO4Zero}
  {\color{darkblue}\tiny \bf \mathrm{O4}}^0
  \phantom{AAAAAA}
  \mathbb{R}^{4,1}
 \; \xymatrix{\ar@{^{(}->}[r]&}
  \mathbb{R}^{4,1}
  \times
  \xymatrix{
    \ast
    \ar@(ul,ur)|-{\, \mathbb{Z}^{\mathrm{rot}}_k\!\!\!}
  }
  \times
  \xymatrix{
    \mathbb{T}^{\mathbf{5}_{\mathrm{sgn}}}
    \ar@(ul,ur)^-{\, \mathbb{Z}_2^{\mathrm{refl}+\mathrm{HW}}}
  }
  \,,
\end{equation}
with non-empty fixed/singular subspace being the $\mathrm{O4}$-worldvolume --
which is thereby exhibited as being un-charged, as its lift to M-theory
in in fact non-singular.

\medskip
In the same manner, there is
the analogous $\mathbb{Z}_k^{\mathrm{rot}}$-resolution
of the $\mathrm{MK6}$-singularity \eqref{MK6Singularity}
\begin{equation}
  \label{TheMK6Zero}
  {\color{darkblue}\tiny \bf \mathrm{MK6}^0}
  \phantom{AAAAAA}
  \mathbb{R}^{6,1} \times \varnothing
  \; \xymatrix{\ar@{^{(}->}[r]&}
  \mathbb{R}^{5,1}
  \times
  \xymatrix{
    S^1
    \ar@(ul,ur)|-{\, \mathbb{Z}^{\mathrm{rot}}_k \!\!\!}
  }
  \times
  \xymatrix{
    \mathbb{T}^{ \mathbf{4}_{\mathbb{H}} }
    \ar@(ul,ur)|-{\;\; G^{\mathrm{ADE}} \!\! }
  }
  ,
\end{equation}
as well as of the $\mathrm{MO9}$-singularity \eqref{TheMO9}:

\vspace{-.2cm}

\begin{equation}
  \label{TheMO9Zero}
  {\color{darkblue}\tiny \bf \mathrm{MO9}^0}
  \phantom{AAAAAA}
  \mathbb{R}^{9,1} \times \varnothing
   \; \xymatrix{\ar@{^{(}->}[r]&}
  \mathbb{R}^{8,1}
  \times
  \xymatrix{
    S^1
    \ar@(ul,ur)|-{\, \mathbb{Z}^{\mathrm{rot}}_k\!\!\!}
  }
  \times
  \xymatrix{
    \mathbb{T}^{\mathbf{1}_{\mathrm{sgn}}}
    \ar@(ul,ur)|-{\, \mathbb{Z}_2^{\mathrm{HW}}\!\!\!\!\!}
  }
  .
\end{equation}
The reduction of the latter along $\mathrm{KK}_{S^1_{\mathrm{rot}}}$
\eqref{KKReductionOnShiftCircle} is

\begin{equation}
  \label{TheO8Zero}
  {\color{darkblue}\tiny  \bf \mathrm{O8}^0}
  \phantom{AAAAAA}
  \mathbb{R}^{8,1}
   \; \xymatrix{\ar@{^{(}->}[r]&}
  \mathbb{R}^{8,1}
  \times
  \xymatrix{
    \ast
    \ar@(ul,ur)|-{\, \mathbb{Z}^{\mathrm{rot}}_k}
  }
  \times
  \xymatrix{
    \mathbb{T}^{\mathbf{1}_{\mathrm{sgn}}}
    \ar@(ul,ur)|-{\, \mathbb{Z}_2^{\mathrm{HW}}\!\!\!\!\!}
  }
  .
\end{equation}

\noindent In summary, {\it the full singularity structure of
heterotic M-theory on ADE-orbifolds},
such as to admit
\begin{enumerate}[{\bf (i)}]
\vspace{-2mm}
\item  black M5-branes coinciding with MO5-planes
and
\vspace{-2mm}
\item  the $\mathrm{MO5}^0$-lift of $\mathrm{O4}^0$-planes
\end{enumerate}
\vspace{-2mm}
is as
shown in \hyperlink{Table7}{\it Table 7}.

\vspace{-4mm}
\begin{center}
\hypertarget{Table7}{}
\begin{tikzpicture}

  \draw (0,0) node
    {\footnotesize
      \begin{tabular}{c}
        M5-branes at
        \\
        MO9-planes intersecting
        \\
        ADE-singularities in
        \\
        M-theory on
        \\
        \raisebox{-20pt}{
        \fbox{
        $
         \xymatrix{
           S^{\mathbf{1}_{\mathrm{sgn}}}
           \ar@(ul,ur)|<<<<{\color{darkblue} \bf {}_{\mathbb{Z}_2^{\mathrm{HW}}} }
         }
         \!\!\times\!\!\!
         \xymatrix{
           \mathbb{T}^{ \mathbf{4}_{\mathbb{H}} }
           \ar@(ul,ur)|<<<<{\color{darkblue} \bf {}_{G^{\mathrm{ADE}}} }
         }
         \!\times\!
         \xymatrix{
           S^1
           \ar@(ul,ur)|<<<<{\color{darkblue} \bf {}_{\mathbb{Z}^{\mathrm{rot}}_k} }
         }
        $
        }
        }
        $\mathrlap{
          \mbox{
          \hspace{-.4cm}
          \raisebox{-10pt}{
          {\tiny
          \begin{tabular}{l}
            (\cite[3]{Sen97}, \\ \cite{FLO99}, \\ \cite{KSTY99} )
          \end{tabular}
        }
        }} }$
      \end{tabular}
    };

  \begin{scope}[shift={(0,-.1)}]

  \draw[->]
    (0,-1.5)
      to
      node
      {
        \small
        \colorbox{white}{\tiny  \color{darkblue}
          \begin{tabular}{c}  \bf
            reduction on
            \\
            $S^1_{\mathrm{ADE}}$
          \end{tabular}
        }
      }
    (0,-2.9);

  \draw[->]
    (1.5,-1.5)
      to
      node
      {
        \small
        \colorbox{white}{
          \tiny
          \color{darkblue} \bf
          \begin{tabular}{c}
            reduction on
            \\
            $S^1_{\mathrm{rot}}$
            \eqref{KKReductionOnShiftCircle}
          \end{tabular}
        }
        $\mathrlap{ \mbox{ \hspace{.5cm} \tiny (\cite{HoravaWitten96a}) } }$
      }
    (4,-2.9);

  \draw[->]
    (-1.5,-1.5)
      to
      node
      {
        \small
        \colorbox{white}{\tiny  \color{darkblue}
          \begin{tabular}{c}  \bf
            reduction on
            \\
            $S^1_{\mathrm{HW}}$
          \end{tabular}
        }
      }
    (-4,-2.9);

  \end{scope}

  \begin{scope}[shift={(0,.4)}]

  \draw (0,-5) node
    {\footnotesize
      \begin{tabular}{c}
        NS5-branes at
        \\
        O8-planes intersecting
        \\
        D6-branes in
        \\
        $\mathrm{I}'$-theory on
        \\
        \raisebox{-20pt}{
        \fbox{
        $
         \xymatrix{
           S^{\mathbf{1}_{\mathrm{sgn}}}
           \ar@(ul,ur)|<<<<{\color{darkblue} \bf {}_{\;\; \mathbb{Z}_2^{\mathrm{HW}}} }
         }
         \!\times \!
         \xymatrix{
           S^1
           \ar@(ul,ur)|<<<<{\color{darkblue} \bf  {}_{\mathbb{Z}_2^{\mathrm{rot}}} }
         }
        $
        }
        }
        $\mathrlap{\mbox{
          \raisebox{-10pt}{
          {\tiny (\cite{GKST01})}
          }
        }}$
      \end{tabular}
    };

 \draw (5,-5) node
    { \footnotesize
      \begin{tabular}{c}
        D4-branes at
        \\
        O8-planes intersecting
        \\
        ADE-singularities in
        \\
        $\mathrm{I}'$-theory on
        \\
        \raisebox{-20pt}{
        \fbox{
        $
         \xymatrix{
           S^{\mathbf{1}_{\mathrm{sgn}}}
           \ar@(ul,ur)|<<<<{\color{darkblue} \bf {}_{\;\;\mathbb{Z}_2^{\mathrm{HW}}} }
         }
         \!\!\times\!\!\!
         \xymatrix{
           \mathbb{T}^{ \mathbf{4}_{\mathbb{H}} }
           \ar@(ul,ur)|<<<<{\color{darkblue} \bf {}_{G^{\mathrm{ADE}}} }
         }
        $
        }
        }
        $\mathrlap{
          \mbox{
          \tiny
          \raisebox{-10pt}{
            \begin{tabular}{c}
              (\cite{BRG12},
              \\
              \cite[3.4.2]{HKKP15})
            \end{tabular}
          }
        }}$
      \end{tabular}
    };

 \draw (-5,-5) node
    {\footnotesize
      \begin{tabular}{c}
        NS5-branes
        \\
        at
        \\
        ADE-singularities in
        \\
        $\mathrm{HET}_{E}$-theory on
        \\
        \raisebox{-20pt}{
        \fbox{
        $
         \xymatrix{
           S^1
           \ar@(ul,ur)|<<<<{\color{darkblue} \bf {}_{\mathbb{Z}_2^{\mathrm{rot}}} }
         }
         \!\!\times\!\!\!
         \xymatrix{
           \mathbb{T}^{ \mathbf{4}_{\mathbb{H}} }
           \ar@(ul,ur)|<<<<{\color{darkblue} \bf {}_{G^{\mathrm{ADE}}} }
         }
        $
        }
        }
        $
        \mathrlap{
        \raisebox{-10pt}{
        \mbox{\tiny (\cite{Witten99})}}}
        $
      \end{tabular}
    };

    \end{scope}

\end{tikzpicture}
\end{center}
\vspace{-4mm}
\noindent {\bf  \footnotesize Table 7 --
Singularity structure of heterotic M-theory on ADE-orbifolds and its string theory duals}
{\footnotesize given by combining
the $\tfrac{1}{2}\mathrm{M5}$-structure of \hyperlink{FigureS}{\it Figure S}
with $\mathrm{IIA}^0$-structure \eqref{TheIIAZero},
hence admitting also $\mathrm{MO5}^0$-structure \eqref{TheMO5Zero}.}

\vspace{.4cm}

\noindent The following \hyperlink{FigureT}{\it Figure T} shows the corresponding subgroup lattice
with its
associated fixed/singular spaces:

\begin{center}
\hypertarget{FigureT}{}
\begin{tikzpicture}[scale=0.75]

  \begin{scope}

  \draw (0,5.4)
    node
    { \small
      \color{darkblue} \bf
      \begin{tabular}{c}
        $\mathrm{M}_{\mathrm{HET}}/{\mathrm{ADE}}$-orbifold subgroups
        \\
        $H \subset$
        \\
        $
          G^{\mathrm{ADE}}
            \times
          \mathbb{Z}_k^{\mathrm{rot}}
            \times
          \mathbb{Z}_2^{\mathrm{HW}}
        $
      \end{tabular}
    };
  \draw (0,4) node {$\overbrace{\phantom{--------------------}}$};

\begin{scope}[shift={(0,1)}]

  \draw (0+30+180:2.2) to (0+30+180+60:3.5);
  \draw (0+30+180+120:2.2) to (0+30+180+60+120:3.5);
  \draw (0+30+180+240:2.2) to (0+30+180+60+240:3.5);

  \draw (0+30+180+60:3.5) to (0+30+180+120:2.2);
  \draw (0+30+180+60+120:3.5) to (0+30+180+120+120:2.2);
  \draw (0+30+180+60+240:3.5) to (0+30+180+120+240:2.2);

  \draw (0,0) to (0+30+180+60:3.5);
  \draw (0,0) to (0+30+180+60+120:3.5);
  \draw (0,0) to (0+30+180+60+240:3.5);

  \draw (0+30+180:2.2)
    node
    {\colorbox{white}{
    $
      \mathbb{Z}_2^{\mathrm{refl}+\mathrm{rot}}
    $}};
  \draw (60+30+180:3.5)
    node
    {\colorbox{white}{$
      \underset{
        \mbox{
          \tiny
          \color{darkblue} \bf
          \eqref{KKReductionOnShiftCircle}
        }
      }{
        \mathbb{Z}_k^{\mathrm{rot}}
      }
    $}};
  \draw (120+30+180:2.2)
    node
    {\colorbox{white}{$
      \mathbb{Z}_2^{\mathrm{rot}+\mathrm{HW}}
    $}};
  \draw
    (180+30+180:3.5)
    node
    {\colorbox{white}{$
      \mathbb{Z}_2^{\mathrm{HW}}
    $}};
  \draw (240+30+180:2.2)
    node
    {\colorbox{white}{$
      \underset{
        \mbox{
          \tiny
          \color{darkblue} \bf
          \eqref{Z2ReflHW}
        }
      }{
        \mathbb{Z}_2^{\mathrm{refl}+\mathrm{HW}}
      }
    $}};
  \draw (300+30+180:3.5)
    node
    {\colorbox{white}{
      $
        \underset{
         \mbox{
           \tiny
           \color{darkblue} \bf
           \eqref{PointReflectionSubgroup}
         }
        }{
          \mathbb{Z}_2^{\mathrm{refl}}
        }
      $}};

  \draw (0,0)
    node
      {\colorbox{white}{$
        \mathbb{Z}_2^{\mathrm{refl}+\mathrm{rot}+\mathrm{HW}}
      $}};

  \end{scope}

  \end{scope}

  \begin{scope}[shift={(9.7,0)}]

  \draw (0,5.5)
    node
    {\small
      \color{darkblue} \bf
      \begin{tabular}{c}
        Fixed/singular subspaces {\footnotesize\eqref{FixedLoci}}
        \\
        $
         \Big(
         \xymatrix{
           \mathbb{T}^{ \mathbf{4}_{\mathbb{H}} }
           \ar@(ul,ur)|<<<<{\color{darkblue} \bf {}_{G^{\mathrm{ADE}}} }
         }
         \!\times\!
         \xymatrix{
           S^1
           \ar@(ul,ur)|<<<<{\color{darkblue} \bf {}_{\mathbb{Z}^{\mathrm{rot}}_k} }
         }
         \!\!
         \times
         \!\!\!
         \xymatrix{
           S^{\mathbf{1}_{\mathrm{sgn}}}
           \ar@(ul,ur)|<<<<{\color{darkblue} \bf {}_{\;\; \mathbb{Z}_2^{\mathrm{HW}}} }
         }
         \Big)^H
        $
      \end{tabular}
    };
  \draw (0,4) node {$\overbrace{\phantom{--------------------}}$};

\begin{scope}[shift={(0,1)}]

  \draw (0+30+180:2.2) to (0+30+180+60:3.5);
  \draw (0+30+180+120:2.2) to (0+30+180+60+120:3.5);
  \draw (0+30+180+240:2.2) to (0+30+180+60+240:3.5);

  \draw (0+30+180+60:3.5) to (0+30+180+120:2.2);
  \draw (0+30+180+60+120:3.5) to (0+30+180+120+120:2.2);
  \draw (0+30+180+60+240:3.5) to (0+30+180+120+240:2.2);

  \draw (0,0) to (0+30+180+60:3.5);
  \draw (0,0) to (0+30+180+60+120:3.5);
  \draw (0,0) to (0+30+180+60+240:3.5);

  \draw (0+30+180:2.2)
    node
    {\colorbox{white}{$
      \underset{
        \mbox{\bf
          \tiny
          \color{darkblue}
          \eqref{TheMK6Zero}
        }
      }{
        \mathrm{MK6}^0
      }
    $}};
  \draw (60+30+180:3.5)
    node
    {\colorbox{white}{$
      \underset{
        \mbox{\bf
          \tiny
          \color{darkblue}
          \eqref{TheIIAZero}
        }
      }{
        \mathrm{IIA}^0
      }
    $}};
  \draw (120+30+180:2.2)
    node
    {\colorbox{white}{$
      \underset{
        \mbox{\bf
          \tiny
          \color{darkblue}
          \eqref{TheMO9Zero}
        }
      }{
        \mathrm{MO9}^0
      }
    $}};
  \draw (180+30+180:3.5)
    node
    {\colorbox{white}{$
      \underset{\bf
        \tiny
        \color{darkblue}
        \eqref{TheMO9}
      }{
        \mathrm{MO9}
      }
    $}};
  \draw (240+30+180:2.2)
    node
    {\colorbox{white}{$
      \underset{
        \mbox{\bf
          \tiny
          \color{darkblue}
          \eqref{TheMO5}
        }
      }{
        \mathrm{MO5}
      }
    $}};
  \draw (300+30+180:3.5)
    node
    {\colorbox{white}{$
      \underset{
        \mbox{\bf
          \tiny
          \color{darkblue}
          \eqref{MK6Singularity}
        }
      }{
        \mathrm{MK6}
      }
    $}};

  \draw (0,0)
    node
    {\colorbox{white}{$
      \underset{
        \mbox{\bf
          \tiny
          \color{darkblue}
          \eqref{TheMO5Zero}
        }
      }{
        \mathrm{MO5}^0
      }
    $}};

\end{scope}
  \end{scope}

\end{tikzpicture}
\end{center}
\vspace{-5mm}
\noindent {\footnotesize \bf Figure T -- Subgroup lattice
and fixed/singular subspaces in
the singularity structure 
for heterotic M-theory on ADE-orbifolds
from \hyperlink{Table7}{\it Table 7}}.
{\footnotesize  On the left, groups associated to the middle of a sub-simplex
are diagonal subgroups inside the direct product of subgroups associated
to the vertices, as indicated by the superscripts.
On the right, all fixed loci with superscript $(-)^0$ are actually empty,
but appear as superficially non-empty (un-charged) singularities after M/IIA KK-reduction \eqref{KKReductionOnShiftCircle}, e.g. $\mathrm{O4}^0$ \eqref{TheO4Zero}, $\mathrm{O8}^0$ \eqref{TheO8Zero}, as
on the right of \hyperlink{FigureOP}{\it Figure OP}.  The numbered subscripts $(xx)$
indicate the corresponding expression in the text.}

\subsection{Equivariant Cohomotopy charge of M5 at $\mathrm{MO5}_{\mathrm{ADE}}$}
\label{EquivariantCohomotopyChargeOfM5AtMO5}

Applying the general mathematical results of \cref{EquivariantCohomotopyAndTadpoleCancellation}
to the $\mathrm{M}_{\mathrm{HET}}/\mathrm{ADE}$-singularities
from \cref{HeteroticMTheoryOnADEOrbifolds},
we finally show here (see \hyperlink{FigureV}{\it Figure V} )
that \hyperlink{HypothesisH}{\it Hypothesis H}
formalizes and validates the following widely accepted but
informal Folklore \ref{AnomalyCancellationOnMTheoreticOrientifolds},
concerning the nature of M-theory:

\vspace{.0cm}

\noindent
\begin{minipage}[l]{10.8cm}
\hypertarget{FigureU}{}
\begin{folklore}[M5/MO5 anomaly cancellation {\cite{DasguptaMukhi95}\cite[3.3]{Witten95b}
\cite[2.1]{Hori98}}]
 \label{AnomalyCancellationOnMTheoreticOrientifolds}
For M-theory on the toroidal orientifold
$\mathbb{R}^{5,1} \times \mathbb{T}^{ \mathbf{5}_{\mathrm{sgn}} }
\!\sslash\! \mathbb{Z}_2$
(\hyperlink{Table5}{\it Table 5}) with MO5-singularities \eqref{TheMO5},
consistency requires the situation
shown in \hyperlink{Table2}{\it Table 2}:
\begin{enumerate} [\bf (i)]
\vspace{-3mm}
\item a charge of $q_{{}_{\mathrm{MO5}}}/q_{{}_{\mathrm{M5}}} = -1/2$ is carried by each of the fixed/singular MO5-planes \eqref{TheMO5};
\vspace{-3mm}
\item the M5-brane charge is integral in natural units,
hence on the covering
$\mathbb{Z}_2$-space $\mathbb{T}^{\mathbf{5}_{\mathrm{sgn}}}$ the M5-branes appear in $\mathbb{Z}_2$-mirror pairs
around the MO5-planes, as in \hyperlink{FigureL}{\it Figure L}
and \hyperlink{FigureN}{\it Figure N};

\vspace{-3mm}
\item  the total charge of the $N_{{}_{\rm M5}}$ M5-branes
has to cancel that of the 32 O-planes \eqref{RepresentationTorusOfSignRep},
 $N_{{}_{\rm M5}}  q_{{}_{\rm M5}} + 32 q_{{}_{\rm MO5}} = 0$,
 as indicated in \hyperlink{FigureA}{\it Figure A}.
\end{enumerate}

\vspace{-4mm}
\end{folklore}
\noindent Via the similarly widely accepted  Folklore \ref{ReductionOfMO5ToO4},
the statement of Folklore \ref{AnomalyCancellationOnMTheoreticOrientifolds}
implies tadpole anomaly cancellation in string theory.
Notice that this is not so much a claim than
part of the defining criterion for M-theory:

\vspace{-2mm}
\begin{folklore}[Double dimensional reduction of M5/MO5 to D4/O4 {\cite[3]{Hori98}\cite[III.A]{Gimon98}\cite[3.1.1]{HananyKol00}}]
\label{ReductionOfMO5ToO4}
Under  M/IIA duality,
the situation of Folklore \ref{AnomalyCancellationOnMTheoreticOrientifolds}
becomes  the string-theoretic
tadpole cancellation condition from
\hyperlink{Table1}{\it Table 1}
for D4-branes and $\mathrm{O}4^-$-planes.
\end{folklore}

\vspace{-5mm}
\begin{folklore}[T-duality relating $\mathrm{O}$-planes, e.g. {\cite[p.317-318]{BLT13}}]
\label{TDualityRelatesOpPlanes}
By iterative T-duality, the situation of
Folklore \ref{ReductionOfMO5ToO4} implies
general tadpole cancellation for $\mathrm{D}p$-branes
and $\mathrm{O}p^-$-planes (\hyperlink{Table3}{\it Table 3}).
\end{folklore}
\end{minipage}
$\phantom{a}$
\fbox{
$
  \!\!\!
  \raisebox{156pt}{
  \xymatrix@R=14pt{
    \fbox{
      \hyperlink{HypothesisH}{\it Hypothesis H}
    }
    \ar@{=>}[dd]_-{
      \mbox{\bf
        \tiny
        \color{darkblue}
        \begin{tabular}{c}
          Equivariant Cohomotopy
          \\
          of $\mathrm{M}_{\mathrm{HET}}/\mathrm{ADE}$-orbifolds
        \end{tabular}
      }
    }^-{
      \mbox{\color{darkblue} \bf
        \tiny
        \begin{tabular}{l}
          Rigorous: Cor. \ref{EquivariantCohomotopyOfSemiComplementSpacetime},
          \ref{GlobalM5MO5CancellationImplied}
        \end{tabular}
      }
    }
    \\
    \\
    \fbox{
      \begin{tabular}{c}
        M5/MO5 anomaly cancellation
        \\
        {
          \tiny
          (Folklore \ref{AnomalyCancellationOnMTheoreticOrientifolds})
        }
      \end{tabular}
    }
    \ar@{<=>}[dd]_-{
      \mbox{
        \tiny
        \begin{tabular}{c}
          M/IIA duality
        \end{tabular}
      }
    }^-{
      \mbox{
        \tiny
        \begin{tabular}{c}
          Folklore \ref{ReductionOfMO5ToO4}
        \end{tabular}
      }
    }
    \\
    \\
    \fbox{
      \begin{tabular}{c}
        D4/O4 tadpole cancellation
      \end{tabular}
    }
    \ar@{<=>}[dd]_-{
      \mbox{
        \tiny
        \begin{tabular}{c}
          T-duality
        \end{tabular}
      }
    }^-{
      \mbox{
        \tiny
        \begin{tabular}{l}
          Folklore \ref{TDualityRelatesOpPlanes}
        \end{tabular}
      }
    }
    \\
    \\
    \fbox{
      \begin{tabular}{c}
        D$p$/O$p$-tadpole cancellation
      \end{tabular}
    }
    \\
    \mbox{
    \footnotesize
    \begin{minipage}[l]{5.7cm}
      \noindent  {\bf Figure U -- Structure of the argument.}
      We demonstrate
      that \hyperlink{HypothesisH}{\it Hypothesis H}
      on C-field charge quantization in Cohomotopy,
      applied to heterotic M-theory on toroidal ADE-oribolds,
      implies M5/MO5-anomaly cancellation in M-theory.
      This directly subsumes and implies the statement of
      tadpole cancellation for D$4$/O$4$ branes in string theory.
    \end{minipage}
    }
  }
  }
  \!\!\!
$
}

\vspace{-2mm}
\begin{center}
\hypertarget{FigureV}{}
\begin{tikzpicture}[scale=0.8]

  \begin{scope}[shift={(0,-.7)}]

  \draw (1.5,6.7) node {$\overbrace{\phantom{------------------------}}$};
  \draw (11,6.7) node {$\overbrace{\phantom{---------------}}$};

  \draw (1.5,7.6)
    node
    {\tiny \color{darkblue} \bf
      \begin{tabular}{c}
        semi-complement \eqref{SemiComplement}
        \\
        of $\mathrm{MO5} \subset \tfrac{1}{2}\mathrm{M5}$-singularities
        (\hyperlink{FigureS}{\it Figure S})
        \\
        in heterotic M-theory on ADE-orbifolds
        (\hyperlink{Table7}{\it Table 7})
      \end{tabular}
    };
  \draw (11,7.6)
    node
    {
      \tiny
      \color{darkblue} \bf
      \begin{tabular}{c}
        C-field charge quantization
        \\
        in ADE-equivariant Cohomotopy \eqref{EquivariantCohomotopySet}
      \end{tabular}
    };

  \draw[->]
    (4,8.5)
    to
    node[above]
      {
        \tiny
        \color{darkblue} \bf
        \begin{tabular}{c}
          M-theory C-field
          \\
          charge-quantized by \hyperlink{Hyothesis}{\it Hypothesis H}
          \\
          as a cocycle in equivariant Cohomotopy
        \end{tabular}
      }
    (9,8.4);

  \end{scope}

  \begin{scope}[shift={(0, 0)}]
  \clip (0,-1) rectangle (2.8,5.7);

  \draw[dotted, thick] (0,-3) to (0,6);

  \draw[step=3, dotted, thick] (-3,-3) grid (6,6);

  \draw[draw=black, fill=white] (0,0) circle (.1);
  \draw[draw=black, fill=white] (0,3) circle (.1);

  \draw[draw=black, fill=white] (3,3) circle (.1);
  \draw[draw=black, fill=white] (3,0) circle (.1);

  \draw[draw=black, fill=white] (3,6) circle (.1);
  \draw[draw=black, fill=white] (0,6) circle (.1);

  \draw[draw=black, fill=white] (3,-3) circle (.1);
  \draw[draw=black, fill=white] (0,-3) circle (.1);

  \draw[dashed] (0,3+.7) to (3,3+.7);
  \draw[dashed] (0,3-.7) to (3,3-.7);

  \draw[dashed] (0,+.7) to (3,+.7);
  \draw[dashed] (0,-.7) to (3,-.7);

  \end{scope}

  \draw (0,-.9) node {\tiny $x_1 = 0$};
  \draw (3,-.9) node {\tiny $x_1 > 0$};
  \draw (-3.7,0) node {\tiny $x_2 = 0$};
  \draw (-3.7,3) node {\tiny $x_2 = \tfrac{1}{2}$};
  \node at (11,2) {\colorbox{white}{$\phantom{a}$}};

  \draw (11,2) circle (2);
  \node (infinity) at (11+2,2)
    {\colorbox{white}{$\infty$}};
  \node (zero) at (11-2,2) {$-\,\tiny \mathrlap{0} $};
  \begin{scope}[shift={(9,2)}]
   \fill[darkblue] (2,0) ++(40+180:2) node (minusepsilon)
     {\begin{turn}{-45} $)$  \end{turn}};
   \fill[darkblue] (2,0) ++(180-40:2) node (epsilon)
     {\begin{turn}{45} $)$ \end{turn}};
   \fill[darkblue] (2.3,0.25) ++(40+180:2)
     node (label+epsilon)
     { \tiny $-\epsilon$ };
   \fill[darkblue] (2.3,-0.25) ++(-40-180:2)
     node (label-epsilon)
     { \tiny $+\epsilon$ };
   \draw[<->, dashed, gray]
     (label+epsilon)
     to
     node {\tiny $\mathbb{Z}_2$}
     (label-epsilon);
  \end{scope}
  \node (torus) at (1.5,7.5)
    {\raisebox{42pt}{$
      \big(
      \mathbb{R}^{\mathbf{1}_{\mathrm{sgn}}}
      /
      \mathbb{Z}_2^{\mathrm{HW}}
      \big)
      \times
      \xymatrix{
        \mathbb{T}^{ \mathbf{4}_{\mathrm{sgn}}  }
        \ar@(ul,ur)|{\,
          \mathbb{Z}^{\mathrm{refl}}_2 \!\!\!\!\!
        }
      }
    $}};
  \node (sphere) at (11,7.2)
    {\raisebox{42pt}{$
      \xymatrix{
        S^{\mathbf{4}_{\mathbb{H}}}
        \ar@(ul,ur)|-{\, \mathbb{Z}_2 }
      }
    $}
    };

  \draw[<->, dashed]
    (1.5,3+1.5)
    to
    node[very near end]
      {
        \tiny
        \color{darkblue} \bf
        \begin{tabular}{c}
        \;\;  residual
          \\
        \;\;  $\mathbb{Z}_2^{\mathrm{refl}}$-action
        \end{tabular}
      }
    (1.5,3-1.5);

  \draw[draw=darkblue, fill=darkblue]
    (0.01,3+.35-.05)
    rectangle
    (2.8,3+.35+.05);

  \draw[draw=darkblue, fill=darkblue]
    (0.01,3-.35-.05)
    rectangle
    (2.8,3-.35+.05);

  \draw[draw=lightgray, fill=lightgray]
    (0.09,-.05)
    rectangle
    (2.8,+.05);

  \draw[thin]
    (0.09,-.05)
    to
    (2.8,-.05);
  \draw[thin]
    (0.09,+.05)
    to
    (2.8,+.05);

  \draw[draw=lightgray, fill=lightgray]
    (0.09,3-.05)
    rectangle
    (2.8,3+.05);
  \draw[thin]
    (0.09,3-.05)
    to
    (2.8,3-.05);
  \draw[thin]
    (0.09,3+.05)
    to
    (2.8,3+.05);

  \draw[|->, olive]
    (2.1,3+.35)
    to[bend right=8]
    (zero);

  \draw[|->, olive]
    (2.1,3-.35)
    to[bend right=8]
    (zero);

  \draw[|->, olive]
    (2.1,3)
    to[bend right=8]
    node
      {
        \colorbox{white}{
          \tiny
          \color{darkblue} \bf
          codimension 1 submanifolds
        }
      }
    (zero);

  \draw[|->, olive]
    (2.1,0)
    to[bend right=8]
    node
      {
        \colorbox{white}{
          \tiny
          \color{darkblue} \bf
          codimension 1 submanifold
        }
      }
    (zero);

  \draw[|->, olive]
    (1.2,5)
    to[bend left=26]
    node {  }
    (infinity);

  \draw[|->, olive]
    (1.8,3+.7)
    to[bend left=26]
    node
      {
        \colorbox{white}
        {
          \tiny
          \color{darkblue} \bf
          cocycle vanishes far away from fixed lines
        }
      }
    (infinity);

  \begin{scope}[shift={(0,-3)}]
  \node (MO9) at (-.7,4.7) {\tiny \color{darkblue} \bf $\mathrm{MO9}$};
  \draw[->, gray] (MO9) to (-.1,4.3);

  \begin{scope}[shift={(.1,3.35)}]

  \node (MK6) at (-.7,3.9) {\tiny \color{darkblue} \bf $\mathrm{MK6}$};
  \draw[->, gray] (MK6) to (.3,3.1);

  \end{scope}

  \node (MO5) at (-1,3) {\tiny \color{darkblue} \bf $\mathrm{MO5}$};
  \draw[->, gray] (MO5) to (-.1,3);

  \node (M5) at (-.7,6.6) {\tiny \color{darkblue} \bf $\mathllap{\tfrac{1}{2}}\mathrm{M5}$};
  \draw[->, gray] (M5) to ++(.6,-.2);

  \node (halfM5) at (-1,6)
    {
      \tiny
      \color{darkblue} \bf
      $\mathllap{-\tfrac{1}{2}\mathrm{M5} = \;} \mathrm{MO5}$};
  \draw[->, gray] (halfM5) to ++(.9,0);

  \node (mirrorM5)
    at (-.7,5.4)
    {
      \tiny
      \color{darkblue} \bf
      $\mathllap{\mbox{mirror}\;\tfrac{1}{2}}\mathrm{M5}$
    };
  \draw[->, gray] (mirrorM5) to ++(.6,+.2);

  \end{scope}

  \begin{scope}[shift={(0,+.35)}]
  \clip (0,3+.2) rectangle (1,3-.2);
  \draw[draw=green, fill=green]
    (0,3) circle (.1);
  \end{scope}

  \begin{scope}[shift={(0,-.35)}]
  \clip (0,3+.2) rectangle (1,3-.2);
  \draw[draw=green, fill=green]
    (0,3) circle (.1);
  \end{scope}

\end{tikzpicture}
\end{center}

\vspace{-.4cm}

\noindent {\bf \footnotesize Figure V --
Equivariant Cohomotopy of ADE-orbifolds
in heterotic M-theory}
{\footnotesize
with singularity structure as in \hyperlink{FigureS}{\it Figure S}.
The resulting charge classification (Cor. \ref{EquivariantCohomotopyOfSemiComplementSpacetime}) implies, via the
unstable PT isomorphism (\cref{PTTheorem}),
the
$\tfrac{1}{2}\mathrm{M5} = \mathrm{MO9} \cap \mathrm{MK6}$-brane configurations \eqref{TheHalfM5} similarly shown in
\cite[Fig. 1]{FLO99}\cite[p. 7]{KSTY99}\cite[Fig. 1]{FLO00a}\cite[Fig. 2]{FLO00b}\cite[Fig. 1]{FLO00c}\cite[p. 4, 68, 71]{GKST01}.
This is as in \hyperlink{FigureL}{\it Figure L}
but with points (M5s) extended to half-line (MK6s), see Remark \ref{TheRoleOfMK6EndingOnM5} and \hyperlink{Table8}{\it Table 8}.
}

\medskip
\noindent {\bf C-Field flux quantization at pure MO5-Singularities.}
To put the discussion below in perspective,
it is instructive to first recall the
success and the shortcoming of the existing argument \cite[2]{Hori98}
for M5/MO5-brane charge quantization around a \emph{pure} MO5-singularity \eqref{TheMO5} (see the left column of\hyperlink{Table8}{\it Table 8}):
Following the classical argument of \cite{Dirac31},
we consider removing the locus of the would-be M5-brane
from spacetime and then computing the appropriate cohomology
of the remaining complement. For the
\emph{pure} MO5-singularity \eqref{TheMO5} the
complement spacetime is, up to homotopy equivalence,
the 4-dimensional real projective space:

\vspace{-.2cm}

\begin{equation}
  \label{M5ThreadedThroughRP4}
  \underset{
    \mathclap{
    \mbox{\bf
      \tiny
      \color{darkblue}
      \begin{tabular}{c}
        complement spacetime
        \\
        around pure MO5-singularity
      \end{tabular}
    }
    }
  }{
    X^{11}_{{}_{\mathrm{MO5}}}
  }
  \;\;\;\;\;\;\;\;=\;
  \underset{
    \mathclap{
    \mbox{\bf
      \tiny
      \color{darkblue}
      \begin{tabular}{c}
        full Euclidean orientifold \eqref{EuclideanGSpace}
        \\
        with pure MO5-singularity \eqref{TheMO5}
      \end{tabular}
    }
    }
  }{
    \big(
      \mathbb{R}^{5,1}
      \times
      \xymatrix{
        \mathbb{R}^{\mathbf{5}_{\mathrm{sgn}}}
      }
      \overset{
        \mathclap{
         \mbox{\bf
            \tiny
            \color{darkblue}
            \begin{tabular}{c}
              pure \eqref{Z2ReflHW}
              \\
              orbifold quotient
            \end{tabular}
          }
         }
      }{
        \!\!
        \!\sslash\! 
        \mathbb{Z}_2^{\mathrm{refl}+\mathrm{HW}}
      }
    \big)
  }
  \setminus
  \underset{
    \mathclap{
    \mbox{\bf
      \tiny
      \color{darkblue}
      \begin{tabular}{c}
        $\mathbb{Z}_2^{\mathrm{refl}+\mathrm{HM}}$-fixed
        \\
        subspace \eqref{FixedLoci}
      \end{tabular}
    }
    }
  }{
    \{ \mathbb{R}^{5,1} \times \{0\} \}
  }
  \;\simeq_{{}_{\mathrm{homotopy}}}\;
  \underset{
    \mathclap{
    \mbox{\bf
      \tiny
      \color{darkblue}
      \begin{tabular}{c}
        unit 4-sphere
        \\
        around MO5
      \end{tabular}
    }
    }
  }{
    S\big( \mathbb{R}^{\mathbf{5}_{\mathrm{sgn}}} \big)
  }
  /
  \underset{
    \mathclap{
    \mbox{\bf
      \tiny
      \color{darkblue}
      \begin{tabular}{c}
        pure \eqref{Z2ReflHW}
        \\
        MO5-quotient
      \end{tabular}
    }
    }
  }{
    \mathbb{Z}_2^{\mathrm{refl}+\mathrm{HW}}
  }
  \;\simeq\;
  \underset{
    \mathclap{
    \mbox{\bf
      \tiny
      \color{darkblue}
      \begin{tabular}{c}
        real projective
        \\
        4-space
      \end{tabular}
    }
    }
  }{
    \mathbb{R}P^4
  }
  \,.
\end{equation}
Since this ambient spacetime \eqref{M5ThreadedThroughRP4}
is a smooth but curved (i.e. non-parallelizable) manifold, the flavor of Cohomotopy theory
that measures its M-brane charge, according to \hyperlink{HypothesisH}{\it Hypothesis H}, is,
according to \hyperlink{Table4}{\it Table 4}, the $J$-twisted Cohomotopy theory of \cite[3]{FSS19b}.
This implies, by \cite[Prop. 4.12]{FSS19b}, that rationalized brane charge (bottom of \eqref{DifferentialEquivariantCohomotopyPullback})
is measured by the integral of a differential 4-form $G_4 \in \Omega^4\big( X^{11}\big) $
(the C-field 4-flux density) which satisfies the half-integral shifted flux quantization condition
\begin{equation}
  \label{HalfIntegralFluxQuantization}
  [G_4] + \big[ \tfrac{1}{4}p_1 \big]
  \;\in\;
  H^4\big( X^{11}, \mathbb{Z}\big)
    \to
  H^4\big( X^{11}, \mathbb{R}\big)
\end{equation}
as is expected from the M-theory folklore (recalled in \cite[2.2]{FSS19b}).
Applying this to the complement $X^{11}_{\mathrm{MO5}}$ \eqref{M5ThreadedThroughRP4}
around a pure MO5-plane implies, as pointed out in \cite[2.1]{Hori98}, that there
must be an \emph{odd integer} of  brane charge in the pure MO5-spacetime

\vspace{-3mm}
\begin{equation}
  \label{TheOddChargeAroundAPureMO5}
  \underset{
    \mathclap{
    \mbox{\bf
      \tiny
      \color{darkblue}
      \begin{tabular}{c}
        {\phantom{A}}
        \\
        $J$-twisted Cohomotopy
        (\cite[3.1]{FSS19b})
        \\
        of pure MO5-complement \eqref{M5ThreadedThroughRP4}
      \end{tabular}
    }
    }
  }{
  \pi^{{}^{
    \overset{
      \mathclap{
      \mbox{\bf
        \tiny
        \color{darkblue}
        twist
      }
      }
    }{
      T \mathbb{R}P^4
    }
  }}
  \!\!
  \big(
    X^{11}_{{{}_\mathrm{MO5}}}
  \big)_{\mathbb{R}}
  }
  \;\;\;\;\;\;=\;
  \underset{
    \raisebox{13pt}{
    \mbox{\bf
      \tiny
      \color{darkblue}
      \begin{tabular}{c}
        due to
        half-integral $G_4$-flux quantization \eqref{HalfIntegralFluxQuantization}
        \\
        implied by twisted Cohomotopy \cite[Prop. 4.12]{FSS19b}
      \end{tabular}
    }
    }
  }{
  \Bigg\{
    2
    \int_{\mathbb{R}P^4} G_4
    \;=\;
    \overset{
      \mathclap{
      \mbox{\bf
        \tiny
        \color{darkblue}
        \begin{tabular}{c}
          odd integer
          \\
          net charge
          \\
          {\phantom{-}}
        \end{tabular}
      }
      }
    }{
      1 - 2N
    }
    \;\vert\;
    N \in \mathbb{N}
  \Bigg\}
  }
  \,.
\end{equation}

\noindent {\bf The need to resolve further microscopic details.}
If one could identify in \eqref{TheOddChargeAroundAPureMO5}
 the offset of $1 \,\mathrm{mod}\, 2$
in \eqref{TheOddChargeAroundAPureMO5} with the charge carried by the
pure MO5-plane \eqref{TheMO5}, and the remaining even charge $2 N$
with that of $N$ M5-branes in its vicinity
\begin{equation}
  \label{MissingBraneChargeIdentification}
  1 - N \cdot 2
  \;\overset{?}{=}\;
  Q_{\mathrm{MO5}}
  -
  N_{\mathrm{brane}} \cdot Q_{\mathrm{M5}}
\end{equation}
this would be the
local/twisted M5/MO5-anomaly cancellation condition of
\hyperlink{Table2}{\it Table 2}.
Without such further information, the charge quantization
\eqref{TheOddChargeAroundAPureMO5} around \emph{pure} MO5-planes \eqref{TheMO5}
is only \emph{consistent with} the local/twisted M5/MO5-anomaly cancellation from \hyperlink{Table2}{\it Table 2},
as noticed in \cite[bottom of p. 5]{Hori98}.

\medskip

But with the results of \cref{EquivariantCohomotopyAndTadpoleCancellation}
and in view of \cref{HeteroticMTheoryOnADEOrbifolds}, we may now complete this old argument (see the right column of \hyperlink{Table8}{\it Table 8}):

\medskip
\noindent {\bf Equivariant Cohomotopy implies local/twisted M5/MO5-anomaly cancellation
at $\tfrac{1}{2}\mathrm{M5}$-singularities.} We know from  \cref{LocalTadpoleCancellation} that the
identification \eqref{MissingBraneChargeIdentification}
missing from the result \eqref{TheOddChargeAroundAPureMO5}
for twisted Cohomotopy on smooth but curved spacestimes
\emph{is} implied by the result of Prop. \ref{TheoremLocalTadpoleCancellation}
for equivariant Cohomotopy of singular but flat spacetimes.
Moreover, we have argued in \cref{HeteroticMTheoryOnADEOrbifolds}
that having black M5-branes actually coinciding with MO5-planes
requires/implies that the pure MO5-planes are but the diagonally
fixed  sub-loci (shown in \hyperlink{Table6}{\it Table 6}) inside
the richer $\tfrac{1}{2}\mathrm{M5} = \mathrm{MK6} \cap \mathrm{MO9}$-singularities
\eqref{TheHalfM5} of heterotic M-theory on ADE-orbifolds (\hyperlink{FigureS}{\it Figure S}).
Hence for a rigorous M5/MO5-anomaly cancellation
result not just consistent with (as in \eqref{TheOddChargeAroundAPureMO5}), but actually
\emph{implying} Folkore \ref{AnomalyCancellationOnMTheoreticOrientifolds},
we need to compute the M-brane charge at MO5-singularities
inside $\tfrac{1}{2}\mathrm{M5}$-singularities \eqref{TheHalfM5}.
Concretely, this means with \hyperlink{HypothesisH}{\it Hypothesis H}
that M5/MO5-charge at a single MO5-singularity
is measured by the equivariant Cohomotopy
of the following $\tfrac{1}{2}\mathrm{M5}$-refinement
of the naive MO5-complement spacetime \eqref{M5ThreadedThroughRP4}:
\vspace{0mm}
\begin{eqnarray}
  \label{SemiComplement}
  \underset{
    \mathclap{
    \mbox{\bf
      \tiny
      \color{darkblue}
      \begin{tabular}{c}
        semi-complement spacetime
        \\
        around $\mathrm{MO5}$
        in $\tfrac{1}{2}\mathrm{M5}$-singularity
      \end{tabular}
    }
    }
  }{
    X^{11}_{{}_{\frac{1}{2}\mathrm{M5}}}
  }
  &\;\coloneqq\; &
  \underset{
    \mathclap{
    \mbox{\bf
      \tiny
      \color{darkblue}
      \begin{tabular}{c}
        full Euclidean orientifold \eqref{EuclideanGSpace}
        \\
        with $\tfrac{1}{2}\mathrm{M5}$-singularity \eqref{TheHalfM5}
      \end{tabular}
    }
    }
  }{
  \big(
    \mathbb{R}^{5,1}
    \!\!\times\!\!
    \xymatrix{
      \mathbb{R}^{\mathbf{1}_{\mathrm{sgn}}}
    }
    \!\!\times\!\!
    \xymatrix{
      \mathbb{R}^{\mathbf{4}_{\mathbb{H}}}
    }
    \!\!\!\sslash\!
    \mathbb{Z}_2^{\mathrm{HW}}
    \!\!\times\!
    \mathbb{Z}_2^{\mathrm{refl}}
  \big)
  }
  \setminus
  \big(
  \underset{
    \mathclap{
    \mbox{\bf
      \tiny
      \color{darkblue}
      \begin{tabular}{c}
        $\mathbb{Z}_2^{\mathrm{HW}}$-fixed subspace \eqref{FixedLoci}
        \\
        with residual
        $\mathbb{Z}_{2}^{\mathrm{het}}$-action \eqref{WeylGroup}
      \end{tabular}
    }
    }
  }{
    \mathbb{R}^{5,1}
    \!\!\times\!\!
    \{0\}
    \!\times\!\!
    \xymatrix{
      \mathbb{R}^{\mathbf{4}_{\mathbb{H}}}
    }
  }
  \!\!\sslash\!
  \mathbb{Z}_2^{\mathrm{refl}}
  \big)
  \nonumber
  \\
  &  \; \underset{\mathrm{homotopy}}{\simeq}\; &
  \underset{
    \simeq \;\ast
  }{
  \underbrace{
    \underset{
      \mathclap{
      \mbox{\bf
        \tiny
        \color{darkblue}
        \begin{tabular}{c}
          unit 0-sphere
          \\
          around $\mathrm{MO9}$
        \end{tabular}
      }
      \;\;
      }
    }{
      S(\mathbb{R}^{\mathbf{1}_{\mathrm{sgn}}})
    }
    /
    \underset{
      \mathclap{
      \;\;
      \mbox{\bf
        \tiny
        \color{darkblue}
        \begin{tabular}{c}
          HW-quotient
          \\
          \eqref{TheMO9}
        \end{tabular}
      }
      }
    }{
      \mathbb{Z}_2^{\mathrm{HW}}
    }
  }}
  \;
  \times
  \!\!\!\!
  \underset{
    \mbox{\bf
      \tiny
      \color{darkblue}
      \begin{tabular}{c}
        ADE-singularity
        \\
        (\hyperlink{Table5}{\it Table 5})
      \end{tabular}
    }
  }{
  \xymatrix{
    \mathbb{R}^{\mathbf{4}_{\mathbf{H}}}
  }
  \!\!\sslash\!
  \mathbb{Z}_2^{\mathrm{refl}}
  }.
\end{eqnarray}
As shown in the second line, this is homotopy-equivalent to a residual ADE-singularity (\hyperlink{Table5}{\it Table 5}).
Therefore, the discussion from \cref{LocalTadpoleCancellation} applies:

\begin{cor}[\bf Equivariant Cohomotopy implies local/twisted M5/MO5-anomaly cancellation]
\label{EquivariantCohomotopyOfSemiComplementSpacetime}
The super-differentiable \eqref{DifferentialEquivariantCohomotopyPullback}
equivariant Cohomotopy charge 
of the vicinity (Def. \ref{CohomotopyOfVicinityOfSingularity})
of the semi-complement spacetime of a single charged $\mathrm{MO5}$-singularity \eqref{SemiComplement}
$$
  \pi^{\mathbf{4}_{\mathbb{H}}}
  \Big(
    \Big(
      X^{11}_{{{}_{\frac{1}{2}\mathrm{MO5}}}}
    \Big)^{\mathrm{cpt}}
  \Big)_-
  \;=\;
  \Big\{
  \underset{
    \mathclap{
    \mbox{\bf
      \tiny
      \color{darkblue}
      \begin{tabular}{c}
      \\
        MO5-plane
        \\
        charge
      \end{tabular}
    }
    }
  }{
    1 \cdot \mathbf{1}_{\mathrm{triv}}
  }
  -
  \underset{
    \mbox{\bf
      \tiny
      \color{darkblue}
      \begin{tabular}{c}
      \\
        M5-brane
        \\
        charge
      \end{tabular}
    }
  }{
    N_{\mathrm{M5}} \cdot \mathbf{2}_{\mathrm{reg}}
  }
  \;\Big\vert\;
  N_{\mathrm{M5}} \in \mathbb{Z}
  \Big\}
$$
as in Folklore \ref{AnomalyCancellationOnMTheoreticOrientifolds},
\hyperlink{Table2}{\it Table 2},
regarding the local/twisted form of M5/MO5-anomaly cancellation.
\end{cor}
\begin{proof}
  By $G$-homotopy invariance of $G$-equivariant
  homotopy theory, this follows as the special case of
  Prop. \ref{TheoremLocalTadpoleCancellation}
  with
  \eqref{SuperDifferentiableLocalCohomotopyCharge} in Example \ref{SuperDifferentiableEquivariantCohomotopyOfADEOrbifolds},
  for $G = \mathbb{Z}_2$, hence with
  $k = \left\vert W_G(1)\right\vert = 2$.
\end{proof}

\begin{remark}[Super-exceptional geometry of $\mathrm{MO5}$ semi-complement]
  While here we consider only topological orientifold structure,
  the full super-exceptional geometry corresponding
  to \eqref{SemiComplement} is introduced in
  \cite[4]{FSS19d}; shown there to induce the
  M5-brane Lagrangian on any super-exceptional embedding
  of the $\tfrac{1}{2}\mathrm{M5}$-locus.
\end{remark}

\noindent {\bf Equivariant Cohomotopy implies local/untwisted M5/MO5-anomaly cancellation at $\tfrac{1}{2}\mathrm{M5}$-singularities.}
It is immediate to consider the globalization
of this situation to the semi-complement
around one $\mathrm{MO9}$ in heterotic M-theory
compactified on the toroidal $\mathbb{Z}_2^{\mathrm{refl}}$-orbifold
$\mathbb{T}^{\mathbf{5}_{\mathrm{sgn}}} \sslash \mathbb{Z}_2^{\mathrm{refl}+\mathrm{HW}} $ with MO5-singularities:

\vspace{-.2cm}

\begin{equation}
  \label{SemiComplementOfToroidalMNO5Compactification}
  \underset{
    \mbox{\bf
      \tiny
      \color{darkblue}
      \begin{tabular}{c}
        semi-complement spacetime
        \\
        around $\mathrm{MO5}s$ in $\mathrm{M}_{\mathrm{HET}}/\mathbb{Z}_2^{\mathrm{refl}}$
      \end{tabular}
    }
  }{
    X^{11}_{\mathrm{M}_{\mathrm{HET}}/\mathbb{Z}^{\mathrm{refl}}_2}
  }
  \;\coloneqq\;
  \mathbb{R}^{5,1}
  \,\times
  \,\,
  \overset{ \simeq \, \ast }{
  \overbrace{
    \underset{
      \mathclap{
      \mbox{\bf
        \tiny
        \color{darkblue}
        \begin{tabular}{c}
          unit 0-sphere
          \\
          around $\mathrm{MO9}$
        \end{tabular}
      }
      \;\;
      }
    }{
      S(\mathbb{R}^{\mathbf{1}_{\mathrm{sgn}}})
    }
    /
    \underset{
      \mathclap{
      \;\;
      \mbox{\bf
        \tiny
        \color{darkblue}
        \begin{tabular}{c}
          HW-quotient
          \\
          \eqref{TheMO9}
        \end{tabular}
      }
      }
    }{
      \mathbb{Z}_2^{\mathrm{HW}}
    }
  }}
  \;\;\;\;\;
  \times
  \!\!\!\!\!
  \underset{
    \mbox{\bf
      \tiny
      \color{darkblue}
      \begin{tabular}{c}
        toroidal reflection-orbifold \eqref{PointReflectionSubgroup}
        \\
        (\hyperlink{Table5}{\it Table 5})
      \end{tabular}
    }
  }{
  \xymatrix{
    \mathbb{T}^{\mathbf{4}_{\mathbf{H}}}
  }
  \!\!
  \!\sslash\! \mathbb{Z}_2^{\mathrm{refl}}\;\;.
  }
 \end{equation}
To this toroidal ADE-orbifold the discussion in
\cref{GlobalTadpoleCancellation} applies as follows.
\begin{cor}[\bf Equivariant Cohomotopy implies global/untwisted M5/MO5-anomaly cancellation]
\label{GlobalM5MO5CancellationImplied}
The super-differentiable
\eqref{DifferentialEquivariantCohomotopyPullback}
equivariant Cohomotopy charge \eqref{EquivariantCohomotopySet} of
the semi-complement spacetime \eqref{SemiComplementOfToroidalMNO5Compactification} of
heterotic M-theory on a toroidal MO5-orientifold (\cref{HeteroticMTheoryOnADEOrbifolds}) with charged MO5-planes in compatible RO-degree (Example \ref{ExamplesOfCompatibleRODegree})
and admitting equivariant super-differential refinement \eqref{KernelOfTheGlobalElmendorfStageProjection}
is
$$
  \pi^{\mathbf{4}_{\mathbb{H}}}
  \Big(
    \big(
      X^{11}_{\mathrm{M}_{\mathrm{HET}}/\mathbb{Z}_2^{\mathrm{refl}}}
    \big)_{+}
  \Big)_{{}_{\mathrm{Sdiffble}_-}}
  \;=\;
  \Big\{
  \underset{
    \mathclap{
    \mbox{\bf
      \tiny
      \color{darkblue}
      \begin{tabular}{c}
      \\
        MO5-plane
        \\
        charge
      \end{tabular}
    }
    }
  }{
    16 \cdot \mathbf{1}_{\mathrm{triv}}
  }
  -
  \underset{
    \mbox{\bf
      \tiny
      \color{darkblue}
      \begin{tabular}{c}
      \\
        M5-brane
        \\
        charge
      \end{tabular}
    }
  }{
    8 \cdot \mathbf{2}_{\mathrm{reg}}
  }
  \big\}
$$

\vspace{-2mm}
\noindent as expected from Folklore \ref{AnomalyCancellationOnMTheoreticOrientifolds},
\hyperlink{Table2}{\it Table 2},
regarding the global/untwisted form of M5/MO5-anomaly cancellation
(recalling that the semi-complement \eqref{SemiComplementOfToroidalMNO5Compactification}
is that around \emph{one} of the two MO9-planes).
\end{cor}
\begin{proof}
  By $G$-homotopy invariance of equivariant Cohomotopy,
  this follows from  the statement \eqref{SuperDifferentiableEquivariantCohomotopyOfKummerSurface}
  in Example \ref{SuperDifferentiableEquivariantCohomotopyOfADEOrbifolds}.
\end{proof}

\medskip
\noindent More generally we have the following:

\medskip
\noindent {\bf M5/MO5-anomaly cancellation in heterotic M-theory on general ADE-orbifolds.}
The statements and proofs of Corollary \ref{EquivariantCohomotopyOfSemiComplementSpacetime}
and Cor. \ref{GlobalM5MO5CancellationImplied}
directly generalize to heterotic M-theory on general
$G^{\mathrm{ADE}}$-singularities $\mathbb{R}^{\mathbf{4}_{\mathbb{H}}}$ \cref{HeteroticMTheoryOnADEOrbifolds}, because the underlying results in
\cref{EquivariantCohomotopyAndTadpoleCancellation} apply in this generality.
Hence \hyperlink{HypothesisH}{\it Hypothesis H} implies that
on the semi-complement spacetime of an $\mathrm{MO9}$ intersecting a
toroidal ADE-orbifold

\vspace{-.2cm}

\begin{equation}
  \label{SemiComplementOfToroidalM5ADECompactification}
  \underset{
    \mbox{\bf
      \tiny
      \color{darkblue}
      \begin{tabular}{c}
        semi-complement spacetime
        \\
        around $\tfrac{1}{2}\mathrm{M5}_{\mathrm{ADE}}$ in
        $\mathrm{M}_{\mathrm{HET}}/\mathrm{ADE}$
        \\
        (\hyperlink{FigureS}{\it Figure S})
      \end{tabular}
    }
  }{
    X^{11}_{{}_{
      \mathrm{M}_{\mathrm{HET}}
       /
       G^{\mathrm{ADE}}
    }}
  }
  \;\coloneqq\;
  \mathbb{R}^{5,1}
  \;\times
  \;\;
  \overset{ \simeq \, \ast }{
  \overbrace{
    \underset{
      \mathclap{
      \mbox{\bf
        \tiny
        \color{darkblue}
        \begin{tabular}{c}
          \\
          unit 0-sphere
          \\
          around $\mathrm{MO9}$
        \end{tabular}
      }
      \;\;
      }
    }{
      S(\mathbb{R}^{\mathbf{1}_{\mathrm{sgn}}})
    }
    /
    \underset{
      \mathclap{
      \;\;
      \mbox{\bf
        \tiny
        \color{darkblue}
        \begin{tabular}{c}
          \\
          HW-quotient
          \\
          \eqref{TheMO9}
        \end{tabular}
      }
      }
    }{
      \mathbb{Z}_2^{\mathrm{HW}}
    }
  }
  }
  \;\;\;\;\;\;
  \times
    \underset{
    \mbox{\bf
      \tiny
      \color{darkblue}
      \begin{tabular}{c}
      \\
        toroidal ADE-orbifold
        \\
        (\hyperlink{Table5}{\it Table 5})
      \end{tabular}
    }
  }{
  \xymatrix{
    \mathbb{T}^{\mathbf{4}_{\mathbf{H}}}
  }
  \!\!
  \!\sslash\! G^{\mathrm{ADE}}
  }
  \end{equation}
the M5/MO5 charge, measured in equivariant Cohomotopy, is
$$
  Q_{\mathrm{tot}}
  \;=\;
  16 \cdot \mathbf{1}_{\mathrm{triv}}
  -
  N_{\mathrm{M5}} \cdot \mathbf{k}_{\mathrm{reg}}
  \phantom{AAA}
  \left\vert Q_{\mathrm{tot}}\right\vert = 0
  \,,
$$
for $k = \left\vert G^{\mathrm{ADE}} \right\vert$ the order
of the global quotient group.
Under double dimensional reduction to type IIA string theory
according to \hyperlink{Table7}{\it Table 7},
this implies the tadpole cancellation conditions for
D4-branes in ADE-orientifolds, from \hyperlink{Table1}{\it Table 1}.

\newpage

\medskip

{\small
\hypertarget{Table8}{}
\hspace{-.9cm}
\begin{tabular}{|c|c|c|}
\hline
\multicolumn{3}{ |c| }{
  \begin{tabular}{c}
  \\
  \bf Spacetimes on which to measure flux sourced by M5/MO5-charge
  \\
  $\phantom{-}$
  \end{tabular}
}
\\
\hline
\multirow{2}{*}{
  {\bf Definition}
}
&
$
\mathclap{\phantom{ {\vert \atop \vert} \atop {\vert \atop \vert}}}
X_{\mathrm{MO5}}
\simeq_{{}_{\mathrm{htpy}}} S(\mathbb{R}^{\mathbf{1}_{\mathrm{sgn}} + \mathbf{4}_{\mathrm{sgn}}})/\mathbb{Z}_2^{\mathrm{het}+\mathrm{HW}} $
&
$
  X_{\sfrac{1}{2}\mathrm{M5}}
  \simeq_{{}_{\mathrm{htpy}}}
  S(\mathbb{R}^{\mathbf{1}_{\mathrm{sgn}}})/\mathbb{Z}_2^{\mathrm{HW}}
  \times
  \mathbb{T}^{\mathbf{4}_{\mathbb{H}}} \sslash \mathbb{Z}_2^{\mathrm{refl}}
$
\\
&
\eqref{M5ThreadedThroughRP4}
&
\eqref{SemiComplement}
\\
\hline
\hline
\raisebox{2.5cm}{
  {\bf Illustration}
}
&
\begin{tikzpicture}[scale=.8]

	\begin{scope}
	    \clip(-3,-3) -- (-3,0) -- (-2,0) arc (180:360:2cm and 0.6cm) -- (3,0) -- (3,-3) -- (-3,-3);
		\shade[ball color=lightblue!60!white, opacity=0.60] (0,0) circle (2cm);
	\end{scope}

    \shade[shading=radial, inner color=lightblue!20!white, outer color=lightblue!60!white, opacity=0.60] (2,0) arc (0:360:2cm and 0.6cm);

    \draw (2,0) arc (0:360:2cm and 0.6cm);

    \draw (-2,0) arc (180:360:2cm and 2cm);

   \draw[->, dashed] (0,0) to (30:1.06);
   \draw[->, dashed] (0,0) to (30+180:1.06);

   \draw[->, dashed] (0,0) to (180-30:1.06);
   \draw[->, dashed] (0,0) to (180-30+180:1.06);

\end{tikzpicture}

&
\begin{tikzpicture}[scale=0.8]

  \draw (0,3.7) node {\phantom{$---$}};

  \begin{scope}[shift={(1.5,1)}]

  \draw[dotted, thick] (0,0) circle (1.5);

  \begin{scope}
    \clip (1.5-.02,0-.1) rectangle (1.5+.1, 0+.1);
    \draw[fill=white] (1.5,0) circle (.1);
  \end{scope}

  \draw[draw=lightgray, fill=lightgray]
    (1.5+0.1,0-.05)
    rectangle
    (4,0+.05);
  \draw[]
    (1.5+0.1,0-.05)
    to
    (4,0-.05);
  \draw[]
    (1.5+0.1,0+.05)
    to
    (4,0+.05);

  \draw[draw=lightgray, fill=lightgray]
    (-1.5-0.1,0-.05)
    rectangle
    (-4,0+.05);
  \draw[]
    (-1.5-0.1,0-.05)
    to
    (-4,0-.05);
  \draw[]
    (-1.5-0.1,0+.05)
    to
    (-4,0+.05);

  \draw
     (4.6,0)
     node
     {\tiny $x_2 = \tfrac{1}{2}$};

  \draw
     (-4.6,0)
     node
     {\tiny $x_2 = 0$};

  \draw
     (0,-1.5)
     node
     {\colorbox{white}{\tiny $x_1 = 0$}};

  \draw
     (-70:2.6)
     node
     {\colorbox{white}{\tiny $x_1 > 0$}};

  \draw[dashed]
    (40:1.5)
    to
    (40:4);
  \draw[dashed]
    (-40:1.5)
    to
    (-40:4);

  \draw[very thick, darkblue]
    (18-1:1.5)
    to
    (18-.4:4);
  \draw[very thick, darkblue]
    (18-.5:1.5)
    to
    (18-.2:4);
  \draw[very thick, darkblue]
    (18-0:1.5)
    to
    (18-0:4);
  \draw[very thick, darkblue]
    (18+.5:1.5)
    to
    (18+.2:4);
  \draw[very thick, darkblue]
    (18+1:1.5)
    to
    (18+.4:4);

  \draw[very thick, darkblue]
    (-18-1:1.5)
    to
    (-18-.4:4);
  \draw[very thick, darkblue]
    (-18-.5:1.5)
    to
    (-18-.2:4);
  \draw[very thick, darkblue]
    (-18-0:1.5)
    to
    (-18-0:4);
  \draw[very thick, darkblue]
    (-18+.5:1.5)
    to
    (-18+.2:4);
  \draw[very thick, darkblue]
    (-18+1:1.5)
    to
    (-18+.4:4);

  \begin{scope}
    \clip (-1.5-.1,0-.1) rectangle (-1.5+0.02, 0+.1);
    \draw[fill=white] (-1.5,0) circle (.1);
  \end{scope}

  \draw
    (-1.5+.7,0)
    node
    {\tiny \color{darkblue} \bf MO5};

  \draw[->, gray]
    (-1.5+.4,0)
    to
    ++(-.3,0);

  \begin{scope}[shift={(5.1,1.1)}]

  \draw
    (-2.1,.5)
    node
    {\tiny \color{darkblue} \bf MK6};

  \draw[->, gray]
    (-2.1,.3)
    to
    ++(0,-.3);

  \end{scope}

  \draw
    (180-60:.9)
    node
    {\tiny \color{darkblue} \bf MO9 };

  \draw[->, gray]
    (180-60:1.1)
    to
    (180-60:1.4);

  \draw
    node
    (halfM5)
    at
    (0:.75)
    {
      \tiny
      \color{darkblue} \bf
      $
        \mathrm{MO5}
      $
    };

  \draw[->, gray]
    (halfM5)
    to
    ++(0:.7);

  \draw
    node
    (M5)
    at
    (36:.8)
    {
      \tiny
      \color{darkblue} \bf
      $
        \mathllap{\tfrac{1}{2}}
        \mathrm{M5}
      $
    };

  \draw[->, gray]
    (M5)
    to
    ++(.65,0);

  \draw
    node
    (mirrorM5)
    at
    (-36:.8)
    {
      \tiny
      \color{darkblue} \bf
      $
        \mathllap{\mbox{mirror}\;\tfrac{1}{2}}
        \mathrm{M5}
      $
    };

  \draw[->, gray]
    (mirrorM5)
    to
    ++(.65,0);

  \draw[dashed]
    (180-40:1.5)
    to
    (180-40:4);
  \draw[dashed]
    (180+40:1.5)
    to
    (180+40:4);

  \draw[<->, dashed, darkblue]
    (180-29:3)
    to[bend right=24]
    node[very near end]
      {
        $
        \mathllap{
        \mbox{\bf
        \tiny
        \color{darkblue}
        \begin{tabular}{c}
          residual
          \\
          $\mathbb{Z}_2^{\mathrm{refl}}$-action
        \end{tabular}
        }
        }
        $
      }
    (180+29:3);

  \end{scope}

  \begin{scope}[shift={(1.5,1)}]

  \begin{scope}[rotate=(+18)]
    \clip (1.5-.02,.3) rectangle (1.5+.3,-.3);
    \draw[draw=green, fill=green] (1.5,0) circle (.1);
  \end{scope}

  \begin{scope}[rotate=(-18)]
    \clip (1.5-.02,.3) rectangle (1.5+.3,-.3);
    \draw[draw=green, fill=green] (1.5,0) circle (.1);
  \end{scope}

  \end{scope}

\end{tikzpicture}
\\
\hline
$\phantom{ {\vert \atop \vert} \atop {\vert \atop \vert} }$
{\bf Geometry}
&
smooth but curved
&
singular but flat
\\
\hline
\hline
\multicolumn{3}{|c|}{
  \begin{tabular}{c}
    \\
    \bf Cohomological charge quantization
    \\
    by \hyperlink{HypothesisH}{\it Hypothesis H}
    \\
    $\phantom{-}$
  \end{tabular}
}
\\
\hline
\begin{tabular}{c}
 \\
 {\bf Cohomology theory}
 \\
 (by \hyperlink{Table4}{\it Table 4})
 \\
 $\phantom{-}$
\end{tabular}
&
\begin{tabular}{c}
  $J$-twisted Cohomotopy $\pi^{T X}\big(X\big)$
  \\
  \cite{FSS19b}\cite{FSS19c}
\end{tabular}
&
\begin{tabular}{c}
  equivariant Cohomotopy $\pi_{\mathbb{Z}_2}^{V}\big( \mathbb{T}^V \big)$
  \\
  \cref{EquivariantCohomotopyAndTadpoleCancellation}
\end{tabular}
\\
\hline
\raisebox{53pt}{
\begin{tabular}{c}
  {\bf Illustration}
  \\
  (Remark \ref{TheRoleOfMK6EndingOnM5})
\end{tabular}
}
&
\raisebox{19pt}{
\begin{tikzpicture}[scale=.9]

  \draw[dashed] (0,0) circle (1.4);

  \draw[draw=green, fill=green] (0,0) circle (.15);

  \draw (0.45,0)
    node
    {
      \tiny
      \color{darkblue} \bf
      $\mathrm{M5}$
    };

  \draw (125:1.6)
    node
    {
      \small
      \color{darkblue} \bf
      $S^4$
    };

\end{tikzpicture}
}
&
\hspace{-1.7cm}
\begin{tikzpicture}[scale=.9]

  \begin{scope}
\
  \clip (-3,2) rectangle (0,-2);

  \draw[dashed] (0,0) circle (1.4);

  \draw[draw=darkblue, fill=darkblue]
    (0,-.05) rectangle (-2.5,+.05);
  \draw[draw=darkblue, fill=darkblue]
    (-2.5-.05,-.05) rectangle (-2.5-.10,+.05);
  \draw[draw=darkblue, fill=darkblue]
    (-2.5-.15,-.05) rectangle (-2.5-.20,+.05);
  \draw[draw=darkblue, fill=darkblue]
    (-2.5-.25,-.05) rectangle (-2.5-.30,+.05);

  \draw[draw=green, fill=green] (0,0) circle (.15);

  \draw (-2.15,-.25)
    node
    {
      \tiny
      \color{darkblue} \bf
      $\mathrm{MK6}$
    };

  \draw (125:1.6)
    node
    {
      \small
      \color{darkblue} \bf
      $S^4$
    };

  \end{scope}

  \draw (0.45,0)
    node
    {
      \tiny
      \color{darkblue} \bf
      $\tfrac{1}{2}\mathrm{M5}$
    };

 \draw (0,-1.8) to (0,1.8);

 \draw (0,-2)
   node
   {
     \tiny
     \color{darkblue} \bf
     $\mathrm{MO9}$
   };

\end{tikzpicture}
\\
\hline
$\mathclap{\phantom{ {\vert \atop \vert} \atop {\vert \atop \vert} }}$
\begin{tabular}{c}
  \bf Charge classification
\end{tabular}
&
\begin{tabular}{c}
  $c_{\mathrm{tot}} = 1 - N \cdot 2$
  \\
  \eqref{TheOddChargeAroundAPureMO5}
\end{tabular}
&
\begin{tabular}{c|c}
  $
    c_{\mathrm{tot}}
    =
    N_{\mathrm{MO5}}
     \cdot
    \mathbf{1}_{\mathrm{triv}}
    -
    N_{\mathrm{M5}} \cdot \mathbf{2}_{\mathrm{reg}}
  $
  &
  $
    \begin{array}{rcl}
      & \left\vert Q_{\mathrm{tot}}\right\vert & \!\!\!= 0
      \\
      \Leftrightarrow &  N_{\mathrm{M5}} & \!\!\!= 8
    \end{array}
  $
  \\
  (Cor. \ref{EquivariantCohomotopyOfSemiComplementSpacetime})
  &
  (Cor. \ref{GlobalM5MO5CancellationImplied})
\end{tabular}
\\
\hline
\end{tabular}

\vspace{.2cm}

\noindent {\bf \footnotesize Table 8 -- Two ways of measuring M5/MO5-charge.} {\footnotesize On the left is the traditional approach
not resolving the singularities.
On the right (which shows the same situation as in \hyperlink{FigureV}{\it Figure V} but with the periodic identification indicated more explicitly) the fine-grained
microscopic picture seen by C-field charge quantization in
equivariant Cohomotopy.}
}

\medskip
\medskip

With these result in hand, we highlight that not only did
equivariant Cohomotopy inform us about M-theory, but
M-theory also shed light on a subtle point regarding the
interpretation of equivariant Cohomotopy:

\begin{remark}[\bf Equivariant Cohomotopy and MK6 ending on M5]
  \label{TheRoleOfMK6EndingOnM5}
 {\bf (i)}  The heuristic way to see that ordinary Cohomotopy $\pi^4$ from
  \eqref{PlainCohomotopySet} canonically measures charges of
  5-branes inside 11-dimensional spacetime is that the `\emph{classifying space}'
  $S^4$ of $\pi^4$ gets essentially identified with the (any) \emph{spacetime}
  4-sphere \emph{around} a 5-brane in an 11-dimensional ambient space
  (see \cite[(6)]{ADE} for the heuristic picture, and
  \cite[4.5]{FSS19b} for the full mathematical detail).

\vspace{-1mm}
\item {\bf (ii)}  But as we pass from plain to equivariant Cohomotopy,
  this picture

  \vspace{-.5cm}

  $$
    \mbox{
      \it
      Brane charge sourced in the center of $S^4$
      \hspace{2.2cm}
\raisebox{-32pt}{
\begin{tikzpicture}[scale=.9]

  \draw[dashed] (0,0) circle (1.4);

  \draw[draw=green, fill=green] (0,0) circle (.15);

  \draw (0.45,0)
    node
    {
      \tiny
      \color{darkblue} \bf
      $\mathrm{M5}$
    };

  \draw (125:1.6)
    node
    {
      \small
      \color{darkblue} \bf
      $S^4$
    };

\end{tikzpicture}
}
    }
  $$
  may superficially appear to be in tension with
  the picture provided by the Pontrjagin-Thom theorem
  as in \hyperlink{FigureD}{\it Figure D}
  and \hyperlink{FigureL}{\it Figure L}, where
  instead

  \vspace{-.5cm}

  $$
    \mbox{
      \it
      Brane charge is sourced at the 0-pole of $S^{\mathbf{4}}$
      \hspace{.7cm}
\raisebox{-32pt}{
\begin{tikzpicture}[scale=.9]

  \begin{scope}

  \clip (-2.8,-1) rectangle (0,1);

  \draw[draw=darkblue, fill=darkblue] (0,-.05) rectangle (-2.5,+.05);
  \draw[draw=darkblue, fill=darkblue] (-2.5-.05,-.05) rectangle (-2.5-.10,+.05);
  \draw[draw=darkblue, fill=darkblue] (-2.5-.15,-.05) rectangle (-2.5-.20,+.05);
  \draw[draw=darkblue, fill=darkblue] (-2.5-.25,-.05) rectangle (-2.5-.30,+.05);

  \draw[draw=green, fill=green] (0,0) circle (.15);

  \draw (-2.15,-.25)
    node
    {
      \tiny
      \color{darkblue} \bf
      $\mathrm{MK6}$
    };

  \end{scope}

  \draw (125:1.6)
    node
    {
      \small
      \color{darkblue} \bf
      $S^4$
    };

  \draw (0.45,0)
    node
    {
      \tiny
      \color{darkblue} \bf
      $\tfrac{1}{2}\mathrm{M5}$
    };

  \draw[dashed] (0,0) circle (1.4);

\end{tikzpicture}
}
    }
  $$
 However, in the orbi-geometry
  of heterotic M-theory on ADE-singularities \cref{HeteroticMTheoryOnADEOrbifolds}
  indeed \emph{both pictures apply simultaneously},
  witnessing different but closely related brane species
  (see \hyperlink{Table8}{\it Table 8}):

\vspace{-1mm}
  \item {\bf (iii)} The black $\tfrac{1}{2}\mathrm{M5}$-brane locus
  (\hyperlink{FigureS}{\it Figure S})  is the terminal point of an MK6-singularity
  which extends radially away from the M5. Hence, given any radial 4-sphere with
  the $\tfrac{1}{2}\mathrm{M5}$ at its center, the MK6 will pierce this 4-sphere at one point.
  Since the $\tfrac{1}{2}\mathrm{M5}$ and the MK6 are \emph{necessarily} related this way,
  the 5-brane charge inside the $S^4$ may equivalently be measured by 6-brane charge
  piercing through $S^4$. This is exactly what the Pontrjagin-Thom theorem
  says happens in Corollary \ref{EquivariantCohomotopyOfSemiComplementSpacetime},
  as
  shown in \hyperlink{FigureV}{\it Figure V} and on the
  right of \hyperlink{Table8}{\it Table 8}.
\end{remark}

\vspace{.5cm}
\noindent {\large \bf Acknowledgements.}

H. S.  acknowledges that this work was  performed in part at Aspen Center for Physics,
which is supported by National Science  Foundation grant PHY-1607611.
This work was partially supported by a grant from the Simons Foundation.
We thank Matt Kukla for a hint on TikZ typesetting.

\vspace{1cm}
\noindent Hisham Sati, {\it Mathematics, Division of Science, New York University Abu Dhabi, UAE.}

 \medskip
\noindent Urs Schreiber,  {\it Mathematics, Division of Science, New York University Abu Dhabi, UAE;
 on leave from Czech Academy of Science, Prague.}

\end{document}